\documentclass[11pt]{article}

\usepackage{titlesec}

\titlespacing\section{4pt}{8pt}{4pt}
\titlespacing\subsection{4pt}{6pt}{4pt}
\titlespacing\subsubsection{4pt}{6pt}{4pt}

\newif\ifarxiv   % Toggle appendix
\arxivfalse %set true for arxiv version

\ifarxiv
\else
\fi

\usepackage{graphicx}
\usepackage{csvsimple}
\usepackage{rotating}
\usepackage{hyperref}
\hypersetup{
    colorlinks=true,
    linkcolor=blue,
    citecolor=red,
    filecolor=magenta,      
    urlcolor=cyan
    }

\usepackage[nocomma]{optidef}
\usepackage{listings}
\usepackage{comment}
\usepackage{appendix}
\usepackage[ruled,vlined]{algorithm2e}
\usepackage{amsfonts} 
\usepackage{amssymb}
\usepackage{bm}
\usepackage[square,comma,numbers]{natbib}

\usepackage{booktabs}

\usepackage{color} %red, green, blue, yellow, cyan, magenta, black, white
\definecolor{mygreen}{RGB}{28,172,0} % color values Red, Green, Blue
\definecolor{mylilas}{RGB}{170,55,241}
\DeclareFixedFont{\ttb}{T1}{txtt}{bx}{n}{12} % for bold
\DeclareFixedFont{\ttm}{T1}{txtt}{m}{n}{12}  % for normal
\usepackage{bbm}
\usepackage{amsmath,amsthm}

\newtheorem{theorem}{Theorem}
\newtheorem{corollary}{Corollary}
\newtheorem{proposition}{Proposition}

\newtheorem{lemma}{Lemma}
\newtheorem{observation}{Observation}

\newtheorem{assumption}{Assumption}

\usepackage{tikz,pgfplots}
\pgfplotsset{compat=1.14}
\pgfplotsset{scaled y ticks=false}

\usepackage{subcaption}

\pgfplotsset{scaled x ticks=false}

\theoremstyle{definition}
\newtheorem{definition}{Definition}

\newtheorem{example}{Example}

\newenvironment{hproof}{%
  \proof}{\endproof}

\usepackage{color}
\definecolor{deepblue}{rgb}{0,0,0.5}
\definecolor{deepred}{rgb}{0.6,0,0}
\definecolor{deepgreen}{rgb}{0,0.5,0}

\DeclareMathOperator*{\argmax}{arg\,max}

\usepackage{accents}

\usepackage{mathtools}

\usepackage{listings}

\renewcommand\footnotemark{}
\def\sigmaa{\boldsymbol{\sigma}}
\def\y{\mathbf{y}}
\def\x{\mathbf{x}}
\def\I{\mathcal{I}}
\def\A{\mathcal{A}}
\def\S{\mathcal{S}}
\def\H{\mathcal{H}}
\def\L{\mathcal{L}}
\def\J{\mathcal{J}}

\def\A{\mathcal{A}}
\def\llambda{\boldsymbol{\lambda}}

\usepackage{array}

\usepackage{amsmath,inputenc}

\usepackage[letterpaper, portrait, margin=1in]{geometry}

\graphicspath{ {fig/} }

\title{When Simple is Near Optimal in Security Games}
\author{Devansh Jalota\\ Stanford\\{\tt djalota@stanford.edu}
        \and 
        Michael Ostrovsky\\ Stanford\\ {\tt ostrovsky@stanford.edu}
        \and 
        Marco Pavone\\ Stanford\\ {\tt pavone@stanford.edu}
}

\begin{document}
\maketitle

\begin{abstract} 
Fraud is ubiquitous across many applications and involves users bypassing the rule of law, often with the strategic aim of obtaining some benefit that would otherwise be unattainable within the bounds of lawful conduct. However, user fraud is detrimental, as it may compromise safety or impose disproportionate negative externalities on particular population groups.

To mitigate the potential harms of user fraud, we study the problem of policing such fraud as a security game between an administrator and users. In this game, an administrator deploys $R$ security resources (e.g., police officers) across $L$ locations and levies fines against users engaging in fraud at those locations. For this security game, we study both payoff and revenue maximization administrator objectives. In both settings, we show that the problem of computing the optimal administrator strategy is NP-hard and develop natural greedy algorithm variants for the respective settings that achieve at least half the payoff or revenue as the payoff-maximizing or revenue-maximizing solutions, respectively. We also establish a resource augmentation guarantee that our proposed greedy algorithms with just one additional resource, i.e., $R+1$ resources, achieve at least the same payoff (revenue) as the payoff-maximizing (revenue-maximizing) outcome with $R$ resources that is NP-hard to compute. Moreover, in the setting when the user types are homogeneous, i.e., at each location, all users have the same type, we develop a near-linear time algorithm to solve the administrator's revenue maximization problem and a polynomial time approximation scheme for the administrator's payoff maximization problem.

Next, we present numerical experiments based on a case study of parking enforcement at a university campus, highlighting the efficacy of our algorithms in increasing earnings from parking permit purchases by \$300,000 annually relative to the status-quo enforcement policy. Finally, we study several model extensions, including incorporating contracts into our framework to bridge the gap between the payoff and revenue-maximizing outcomes and generalizing our model to incorporate additional constraints beyond a resource budget constraint.
\end{abstract}

\section{Introduction}

Fraudulent activities involving users bypassing the rule of law are ubiquitous and arise as users seek to strategically obtain some benefit that would otherwise not be attainable within the bounds of lawful conduct. For instance, in transportation networks, users often drive above the speed limit to reduce travel times, and, in school choice contexts, parents often misreport their home addresses to admit their children to better schools~\cite{bjerre2023playing}. Similar issues of users engaging in fraud arise in other domains, including labor markets~\cite{RePEc:hal:psewpa:halshs-03828729}, strategic classification~\cite{estornell2021incentivizing,estornell2023incentivizing}, non-market allocation mechanisms~\cite{perez2023fraud}, and resource allocation~\cite{lundy-taylor-etal-2019}.

However, fraud can be detrimental, as it can compromise safety, result in disproportionate negative externalities to particular groups of the population, or hamper efficiency. In transportation networks, users driving above the speed limit can compromise road safety, and, in school choice settings, sophisticated gaming by some parents, who typically belong to higher-income strata, often adversely affects users not engaging in such practices~\cite{bjerre2023playing}. Moreover, in healthcare, the manipulation of patient's priority by transplant providers in organ transplant waiting lists often results in reductions in organ donations~\cite{bolton2018,mcmichael2022stealing}. We present further examples of users engaging in fraud and the associated harms in Section~\ref{sec:examples-pertinent}. 

The prevalence of fraud and its associated harms across applications poses critical security challenges and requires policing to deter users from engaging in such activities. Central to the problem of policing fraud is a resource allocation task, which involves allocating a limited set of security resources (e.g., police officers) to mitigate fraud. However, the challenge in policing fraud is that only limited security resources are available, i.e., providing complete coverage to prevent fraud is not possible. This resource limitation raises a fundamental question of how to best allocate the available resources to protect against fraud. % at various susceptible nodes or locations in a system.

To answer this question, we study the problem of policing fraud as a security game between an administrator and fraudulent users at susceptible locations in a system. In our security game, the administrator can allocate a budget of security resources across various locations and levy fines against users engaging in fraud. For our studied security game, we show that \emph{simple} algorithms with a low computational overhead can achieve a good performance relative to the administrator's \emph{optimal} strategies, which are NP-hard to compute. To analyze the associated equilibria and the corresponding optimal administrator strategies, we leverage and combine techniques from optimization theory, linear programming, and approximation algorithms. Moreover, we apply our developed algorithms to a real-world data set on enforcing parking regulations at a university campus.

Our modeling assumptions and, in particular, considering fines to deter fraud are motivated by several real-world applications, e.g., users typically pay fines for road traffic or speeding violations. Moreover, in line with prior work on security games~\cite{tambe2011security}, we adopt a game-theoretic framework to study the problem of deploying security resources to mitigate fraud. We do so as a game-theoretic framework enables us to incorporate the preferences of both the administrator and users and predict how a fraudulent user will respond to a resource allocation policy and the corresponding fines the administrator sets. While we study a security game, a widely studied topic in the literature, our work contributes to the security games literature and differs from prior work on security games in several significant ways, which we summarize in Section~\ref{sec:related-literature}.

%studies a security game with multiple security resources and multiple adversary types

\subsection{Examples} \label{sec:examples-pertinent}

This section presents detailed examples of real-world settings where users engage in fraud. Our examples help further contextualize our studied problem setting and provide a grounding for the model in Section~\ref{sec:model}. While we present three examples, our developed framework more generally applies to a broad range of settings where users engage in fraud. 

\emph{Environmental Non-Compliance:} In environmental applications, public agencies institute regulations designed to meet particular health, safety, and environmental standards that the institutions these regulations seek to govern must comply with~\cite{heyes2000implementing}. For instance, the United States Environmental Protection Agency (USEPA) imposes water quality standards that all water treatment facilities must adhere to. However, institutions often have an incentive to, for instance, save on the costs of upgrading their facilities and thereby not comply with the regulations, which can lead to potentially detrimental environmental or health outcomes for the residents these facilities serve.

\emph{Queue skipping in Intermediate Public Transport (IPT):} IPT comprises informal modes of transport pivotal in serving the last-mile connectivity needs in many developing nation cities, particularly when the formal public transportation system is inadequate in serving users' travel demands~\cite{cervero2000informal,CERVERO2007445}. IPT services typically entail mini-buses~\cite{southAfrica-minibus}, share-taxis~\cite{SALOMON1985259,zwick2017analysis}, or auto-rickshaws~\cite{ipt-india}.

While IPT services offer affordable on-demand last-mile connectivity, commuters often face long wait times during rush hours~\cite{Subramanya_2012}, which frequently results in disorderly queues and, in particular, illicit behavior by some commuters who jump the queue to reduce their wait times~\cite{queue-jumping-website,cervero2000informal}, a problem that occurs in a broad range of queuing applications~\cite{skipping-the-queue-2012}. For instance, some male passengers who have recently arrived to avail of an IPT service, e.g., at a mini-bus stand, might bypass the queue and board a moving vehicle as it arrives to pick up passengers. Such alighting of moving vehicles is more challenging for women or less-abled travelers, who typically bear the burden of disproportionately high wait times when using IPT systems.

\emph{Parking Violations:} In another transportation context, users often engage in illegal parking practices, which typically involve bypassing parking regulations by, for instance, not purchasing parking permits or paying the appropriate parking fees when availing a parking spot. The violation of parking regulations is so commonplace that in 2022 alone, over 15 million parking tickets worth over \$1 billion were issued in New York City~\cite{traffic-stats-nyc}. Illegal parking by users denies other commuters from availing parking spots, can be a source of traffic jams~\cite{illegal-parking-traffic}, and also results in lower permit and parking revenues for owners of parking lots (see Section~\ref{sec:numerical-experiments-parking-enforcement} for more details).

%In another transportation context, users often bypass parking regulations by engaging in illegal parking activities, e.g., not purchasing parking permits or paying the appropriate parking fees when availing a parking spot. Such behavior by users not only denies other users who may have paid for the parking permit from using the parking spot but also negatively impacts the revenue collection of the entity running the parking lot. For instance, on university campuses, any parking permits that are not purchased represent lost revenue as any parking fines issued to non-permit holders are typically accrued by the police department rather than the university.

%, e.g., the United States Environmental Protection Agency (USEPA), impose 

%These perverse financial incentives for water treatment facilities can lead to potentially detrimental health outcomes for the residents that these facilities are serving.

%In several environmental applications, non-compliance with existing regulations is common as typically meeting the regulations is costly. For instance, the United States Environmental Protection Agency (USEPA) imposes water quality standards that all water treatment facilities are required to adhere to. However, water treatment facilities often have an incentive to save on costs to upgrade their facilities to meet the required water quality standards. Such perverse financial incentives for water treatment facilities can lead to potentially detrimental health outcomes for the residents that these facilities are serving.

Each of these examples highlights the need for policing to deter defaulting users from engaging in such fraudulent activities at susceptible locations in a system, e.g., water treatment facilities in an environmental regulation context, parking lots in a parking enforcement context, and mini-bus stands where users queue to avail a mini-bus in an IPT context.

\subsection{Our Contributions}

Motivated by the prevalence and detrimental effects of fraud across applications, we study a security game between an administrator and users engaging in fraud at susceptible locations in a system. In this security game, the administrator has $R$ security resources to allocate across these locations and levies fines against defaulting users to deter fraud. For this security game, we study the resulting equilibria and the optimal administrator strategies under both the payoff and revenue maximization objectives, which we introduce in Section~\ref{sec:model}. We model users as belonging to a discrete set of \emph{types}, where a type is specified by the number of users belonging to that type, the benefit that users in that type derive when engaging in fraud, and the payoff corresponding to preventing users of that type from engaging in fraud, and study the optimal administrator strategies in the setting with (i) \emph{homogeneous} user types, i.e., each location has a single user type while user types can differ across locations, and (ii) \emph{heterogeneous} user types.

%, under two settings: (i) \emph{homogeneous} user types, where a type is specified by the number of users belonging to that type, the benefit that users in that type derive when engaging in fraud, and the payoff of preventing users from that type from engaging in fraud

%where the administrator knows each location's \emph{type}, i.e., the administrator has complete knowledge of the total number of fraudulent users, the benefit that fraudulent users derive when engaging in fraud, and the value of preventing users from engaging in fraud at each location, and (ii) \emph{probabilistic}, where the administrator only has distributional information on each location's type.

We first study the setting with homogeneous user types in Section~\ref{sec:deterministic-setting}, where we develop a natural \emph{greedy} procedure to compute the administrator's revenue-maximizing strategy and show that it is NP-hard to compute the administrator's payoff-maximizing strategy. To that end, in the payoff-maximization setting, we develop a variant of a greedy algorithm, different from that in the revenue-maximization setting, that achieves at least half the optimal payoff. Moreover, we establish a resource augmentation guarantee that our proposed greedy algorithm with just one additional resource, i.e., $R+1$ resources, achieves at least the optimal payoff given $R$ resources. Finally, we also develop a polynomial-time approximation scheme (PTAS), which, when augmented with $\delta$ extra resources for any $\delta>0$, achieves a $1-\epsilon$ approximation to the optimal payoff for every fixed $\epsilon>0$. To establish the hardness, approximation, and resource augmentation guarantees in the payoff maximization setting, we develop and leverage properties of the payoff-maximizing solution and geometric insights based on the structure of the administrator's payoff function.

%We first study the setting with homogeneous user types in Section~\ref{sec:deterministic-setting}, where we develop a natural \emph{greedy} procedure to compute the administrator's revenue-maximizing strategy and show that the problem of computing the administrator's payoff-maximizing strategy is NP-hard. To that end, in the payoff-maximization setting, we develop a variant of a greedy algorithm, different from that in the revenue-maximization setting, that achieves at least half the payoff corresponding to the payoff-maximizing solution. Moreover, we establish a resource augmentation guarantee that our proposed greedy algorithm with just one additional resource, i.e., $R+1$ resources, achieves at least the same payoff as the payoff-maximizing solution with $R$ resources that is NP-hard to compute. To establish the hardness, approximation, and resource augmentation guarantees, we develop and leverage properties of the optimal solution to the administrator's payoff-maximization problem and geometric insights based on the structure of the administrator's payoff function. Finally, in the payoff-maximization setting, we also develop a polynomial-time approximation scheme, which, when augmented with $\delta$ additional resources for any constant $\delta>0$, achieves a $1-\epsilon$ approximation to the administrator's optimal payoff for every fixed parameter $\epsilon>0$.

%Moreover, we show that our proposed greedy algorithm solves a linear program whose objective upper bounds the administrator's welfare for any chosen resource allocation strategy. 

Next, in Section~\ref{sec:probabilistic-setting}, we study the setting with heterogeneous user types, where we show that, unlike in the setting with homogeneous user types, computing the administrator's revenue-maximizing strategy is NP-hard. Yet, we develop variants of greedy algorithms that achieve similar half-approximation and resource augmentation guarantees to those in the homogeneous user type setting for the revenue and payoff maximization administrator objectives with heterogeneous user types. The crux of establishing our algorithmic guarantees involves constructing a \emph{monotone concave upper approximation} (MCUA) of the revenue and payoff functions of the administrator (see Section~\ref{sec:near-opt-greedy-rev-max} for more details on the MCUA) and showing that the corresponding MCUAs can be maximized via a greedy process. Our results shed light on the benefits of using simple algorithms, e.g., variants of the greedy algorithm, and highlight the value of recruiting one additional security resource rather than expending computational effort in solving for the optimal administrator strategies that are NP-hard to compute.

%Our results establish that policy-makers can achieve the same (or better) outcome in terms of expected revenues or welfare through simple algorithmic approaches, i.e., variants of the greedy algorithm, with one additional resource rather than developing approaches requiring a large amount of computational power to solve for the optimal optimal administrator strategies that are NP-hard to compute.

Then, in Section~\ref{sec:numerical-experiments-parking-enforcement}, we present numerical experiments based on a case study of parking enforcement at Stanford University's campus. Our results demonstrate that our algorithms outperform the status-quo enforcement policy by increasing parking permit earnings by over $\$300,000$ (a $2\%$ increase) annually. Moreover, our results demonstrate the increasing power of our security game framework and algorithms in the regime when the proportion of users strategically deciding whether to engage in fraud increases.

%Then, in Section~\ref{sec:numerical-experiments-parking-enforcement}, we evaluate the efficacy of our developed greedy algorithms through numerical experiments based on a real-world case study of parking enforcement at a university campus. Our results demonstrate that our proposed algorithms outperform the status-quo parking enforcement policy by increasing revenues from parking permit purchases by over $\$300,000$ (a $2\%$ increase) each year relative to the status-quo policy. Moreover, our sensitivity analysis demonstrates the increasing power of our security game framework and proposed algorithms in the regime when the proportion of users strategically deciding whether to engage in fraud increases.

Finally, we study several model extensions, highlighting the generalizability of our proposed framework and algorithms. In Section~\ref{sec:optimal-contracts}, we extend our security game to incorporate \emph{contracts}, wherein a revenue-maximizing administrator is compensated through contracts for the payoff it contributes to the system. Beyond extending our theoretical results in the payoff and revenue maximization settings to studying equilibrium strategies in the contract game, we also present numerical experiments \ifarxiv based on a case study of queue jumping in IPT services, \fi highlighting the effectiveness of using contracts in bridging the gap between the revenue and payoff-maximizing administrator outcomes. Next, in Section~\ref{sec:additional-constraints}, we generalize our model to incorporate additional constraints on allocating security resources beyond a resource budget constraint. While a natural generalization of our proposed greedy algorithms that respect the additional constraints (termed as \emph{Constrained-Greedy}) may achieve arbitrarily bad approximation ratios for general constraints, we identify a practically relevant class of constraints, namely that of a \emph{hierarchy}, under which Constrained-Greedy achieves constant factor approximation ratios when augmented with some additional resources.

In the appendix, we provide proofs and extensions of theoretical results omitted from the main text, present %numerical implementation details along with 
additional numerical results, and highlight extensions of the model presented in this work, which opens directions for future work.

\section{Related Literature} \label{sec:related-literature}

Game theory has served as a foundational paradigm in studying multi-agent systems wherein agents pursue their selfish interests~\cite{fudenberg1991game}. Among the many successful applications of game theory, it has, in particular, gained traction in security applications, where the problem of allocating security resources is formulated as a Stackelberg game~\cite{tambe2011security,sinha2018stackelberg}. In a Stackelberg security game (SSG), the objective of the security agency is to compute a strategy to deploy its resources to prevent security breaches, given that the adversary will optimize its utility on observing this strategy. SSGs have found many applications, including protecting security checkpoints at airports~\cite{pita2008deployed}, protecting shipping ports~\cite{shieh2012protect}, homeland security and the defense of critical infrastructure~\cite{brown2006defending}, fare inspections in transit systems~\cite{yin2012trusts}, and more recently, in green security contexts~\cite{Xu_Bondi_Fang_Perrault_Wang_Tambe_2021,yang2014adaptive,fang2015security,fang2016deploying}. For a more detailed review of the state-of-the-art on security games, see~\cite{tambe2011security,sinha2018stackelberg}.

The work on SSGs has focused on equilibrium computation in two classes of problems: (i) single-resource Bayesian games, where the adversary's type is drawn from a distribution while the security agency has one resource to prevent a security breach~\cite{paruchuri2008playing} and (ii) multi-resource single-type games, where adversaries have a single type but the security agency has multiple resources~\cite{korzhyk2010complexity,kjtjft-2009}. While designing algorithms for multi-resource Bayesian SSGs has been challenging, in a recent work, Li et al.~\cite{li2016catcher} developed algorithms to compute equilibria in such games. Akin to Li et al.~\cite{li2016catcher}, we also study a security game with multiple resources and multiple adversary types; however, unlike Li et al.~\cite{li2016catcher}, who only evaluate their algorithms through experiments, we present provable approximation guarantees of our algorithms that are computationally efficient. Thus, to our best knowledge, our work is the first to study SSGs with multiple resources and multiple adversary types while presenting tractable algorithms with provable guarantees.

Beyond the above-mentioned contribution, our game structure and solution methodology also differ from prior works on security games. First, unlike classical security games, where adversaries typically have an allocation task of determining a utility-maximizing set of locations to target, in our setting, the role of users (i.e., the adversaries) is to decide whether or not to commit fraud at their respective locations. Next, we depart from much of the previous security games literature as we explicitly model fines administrators levy on defaulting users. Consequently, given the differences in our modeling assumptions and game structure, rather than developing large-scale mixed-integer linear programs to compute equilibria as in past works on security games, we leverage the structure of the payoffs of the administrator and users induced by the fines in our setting to uncover novel geometric insights to develop computationally efficient algorithms with provable performance guarantees.

While SSGs have been a more recent application of game theory, applying economic theory in law enforcement contexts dates to the seminal work of Becker~\cite{becker1968crime}, who modeled criminals as utility-maximizing agents and considered a framework where it is profitable to perform enforcement if its benefits outweigh its costs. Since Becker's work~\cite{becker1968crime}, there has been a growing literature on inspection games that study the equilibrium interplay between an enforcement agency and fraudulent agents~\cite{AVENHAUS20021947}. Such equilibrium inspection models have been studied across applications, including environmental compliance~\cite{HARRINGTON198829,heyes2000implementing,FRIESEN200372,EconomicsofEnvironmentalComplianceandEnforcement}, quality control in supply chains~\cite{su12051874}, fare collections for public transportation ~\cite{SASAKI2014107}, and parking enforcement~\cite{NOURINEJAD201733}, and even been investigated in repeated game settings~\cite{che2024predictive,timing-inspections-predictable-2023}. Akin to the work on inspection games, we also study the influence of the probability of inspection (or allocating a security resource) and the associated inspection penalties (or fines) on the behavior of fraudulent users~\cite{telle2009threat,calel2023policing}. However, unlike this literature, which typically does not account for resource limitations, we investigate a setting where the administrator (or enforcement agency) has a fixed budget of resources to mitigate fraud.

Methodologically, our work aligns with the literature on approximation algorithms for NP-hard problems~\cite{hochba1997approximation,williamson2011design,kempe2003maximizing}, as we also develop polynomial time algorithms that achieve a constant factor approximation to the optimal solutions that are NP-hard to compute. Moreover, our work contributes to the literature on beyond worst-case algorithm design by developing resource augmentation guarantees~\cite{Roughgarden_2021}, wherein an algorithm's performance is compared to the optimal solution handicapped with fewer resources. In particular, we obtain Bulow-Klemperer~\cite{auctions-vs-negotiations} style results for our problem setting as we establish that simple greedy-like algorithms with one additional resource achieve higher revenues or payoffs than the optimal solutions (which are NP-hard to compute) with no extra resources~\cite{akbarpour2020just,akbarpout-matching-spatial,siam-bgg-2020}. Overall, our guarantees for the developed greedy algorithms contribute to the broader literature on \emph{simplicity versus optimality} in algorithm design~\cite{hartline2009simple,HART2017313}.

Our work is also related to \emph{contract theory}~\cite{hart1986theory,bolton2004contract}, which typically considers a principal-agent problem where a principal delegates a task to an agent who takes a (possibly) costly action, unobservable to the principal, that triggers a distribution over rewards~\cite{holmstrom1979moral,grossman1992analysis,dutting2022combinatorial,dutting2024combinatorial,comb-agency-babaioff-2006,dutting-2023-multiagent}. While we also study an optimal contract problem, unlike classical principal-agent models, we use contracts as a mechanism to bridge the gap between the revenue and payoff maximizing outcomes.

\vspace{-5pt}

\section{Model and Preliminaries} \label{sec:model}

This section presents a model of our security game and introduces the strategies and payoffs of the administrator and users. For additional discussions of some of our modeling assumptions beyond those presented in this section, see Appendix~\ref{apdx:model-assumptions-discussion}.

\subsection{Parameters of Security Game} \label{subsec:game-params}

We consider a security game where an administrator seeks to allocate a budget of $R$ security resources (e.g., police officers) across a set of $L$ locations susceptible to fraud and levies a fine of $k$ against users found engaging in fraud. In this security game, users at each location $l$ belong to a discrete (and finitely supported) set $\I$ of \emph{types}, where $|\I| \in \mathbb{N}$, known to the administrator. At each location $l$, each user type $i \! \in \! \I$ is specified by a triple $\Theta_l^i \! = \! (\Lambda_l^i, d_l^i, p_l^i)$, where $\Lambda_l^i$ denotes the number of users belonging to type $i$, $d_l^i$ is the benefit received by users of type $i$ who engage in fraud, and $p_l^i$ represents the administrator payoff for mitigating fraud at location $l$ from users of type $i$. We emphasize that user types are location-specific; hence, we use the subscript $l$ in defining a user type. Further, we order user types by the user benefits such that $d_l^1 \leq d_l^2 \leq \ldots \leq d_l^{|\I|}$. We also note that depending on the administrator’s goals in a specific context, many formulations for the payoff $p_l^i$ are possible, and we subsume the different possible administrator objectives into the term $p_l^i$ for generality.

%Mention that knowing user types is akin to corresponding assumption in bayesian stackelberg games

A few comments about our modeling assumptions are in order. First, while we study a finite user type setting, which captures the main technical nuances of our security game, our algorithms and results can be naturally extended to the continuous user type setting under a mild regularity condition on the boundedness of the administrator's objective (see Appendix~\ref{apdx:extension-continuous-distributions}). Moreover, while our model considers a deterministic set of user types, our model and results can be naturally extended to capture the Bayesian setting, as in Bayesian SSGs, where the administrator only knows the distribution over user types (see Appendix~\ref{apdx:probabilistic-extension}). We present the additional notation required to directly extend our algorithms and results to the probabilistic setting in Appendix~\ref{apdx:probabilistic-extension}. Finally, while we study a setting where the fine $k$ is fixed, our results generalize to the setting when the administrator additionally optimizes over a fine $k$ belonging to an interval $[\underline{k}, \Bar{k}]$ for some $\underline{k}, \Bar{k} \! > \!0$ (see Appendix~\ref{apdx:variability-in-fines-extension}).

To make our model more concrete, we elucidate an example of user types and administrator payoffs in the context of queue jumping in intermediate public transport services. In this context, the location set $L$ corresponds to the mini-bus or share-taxi stands at which users avail the IPT service, and the user type set $\I$ specifies users' different sensitivities or willingness to wait in the queue, captured by a value of time parameter $v_l^i$, which corresponds to the monetary equivalent of one unit of time waiting in the queue. Moreover, letting $t_l$ denote a user's time savings when jumping the queue at location $l$, user types can be defined as follows: $\Lambda_l^i$ represents the number of users with value of time $v_l^i$ availing the IPT service at a location $l$, $d_l^i = v_l^i t_l$ represents the monetary equivalent of the wait time that user type $i$ saves when jumping the queue at location $l$, and $p_l^i$ can, for instance, represent the fraud that the presence of a security resource at location $l$ can prevent, given by $p_l^i = \Lambda_l^i t_l$, i.e., the total wait time benefits that fraudulent users accrue, which is equal to the additional wait time faced by non-defaulting users. As another example, a payoff $p_l^i = x \Lambda_l^i t_l$ or $p_l^i = \Lambda_l^i t_l^x$ for some $x \geq 1$ can serve as a proxy to capture the fact that an administrator may place a higher value in reducing additional wait times of non-defaulting users.

For another example of user types and administrator payoffs in the context of another security game application, that of parking enforcement, see Section~\ref{sec:numerical-experiments-parking-enforcement}.

\subsection{Strategies of Administrator and Users} \label{subsec:player-strategies}

As in prior security games literature~\cite{tambe2011security}, we model our problem as a Stackelberg game, where an administrator (the leader) selects a strategy to allocate its security resources to which the users (the followers) respond by deciding whether to engage in fraud. Here, we present the strategies of the administrator and users in our studied security game.

\emph{Administrator Strategy:} We denote $\sigmaa = (\sigma_l)_{l \in L}$ as the (mixed)-strategy of the administrator, where $\sigma_l \in [0, 1]$ denotes the probability with which a security resource is allocated to location $l$. This mixed-strategy satisfies the administrator's resource budget, i.e., $\sum_{l \in L} \sigma_l \leq R$. For brevity of notation, we define the feasible set of this mixed-strategy vector as $\Omega_R = \{ \sigmaa = (\sigma_l)_{l \in L}: \sigma_l \in [0, 1] \text{ for all } l \in L \text{ and } \sum_{l \in L} \sigma_l \leq R \}$, where the subscript $R$ represents the number of security resources available to the administrator.

\emph{User Strategy:} In response to the administrator's strategy $\sigmaa$, users at each location $l$ decide whether to engage in fraud. We let $y_l^i(\sigmaa) \in [ 0, 1 ]$ denote the probability with which users of type $i$ at location $l$ will engage in fraud, where the outcome $y_l^i(\sigmaa) = 1$ ($y_l^i(\sigmaa) = 0$) corresponds to a setting where users of type $i$ at location $l$ do (do not) engage in fraud.

\subsection{User and Administrator Objectives} \label{subsec:payoffs}

We assume users to be utility maximizers, as is standard in security games~\cite{yin2012trusts}, and study the administrator's resource allocation strategies under revenue and payoff maximization objectives. A revenue maximization objective, wherein the administrator maximizes the total fines collected from enforcement, aligns with the model of a selfish administrator, which is a standard assumption in security games~\cite{yin2012trusts} and resembles practice. For instance, police often place speed traps where users are likely to violate the speed limit despite other locations being more accident-prone~\cite {prensky_johnson_2023}. On the other hand, a payoff maximization objective is a natural choice for an administrator seeking to, for instance, target locations most impacted by fraud.

%In this work, we assume users to be utility maximizers, as is standard in the security games literature~\cite{yin2012trusts}, and study the administrator's resource allocation strategies under two objectives: (i) revenue maximization and (ii) payoff maximization. A revenue maximization objective, wherein the administrator maximizes the total fines collected through its enforcement strategy, aligns with the model of a selfish administrator, which is a standard assumption in prior work on security games~\cite{yin2012trusts} and closely resembles practice. For instance, in road traffic scenarios, police often place speed traps where users are likely to violate the speed limit even though other locations may be more accident-prone~\cite {prensky_johnson_2023}. On the other hand, a payoff maximization objective is a natural choice for an administrator seeking to, for instance, ensure it targets locations most impacted by fraud. We refer to Sections~\ref{subsec:game-params} and~\ref{sec:numerical-experiments-parking-enforcement} for examples of administrator payoffs in two security game applications.

%Examples~\ref{eg:parking-enforcement} and~\ref{eg:ipt-queue-jumping} in Section~\ref{subsec:game-params} for examples of administrator payoffs for two relevant security game applications.

We first elucidate the utility maximization problem of users, who at each location $l$ best respond to the administrator's strategy $\sigmaa$ by choosing whether to engage in fraud based on whether the benefits from fraud outweigh the risk of potential losses through fines. The trade-off between the benefits and losses from fraud can differ by application. For instance, in a parking enforcement context, users have typically already engaged in fraud, e.g., by not purchasing a parking permit, when a police officer patrols a given parking lot, thereby resulting in a gain of $d_l^i$ for user type $i$ from not purchasing the parking permit and a loss through fines of $\sigma_l k$, as the police patrols location $l$ with probability $\sigma_l$. In contrast, in the context of queue jumping in informal transport (see Section~\ref{sec:examples-pertinent}), the presence of a police officer is likely to deter users jumping the queue in the first place, thereby resulting in a gain of $(1-\sigma_l) d_l^i$ for user type $i$ from committing fraud at location $l$, as users are only able to jump the queues when the police officer is not present at location $l$, and a loss through fines of $\sigma_l k$.

To capture these different user utilities, we let $\beta \in [0, 1]$ denote the probability that the security resource when allocated deters fraud. Then, the utility maximization problem of the users of type $i$, given a strategy $\sigmaa$, at location $l$ is given by
%To capture the different possible user utilities across applications, we introduce a parameter $\beta \in [0, 1]$ to denote the probability that the security resource when allocated deters fraud at a given location. Then, the utility maximization problem of the users of type $i$, given an administrator strategy $\sigmaa$, at location $l$ is given by
\begin{align} \label{eq:userOpt-eachLoc}
    \max_{y_l^i \in [0, 1]} U_l^i(\sigmaa, y_l^i) = y_l^i [ (1-\beta \sigma_l) d_l^i - \sigma_l k ].
\end{align}
The above problem represents the \emph{additional} gains to users of type $i$ at location $l$ when engaging in fraud with probability $y_l^i$, where, without loss of generality, we normalize the utility of not engaging in fraud (which happens with probability $1-y_l^i$) to zero. Note that the case when $\beta = 0$ corresponds to the user utility in the parking enforcement context while the case when $\beta = 1$ corresponds to the user utility in the context of queue jumping in IPT services. To simplify exposition, we focus on the setting when $\beta = 1$; however, all our results and analysis naturally extend to any $\beta \in [0, 1]$.

Next, we elucidate the administrator's revenue and payoff maximization problems.

\emph{Revenue Maximization:} The revenue from collected fines under an administrator strategy $\sigmaa$ at each location $l$ from user type $i$ is given by $\sigma_l y_l^i(\sigmaa) k \Lambda_l^i$, where $y_l^i(\sigmaa)$ is the best-response of users at location $l$, given by the solution of Problem~\eqref{eq:userOpt-eachLoc}. Note if $y_l^i(\sigmaa) = 0$, i.e., users of type $i$ do not commit fraud at location $l$, the administrator collects no revenue from these users, while if $y_l^i(\sigmaa)>0$, the administrator levies a fine $k$ on the $\Lambda_l^i$ fraudulent users of type $i$ at location $l$, resulting in a revenue of $k \Lambda_l^i$ from those users. Then, the administrator's revenue maximization problem is given by the following bi-level program: \vspace{-5pt}
\begin{maxi!}|s|[2]   
    {\substack{\sigmaa \in \Omega_R \\ y_l^i(\sigmaa) \in [0, 1], \forall l \in L, i \in \I}}                            
    { Q_R(\sigmaa) = \sum_{l \in L} \sum_{i \in \I} \sigma_l y_l^i(\sigmaa) k \Lambda_l^i,  \label{eq:admin-obj-revenue}}   
    {\label{eq:Eg002}}             
    {}          
    \addConstraint{y_l^i(\sigmaa)}{\in \argmax_{y \in [0, 1]} U_l^i(\sigmaa, y), \quad \text{for all } l \in L, i \in \I, \label{eq:bi-level-con-revenue}}    
\end{maxi!}
where in the upper level problem the administrator deploys an enforcement strategy $\sigmaa$ to maximize its revenue to which users best-respond by maximizing their utilities in the lower-level problem. Here, we express the administrator revenue $Q_R(\sigmaa)$ only as a function of $\sigmaa$ and not the user strategy vector $\y = (y_l^i)_{l \in L, i \in \I}$ for notational simplicity, as the best-response function $y_l^i(\sigmaa)$ of users of type $i$ at each location $l$ can itself be expressed as a function of the strategy $\sigmaa$, as given by Equation~\eqref{eq:bi-level-con-revenue}. In the remainder of this work, we will clarify the formulation of $y_l^i(\sigmaa)$ based on context.

\emph{Payoff Maximization:} When $\beta = 1$, the payoff under a strategy $\sigmaa$ at each location $l$ for user type $i$ is $\sigma_l p_l^i \! + \! (1-\sigma_l) (1 \!- \!y_l^i(\sigmaa)) p_l^i \! = \! p_l^i \! - \! (1 \!- \! \sigma_l) y_l^i(\sigmaa) p_l^i$. In other words, the administrator accrues a payoff $p_l^i$ from user type $i$ if it allocates a resource to location $l$, as the presence of a security resource can prevent fraud at that location when $\beta = 1$, but only accrues $p_l^i$ with probability $1 \! - \! y_l^i(\sigmaa)$ if it does not allocate a resource to location $l$. Then, defining the payoff not accrued under a strategy $\sigmaa$ with $R$ resources as $P_R^-(\sigmaa) = \sum_{l \in L} \sum_{i \in \I} (1 \! - \! \sigma_l) y_l^i(\sigmaa) p_l^i$, the administrator's payoff is $P_R(\sigmaa) \! = \! \sum_{l \in L} \sum_{i \in \I} p_l^i \! - \! P_R^-(\sigmaa)$. Finally, since the term $\sum_{l \in L} \sum_{i \in \I} p_l^i$ is a constant independent of $\sigmaa$, the payoff-maximizing strategy is equivalent to one that minimizes $P_R^-(\sigmaa)$ and can be computed using the following bi-level program: 
\begin{mini!}|s|[2]   
    {\substack{\sigmaa \in \Omega_R \\ y_l(\sigmaa) \in [0, 1], \forall l \in L}}                   
    { P_R^-(\sigmaa) = \sum_{l \in L} \sum_{i \in \I} (1 - \sigma_l) y_l^i(\sigmaa) p_l^i,  \label{eq:admin-obj-fraud}}   
    {\label{eq:Eg001}}             
    {}          
    \addConstraint{y_l^i(\sigmaa)}{\in \argmax_{y \in [0, 1]} U_l^i(\sigmaa, y), \quad \text{for all } l \in L, i \in \I. \label{eq:bi-level-con-fraud}}    
\end{mini!}
While Problem~\eqref{eq:admin-obj-fraud}-\eqref{eq:bi-level-con-fraud} corresponds to the setting when the parameter $\beta = 1$, analogous formulations for the administrator's payoff-maximization problem can also be derived for any $\beta \in [0, 1)$.

Having defined the administrator's and users' objectives, our \emph{goal} is to find tuples $(\sigmaa^*, (y_l^i(\sigmaa^*))_{l \in L, i \in \I})$ that solve the above bi-level programs. Note that such tuples constitute \emph{equilibria} of our security game corresponding to the respective administrator objectives. Moreover, as the structure of the equilibrium best-response function $y_l^i(\sigmaa)$ takes on a simple form (see Section~\ref{sec:deterministic-setting}), analyzing equilibria in our security game reduces to studying the optimal administrator strategies, which will be the main focus in this work.

\section{Revenue and Payoff Maximization with Homogeneous Users} \label{sec:deterministic-setting}

We begin with the study of our security game and the corresponding resource allocation strategies of the administrator under both revenue (Section~\ref{subsec:deterministic-revenue-max}) and payoff (Section~\ref{subsec:welfare-max-deterministic}) maximization objectives in the setting when all users at each location are homogeneous, i.e., each location has a single type of users ($|\I| = 1$), while user types can differ across locations. In this setting, for expositional simplicity and brevity of notation, we drop the superscript $i$ in the notation of user types and denote the type of all users at a given location $l$ by the triple $\Theta_l = (\Lambda_l, d_l, p_l)$.

\vspace{-3pt}

\subsection{Revenue Maximization} \label{subsec:deterministic-revenue-max}

This section studies the administrator's revenue maximization problem and analyzes a greedy algorithm, shown in Algorithm~\ref{alg:GreedyRevenueMaximization}, which allocates resources to locations in the descending order of their $\Lambda_l$ values, where the total spending at each location $l$ is no more than a threshold of $\frac{d_l}{d_l+k}$.

%In this setting, we show that a natural greedy algorithm, shown in Algorithm~\ref{alg:GreedyRevenueMaximization}, which allocates resources to locations in the descending order of their $\Lambda_l$ values, where the total spending at each location $l$ is no more than a threshold of $\frac{d_l}{d_l+k}$, achieves a revenue-maximizing outcome.

%, as is elucidated by the following theorem.

% when the administrator knows each location's type $\Theta_l$.

%Before proving Theorem~\ref{thm:greedy-opt-rev-max-deterministic}, a few comments about Algorithm~\ref{alg:GreedyRevenueMaximization} are in order. 

%In particular, Algorithm~\ref{alg:GreedyRevenueMaximization} allocates resources to locations in the descending order of their $\Lambda_l$ values, where the total spending on each location is no more than a location specific threshold, given by $\frac{d_l}{d_l+k}$.

%, as is elucidated by the following theorem.

%In this setting, we show that a simple greedy algorithm, shown in Algorithm~\ref{alg:GreedyRevenueMaximization}, that allocates resources to locations in the descending order of their $\Lambda_l$ values achieves a revenue-maximizing outcome, as is elucidated by the following theorem.

%The proof of this result 

\vspace{-2pt}

\begin{algorithm}
\footnotesize
\SetAlgoLined
\SetKwInOut{Input}{Input}\SetKwInOut{Output}{Output}
\Input{Total Resource capacity $R$, User types $\Theta_l = (\Lambda_l, d_l, v_l)$ for all locations $l$}
Order locations in descending order of $\Lambda_l$ \;
 \For{$l = 1, 2, ..., |L|$}{
      $\sigma_l \leftarrow \min \{ R, \frac{d_l}{d_l+k} \}$ ; \texttt{\footnotesize \sf Allocate the minimum of the remaining resources and $\frac{d_l}{d_l+k}$ to location $l$} \;
      $R \leftarrow R -  \sigma_l$; \quad \texttt{\footnotesize \sf Update amount of remaining resources} \;
  }
\caption{\footnotesize Greedy Algorithm for Administrator's Revenue Maximization Objective}
\label{alg:GreedyRevenueMaximization}
\end{algorithm}

%\vspace{-2pt}

We show that Algorithm~\ref{alg:GreedyRevenueMaximization} achieves a revenue-maximizing outcome for the administrator.

\begin{theorem} [Optimality of Greedy Algorithm for Revenue Maximization Setting] \label{thm:greedy-opt-rev-max-deterministic}
Suppose that user types are homogeneous. Then, the allocation strategy corresponding to Algorithm~\ref{alg:GreedyRevenueMaximization} achieves a revenue-maximizing outcome, i.e., it solves Problem~\eqref{eq:admin-obj-revenue}-\eqref{eq:bi-level-con-revenue}.
\end{theorem}

Theorem~\ref{thm:greedy-opt-rev-max-deterministic}'s proof relies on the fact that, in the revenue maximization setting, given an administrator strategy $\sigmaa$, the best-response function $y_l(\sigmaa)$ of users, given by the solution of Problem~\eqref{eq:userOpt-eachLoc}, at each location $l$ is given by a threshold policy: %\vspace{-2pt}
\begin{align} \label{eq:best-response-users-rev-max}
    y_l(\sigmaa) = 
    \begin{cases}
        0, &\text{if } \sigma_l > \frac{d_l}{d_l+k}, \\
        1, &\text{otherwise}.
    \end{cases}
\end{align}
Notice that when $\sigma_l = \frac{d_l}{d_l+k}$, any $y_l(\sigmaa) \in [0, 1]$ is a best-response for users at location $l$, i.e., users at location $l$ are indifferent between engaging and not engaging in fraud. However, at the threshold $\sigma_l = \frac{d_l}{d_l+k}$, the administrator's revenue is maximized when $y_l(\sigmaa) = 1$ with any $y_l(\sigmaa) < 1$ resulting in strictly lower revenues. Thus, in the revenue maximization setting, our security game has an equilibrium and, consequently, $y_l(\sigmaa)$ corresponds to a solution of the lower-level problem of the bi-level Program~\eqref{eq:admin-obj-revenue}-\eqref{eq:bi-level-con-revenue} if and only if $y_l(\sigmaa) = 1$ when $\sigma_l = \frac{d_l}{d_l+k}$. We also note that selecting $y_l(\sigmaa) = 1$ when $\sigma_l = \frac{d_l}{d_l+k}$ aligns with the notion of strong Stackelberg equilibria~\cite{kjtjft-2009}, where the ties of the followers (users) are broken to optimize the leader's (administrator's) payoff.

Equation~\eqref{eq:best-response-users-rev-max} implies that if the probability of allocating resources to a location exceeds a threshold, users at that location will stop engaging in fraud as the risk of fines outweighs the gains from fraud. Given users' best-response function in Equation~\eqref{eq:best-response-users-rev-max}, the revenue at each location as a function of the amount of resources allocated to that location is depicted in Figure~\ref{fig:social-welfare-revenue-best-response-plots} (left).

%On the other hand, if the probability of allocating security personnel is below the specified threshold, then users will engage in fraudulent activities. 

%threshold based policy in our setting as opposed to purely randomization, i.e., if suffieicnet resources deployed to a given location users will stop engaging in fraud else not

%\vspace{-7pt}

\begin{figure}[tbh!]
    \centering
    \includegraphics[width=0.65\linewidth]{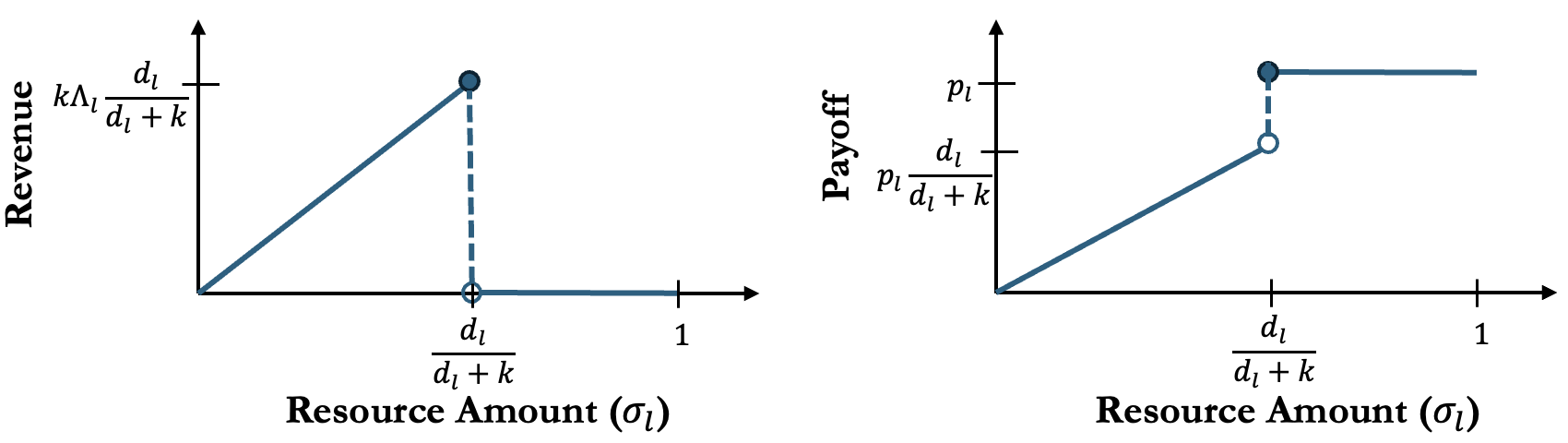}
    \vspace{-12pt}
    \caption{{\small \sf Depiction of the revenue (left) and payoff (right) as a function of the amount of resources allocated to a location $l$. The revenue increases up to a threshold on the resource spending, following which the revenue drops to zero. Analogously, the payoff increases linearly up to the same threshold on the resource spending at which the payoff has a jump discontinuity to the maximum payoff $p_l$ achievable at that location. %Any increase in the resource spending beyond this threshold results in no increase in the payoff. 
    While users are indifferent between whether or not to engage in fraud at the threshold $\sigma_l = \frac{d_l}{d_l+k}$, % of resources allocated to location $l$, 
    revenue is maximized when $y_l(\sigmaa) = 1$ (Equation~\eqref{eq:best-response-users-rev-max}), while the payoff is maximized when $y_l(\sigmaa) = 0$ (see Section~\ref{subsec:problem-structure}), aligning with the filled and open circles for the revenue (left) and payoff (right) maximization settings, respectively.
    %We note that at the threshold fraction $\frac{d_l}{d_l+k}$ of resources allocated to location $l$, users at location $l$ are indifferent between whether or not to engage in fraudulent behavior. However, at this threshold fraction $\frac{d_l}{d_l+k}$ of resources allocated to location $l$, revenue is maximized when $y_l(\sigmaa) = 1$ while welfare is maximized when $y_l(\sigmaa) = 0$, which aligns with the depiction in the filled and open circles for the corresponding revenue and welfare maximization settings, respectively.
    %The discontinuities in the expected revenue function occur due to the best-response function of users given by Problem~\eqref{eq:best-response-users-rev-max} 
    }} 
    \label{fig:social-welfare-revenue-best-response-plots}
\end{figure} %\vspace{-10pt}

%\vspace{-5pt}

%Given Observation~\ref{obs:user-best-response-rev-max}, we now present the main result of this section, which establishes that the greedy algorithm, shown in Algorithm~\ref{alg:GreedyRevenueMaximization}, achieves a revenue maximizing outcome for the administrator, as is elucidated by the following theorem.

%We now use Observation~\ref{obs:user-best-response-rev-max} to complete the proof of Theorem~\ref{thm:greedy-opt-rev-max-deterministic}. 

We now use users' best response function in Equation~\eqref{eq:best-response-users-rev-max} to complete the proof of Theorem~\ref{thm:greedy-opt-rev-max-deterministic}. 

\begin{hproof}[Proof (Sketch) of Theorem~\ref{thm:greedy-opt-rev-max-deterministic}]
Leveraging Equation~\eqref{eq:best-response-users-rev-max}, Theorem~\ref{thm:greedy-opt-rev-max-deterministic}'s proof relies on establishing that the bi-level Program~\eqref{eq:admin-obj-revenue}-\eqref{eq:bi-level-con-revenue} can be reduced to solving a linear program. The key observation to this reduction is that it suffices to consider administrator strategies satisfying $\sigma_l \in [0, \frac{d_l}{d_l+k}]$ for all locations $l$, as allocating more than $\frac{d_l}{d_l+k}$ resources to any location results in reduced revenues (see Figure~\ref{fig:social-welfare-revenue-best-response-plots}), as $y_l(\sigmaa) = 0$ when $\sigma_l > \frac{d_l}{d_l+k}$. We then leverage the structure of the linear program to show that Algorithm~\ref{alg:GreedyRevenueMaximization} that allocates resources to locations in the descending order of the $\Lambda_l$ values achieves an optimal solution to this linear program. %, which establishes our claim. %To do so, we show that any solution that does not result in the same allocation as Algorithm~\ref{alg:GreedyRevenueMaximization} cannot have a higher revenue as the greedy algorithm allocates resources to locations in the descending order of the $\Lambda_l$ values, which establishes our claim.
\end{hproof}

For a complete proof of Theorem~\ref{thm:greedy-opt-rev-max-deterministic}, see Appendix~\ref{apdx:pf-thm-rev-max-deterministic}. Theorem~\ref{thm:greedy-opt-rev-max-deterministic} is in contrast to general hardness results for solving bi-level programs, as it establishes that a simple greedy algorithm (Algorithm~\ref{alg:GreedyRevenueMaximization}) computes the revenue-maximizing strategy, i.e., solves Problem~\eqref{eq:admin-obj-revenue}-\eqref{eq:bi-level-con-revenue}, in polynomial time. In particular, Algorithm~\ref{alg:GreedyRevenueMaximization}'s complexity is $O(|L| \log(|L|))$, since the complexity of sorting locations is $O(|L| \log(|L|))$ while that of iterating through the locations in linear in $|L|$.

\subsection{Payoff Maximization} \label{subsec:welfare-max-deterministic}

\vspace{-1pt}

This section studies the administrator's payoff maximization problem in the setting with homogeneous users. To this end, we first present several properties of the administrator's payoff-maximizing strategy in Section~\ref{subsec:problem-structure}. Then, we establish that computing the payoff-maximizing strategy is NP-hard (Section~\ref{subsec:nphardness-swm-fm}). Finally, we present a greedy algorithm with its associated approximation ratio and resource augmentation guarantees (Section~\ref{subsec:greedy-fraud-min}) and a PTAS for the payoff maximization problem (Section~\ref{sec:ptas-payoff-maximization}).

%and present a greedy algorithm with its associated approximation and resource augmentation guarantees (Section~\ref{subsec:greedy-fraud-min}). Finally, we present a polynomial-time approximation scheme (PTAS), which, when augmented with $\delta$ additional resources for any constant $\delta>0$, achieves a $1-\epsilon$ approximation to the optimal payoff for every fixed parameter $\epsilon>0$ (Section~\ref{sec:ptas-payoff-maximization}).

\vspace{-3pt}

\subsubsection{Properties of Payoff Maximizing Strategy} \label{subsec:problem-structure}

This section presents properties of the solution of the administrator's payoff-maximization problem, which play a key role in our analysis in Sections~\ref{subsec:nphardness-swm-fm}-\ref{sec:ptas-payoff-maximization}. To elucidate these properties, we first note in the payoff-maximization setting that the best-response function of users, given by the solution of Problem~\eqref{eq:userOpt-eachLoc}, at each location $l$ given an administrator strategy $\sigmaa$ is given by: \vspace{-2pt}
%In this section, we present properties of the structure of the optimal solution of the administrator's payoff-maximization problem, which play a pivotal role in establishing the hardness of solving Problem~\eqref{eq:admin-obj-fraud}-\eqref{eq:bi-level-con-fraud} (see Section~\ref{subsec:nphardness-swm-fm}) and analyzing the performance of the greedy algorithm and the PTAS to approximate the administrator's payoff-maximizing strategy (see Sections~\ref{subsec:greedy-fraud-min} and~\ref{sec:ptas-payoff-maximization}). To elucidate these properties, we first note in the payoff-maximization setting that the best-response function of users, given by the solution of Problem~\eqref{eq:userOpt-eachLoc}, at each location $l$ given an administrator strategy $\sigmaa$ is given by the following threshold policy: 
\begin{align} \label{eq:best-response-users}
    y_l(\sigmaa) = 
    \begin{cases}
        0, &\text{if } \sigma_l \geq \frac{d_l}{d_l+k}, \\
        1, &\text{otherwise}.
    \end{cases}
\end{align}
As in the revenue maximization setting, we recall that when $\sigma_l = \frac{d_l}{d_l+k}$, any $y_l(\sigmaa) \in [0, 1]$ is a best-response for users at location $l$. However, at the threshold $\sigma_l = \frac{d_l}{d_l+k}$, the administrator's payoff is maximized and equal to $p_l$ when $y_l(\sigmaa) = 0$ with any $y_l(\sigmaa) > 0$ resulting in a strictly lower administrator payoff of $p_l - (1-\sigma_l) y_l(\sigmaa) p_l$ at that location. Thus, in the payoff maximization setting, our security game has an equilibrium and, consequently, $y_l(\sigmaa)$ corresponds to a solution of the lower-level problem of the bi-level Program~\eqref{eq:admin-obj-fraud}-\eqref{eq:bi-level-con-fraud} if and only if $y_l(\sigmaa) = 0$ when $\sigma_l = \frac{d_l}{d_l+k}$. We note here that unlike the best response of users in the revenue-maximization setting, where $y_l(\sigmaa) = 1$ when $\sigma_l = \frac{d_l}{d_l+k}$, in the payoff maximization setting, we take $y_l(\sigmaa) = 0$ when $\sigma_l = \frac{d_l}{d_l+k}$. Such a difference in the outcomes is attributable to the nature of the revenue and payoff functions at each location $l$, as depicted in Figure~\ref{fig:social-welfare-revenue-best-response-plots}, where the payoff increases to $p_l$ once the probability of allocating resources to location $l$ exceeds the threshold $\frac{d_l}{d_l+k}$ while the revenue drops to zero.

Given the best-response function of users at each location $l$, we now characterize several properties of the administrator's payoff-maximizing strategy.

\vspace{-1pt}

\begin{proposition} [Properties of Payoff-Maximizing Strategy] \label{prop:opt-mixed-strategy-solution}
Suppose user types are homogeneous and that users at each location best-respond using Equation~\eqref{eq:best-response-users}. Then, there exists a solution $\Tilde{\sigmaa}^*$ of Problem~\eqref{eq:admin-obj-fraud}-\eqref{eq:bi-level-con-fraud} satisfying: \vspace{-2pt}
%Then, there exists a welfare-maximizing administrator strategy $\Tilde{\sigmaa}^*$ corresponding to the solution of Problem~\eqref{eq:admin-obj-fraud}-\eqref{eq:bi-level-con-fraud} satisfying:
\begin{enumerate}
    \item $\Tilde{\sigma}_l^* \in \big[ 0, \frac{d_l}{d_l + k} \big]$ for all locations $l \in L$.
    \item There exists a set $L_1$ with $\Tilde{\sigma}_l^* = \frac{d_l}{d_l+k}$ for all $l \in L_1$, a set $L_2$ with $\Tilde{\sigma}_l^* = 0$ for all $l \in L_2$, and at most one location $l'$ with $\Tilde{\sigma}_{l'}^* \in \big( 0, \frac{d_{l'}}{d_{l'}+k} \big)$, where $L_1$, $L_2$, and $\{ l'\}$ are disjoint and $L_1 \cup L_2 \cup \{ l'\} = L$.
\end{enumerate}
Moreover, the administrator's optimal payoff under the strategy $\Tilde{\sigmaa}^*$ is: $P_R(\Tilde{\sigmaa}^*) = \sum_{l \in L_1} v_l + \Tilde{\sigma}_{l'}^* v_{l'}$.
\end{proposition}

\vspace{-1pt}

The proof of the first claim in the proposition statement follows as spending more than $\frac{d_l}{d_l+k}$ at any location does not increase the payoff at that location (see right of Figure~\ref{fig:social-welfare-revenue-best-response-plots}). The proof of the second claim follows from the linearity of the payoff function in the region $\big[0, \frac{d_l}{d_l+k} \big)$ for all locations $l$. Together, both these claims in Proposition~\ref{prop:opt-mixed-strategy-solution} establish that, without loss of generality, for the administrator's payoff maximization problem, it suffices to restrict attention to administrator strategies where the total allocation of security resources $\sigma_l$ at any location $l$ does not exceed $\frac{d_l}{d_l+k}$ and where there is at most one location $l'$ such that $\sigma_{l'} \in \big( 0, \frac{d_{l'}}{d_{l'}+k} \big)$. Consequently, the optimal payoff of the administrator takes a relatively simple form, as in the statement of the proposition. We refer to Appendix~\ref{apdx:pf-thm-opt-mixed-strategy-soln} for a complete proof of Proposition~\ref{prop:opt-mixed-strategy-solution}.

\subsubsection{NP-Hardness of Payoff Maximization}
\label{subsec:nphardness-swm-fm}

We show that the problem of computing the administrator's payoff-maximizing strategy is NP-hard.

%While we developed an efficient algorithm to compute the revenue-maximizing strategy of the administrator, in this section, we show that the problem of computing the administrator's welfare-maximizing strategy is NP-hard. %, as is elucidated by the following theorem.

%In this section, we show that computing the optimal resource allocation strategy of the administrator that minimizes the level of fraud as given by the solution to Problem~\eqref{eq:admin-obj-fraud}-\eqref{eq:bi-level-con-fraud} is NP-hard, as is elucidated by the following theorem.

\begin{theorem} [NP-Hardness of Payoff Maximization] \label{thm:npHardness-swm-fm}
The problem of computing the administrator's payoff-maximizing strategy, i.e., solving Problem~\eqref{eq:admin-obj-fraud}-\eqref{eq:bi-level-con-fraud}, is NP-hard.
\end{theorem}

\begin{proof}[Proof (Sketch) of Theorem~\ref{thm:npHardness-swm-fm}]
We prove this result through a reduction from an instance of the partition problem, which consists of a sequence of numbers $a_1, \ldots, a_n$ with $\sum_{l \in [n]} a_l = A$ and involves the task of deciding whether there is some subset $S_1$ of numbers such that $\sum_{l \in S_1} a_l = \frac{A}{2}$. 

We now construct an instance of the payoff maximization problem (PMP) with $n+1$ locations, where the first $n$ locations correspond to each number of the partition instance, where we define $p_l = a_l$ and let $\frac{d_l}{d_l+k} = \frac{a_l}{A}$ for all locations $l \in [n]$. We also consider a location $n+1$ with $p_{n+1} = \max_{l \in [n]} p_l + \epsilon$ and $\frac{d_{n+1}}{d_{n+1}+k} = \frac{1}{2} + \delta$, where $\epsilon, \delta >0$ are small constants. Finally, we let the number of resources $R = 0.5$. We then show that a sequence of numbers correspond to a ``Yes'' instance of partition if and only if the optimal payoff of this PMP instance is at least $\frac{A}{2}$.

%objective of Problem~\eqref{eq:admin-obj-fraud}-\eqref{eq:bi-level-con-fraud} for above defined instance is at most $\frac{A}{2} + v_{n+1}$, i.e., the total welfare is at least $\frac{A}{2}$.
%We choose $\epsilon$ such that $\max_{l \in [n]} a_l + \epsilon < A$, which is well defined as if $\max_{l \in [n]} a_l > \frac{A}{2}$, then the instance of partition can be solved in polynomial time and thus, without loss of generality we consider partition instances where $\max_{l \in [n]} a_l \leq \frac{A}{2}$. 

To prove the forward direction of this claim, for any “Yes” instance of partition with a set $S$ such that $\sum_{l \in S} a_l = \frac{A}{2}$, we construct an allocation strategy $\sigmaa$ such that $\sigma_l = \frac{a_l}{A}$ for all $l \in S$ and $\sigma_l = 0$ for all $l \in [n+1] \backslash S$. We then verify that $\sigmaa$ is feasible and achieves $P_R(\sigmaa) \geq \frac{A}{2}$.

%To prove the forward direction of this claim, for any ``Yes'' instance of partition, we present a method to construct a feasible administrator allocation strategy $\sigmaa$ that achieves $W_R(\sigmaa) \geq \frac{A}{2}$.

To establish the reverse direction, we leverage Proposition~\ref{prop:opt-mixed-strategy-solution} to show that the only way for both the upper bound resource constraint of $\frac{A}{2}$ and the lower bound payoff constraint of $\frac{A}{2}$ to be satisfied is if the location $l'$ defined in the statement of Proposition~\ref{prop:opt-mixed-strategy-solution} is such that $l' = \emptyset$. %This consequently establishes that we have a ``Yes'' instance of partition, which proves our claim.
\end{proof}

For a complete proof of Theorem~\ref{thm:npHardness-swm-fm}, see Appendix~\ref{apdx:pf-np-hardness-swm-fm}. Theorem~\ref{thm:npHardness-swm-fm} establishes that Problem~\eqref{eq:admin-obj-fraud}-\eqref{eq:bi-level-con-fraud} cannot be solved in polynomial time unless $P = NP$. This result contrasts the polynomial time algorithm we developed in the revenue-maximization setting (see Theorem~\ref{thm:greedy-opt-rev-max-deterministic}). We note that the NP-hardness of the payoff maximization setting stems from the fact that, unlike the revenue function, the administrator's payoff function is discontinuous at $\sigma_l = \frac{d_l}{d_l+k}$ (see Figure~\ref{fig:social-welfare-revenue-best-response-plots}).

\subsubsection{Greedy Algorithm for Payoff Maximization} \label{subsec:greedy-fraud-min}

Given the impossibility of developing a polynomial time algorithm for the administrator's payoff-maximization problem unless $P = NP$ (see Theorem~\ref{thm:npHardness-swm-fm}), this section presents a computationally efficient algorithm to compute an administrator strategy with strong approximation guarantees to the solution of Problem~\eqref{eq:admin-obj-fraud}-\eqref{eq:bi-level-con-fraud}. In particular, we develop a variant of a greedy algorithm, described in Algorithm~\ref{alg:GreedyFraudminimizationDeterministic}, to compute an allocation of resources that achieves at least half the optimal payoff. Moreover, we show that running the greedy algorithm with one additional resource, i.e., $R+1$ resources, results in an outcome with at least the optimal payoff under $R$ resources.

%For each location $l$, we define a threshold $t_l = \min \{ R, \frac{d_l}{d_l+k}\}$ and discretize the interval $[0, t_l]$ with a step-size $\delta>0$, i.e., we define a set of points $\A_l = \{ 0, \delta, 2 \delta, \ldots, t_l \}$ for all locations $l$. Moreover, let $\Hat{p}_l = P_l(t_l)$ be the payoff corresponding to allocating $t_l$ resources to location $l$, where, note that $\Hat{p}_l = p_l$ if $t_l = \frac{d_l}{d_l+k}$ and $\Hat{p}_l = t_l p_l$ if $t_l < \frac{d_l}{d_l+k}$.

We begin by first presenting our algorithmic approach described in Algorithm~\ref{alg:GreedyFraudminimizationDeterministic}. First, for each location $l$, we define an affordability threshold $t_l = \min \{ R, \frac{d_l}{d_l+k}\}$, which represents the maximum amount of resources the administrator can feasibly allocate to location $l$, where recall by Proposition~\ref{prop:opt-mixed-strategy-solution} that it suffices to restrict attention to strategies where the total allocation of security resources to any location $l$ does not exceed the threshold $\frac{d_l}{d_l+k}$. Moreover, let $\Hat{p}_l$ be the payoff from allocating $t_l$ resources to location $l$, where $\Hat{p}_l = p_l$ if $t_l = \frac{d_l}{d_l+k}$ and $\Hat{p}_l = t_l p_l$ if $t_l < \frac{d_l}{d_l+k}$ (see Figure~\ref{fig:social-welfare-revenue-best-response-plots}). Then, we find the greedy solution $\Tilde{\sigmaa}$ that allocates at most $t_l$ resources to each location $l$ in descending order of their \emph{affordable bang-per-buck} ratios given by $\frac{\Hat{p}_l}{t_l}$. We define the quantity $\frac{\Hat{p}_l}{t_l}$ as the affordable bang-per-buck ratio as the total payoff received on allocating a fraction $t_l$ of resources to location $l$ that is affordable given $R$ resources is $\Hat{p}_l$. Next, we compute an allocation $\sigmaa'$ corresponding to spending all the available resources at a single location that yields the highest administrator payoff. Finally, we return the strategy between the two computed strategies $\Tilde{\sigmaa}$ and $\sigmaa'$ that achieves a higher payoff.

%Then, we order the locations in descending order of their \emph{affordable bang-per-buck} ratios, given by $\frac{\Hat{p}_l}{t_l}$, and find the solution $\Tilde{\sigmaa}$ corresponding to the greedy algorithm that allocates at most $t_l$ resources to each location $l$ in descending order of their affordable bang-per-buck ratios. We define the quantity $\frac{\Hat{p}_l}{t_l}$ as the affordable bang-per-buck ratio as the total payoff received on allocating a fraction $t_l$ of resources to location $l$ that is affordable given $R$ resources is $\Hat{p}_l$. Next, we compute an allocation $\sigmaa'$ corresponding to spending all the available resources at a single location that yields the highest payoff for the administrator. Finally, we return the resource allocation strategy between the two computed strategies $\Tilde{\sigmaa}$ and $\sigmaa'$ that achieves a higher administrator payoff. %This procedure is formally presented in Algorithm~\ref{alg:GreedyFraudminimizationDeterministic}.

%\vspace{-5pt}
\begin{algorithm}
\footnotesize
\SetAlgoLined
\SetKwInOut{Input}{Input}\SetKwInOut{Output}{Output}
\Input{Total Resource capacity $R$, User Types $\Theta_l = (\Lambda_l, d_l, v_l)$ for all locations $l$}
\textbf{Step 1: Find Greedy Solution $\Tilde{\sigmaa}$:} \\
Define affordability threshold $t_l \leftarrow \min \{ R, \frac{d_l}{d_l+k}\}$ for all locations $l$ \;
Define payoffs $\Hat{p}_l$, %corresponding to allocating $t_l$ resources to each location $l$, 
where $\Hat{p}_l = p_l$ if $t_l = \frac{d_l}{d_l+k}$ and $\Hat{p}_l = t_l p_l$ if $t_l < \frac{d_l}{d_l+k}$ \;
Order locations in descending order of $\frac{\Hat{p}_l}{t_l}$ \; \vspace{-1pt}
Initialize strategy $\Tilde{\sigmaa} \leftarrow \mathbf{0}$ \;
%Initialize $\Tilde{R} \leftarrow R$ \; \vspace{-1pt}
  \For{$l = 1, 2, ..., |L|$}{
      $\Tilde{\sigma}_l \leftarrow \min \{ R, \frac{d_l}{d_l+k} \}$ ; \texttt{\footnotesize \sf Allocate the minimum of the remaining resources and $\frac{d_l}{d_l+k}$ to location $l$} \; \vspace{-1pt}
      $R \leftarrow R -  \Tilde{\sigma}_l$; \quad \texttt{\footnotesize \sf Update amount of remaining resources} \; \vspace{-1pt}
  }
\textbf{Step 2: Find Solution $\sigmaa'$ that Maximizes Payoff from Spending on Single Location} \\ \vspace{-1pt}
$\sigmaa^l \leftarrow \argmax_{\sigmaa \in \Omega_R: \sigma_{l'} = 0, \forall l' \neq l} P_R(\sigmaa)$, $\forall l$ ; \texttt{\footnotesize \sf Compute allocation $\sigmaa^l$ maximizing payoff from only spending on $l$} \; \vspace{-1pt}
$\sigmaa' \leftarrow \argmax_{l \in L} P_R(\sigmaa^l)$ \; \vspace{-1pt}
\textbf{Step 3: Return Solution with Higher Payoff:} $\sigmaa^*_A \leftarrow \argmax \{ P_R(\Tilde{\sigmaa}), P_R(\sigmaa') \}$ \; 
\caption{\footnotesize Greedy Algorithm for Administrator's Payoff Maximization Objective}
\label{alg:GreedyFraudminimizationDeterministic}
\end{algorithm}
%\vspace{-5pt}

A few comments about Algorithm~\ref{alg:GreedyFraudminimizationDeterministic} are in order. First, as with Algorithm~\ref{alg:GreedyRevenueMaximization}, Algorithm~\ref{alg:GreedyFraudminimizationDeterministic} is computationally efficient with a run-time of $O(|L| \log(|L|))$, corresponding to the complexity of sorting the locations in descending order of the affordable bang-per-buck ratios, and the remaining steps of Algorithm~\ref{alg:GreedyFraudminimizationDeterministic} can be performed in linear time in the number of locations. Next, Algorithm~\ref{alg:GreedyFraudminimizationDeterministic} resembles analogous algorithms from the literature on the knapsack problem, which is the problem of finding the value-maximizing subset of items of given sizes that fits within the capacity of a knapsack. Despite the connection between Algorithm~\ref{alg:GreedyFraudminimizationDeterministic} and the knapsack literature, our problem setting differs from several well-studied variants of the knapsack problem. Unlike the 0-1 knapsack problem, wherein the decision space is binary, the administrator's decision space is continuous in our setting. Moreover, unlike the fractional knapsack problem, where a greedy algorithm that allocates resources in descending order of the bang-per-buck ratios is optimal, in our setting, computing the payoff maximizing strategy is NP-hard (Theorem~\ref{thm:npHardness-swm-fm}). 

Given the similarities and differences between the payoff maximization setting and the knapsack literature, we can interpret the payoff maximization problem (with homogeneous user types) as a novel variant of the knapsack problem. In particular, we can consider locations as items, whose value (or payoff) function increases linearly as a higher fraction of it is packed in the knapsack (i.e., as the probability of allocating resources to a location is increased) up to some threshold. At this location-specific threshold, $\frac{d_l}{d_l+k}$, which we interpret as the size of the item, the payoff function of the item has a jump discontinuity and equals the item's payoff $p_l$, as depicted on the right of Figure~\ref{fig:social-welfare-revenue-best-response-plots}. Note that when the fine $k = 0$, the payoff function has no discontinuity at $\frac{d_l}{d_l+k}$ (see Figure~\ref{fig:social-welfare-revenue-best-response-plots}); thus, when $k = 0$, our payoff maximization problem reduces to a fractional knapsack problem that is polynomial time solvable. Hence, the presence of fines introduces discontinuities in the administrator's payoff function, which corresponds to the source of the NP-hardness of Problem~\eqref{eq:admin-obj-fraud}-\eqref{eq:bi-level-con-fraud} (Theorem~\ref{thm:npHardness-swm-fm}).

We now present the approximation guarantees of Algorithm~\ref{alg:GreedyFraudminimizationDeterministic} to the optimal payoff corresponding to the solution of Problem~\eqref{eq:admin-obj-fraud}-\eqref{eq:bi-level-con-fraud}. Our first result establishes that Algorithm~\ref{alg:GreedyFraudminimizationDeterministic} achieves at least half the optimal payoff. %, i.e., Algorithm~\ref{alg:GreedyFraudminimizationDeterministic} is a $(0.5, 0)$-approximation. % to the payoff-maximizing allocation.

%a half approximation to the welfare-maximizing allocation, i.e., Algorithm~\ref{alg:GreedyFraudminimizationDeterministic} achieves at least half the welfare as that corresponding to the solution of Problem~\eqref{eq:admin-obj-fraud}-\eqref{eq:bi-level-con-fraud}.

\vspace{-1pt}

\begin{theorem} [1/2 Approximation of Greedy Algorithm for Payoff Maximization] \label{thm:greedy-half-approx-fraud-min}
Denote $\sigmaa^*_A$ as the solution corresponding to Algorithm~\ref{alg:GreedyFraudminimizationDeterministic} and let $\sigmaa^*$ be the payoff-maximizing allocation that solves Problem~\eqref{eq:admin-obj-fraud}-\eqref{eq:bi-level-con-fraud}. Then, $\sigmaa^*_A$ achieves at least half the payoff as $\sigmaa^*$, i.e., $P_R(\sigmaa_A^*) \geq \frac{1}{2} P_R(\sigmaa^*)$.
\end{theorem}

\vspace{-1pt}

\begin{comment}
\begin{hproof}
We first define a linear program resembling a fractional knapsack like optimization and leverage Proposition~\ref{prop:opt-mixed-strategy-solution} to show that the optimal administrator welfare is upper bounded by the optimal objective of this linear program. Next, since the optimal objective of the fractional knapsack problem satisfies $\sum_{l \in S} v_l + x_{\Tilde{l}} v_{\Tilde{l}}$ for some subset of locations $S$ and a location $\Tilde{l}$ with $x_{\Tilde{l}} \leq \frac{d_{\Tilde{l}}}{d_{\Tilde{l}}+k}$, to establish the desired half approximation, we show that (i) $W(\sigmaa') \geq x_{\Tilde{l}} v_{\Tilde{l}}$ and (ii) $W(\Tilde{\sigmaa}) \geq \sum_{l \in S} v_l$, where the allocations $\sigmaa'$ and $\Tilde{\sigmaa}$ are as defined in Algorithm~\ref{alg:GreedyFraudminimizationDeterministic}. Notice that the proof of claim (i) is by construction as $\sigmaa'$ by definition is chosen to maximize the welfare from spending on a single location. The proof of claim (ii) relies on the fact that the greedy algorithm in Step 1 on Algorithm~\ref{alg:GreedyFraudminimizationDeterministic} precisely corresponds to the greedy algorithm to optimize the knapsack linear program. Finally, combining the results of claims (i) and (ii), the desired result follows.
\end{hproof}    
\end{comment}

For a detailed proof sketch and proof of Theorem~\ref{thm:greedy-half-approx-fraud-min}, see Appendix~\ref{apdx:pf-greedy-half-approx-welfare}. The key insight in developing Algorithm~\ref{alg:GreedyFraudminimizationDeterministic} and establishing Theorem~\ref{thm:greedy-half-approx-fraud-min} despite the discontinuity in the payoff function (see Figure~\ref{fig:social-welfare-revenue-best-response-plots}) is in recognizing that the administrator's payoff-maximization objective can be upper bounded by the objective of a linear program that resembles a fractional knapsack optimization. Moreover, the approximation ratio of $\frac{1}{2}$ of Algorithm~\ref{alg:GreedyFraudminimizationDeterministic} aligns with the approximation guarantee of an analogous greedy algorithm for the NP-hard 0-1 knapsack problem. We also note that the additional pre-processing step of determining affordability thresholds to compute the bang-per-buck ratios is not necessary for the 0-1 knapsack problem, which has a binary decision space, but is necessary for the administrator's payoff maximization problem with a continuous decision space. In Appendix~\ref{apdx:sub-optimal-less0.5}, we present an example demonstrating that an algorithm analogous to Algorithm~\ref{alg:GreedyFraudminimizationDeterministic} that orders locations in the descending order of their bang-per-buck ratios $\frac{p_l}{\frac{d_l}{d_l+k}}$ (rather than their affordable bang-per-buck ratios) does not achieve the half approximation guarantee as in Theorem~\ref{thm:greedy-half-approx-fraud-min}. 

%Note that such an additional pre-processing step of determining affordability thresholds to compute the bang-per-buck ratios is not necessary for the 0-1 knapsack problem, which has a binary decision space, but is necessary for the administrator's payoff maximization problem with a continuous decision space.

%First, we reiterate that we order locations in the descending order of their affordable bang-per-buck ratios by computing the affordability thresholds and the associated payoffs rather than directly ordering locations by their bang-per-buck ratios $\frac{p_l}{\frac{d_l}{d_l+k}}$ in the first step of Algorithm~\ref{alg:GreedyFraudminimizationDeterministic}. We note that such  is a necessary pre-processing step

%Yet, unlike the 0-1 knapsack problem for which there exists a polynomial-time approximation scheme (PTAS) using dynamic programming, extending this idea to the setting considered in this work is challenging due to the administrator's continuous action space. We defer settling the question of whether a PTAS exists for the welfare maximization problem as a direction for future research.

%We note that the key insight in developing Algorithm~\ref{alg:GreedyFraudminimizationDeterministic} and establishing Theorem~\ref{thm:greedy-half-approx-fraud-min} is in recognizing that the welfare maximization problem of the administrator 

%- mention novelty in proof and how different from knapsack greedy, mention new variant of knapsack

Next, we show that if the administrator had an extra resource, i.e., $R+1$ resources, then Algorithm~\ref{alg:GreedyFraudminimizationDeterministic} achieves a higher payoff than that of the payoff-maximizing outcome with $R$ resources when user types are homogeneous.

%Our next result establishes that if the administrator had just one additional resource, i.e., $R+1$ resources, then Algorithm~\ref{alg:GreedyFraudminimizationDeterministic} achieves a higher welfare than that of the welfare-maximizing outcome with $R$ resources that is NP-hard to compute.

%the optimal solution of the NP-hard bi-level Program~\eqref{eq:admin-obj-fraud}-\eqref{eq:bi-level-con-fraud}.

\begin{theorem} [Resource Augmentation Guarantee for Payoff Maximization] \label{thm:greedy-resource-augmentation-frauad-min}
Suppose $|\I| = 1$, $\sigmaa^*_A$ is the solution of Algorithm~\ref{alg:GreedyFraudminimizationDeterministic} with $R+1$ resources and $\sigmaa^*$ is the solution to Problem~\eqref{eq:admin-obj-fraud}-\eqref{eq:bi-level-con-fraud} with $R$ resources. Then, the total payoff under the allocation $\sigmaa^*_A$ is at least that corresponding to the allocation $\sigmaa^*$, i.e., $P_{R+1}(\sigmaa^*_A) \geq P_{R}(\sigmaa^*)$.
\end{theorem}

The proof of this result relies on much of the machinery developed in proving Theorem~\ref{thm:greedy-half-approx-fraud-min}, and we present its proof in Appendix~\ref{apdx:pf-resource-augmentation-welfare-max}. We note that the administrator only requires $R+\max_{l \in L} \frac{d_l}{d_l+k}$ resources to obtain the guarantee in Theorem~\ref{thm:greedy-resource-augmentation-frauad-min}; however, we present the result with $R+1$ resources for ease of exposition. Theorem~\ref{thm:greedy-resource-augmentation-frauad-min} highlights the benefits of recruiting one additional security resource (e.g., police officer) and applying a simple algorithm, i.e., Algorithm~\ref{alg:GreedyFraudminimizationDeterministic}, rather than investing computational effort to solve the NP-hard payoff-maximization problem. 

\subsubsection{PTAS for Payoff Maximization} \label{sec:ptas-payoff-maximization}

\ifarxiv In the previous section, we developed a near-linear time algorithm that achieves a half approximation to the optimal administrator payoff without additional resources, i.e., a $(0.5, 0)$-approximation, and one approximation with access to one additional resource, i.e., a $(1, 1)$-approximation. \fi This section develops a polynomial-time approximation scheme for the administrator's payoff maximization problem that, when augmented with $\delta$ additional resources for any $\delta > 0$, achieves a $1-\epsilon$ approximation to the administrator's optimal payoff for every fixed parameter $\epsilon>0$.

%, where the running time of our proposed algorithm is polynomial in the number of locations $|L|$ for every fixed parameter $\epsilon>0$.

%$\epsilon$ can be arbitrarily small constants. The running time of this PTAS is polynomial in the number of locations $L$ for every pair of fixed parameters $\epsilon, \delta > 0$.

We begin by presenting the PTAS, which involves a two-stage algorithmic approach, including a brute-force step to find the optimal solution on a constant size subset of locations and a greedy step that extends this partial solution to the set of all locations. To elucidate the two stages of the algorithm, we first discretize the interval $[0, \frac{d_l}{d_l+k}]$ for all locations $l$ with a step-size $\delta>0$, i.e., we define a set of points $\A_l = \{ 0, \delta, 2 \delta, \ldots, \frac{d_l}{d_l+k} \}$ for all locations $l$. Then, in the brute-force step, we consider each feasible subset of at most $m+1$ locations, where we allocate $\frac{d_l}{d_l+k}$ resources to all but one of those locations $l'$ to which we allocate resources belonging to some point in the discretized grid $\A_{l'}$. In particular, we specify an allocation in the brute-force step via the pair $(S, \sigma_{l'}^{\delta})$, where we allocate $\frac{d_l}{d_l+k}$ to all $l \in S$ and allocate resources $\sigma_{l'}^{\delta} \in \A_{l'}$ to location $l'$. Note that the pair $(S, \sigma_{l'}^{\delta})$ is \emph{feasible} if it holds that $\sum_{l \in S} \frac{d_l}{d_l+k} + \sigma_{l'} \leq R + \delta$, i.e., the total amount of allocated resources does not exceed $R+\delta$.

Then, following the brute-force step, for each such feasible pair $(S, \sigma_{l'}^{\delta})$, we allocate resources to the remaining locations via the greedy procedure analogous to step 1 of Algorithm~\ref{alg:GreedyFraudminimizationDeterministic}. In particular, for each feasible pair $(S, \sigma_{l'}^{\delta})$, we allocate the remaining resources to the locations in the set $S' = L \backslash (S \cup \{ l'\})$ in descending order of their bang-per-buck ratios $\frac{p_l}{\frac{d_l}{d_l+k}} = \frac{p_l (d_l+k)}{d_l}$, where we allocate resources only to affordable locations, i.e., the locations whose resource requirement $\frac{d_l}{d_l+k}$ is below the amount of available resources. Finally, we select the allocation $\sigmaa'$ corresponding to the brute-force and greedy steps with the highest payoff. This procedure, including the brute-force step and the greedy step, is presented formally in Algorithm~\ref{alg:PTASWelMaxHom}.

%based on the greedy procedure in step 1 of Algorithm~\ref{alg:GreedyFraudminimizationDeterministic}, i.e., we allocate the remaining resources to locations in descending order of their affordable bang-per-buck ratios given by $\frac{\Hat{p}_l}{t_l}$. We skip locations to which we have to allocate fractionally and only allocate to those that we can allocate fully. Finally, we select the allocation corresponding to the above two-stage procedure that achieves the highest payoff. 

%Our approach is inspired by techniques for developing polynomial time approximation schemes that involve a two step process, including a brute-force to find the optimal solution on a constant size subset of the instance and a step that extends this partial solution to the set of all locations.

%For each location $l$, we define a threshold $t_l = \min \{ R, \frac{d_l}{d_l+k}\}$ and discretize the interval $[0, t_l]$ with a step-size $\delta>0$, i.e., we define a set of points $\A_l = \{ 0, \delta, 2 \delta, \ldots, t_l \}$ for all locations $l$. Moreover, let $\Hat{p}_l = P_l(t_l)$ be the payoff corresponding to allocating $t_l$ resources to location $l$, where, note that $\Hat{p}_l = p_l$ if $t_l = \frac{d_l}{d_l+k}$ and $\Hat{p}_l = t_l p_l$ if $t_l < \frac{d_l}{d_l+k}$.

\begin{algorithm}
\SetAlgoLined
\footnotesize
\SetKwInOut{Input}{Input}\SetKwInOut{Output}{Output}
\Input{Total Resources $R$, User Types $\Theta_l = (\Lambda_l, d_l, p_l)$ for all locations $l$, Subset size $m$, Parameter $\delta>0$}
\Output{Resource Allocation Strategy $\sigmaa'$} 
Order all locations in the descending order of $\frac{p_l(d_l+k)}{d_l}$ \;
\textbf{Step 1: Brute-force Search} \\ 
Let $(S, \sigma_{l'}^{\delta})$ correspond to the pair of all feasible subsets of locations with $|S| \leq m$ with an allocation $\Tilde{\sigma}_l^{(S, \sigma_{l'}^{\delta})} = \frac{d_l}{d_l+k}$ for all $l \in S$ and $\Tilde{\sigma}_{l'}^{(S, \sigma_{l'}^{\delta})} = \sigma_{l'}^{\delta} \in \A_{l'}$ for some location $l'$ \;
Update remaining resources: $R + \delta \leftarrow R + \delta - \sum_{l \in S} \frac{d_l}{d_l+k} - \sigma_{l'}^{\delta}$ \;
\textbf{Step 2: Greedy Allocation to Remaining Locations} \\
For each feasible pair $(S, \sigma_{l'}^{\delta})$, run the following sub-routine: \\
$S' \leftarrow L \backslash (S \cup \{l'\})$ \texttt{\sf Subset of remaining locations that have not been allocated resources in the brute-force step} \; 
\For{$l \in S'$}{
      \If{$\frac{d_l}{d_l+k} \leq R$}{
        $\Tilde{\sigma}_{l}^{(S, \sigma_{l'}^{\delta})} \leftarrow \frac{d_l}{d_l+k}$ ; \texttt{\sf Allocate $\frac{d_l}{d_l+k}$ to location $l$} \; 
        $R + \delta \leftarrow R + \delta -  \frac{d_l}{d_l+k}$; \quad \texttt{\sf Update amount of remaining resources} \;
      }
  }
\textbf{Step 3: Select Allocation $\sigmaa'$ with Highest Payoff:} $\sigmaa' \leftarrow \argmax_{\text{feasible } (S, \sigma_{l'}^{\delta})} P_R(\Tilde{\sigmaa}^{(S, \sigma_{l'}^{\delta})})$ \; 
\caption{PTAS for Homogeneous Payoff Maximization}
\label{alg:PTASWelMaxHom}
\end{algorithm}

A few comments about Algorithm~\ref{alg:PTASWelMaxHom} are in order. First, given parameters $m$ and $\delta$, the maximum number of feasible pairs $(S, \sigma_{l'}^{\delta})$ is $O(m |L|^m \frac{|L|}{\delta})$, as there are at most $O(m |L|^m)$ subsets of size at most $m$ and the number of possible values $\sigma_{l'}^{\delta}$ can take is at most $O\left(\frac{|L|}{\delta}\right)$. Furthermore, for each feasible pair $(S, \sigma_{l'}^{\delta})$, the run-time of the greedy step of Algorithm~\ref{alg:PTASWelMaxHom} is linear in the number of locations and the computational complexity of ordering locations in the descending order of their bang-per-buck ratios is $O(|L| \log(|L|))$. Consequently, the run-time of Algorithm~\ref{alg:PTASWelMaxHom} is $O \left( \frac{m |L|^{m+2}}{\delta} \right)$.

Algorithm~\ref{alg:PTASWelMaxHom} is similar in spirit to the PTAS for the knapsack problem~\cite{lai2006knapsack}\ifarxiv that also involves a two-step process, including a brute-force step to find the optimal solution on a constant size subset of the instance and a greedy step that extends this partial solution to the entire instance~\cite{lai2006knapsack}. \else. \fi However, in contrast to the PTAS for the 0-1 knapsack problem with a binary decision space, Algorithm~\ref{alg:PTASWelMaxHom} has an additional discretization step to generate the feasible pairs $(S, \sigma_{l'}^{\delta})$ due to the continuous decision space of the administrator's enforcement strategies in our setting, which results in an additional $O\left(\frac{|L|}{\delta} \right)$ factor to the run-time of Algorithm~\ref{alg:PTASWelMaxHom} compared to the analogous algorithm for the 0-1 knapsack problem. Furthermore, unlike the PTAS for the 0-1 knapsack problem, which does not require additional resources, establishing an approximation guarantee for Algorithm~\ref{alg:PTASWelMaxHom} requires access to $\delta>0$ additional resources for any fixed constant $\delta>0$ due to the continuity of the decision space of the administrator's payoff maximization problem.

We now establish that given parameters $m$ and $\delta>0$, Algorithm~\ref{alg:PTASWelMaxHom} achieves a $1 - \frac{1}{m+1}$ fraction of the optimal payoff with $\delta$ additional resources.

\begin{theorem} [PTAS for Homogeneous Payoff Maximization] \label{thm:ptas}
Let $\sigmaa'$ be the allocation corresponding to Algorithm~\ref{alg:PTASWelMaxHom} given a subset size $m$ and $R+\delta$ resources, where $\delta>0$ is a parameter. Moreover, let $\sigmaa^*$ be the payoff-maximizing allocation given $R$ resources. Then, the strategy $\sigmaa'$ achieves at least a $1 - \frac{1}{m+1}$ fraction of the optimal payoff, i.e., $P_R(\sigmaa^*) \left( 1 - \frac{1}{m+1} \right) \leq P_R(\sigmaa')$.
%Let $m$ denote the size of the subset of locations in the brute-force step of Algorithm~\ref{alg:PTASWelMaxHom} and $\delta>0$ correspond to the additional resources that Algorithm~\ref{alg:PTASWelMaxHom} has access to. Furthermore, let $\sigmaa'$ be the allocation corresponding to Algorithm~\ref{alg:PTASWelMaxHom} and $\sigmaa^*$ be the payoff-maximizing allocation. Then, the strategy $\sigmaa'$ achieves at least a $1 - \frac{1}{m+1}$ fraction of the optimal payoff, i.e., $P_R(\sigmaa^*) \left( 1 - \frac{1}{m+1} \right) \leq P_R(\sigmaa')$. %In other words, $\sigmaa'$ is a $\left(1 - \frac{1}{m+1}, \delta \right)$ approximation to the payoff maximizing solution.
\end{theorem}

\begin{hproof}
We first recall from Proposition~\ref{prop:opt-mixed-strategy-solution} that the optimal payoff is given by $P_R(\sigmaa^*) = \sum_{l \in L_1} p_l + \sigma_{l'}^* p_{l'}$ for some set of locations $L_1$ for which $\sigma_l^* = \frac{d_l}{d_l+k}$ for all $l \in L_1$ and at most one location $l'$ with $\sigma_{l'}^* < \frac{d_{l'}}{d_{l'}+k}$. We proceed by analysing two cases: (i) $|L_1| \leq m$ and (ii) $|L_1| > m$. 

In the case that $|L_1| \leq m$, we show that the allocation corresponding to the brute-force step of Algorithm~\ref{alg:PTASWelMaxHom} with $\delta>0$ additional resources achieves at least the optimal payoff. 

For the case $|L_1| > m$, we consider the pair $(S, \sigma_{l'}^{\delta})$ such that $S$ corresponds to the set of locations in the administrator's payoff-maximizing allocation with the $m$ highest payoffs and $\sigma_{l'}^{\delta} = \min \{ \sigma_{l'}^* + \delta, \frac{d_{l'}}{d_{l'}+k} \}$. It is straightforward to check that such a pair $(S, \sigma_{l'}^{\delta})$ is feasible and thus is considered in the brute-force step of Algorithm~\ref{alg:PTASWelMaxHom} when the administrator has $\delta$ additional resources. Finally, we proceed by contradiction and leverage the fact that step two of Algorithm~\ref{alg:PTASWelMaxHom} allocates resources to the remaining locations in the set $S' = L \backslash (S \cup \{ l' \})$ in the descending order of their bang-per-buck ratios and that the set $S$ corresponds to the set of locations in the payoff-maximizing allocation with the $m$ highest payoffs to establish our desired guarantee.
\end{hproof}

For a complete proof of Theorem~\ref{thm:ptas}, see Appendix~\ref{apdx:pf-thm-ptas}. Given the guarantee in Theorem~\ref{thm:ptas}, to obtain a PTAS with a $1 - \epsilon$-approximation with at most $\epsilon$ extra resources for any fixed parameter $\epsilon>0$, we set $\epsilon = \delta = \frac{1}{m+1}$. Note that the resulting run time of Algorithm~\ref{alg:PTASWelMaxHom} is $O\big(\frac{m |L|^{m+2}}{\delta} \big) = O(\frac{1}{\epsilon^2} |L|^{\frac{1}{\epsilon}+1})$. Consequently, Theorem~\ref{thm:ptas} implies that Algorithm~\ref{alg:PTASWelMaxHom} is a PTAS for the administrator's payoff maximization problem when given $\epsilon$ extra resources. Finally, we note that there exists a fully polynomial-time approximation scheme (FPTAS) for the 0-1 knapsack problem using dynamic programming; however, extending this idea to the setting considered in this work is challenging due to the administrator’s continuous action space. We defer settling the question of whether a FPTAS exists for the payoff maximization problem as a direction for future research.

%Given the guarantee in Theorem~\ref{thm:ptas}, to obtain a PTAS with a $(1-\epsilon, \epsilon)$-approximation for any fixed parameter $\epsilon>0$, we set $\epsilon = \delta = \frac{1}{m+1}$. Note that the resulting running time of Algorithm~\ref{alg:PTASWelMaxHom} is $O\left(\frac{m |L|^{m+2}}{\delta} \right) = O(\frac{1}{\epsilon^2} |L|^{\frac{1}{\epsilon}+1})$. Consequently, Theorem~\ref{thm:ptas} implies that Algorithm~\ref{alg:PTASWelMaxHom} is a PTAS for the administrator's payoff maximization problem when given access to $\delta>0$ additional resources. Finally, we note that there exists a fully polynomial-time approximation scheme (FPTAS) for the 0-1 knapsack problem using dynamic programming; however, extending this idea to the setting considered in this work is challenging due to the administrator’s continuous action space. We defer settling the question of whether a FPTAS exists for the payoff maximization problem as a direction for future research.

\section{Revenue Maximization with Heterogeneous Users} \label{sec:probabilistic-setting}

This section studies the administrator's revenue maximization objective in the heterogeneous user type setting, i.e., user types at each location can vary with $|\I| \geq 1$. In this setting, we show that computing the revenue-maximizing strategy is NP-hard in Section~\ref{sec:np-hardness-prob-rev-max} and present a variant of a greedy algorithm and its associated approximation ratio and resource augmentation guarantees in Section~\ref{sec:near-opt-greedy-rev-max}. While we focus on revenue maximization in this section, our algorithmic ideas and guarantees naturally generalize to the payoff maximization objective in the setting with heterogeneous users (see Appendix~\ref{apdx:welfare-max-prob-setting}).

\subsection{NP-Hardness of Setting with Heterogeneous Users} \label{sec:np-hardness-prob-rev-max}

%In this section, we consider the probabilistic setting when the administrator only knows the distribution from which users' types are drawn at each location (but not the exact realization of those types) and thus maximize expected revenues. 

In the setting with heterogeneous users, we show that computing the administrator's revenue maximizing strategy is NP-hard.

%, as is elucidated by the following theorem.

\begin{theorem} [NP-Hardness of Heterogeneous Revenue Maximization] \label{thm:npHardness-erm}
The problem of computing the administrator's revenue maximizing strategy in the presence of heterogeneous user types with $|I|>1$, i.e., solving Problem~\eqref{eq:admin-obj-revenue}-\eqref{eq:bi-level-con-revenue}, is NP-hard.
\end{theorem}

\begin{hproof}
We prove this result through a reduction from partition. In particular, given a partition instance with a sequence of numbers $a_1, \ldots, a_n$, we construct an instance of the heterogeneous revenue maximization problem (HRMP) with two types, i.e., $|\I| = 2$, and $n$ locations, where each number $a_l$ corresponds to a location. In this setting, we drop the fine $k$ from Objective~\eqref{eq:admin-obj-revenue} as it is a uniform constant that applies to all locations $l$ and types $i$. Then, we define: (i) $\frac{d_l^1}{d_l^1+k} = \frac{1}{2} \frac{a_l}{A}$, (ii) $\frac{d_l^2}{d_l^2+k} = \frac{a_l}{A}$, (iii) $\Lambda_l^1 = A \left( 2 \frac{\max_{l' \in [n]} a_{l'}}{a_l} - 1 \right)$, and (iv) $\Lambda_l^2 = A \left(1+ \frac{2 \max_{l' \in [n]} a_{l'}}{a_l} \right)$ for all $l \in [n]$. Moreover, we let the number of resources $R = \frac{3}{4}$. Then, we claim that we have a ``Yes'' instance of partition if and only if the optimal revenue for this HRMP instance is at least $\frac{A}{2} + 2 n \times \max_{l \in [n]} a_l$.

To prove the forward direction of this claim, for any “Yes” instance of partition with a set $S$ such that $\sum_{l \in S} a_l = \frac{A}{2}$, we construct a strategy $\sigmaa$ such that $\sigma_l = \frac{a_l}{A}$ for all $l \in S$ and $\sigma_l = \frac{1}{2} \frac{a_l}{A}$ for all $l \in [n] \backslash S$. We then verify that $\sigmaa$ is feasible and achieves $Q_R(\sigmaa) \geq \frac{A}{2} + 2 n \times \max_{l \in [n]} a_l$.

To prove the reverse direction, we first show that when $R = \frac{3}{4}$ there exists an optimal strategy $\Tilde{\sigmaa}^*$ satisfying $\Tilde{\sigma}_l^* \geq \frac{d_l^1}{d_l^1+k}$ for all locations $l$. Leveraging this property and another structural property of the revenue-maximizing strategy (see Appendix~\ref{apdx:pf-np-hardness-erm}) akin to that established in Proposition~\ref{prop:opt-mixed-strategy-solution}, we show that the administrator's optimal strategy can only satisfy the resource constraint and achieve a revenue that is at least $\frac{A}{2} + 2 n \times \max_{l \in [n]} a_l$ if there is some set $S' \subseteq [n]$ with $\sum_{l \in S'} a_l = \frac{A}{2}$. %This consequently establishes that we have a ``Yes'' instance of partition, which proves our claim.
\end{hproof}

\vspace{-2pt}

For a complete proof of Theorem~\ref{thm:npHardness-erm}, see Appendix~\ref{apdx:pf-np-hardness-erm}. Theorem~\ref{thm:npHardness-erm} establishes that Problem~\eqref{eq:admin-obj-revenue}-\eqref{eq:bi-level-con-revenue} cannot be solved in polynomial time unless $P = NP$. This result contrasts the polynomial time algorithm we developed in the revenue maximization setting with homogeneous users (see Theorem~\ref{thm:greedy-opt-rev-max-deterministic}), as, while, for any location $l$, the administrator's revenue function in the setting with homogeneous users is continuous and monotone in the range $\sigma_l \in \big[ 0, \frac{d_l}{d_l+k} \big]$, the revenue function in the presence of heterogeneous users is discontinuous and non-monotone in the range $\sigma_l \in \Big[0, \max_i \frac{d_l^{|i|}}{d_l^{|i|} + k}\Big]$ (see Figure~\ref{fig:concave-upper-approximation}), even when the number of types $|\I| = 2$. The discontinuity and non-monotonicity of the revenue function with heterogeneous user types holds as the best-response function of users, given by the solution of Problem~\eqref{eq:userOpt-eachLoc}, is a threshold function, where users of each type have a different threshold probability at which they shift from engaging to not engaging in fraud. In particular, the best response of users of type $i$ at location $l$ given an administrator strategy $\sigmaa$ is given by
\vspace{-2pt}
\begin{align} \label{eq:best-response-users-rev-max-prob}
    y_l^i(\sigmaa) = 
    \begin{cases}
        0, &\text{if } \sigma_l > \frac{d_l^i}{d_l^i+k}, \\
        1, &\text{otherwise},
    \end{cases}
\end{align}
where the resource amounts $\frac{d_l^i}{d_l^i+k}$ for all $i$ correspond to the points at which the revenue function is discontinuous and non-monotone, as is depicted in Figure~\ref{fig:concave-upper-approximation} (left). \ifarxiv We reiterate, as with the best response of users in the setting with homogeneous users in Equation~\eqref{eq:best-response-users-rev-max}, that when $\sigma_l = \frac{d_l^i}{d_l^i+k}$, any $y_l^i(\sigmaa) \in [0, 1]$ is a best-response for users at a location $l$ with type $i$. Yet, we let $y_l^i(\sigmaa) = 1$ when $\sigma_l = \frac{d_l^i}{d_l^i+k}$ as it corresponds to the highest revenue outcome for the administrator (see Section~\ref{subsec:deterministic-revenue-max} for a further discussion). \fi %Finally, due to the best-response function of users given by Equation~\eqref{eq:best-response-users-rev-max-prob}, the resulting expected revenue at each location $l$ as a function of the amount of resources allocated is both non-monotone and discontinuous at the resource amounts $\frac{d_l^i}{d_l^i+k}$ for all $i$, as is depicted in Figure~\ref{fig:concave-upper-approximation} (left). 

%We note that the discontinuity and non-monotonicity of the revenue function with heterogeneous users follows due to the following best-response function of users 

The fundamental challenge in developing both our hardness results, i.e., for payoff maximization with homogeneous users (Theorem~\ref{thm:npHardness-swm-fm}) and revenue maximization with heterogeneous users (Theorem~\ref{thm:npHardness-erm}), lies in constructing the right instances or \emph{gadgets} to achieve our desired reductions, which we develop leveraging the structural properties of the respective problem settings. For example, in the payoff maximization setting, our reduction crucially relies on Proposition~\ref{prop:opt-mixed-strategy-solution}, and, hence, the discontinuity of the payoff function at $\sigma_l = \frac{d_l}{d_l+k}$. Analogously, in the heterogeneous revenue maximization setting, our reduction leverages the non-monotonicity of the revenue function.

\subsection{Greedy Algorithm for Heterogeneous Revenue Maximization} \label{sec:near-opt-greedy-rev-max}

Given the hardness of solving Problem~\eqref{eq:admin-obj-revenue}-\eqref{eq:bi-level-con-revenue} with heterogeneous users, we develop an efficient algorithm to compute an administrator strategy for the heterogeneous revenue maximization problem and present its associated approximation ratio and resource augmentation guarantees.

To motivate our algorithmic approach, first note that the difficulty in solving Problem~\eqref{eq:admin-obj-revenue}-\eqref{eq:bi-level-con-revenue} in the setting with heterogeneous users is attributable to the non-monotonicity and discontinuity of the revenue function (see Figure~\ref{fig:concave-upper-approximation} for an example of a setting with five types, i.e., $|\I| = 5$). Given this difficulty of directly optimizing the revenue function, we define its \emph{monotone concave upper approximation} (MCUA) that can be tractably maximized. In particular, we define the MCUA of the revenue function at a location $l$ as its concave upper closure, i.e., the point-wise smallest continuous and concave increasing function that bounds the revenue function from above, as depicted by the orange line in Figure~\ref{fig:concave-upper-approximation} (right). We consider the MCUA of the revenue function instead of another approximation as it is both monotonically increasing and concave, properties not guaranteed to hold for other approximations, e.g., the upper bound of the revenue function that connects the origin $(0, 0)$ to its points of discontinuity, as depicted by the green curve in Figure~\ref{fig:concave-upper-approximation} (center), is continuous but may be non-monotone and non-concave, and hence challenging to optimize.

Since the MCUA of the revenue function for each location $l$ is piece-wise linear under a discrete set of user types, we characterize it via a set $\S$ corresponding to the set of all piece-wise linear \emph{segments} of this MCUA across all locations. We associate each segment $s \in \S$ with a triple $(l_s, c_s, x_s)$, where $l_s$ represents the location corresponding to segment $s$, $c_s$ corresponds to its slope, and $x_s$ represents its horizontal width, i.e., resource requirement, as depicted on the right in Figure~\ref{fig:concave-upper-approximation}.

\begin{figure}[tbh!]
    \centering
    \includegraphics[width=0.9\linewidth]{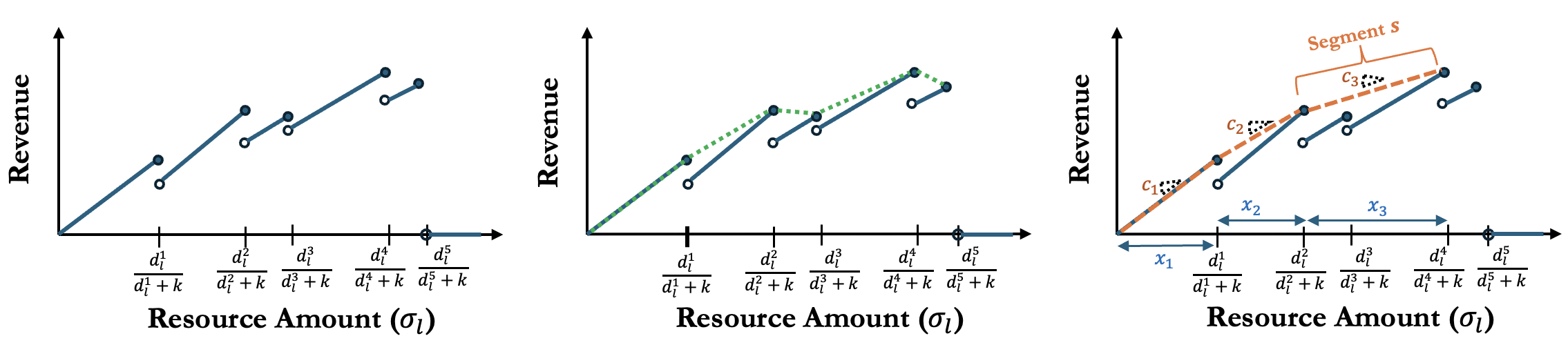}
    \vspace{-13pt}
    \caption{{\small \sf Depiction of an example of the revenue from spending on a location $l$ (left), its upper bound (center), and its corresponding MCUA (right) for a setting with five types. %Note that the expected revenue drops to zero for a location $l$ after the resources allocated exceed $\max_{i} \frac{d_l^i}{d_l^i+k}$. 
    The revenue function's MCUA has three segments $s$ for location $l$ with corresponding slopes $c_s$ and widths $x_s$.
    %For the MCUA of the revenue function, there are three segments $s$ for location $l$ with corresponding slopes $c_s$ and widths $x_s$.  %The discontinuities in the expected revenue function occur due to the best-response function of users given by Problem~\eqref{eq:best-response-users-rev-max} 
    }} 
    \label{fig:concave-upper-approximation}
\end{figure} 

%\vspace{-7pt}

Having introduced the MCUA of the revenue function, we present Algorithm~\ref{alg:GreedyRevMaxProb}, which involves two key steps. First, rather than directly optimizing the administrator's revenue, an NP-hard problem (see Theorem~\ref{thm:npHardness-erm}), we optimize its MCUA using a greedy procedure. To elucidate this procedure, as in Algorithm~\ref{alg:GreedyFraudminimizationDeterministic}, we first define an affordability threshold $t_l \! = \! \min \big\{ R, \max_i \frac{d_l^i}{d_l^i+k}\big\}$ for each location $l$ and define the MCUA of the revenue function for each location $l$ over the range $[0, t_l]$. Since the revenue function's MCUA is piece-wise linear, we order its segments in the set $\S$ in the descending order of the slopes $c_s$ and allocate at most $x_s$ to each segment in that order. Our greedy procedure results in an allocation $\Tilde{\sigmaa}$ and terminates when a segment's resource requirement $x_s$ exceeds the available resources.\footnote{We terminate our greedy procedure at the point in the algorithm when a segment's resource requirement $x_s$ exceeds the available resources, as the revenue function is non-monotone in the resources allocated, and the MCUA, by construction, is only guaranteed to coincide with the revenue function at each location $l$ at the points where the resources allocated equal $\frac{d_l^i}{d_l^i+k}$ for some values of $i$ (e.g., see Figure~\ref{fig:concave-upper-approximation}).} 
In the second step, we compute an allocation $\sigmaa'$ corresponding to spending all the resources at a single location that yields the highest revenue. Finally, we return the strategy between $\Tilde{\sigmaa}$ and $\sigmaa'$ with a higher revenue.

\begin{algorithm}
\footnotesize
\SetAlgoLined
\SetKwInOut{Input}{Input}\SetKwInOut{Output}{Output}
\Input{Total Resource capacity $R$, User Types $\Theta_l^i = (\Lambda_l^i, d_l^i, v_l^i)$ for all locations $l$ and types $i$}
\Output{Resource Allocation Strategy $\sigmaa_A^*$} \vspace{-1.25pt}
\textbf{Step 1: Greedy Allocation $\Tilde{\sigmaa}$ Based on Slopes of MCUA of Revenue Function:} \\ \vspace{-1.25pt}
Define affordability threshold $t_l \leftarrow \min \left\{ R, \max_i \frac{d_l^i}{d_l^i+k}\right\}$ for all locations $l$ \;
Generate MCUA of the revenue function in range $[0, t_l]$ for each location $l$ \;
%Generate MCUA of the revenue function for each location $l$ \; 
\vspace{-1.25pt}
%Define $c_s$ as the slopes of all piece-wise linear segments $s$ of MCUA of expected revenue function across all locations and let $x_s$ be the horizontal width, i.e., required resource amount, corresponding to segment $s$, and $l_s$ be the location associated with segment $s$ \;
$\Tilde{\S} \leftarrow $ Ordered list of segments $s$ across all locations of this MCUA in descending order of slopes $c_s$  \; \vspace{-1.25pt}
Initialize allocation strategy $\Tilde{\sigmaa} \leftarrow \mathbf{0}$ \; \vspace{-1.5pt}
\For{\text{segment $s \in \Tilde{\S}$}}{
      \eIf{$x_s \leq R$}{
        $\Tilde{\sigma}_{l_s} \leftarrow \Tilde{\sigma}_{l_s} + x_s$ ; \texttt{\footnotesize \sf Allocate $x_{s}$ to location $l_s$} \; \vspace{-1.5pt}
        $R \leftarrow R -  x_s$; \quad \texttt{\footnotesize \sf Update amount of remaining resources} \; \vspace{-1.5pt}
      }
      {
      \textbf{break} ; \quad \texttt{\footnotesize \sf Only allocate resources if $x_s \leq R$} \vspace{-1.5pt}
      }
  }
\textbf{Step 2: Find Solution $\sigmaa'$ that maximizes revenue from spending on single location:} \\ \vspace{-1.5pt}
$\sigmaa^l \leftarrow \argmax_{\sigmaa \in \Omega_R: \sigma_{l'} = 0, \forall l' \neq l} Q_R(\sigmaa)$ for all locations $l$ \;  \vspace{-1.5pt}%\quad \texttt{\footnotesize \sf Compute the amount of resources $\sigmaa^l$ that maximizes the expected revenue from solely spending on a given location $l$} \;
$\sigmaa' \leftarrow \argmax_{l \in L} Q_R(\sigmaa^l)$ \; %\quad %\texttt{Find optimal solution from $\sigmaa^l$ from spending on single location} \;
\textbf{Step 3: Return Solution with a Higher Revenue:} $\sigmaa^*_A \leftarrow \argmax \{ Q_R(\Tilde{\sigmaa}), Q_R(\sigmaa') \}$ \; 
\caption{\footnotesize Greedy Algorithm for Administrator's Heterogeneous Revenue Maximization Objective}
\label{alg:GreedyRevMaxProb}
\end{algorithm}

We now present the main results of this section, which establish the approximation guarantees of Algorithm~\ref{alg:GreedyRevMaxProb}. Our first result establishes that Algorithm~\ref{alg:GreedyRevMaxProb} achieves at least half the optimal revenue.

%\begin{definition} [Optimistic Bang-per-buck] \label{def:opt-bang-per-buck}

%\end{definition}

\vspace{-2pt}

\begin{theorem} [1/2 Approximation for Heterogeneous Revenue Maximization] \label{thm:greedy-half-approx-rev-max}
Suppose $|\I| \geq 1$, $\sigmaa^*_A$ is the allocation corresponding to Algorithm~\ref{alg:GreedyRevMaxProb}, and $\sigmaa^*$ is the %expected revenue maximizing 
solution of Problem~\eqref{eq:admin-obj-revenue}-\eqref{eq:bi-level-con-revenue} with $R$ resources. Then, $\sigmaa^*_A$ achieves at least half the revenue as $\sigmaa^*$, i.e., $Q_R(\sigmaa_A^*) \geq \frac{1}{2} Q_R(\sigmaa^*)$.
\end{theorem}

\vspace{-2pt}

The crux of establishing Theorem~\ref{thm:greedy-half-approx-rev-max} involves showing that the MCUA of the revenue function can be maximized via a greedy process akin to step one of Algorithm~\ref{alg:GreedyRevMaxProb}, where resources are allocated until either all the resources are exhausted or there are no further segments remaining to iterate over. Then, noting a key structural property of the solution of this greedy allocation under which the MCUA coincides with the original revenue function at all but at most one location, our desired half approximation guarantee follows from arguments similar to that in the proof of Theorem~\ref{thm:greedy-half-approx-fraud-min}. For a proof of Theorem~\ref{thm:greedy-half-approx-rev-max}, see Appendix~\ref{apdx:greedy-half-approx-rev-max}. We also note that the problem of maximizing the MCUA of the revenue function can be cast as a linear program as in the payoff maximization setting with homogeneous users (see Theorem~\ref{thm:greedy-half-approx-rev-max}). However, this linear program is different from the fractional knapsack linear program to prove Theorem~\ref{thm:greedy-half-approx-fraud-min} in the payoff maximization setting, as, unlike the payoff function in the setting with homogeneous users, the MCUA of the expected revenue function at each location is piece-wise linear with (potentially) multiple segments with differing slopes.

Next, we show that if the administrator had an extra resource, i.e., $R+1$ resources, then Algorithm~\ref{alg:GreedyRevMaxProb} achieves at least the same revenue as the revenue maximizing outcome with $R$ resources.

\vspace{-2pt}

\begin{theorem} [Resource Augmentation Guarantee for Heterogeneous Revenue Maximization] \label{thm:greedy-resource-augmentation-rev-max}
Suppose $|\I| \geq 1$, $\sigmaa^*_A$ is the solution corresponding to Algorithm~\ref{alg:GreedyRevMaxProb} with $R+1$ resources, and $\sigmaa^*$ is the revenue maximizing allocation that solves Problem~\eqref{eq:admin-obj-revenue}-\eqref{eq:bi-level-con-revenue} with $R$ resources. Then, the total revenue under the allocation $\sigmaa^*_A$ is at least that corresponding to $\sigmaa^*$, i.e., $Q_{R+1}(\sigmaa^*_A) \geq Q_{R}(\sigmaa^*)$.
\end{theorem}
\vspace{-2pt}

The proof of Theorem~\ref{thm:greedy-resource-augmentation-rev-max}, as with that of Theorem~\ref{thm:greedy-half-approx-rev-max}, relies on the fact that a greedy-like process akin to step one of Algorithm~\ref{alg:GreedyRevMaxProb} optimizes the MCUA of the revenue function. For a complete proof of Theorem~\ref{thm:greedy-resource-augmentation-rev-max}, see Appendix~\ref{apdx:pf-resource-aug-rev-max}. Akin to Theorem~\ref{thm:greedy-resource-augmentation-frauad-min}, we note that the administrator only requires $R+ \max_{i \in \I} \max_{l \in L} \frac{d_l^i}{d_l^i+k}$ resources to obtain the guarantee in Theorem~\ref{thm:greedy-resource-augmentation-rev-max}; however, we present the result with $R+1$ resources for ease of exposition. Theorem~\ref{thm:greedy-resource-augmentation-rev-max} highlights the benefit to administrators for recruiting one additional security resource (e.g., police officer) and applying a simple algorithm, i.e., Algorithm~\ref{alg:GreedyRevMaxProb}, that relies on computing a tractable MCUA of the revenue function rather than investing computational effort to solve the NP-hard Problem~\eqref{eq:admin-obj-revenue}-\eqref{eq:bi-level-con-revenue}.

While Algorithm~\ref{alg:GreedyRevMaxProb} achieves a half approximation with no additional resources (Theorem~\ref{thm:greedy-half-approx-rev-max}) and a one approximation with one additional resource (Theorem~\ref{thm:greedy-resource-augmentation-rev-max}), a PTAS can also be developed for the heterogeneous revenue maximization setting, as in the homogeneous payoff maximization setting (see Section~\ref{sec:ptas-payoff-maximization}), if the number of user types is a constant. This PTAS is analogous to Algorithm~\ref{alg:PTASWelMaxHom} and, given parameters $m$ specifying the size of the feasible subset and a discretization parameter $\delta$, its running time is $O(\frac{m|\I|}{\delta} |L|^{m|\I|+2})$. The multiplicative factor of $|\I|$ and the additional factor of $|\I|$ in the exponent of the running time stems from the fact that when constructing feasible allocations in the brute-force step of the algorithm, all combinations of allocations $\frac{d_l^i}{d_l^i+k}$ for all types $i \in \I$ must be considered (recall that in the homogeneous payoff maximization setting, the number of user types was one, i.e., $|\I| = 1$; thus, the term $|\I|$ did not enter the runtime calculations). Since the PTAS for heterogeneous revenue maximization (in the setting when the type space is a constant) and its corresponding approximation guarantee are analogous to that in the homogeneous payoff maximization setting, we omit the details for brevity and leave the question of settling whether a PTAS exists when the type space is not constant to future research.

%In particular, given parameters $m$ specifying the size of the feasible subset and a discretization parameter $\delta$, the running time of the corresponding algorithm is $O(\frac{m|\I|}{\delta} |L|^{m|\I|+2})$. The multiplicative factor of $|\I|$ and the additional factor of $|\I|$ in the exponent of the running time stems from the fact that when constructing feasible allocations in the brute-force step of the algorithm, all combinations of allocations $\frac{d_l^i}{d_l^i+k}$ for all types $i \in \I$ must be considered (recall that in the homogeneous payoff maximization setting, the number of user types was one, i.e., $|\I| = 1$; thus, the term $|\I|$ did not enter the runtime calculations). Since the PTAS for heterogeneous revenue maximization (in the setting when the type space is a constant) and its corresponding approximation guarantee are analogous to that in the homogeneous payoff maximization setting, we omit the details for brevity and leave the question of settling whether a PTAS exists when the number of types is not constant to future research.

\ifarxiv
\section{Numerical Experiments: Parking Enforcement at a University Campus} \label{sec:numerical-experiments-parking-enforcement}
\else
\section{Numerical Experiments} \label{sec:numerical-experiments-parking-enforcement}
\fi

This section presents experiments based on a real-world case study of parking enforcement at Stanford University's campus. Our results highlight the efficacy of our proposed algorithms relative to the \emph{status-quo} parking enforcement mechanism and a uniform random enforcement benchmark that allocates resources with equal probability across all university parking lots. Moreover, our results demonstrate the power of our algorithms and studied security game framework as the proportion of strategic users that maximize their utilities according to Problem~\eqref{eq:userOpt-eachLoc} in the population increases. All the code for our experiments is available at the following \href{https://github.com/djalota/SecurityGames}{link}.

In this section, we first elucidate the university's parking enforcement mechanism and the associated data set (Section~\ref{subsec:data-setup}). Then, in Section~\ref{subsec:assumptions-counterfactuals}, we calibrate the model parameters of our security game and present counterfactuals for different models of user behavior. Finally, we present our results (Section~\ref{subsec:results-parking}).

%This section evaluates the performance of our proposed greedy algorithms in our studied security game through numerical experiments based on a real-world case study of parking enforcement at a university campus. We omit the details of the university for anonymization purposes. Our results highlight the efficacy of our proposed algorithms relative to the \emph{status-quo} parking enforcement mechanism in place at the university and a uniform random enforcement benchmark that allocates resources with equal probability across all university parking lots. Moreover, our results demonstrate the power of our algorithms and studied security game framework as the proportion of strategic users that maximize their utilities according to Problem~\eqref{eq:userOpt-eachLoc} in the population increases.

%In this section, we first present an overview of the university's parking enforcement mechanism and the associated data set (Section~\ref{subsec:data-setup}). Then, in Section~\ref{subsec:assumptions-counterfactuals}, we present our assumptions and methodology to calibrate the model parameters of our security game and counterfactuals corresponding to different models of user behavior that we use to test the efficacy of our algorithms to the status-quo parking enforcement mechanism. Finally, we present our results in Section~\ref{subsec:results-parking}.

\subsection{Parking Enforcement Setup and Data} \label{subsec:data-setup}

%We consider a real-world case study of enforcing parking regulations at a university campus for our numerical experiments. 

At Stanford University, commuters must purchase one of several permits, where the permits are heterogeneous with different costs (see Table~\ref{tab:permit-cost} in Appendix~\ref{apdx:costs-parking-permit-types}), depending on the parking spot they are seeking to avail. To enforce parking regulations, the university's Department of Public Safety deploys officers to different parts of the university's campus, where the officers use license plate detection cameras that enable quick and accurate detection of parking violations. If the officer finds a user has violated parking regulations, e.g., not purchased a parking permit, a parking citation with a fine of $k = \$45$ is issued to the user. Crucially, the university accrues permit earnings, while all citation fees go to Santa Clara county's sheriff's department. % in the county where the university is located.

For our study, we obtained seven months of parking enforcement and citation information between September 2022 and March 2023 for nine representative parking lots on the university's campus. The enforcement data includes the time, day, and duration of enforcement at each of the lots while the citation data includes the total number of citations issued in each parking lot every month. Moreover, we obtained data on the number of parking spaces in each of the parking lots for the permit types listed in Table~\ref{tab:permit-cost} in Appendix~\ref{apdx:costs-parking-permit-types}.

%For our experimental study, we obtained data from the university's Department of Public Safety, the entity in charge of enforcing parking regulations on the university's campus. The data set includes seven months of parking enforcement and citation information between September 2022 and March 2023 for nine representative parking lots on the university's campus. \ifarxiv The identity of the parking lots was anonymized in the obtained data set for privacy reasons. \fi The parking enforcement data provides information on the time, day, and duration of the parking enforcement activity at each of the lots during the seven-month period. The citation data includes the total number of citations issued in each parking lot in a given month. Moreover, the university's Department of Public Safety provided data on the number of parking spaces in each of the nine parking lots by permit type for the six permit types listed in Table~\ref{tab:permit-cost}.

\begin{comment}
\begin{table}[]
\caption{\small \sf Costs of different parking permit types per day at university campus for parking enforcement case-study.}
\centering
\small
\begin{tabular}[b]{c|c} \toprule
           Permit Type  & Permit Cost (\$ per day) \\ \midrule
A            & 6.65                        \\
C            & 1.23                         \\
Resident     & 1.50                         \\
Resident/C   & 1.38                       \\
Visitor      & 35.68                        \\
Other Permit & 22.40  \\ \bottomrule      
    \end{tabular} \label{tab:permit-cost}
\end{table}    
\end{comment}

\subsection{Model Calibration and Counterfactuals} \label{subsec:assumptions-counterfactuals}

In this section, we describe the assumptions and methodology used to calibrate our model parameters and counterfactuals corresponding to different models of user behavior that we use to test the performance of our algorithms to the status-quo parking enforcement mechanism.

\vspace{-8pt}

\paragraph{Modelling Assumptions:}

To frame the parking enforcement problem as a security game, we define the location set $L$ as the set of parking lots on the university's campus, i.e., the nine parking lots in our data set, and assume that commuters pay a fine of $k = \$45$ for all detected parking violations, which corresponds to the fine of the majority of the issued parking citations. Moreover, since the university only earns through parking permit purchases, while citation fees (i.e., revenues from enforcement) go to the county's sheriff's department, we model the university as a payoff-maximizer, where payoffs represent permit earnings. To model the payoffs, let $n_l^j$ be the number of parking spaces of permit type $j$ at location $l$ and $f_j$ be the fee for permit type $j$. Then, the payoff (i.e., permit earnings) from users availing permit type $j$ at location $l$ is $p_l^j = n_l^j f_j$.

Furthermore, since no granular information on commuter movements or the usage of the parking spaces is available, we assume that all parking slots are used by commuters every day, with commuters parking at the same lot for the entire day. We make such an assumption for simplicity and that most commuters include students, staff, and faculty, who typically park their vehicles at the same location for most of the day. Further, we assume that patrol officers visit a parking lot once a day, corresponding to most of the enforcement activity performed by Stanford University's Department of Public Safety. Finally, using the available enforcement data, we compute the status-quo probability of allocating a patrol officer to a parking lot as the proportion of days over the seven-month horizon that enforcement activity took place at that parking lot. We also assume that the total number of security resources that Stanford university can allocate to enforce parking regulations is the sum of the probabilities of allocating the patrol officers across the different parking lots under the status-quo enforcement mechanism.

%Since the university only earns permit revenues, while citation fees are accrued by the county's sheriff's department, we model the university as a payoff-maximizer, where the payoffs correspond to the permit revenues. We reiterate that permit revenues differ from the collected citation fees; hence, maximizing permit revenues corresponds to maximizing the administrator's payoff rather than its enforcement revenue from the collected citation fees. To model the administrator's payoffs, we let $n_l^j$ be the number of parking spaces of permit type $j$ at location $l$ and $f_j$ be the fee associated with permit type $j$, where $\mathcal{J}$ denotes the set of parking permit types. Then, the administrator payoff (i.e., permit revenue) from users availing permit type $j$ at location $l$ is $p_l^j = n_l^j f_j$.

\vspace{-8pt}

\paragraph{Counterfactuals:} We study two counterfactuals to simulate different models of user behavior and calibrate the threshold allocation probabilities of security resources at which users of different types shift from not purchasing a permit to purchasing one. %To elucidate our counterfactuals, we let $n_j^l$ be the number of parking spaces of permit type $j$ at location $l$ and $f_j$ corresponds to the fee associated with permit type $j$, where $\mathcal{J}$ denotes the set of parking permit types.

\textbf{Counterfactual 1:} We assume users belong to one of two classes: \emph{strategic} and \emph{non-strategic}. While non-strategic users always buy permits regardless of the enforcement mechanism, strategic users maximize utilities according to Problem~\eqref{eq:userOpt-eachLoc}. Recall from Section~\ref{subsec:payoffs} that in the parking enforcement context, the parameter $\beta = 0$ in the utility maximization Problem~\eqref{eq:userOpt-eachLoc} of users; hence, strategic users will purchase a parking permit only if the allocation probability of security resources to the given parking lot is above a threshold specific to each permit type. In particular, strategic users will buy a permit of type $j$ if the probability of allocating a security resource to a location exceeds a specific threshold; otherwise, they will not buy a permit, as depicted on the left of Figure~\ref{fig:counterfactual-thresholds}. Note that under this counterfactual, the number of user types $|\I| = 2 |\J|$, as users of each of the $|\J|$ permit types are strategic or non-strategic.

\textbf{Counterfactual 2:} We assume that the probability users do not buy permits as a function of the allocation probability of security resources is given by an exponential distribution for each location and permit type pair $(l, j)$.\footnote{Our choice of an exponential distribution serves as a proxy to capture the fact that, in practice, few users are risk-seeking while most are risk-averse.} We learn the parameter $\gamma_l^j$ of the exponential distribution for each pair $(l, j)$ using data corresponding to the number of citations by permit type for each parking lot and the corresponding status-quo enforcement probabilities across locations, as depicted on the right of Figure~\ref{fig:counterfactual-thresholds}.\footnote{Our data set only contains information on the total number of citations issued in a parking lot in each month across the seven-month horizon. Given that more granular information on the distribution of citations across permit types was unavailable, we assumed that the total number of citations across permit types was distributed in proportion to the number of parking spots of the different permit types in each parking lot.} Notice that a single data point of the proportion of citations (denoted $c_l^{j, SQ}$ for the pair $(l, j)$) given the status-quo probability of allocating resources to each parking lot (denoted $\sigma_l^{SQ}$ for location $l$) is sufficient to learn the parameter $\gamma_l^j$ of the exponential distribution (see right of Figure~\ref{fig:counterfactual-thresholds}. 

Having calibrated the exponential distribution parameter $\gamma_l^j$ for each location $l$ and permit type $j$, since we consider a finite set of user types, we discretize the exponential distribution with a step-size of $0.01$, where each user's allocation probability threshold at which they shift from not purchasing a permit to purchasing one is assumed to be the mid-point of the corresponding interval. Moreover, the number of users at each location for each permit type belonging to that interval is specified based on the cumulative distribution function (CDF) of the associated exponential distribution. For instance, for a location $l$ and permit type $j$, an interval $[z_1, z_2]$ and an exponential distribution with CDF $F_j^l(\cdot)$, the number of users with threshold allocation probability of $\frac{z_1+z_2}{2}$ is $n_j^l (F_j^l(z_2) - F_j^l(z_1))$.

Under our chosen discretization, for each location $l$ and permit type $j$, there are 101 user types, where 100 user types have an allocation probability threshold in the set $\{ 0.005, 0.015, \ldots, 0.995\}$ and one user type has an allocation probability threshold above one. All users with an allocation probability threshold above one will never purchase permits regardless of the allocation of security resources. Such a setting can, for instance, model users making \emph{honest mistakes}, as these users violate parking regulations regardless of the enforcement strategy of the administrator, even when it is optimal to purchase permits, e.g., when the administrator allocates security resources to a location with probability one. %Based on these allocation probability thresholds, as with the first counterfactual, we can analogously model the type $\Theta_l^i$ of users in our security game framework for the 101 user types for each location and permit type combination.

\begin{figure}[tbh!]
    \centering
    \includegraphics[width=0.81\linewidth]{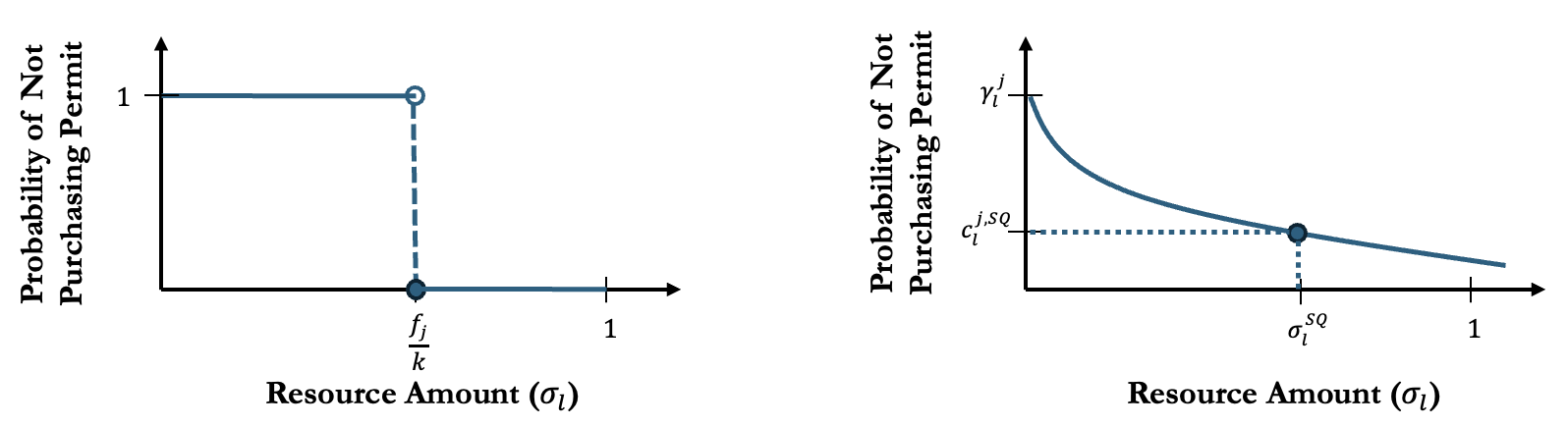}
    \vspace{-13pt}
    \caption{{\small \sf Depiction of the probability of not purchasing a parking permit of type $j$ as a function of the probability of allocating resources to a given location $l$ under counterfactuals one (left) and two (right). In counterfactual one, all strategic users have the same threshold probability $\frac{f_j}{k}$ at which they shift from not purchasing a permit to purchasing one. In counterfactual two, the probability of not purchasing a permit of type $j$ at location $l$ is modeled as an exponential distribution with parameter $\gamma_l^j$. The point $(\sigma_l^{SQ}, c_l^{j, SQ})$ on the exponential distribution corresponds to the fraction of the total parking spots of permit type $j$ at location $l$ that were issued citations, given by $c_l^{j, SQ}$, under the status-quo enforcement probability $\sigma_l^{SQ}$ at location $l$. %The point depicted on the exponential distribution corresponds to the tuple $(\sigma_l^{SQ}, c_l^{j, SQ})$, where $\sigma_l^{SQ}$ is the status-quo probability of allocating security resources to location $l$ and $c_l^{j, SQ}$ is the fraction of the total parking spots of permit type $j$ at location $l$ that were issued citations under the status-quo enforcement mechanism. Note that the point $(\sigma_l^{SQ}, c_l^{j, SQ})$ is sufficient to calibrate the parameter $\gamma_l^j$ of the exponential distribution. 
    }} 
    \label{fig:counterfactual-thresholds}
\end{figure} 

%we assume that the probability users do not purchase permits as a function of the allocation probabilities is specified by an exponential distribution with parameter $\lambda_l^j$ for each location $l$ and permit type $j$.

%To determine these risk thresholds for different user types, we use data corresponding to the number of citations by permit type for each parking lot and the corresponding status-quo security resource allocation probability. In particular, assume that the probability users do not purchase permits as a function of the allocation probabilities is specified by an exponential distribution and learn the parameter $\lambda_l^j$ of this exponential distribution for each location $l$ and permit type $j$ using 

\subsection{Results} \label{subsec:results-parking}

This section compares the performance of Algorithm~\ref{alg:GreedyWelMaxProb} in the heterogeneous payoff maximization setting (see Appendix~\ref{apdx:welfare-max-prob-setting}) in terms of the resulting parking permit earnings (i.e., the administrator's payoff) to several benchmarks, including the status quo parking enforcement mechanism, under the user behavior models in both counterfactuals one and two. For additional tables comparing the permit earnings achieved by Algorithm~\ref{alg:GreedyWelMaxProb} to that achieved by the status quo enforcement mechanism under counterfactuals one and two, we refer to Appendix~\ref{apdx:additional-numerics-parking}.

%In Appendix~\ref{apdx:additional-numerics-parking}, we present additional tables comparing the permit revenues achieved by Algorithm~\ref{alg:GreedyWelMaxProb} to that achieved by the status quo parking enforcement mechanism as the proportion of strategic users is varied in counterfactual one and as the citation multiplier is varied in counterfactual two.

%the status quo parking enforcement mechanism and a uniform random enforcement benchmark (that allocates available resources with equal probability to all parking lots) under the user behavior models in both counterfactuals one and two.

\vspace{-5pt}

\paragraph{Counterfactual 1:} Figure~\ref{fig:frac-permit-revenue-under-counterfactuals} (left) depicts the fraction of the total permit earnings accrued by the administrator as a function of the proportion of strategic users under counterfactual one for four parking enforcement mechanisms: (i) the status quo enforcement mechanism, (ii) the strategy computed using Algorithm~\ref{alg:GreedyWelMaxProb} under counterfactual one, (iii) a uniform random enforcement benchmark that allocates resources with equal probability to all parking lots, and (iv) a \emph{No Enforcement} benchmark, which allocates no resources.

%The left of Figure~\ref{fig:frac-permit-revenue-under-counterfactuals} depicts the fraction of the total permit revenues accrued by the administrator as a function of the proportion of strategic users under counterfactual one for four enforcement mechanisms: (i) the status quo enforcement mechanism currently in place at the university, (ii) the enforcement mechanism computed using Algorithm~\ref{alg:GreedyWelMaxProb} given the user behavior model in counterfactual one, (iii) a uniform random enforcement benchmark that allocates the available resources with equal probability to all parking lots, and (iv) a \emph{No Enforcement} benchmark, which corresponds to not allocating security resources to any locations.

Figure~\ref{fig:frac-permit-revenue-under-counterfactuals} (left) highlights that Algorithm~\ref{alg:GreedyWelMaxProb} achieves significant gains in permit earnings relative to the status quo mechanism, resulting in nearly twice the permit earnings when all users are strategic. Moreover, the permit earnings of all enforcement policies reduce as the proportion of strategic users increases, which is natural as fewer non-strategic users implies that less users purchase permits regardless of the probability with which resources are allocated. Despite this inverse relationship, Algorithm~\ref{alg:GreedyWelMaxProb} achieves about twice the increase in permit earnings relative to the No Enforcement benchmark as compared to that achieved by the status-quo mechanism regardless of the proportion of strategic users. Moreover, the absolute gains in permit earnings achieved using Algorithm~\ref{alg:GreedyWelMaxProb} relative to the status quo mechanism increases as the proportion of strategic users increases, which is natural as our security game framework's power stems from being able to capture the strategic behavior of utility-maximizing users, as given by Problem~\eqref{eq:userOpt-eachLoc}. Finally, Figure~\ref{fig:frac-permit-revenue-under-counterfactuals} (left) also demonstrates that under counterfactual one, the status-quo mechanism provides little gains in permit earnings relative to the uniform random benchmark.

\paragraph{Counterfactual 2:}

Figure~\ref{fig:frac-permit-revenue-under-counterfactuals} (right) depicts the fraction of total permit earnings as a function of the citation multiplier, a proxy for the proportion of strategic users in counterfactual one, for four mechanisms: (i) the status quo mechanism, (ii) the strategy computed using Algorithm~\ref{alg:GreedyWelMaxProb} under counterfactual two (labeled Algorithm~\ref{alg:GreedyWelMaxProb} (C2)), (iii) a uniform random benchmark, and (iv) the strategy computed using Algorithm~\ref{alg:GreedyWelMaxProb} under counterfactual one when all users are strategic (labeled Algorithm~\ref{alg:GreedyWelMaxProb} (C1)). In our experiments, we vary the total number of citations under the status quo resource allocation probabilities via a citation multiplier $\mu_l^j$. For each multiplier $\mu_l^j$, we calibrate the parameter $\gamma_l^j$ of the associated exponential distribution consisting of the point $(\sigma_l^{SQ}, \mu_l^j c_l^{j, SQ})$ for each location $l$ and permit type $j$. Notice that a citation multiplier of one corresponds to the status-quo outcome and a multiplier greater than one is akin to increasing the proportion of strategic users.

%The right of Figure~\ref{fig:frac-permit-revenue-under-counterfactuals} depicts the fraction of total permit revenues accrued by an administrator as a function of the citation multiplier, a proxy for the proportion of strategic users in counterfactual one, for four enforcement mechanisms: (i) the status quo enforcement mechanism, (ii) the enforcement mechanism computed using Algorithm~\ref{alg:GreedyWelMaxProb} calibrated based on the user behavior model in counterfactual two (labeled Algorithm~\ref{alg:GreedyWelMaxProb} (C2), corresponding to the second counterfactual), (iii) a uniform random enforcement benchmark, and (iv) the enforcement mechanism computed using Algorithm~\ref{alg:GreedyWelMaxProb} calibrated based on the user behavior model in counterfactual one when all users are assumed to be strategic (labeled Algorithm~\ref{alg:GreedyWelMaxProb} (C1)). In our experiments, we vary the total number of citations under the status quo resource allocation probabilities via a citation multiplier $\mu_l^j$. For each citation multiplier $\mu_l^j$, we calibrate the parameter $\gamma_l^j$ of the corresponding exponential distribution that consists of the point $(\sigma_l^{SQ}, \mu_l^j c_l^{j, SQ})$ for each location $l$ and permit type $j$. Notice that a citation multiplier of one corresponds to the status-quo outcome while a citation multiplier greater than one is analogous to increasing the proportion of strategic users.

In the right of Figure~\ref{fig:frac-permit-revenue-under-counterfactuals}, we do not depict a No Enforcement benchmark, as we did under counterfactual one, as all users in counterfactual two have some strictly positive threshold allocation probability at which they shift from not purchasing a permit to purchasing one (see Section~\ref{subsec:assumptions-counterfactuals}). Consequently, the permit revenues corresponding to a No Enforcement benchmark are zero under counterfactual two (regardless of the citation multiplier). Thus, instead, we use the enforcement policy computed using Algorithm~\ref{alg:GreedyWelMaxProb} under the user behavior model in the first counterfactual where all users are assumed to be strategic as a benchmark.  %Our purpose behind considering this benchmark is two-fold, which we elucidate in the following discussion. % of the results in the right of Figure~\ref{fig:frac-permit-revenue-under-counterfactuals}.

%Figure~\ref{fig:frac-permit-revenue-under-counterfactuals} depicts the fraction of the total permit revenues achieved by Algorithm~\ref{alg:GreedyWelMaxProb}, the status quo parking enforcement mechanism and a uniform random enforcement benchmark as the proportion of strategic users is varied in counterfactual one (left), and the citation multiplier is varied in counterfactual two (right).

%Moreover, we also depict the performance of a ``No Enforcement'' benchmark, which corresponds to not allocating security resources to any locations, under counterfactual one. 
%Notice that unlike Algorithm~\ref{alg:GreedyWelMaxProb} in counterfactual two, this enforcement policy does not have access to information on the total number of citations associated with each permit type under the status-quo enforcement mechanism; hence, we term this policy as Algorithm~\ref{alg:GreedyWelMaxProb} (no citation information) to distinguish it from the corresponding enforcement policy learned using citation information.

From the right of Figure~\ref{fig:frac-permit-revenue-under-counterfactuals}, we first note that Algorithm~\ref{alg:GreedyWelMaxProb} (C2) achieves higher permit earnings compared to the status quo policy for all citation multipliers. Moreover, the absolute gains in permit earnings of Algorithm~\ref{alg:GreedyWelMaxProb} (C2) relative to the status-quo increases with the citation multiplier, aligning with the corresponding result in the left of Figure~\ref{fig:frac-permit-revenue-under-counterfactuals}. Further, at the status-quo outcome (i.e., the citation multiplier is one), Algorithm~\ref{alg:GreedyWelMaxProb} (C2) achieves 98.7\% of the total permit earnings, which corresponds to about a \$300,000 (a 2\%) annual increase in the permit earnings relative to the status-quo policy under counterfactual two.

We also observe from the right of Figure~\ref{fig:frac-permit-revenue-under-counterfactuals} that unlike counterfactual one, under counterfactual two, the status-quo policy achieves significantly higher permit earnings compared to the uniform random benchmark for all citation multipliers, as unlike counterfactual one, counterfactual two leverages observed citation counts under the status-quo policy to calibrate users' threshold allocation probabilities at which they shift from not purchasing to purchasing permits. Figure~\ref{fig:frac-permit-revenue-under-counterfactuals} (right) also demonstrates that while the status-quo policy achieves higher permit earnings than Algorithm~\ref{alg:GreedyWelMaxProb} (C1) for low citation multipliers, as the multiplier increases (akin to a higher proportion of strategic users under counterfactual one), Algorithm~\ref{alg:GreedyWelMaxProb} (C1), which does not have access to observed citation counts, outperforms the status-quo policy. Moreover, the difference in the permit earnings between Algorithm~\ref{alg:GreedyWelMaxProb} (C1) and Algorithm~\ref{alg:GreedyWelMaxProb} (C2) represents the \emph{benefit of information}, as Algorithm~\ref{alg:GreedyWelMaxProb} (C2), unlike Algorithm~\ref{alg:GreedyWelMaxProb} (C1), has access to the citation counts under the status-quo policy based on past enforcement and citation data.

Overall, our results demonstrate that Algorithm~\ref{alg:GreedyWelMaxProb} outperforms the status-quo parking enforcement mechanism, where the gains in the permit earnings corresponding to Algorithm~\ref{alg:GreedyWelMaxProb} are particularly pronounced as the proportion of strategic users increases.

%Overall, our results demonstrate that Algorithm~\ref{alg:GreedyWelMaxProb} outperforms the status-quo enforcement mechanism currently in place at the university, where the gains in the parking permit revenues corresponding to Algorithm~\ref{alg:GreedyWelMaxProb} compared to the status-quo policy are particularly pronounced as the proportion of strategic users in the population increases.

%In particular, our results highlight the efficacy of our studied security game framework and  Algorithm~\ref{alg:GreedyWelMaxProb} 

%, where the permit revenue gains of Algorithm~\ref{alg:GreedyWelMaxProb} compared to the status-quo policy increases as the number of users behaving strategically increases.

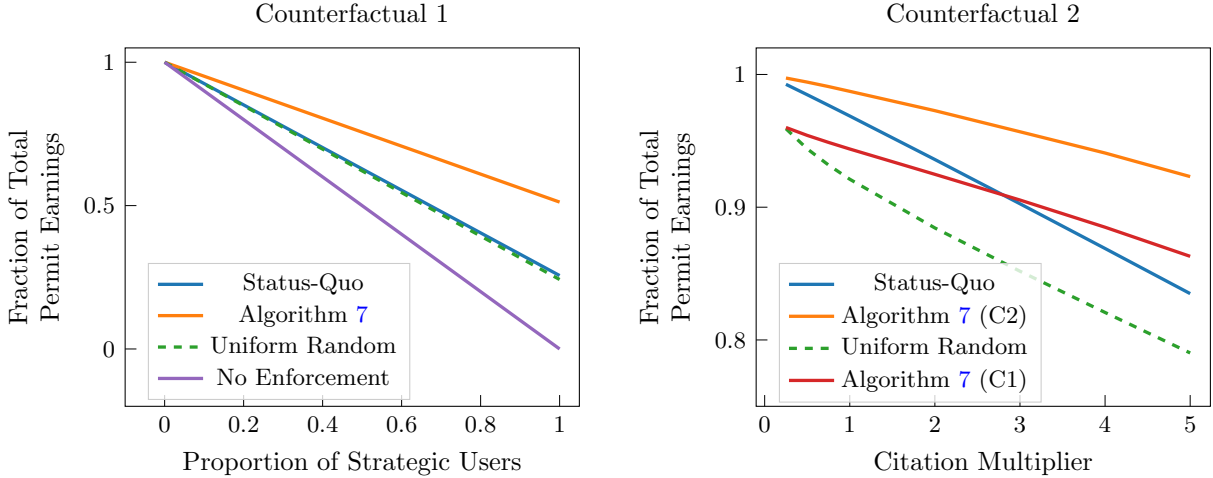
\begin{figure}
    \centering \hspace{-25pt}
\begin{subfigure}[t] {0.45\columnwidth}
    % This file was created with tikzplotlib v0.10.1.
\begin{tikzpicture}

\definecolor{crimson2143940}{RGB}{214,39,40}
\definecolor{darkgray176}{RGB}{176,176,176}
\definecolor{darkorange25512714}{RGB}{255,127,14}
\definecolor{forestgreen4416044}{RGB}{44,160,44}
\definecolor{steelblue31119180}{RGB}{31,119,180}
\definecolor{mediumpurple148103189}{RGB}{148,103,189}

\begin{axis}[
width=3in,
height=2.5in,
tick pos=left,
x grid style={darkgray176},
xlabel={Proportion of Strategic Users},
xmin=-0.1, xmax=1.05,
xtick style={color=black},
y grid style={darkgray176},
ylabel={Fraction of Total \\ Permit Earnings},
ylabel style = {align = center, font = \small},
ymin=-0.2, ymax=1.05,
ytick style={color=black},
xlabel style={font=\color{white!15!black}, font = \small},
tick label style={font=\footnotesize},
title = {Counterfactual 1},
title style = {font = \small},
legend style={
  fill opacity=0.8,
  draw opacity=1,
  text opacity=1,
  at={(0.05,0.4)},
  anchor=north west,
  draw=white!80!black,
  font = {\footnotesize}
}
]
\addplot [line width = 1.25pt, steelblue31119180]
table {%
0 1
0.1 0.925604310943143
0.2 0.851208621886286
0.3 0.776812932829429
0.4 0.702417243772572
0.5 0.628021554715715
0.6 0.553625865658858
0.7 0.479230176602001
0.8 0.404834487545144
0.9 0.330438798488287
1 0.25604310943143
};
\addlegendentry{Status-Quo}
\addplot [line width = 1.25pt, darkorange25512714]
table {%
0 1
0.1 0.951227266822739
0.2 0.902454533645479
0.3 0.853681800468218
0.4 0.804909067290957
0.5 0.756136334113697
0.6 0.707363600936436
0.7 0.658590867759175
0.8 0.609818134581914
0.9 0.561045401404654
1 0.512272668227393
};
\addlegendentry{Algorithm~\ref{alg:GreedyWelMaxProb}}
\addplot [line width = 1.25pt, dashed, forestgreen4416044]
table {%
0 1
0.1 0.924227079986956
0.2 0.848454159973911
0.3 0.772681239960867
0.4 0.696908319947822
0.5 0.621135399934778
0.6 0.545362479921733
0.7 0.469589559908689
0.8 0.393816639895644
0.9 0.3180437198826
1 0.242270799869555
};
\addlegendentry{Uniform Random}
%\begin{comment}
\addplot [line width = 1.25pt, mediumpurple148103189]
table {%
0 1
0.1 0.9
0.2 0.8
0.3 0.7
0.4 0.6
0.5 0.5
0.6 0.4
0.7 0.3
0.8 0.2
0.9 0.1
1 0
};
\addlegendentry{No Enforcement}
%\end{comment}
\end{axis}

\end{tikzpicture}
\end{subfigure} \hspace{20pt} 
\begin{subfigure}[t] {0.45\columnwidth} 
    % This file was created with tikzplotlib v0.10.1.
\begin{tikzpicture}

\definecolor{crimson2143940}{RGB}{214,39,40}
\definecolor{darkgray176}{RGB}{176,176,176}
\definecolor{darkorange25512714}{RGB}{255,127,14}
\definecolor{forestgreen4416044}{RGB}{44,160,44}
\definecolor{steelblue31119180}{RGB}{31,119,180}

\begin{axis}[
width=3in,
height=2.5in,
tick pos=left,
x grid style={darkgray176},
xlabel={Citation Multiplier},
xmin=-0.1, xmax=5.2375,
xtick style={color=black},
y grid style={darkgray176},
ylabel={Fraction of Total \\ Permit Earnings},
ylabel style = {align = center, font = \small},
ymin=0.779848751938878, ymax=1.00766437581389,
ytick style={color=black},
xlabel style={font=\color{white!15!black}, font = \small},
tick label style={font=\footnotesize},
title = {Counterfactual 2},
title style = {font = \small},
ymin=0.75, ymax=1.02,
legend style={
  fill opacity=0.8,
  draw opacity=1,
  text opacity=1,
  at={(0.05,0.4)},
  anchor=north west,
  draw=white!80!black,
  font = {\footnotesize}
}
]
\addplot [line width = 1.25pt, steelblue31119180]
table {%
0.25 0.992574880256405
0.5 0.984730968522417
0.75 0.976758534133514
1 0.968700512610591
2 0.93592683064519
3 0.902595927067722
4 0.868891129786381
5 0.834898966880436
};
\addlegendentry{Status-Quo}
\addplot [line width = 1.25pt, darkorange25512714]
table {%
0.25 0.997309120183209
0.5 0.994243391000487
0.75 0.990892106545871
1 0.987292417925943
2 0.972728824478167
3 0.956835976505668
4 0.94082255230015
5 0.923056258451604
};
\addlegendentry{Algorithm~\ref{alg:GreedyWelMaxProb} (C2)}
\addplot [line width = 1.25pt, dashed, forestgreen4416044]
table {%
0.25 0.959140022103517
0.5 0.943959174320426
0.75 0.931796310956706
1 0.921066981622439
2 0.884371012188524
3 0.851777169575558
4 0.820698756553349
5 0.79020400756956
};
\addlegendentry{Uniform Random}
\addplot [line width = 1.25pt, crimson2143940]
table {%
0.25 0.960056570013104
0.5 0.954102169341155
0.75 0.948821244868306
1 0.943838732632641
2 0.924765735709355
3 0.905396687759872
4 0.884918275529635
5 0.862982452474343
};
\addlegendentry{Algorithm~\ref{alg:GreedyWelMaxProb} (C1)}
\end{axis}

\end{tikzpicture}
\end{subfigure} 
\vspace{-25pt}
    \caption{\small \sf Comparison of the fraction of the total permit earnings achieved by the status quo enforcement mechanism and a uniform random enforcement benchmark to that achieved by Algorithm~\ref{alg:GreedyWelMaxProb} as the proportion of strategic users is varied in counterfactual one (left) and as the citation multiplier is varied in counterfactual two (right). For counterfactual one, we also depict the permit earnings corresponding to a No Enforcement benchmark, wherein no security resources are allocated to any locations. For counterfactual two, we depict the performance of the enforcement mechanism computed using Algorithm~\ref{alg:GreedyWelMaxProb} calibrated based on the user behavior model in counterfactual two (Algorithm~\ref{alg:GreedyWelMaxProb} (C2)) and another calibrated based on the user behavior model in counterfactual one where all users are assumed to be strategic (Algorithm~\ref{alg:GreedyWelMaxProb} (C1)).}
    \label{fig:frac-permit-revenue-under-counterfactuals}
\end{figure}

\section{Extension I: Contracts to Bridge Revenue and Payoff Maximization} \label{sec:optimal-contracts}

\vspace{-2pt}

Thus far, we have studied our security game under the revenue and payoff maximization administrator objectives. While a payoff maximization objective captures an idealized representation of administrator behavior, in practice, an administrator often seeks to maximize revenues through enforcement fines, which may significantly compromise the resulting payoffs (see Appendix~\ref{apdx:additional-numerics} for experiments depicting the contrast in the revenue and payoff maximization outcomes). For instance, in road traffic scenarios, police often place speed traps to collect fines and increase enforcement revenues at locations where users are likely to violate the speed limit even though other locations may be more accident-prone~\cite{prensky_johnson_2023}.

To address this concern of the gulf between the revenue and payoff maximization outcomes, in this section, we extend our security game framework to incorporate \emph{contracts}, wherein a revenue-maximizing administrator is compensated for the payoff it contributes to the system. Section~\ref{sec:optimal-contract-framework} introduces the contract framework and studies equilibria in the associated contract game. Then, Section~\ref{sec:numerical-experiments-optimal-contracts} presents numerical experiments based on a case study of queue jumping in IPT services that highlight the efficacy of contracts in bridging the gap between the payoff and revenue maximization outcomes. For ease of exposition, we present the results for the setting with homogeneous user types at each location and note that following ideas similar to Section~\ref{sec:probabilistic-setting}, our results naturally extend to the setting with heterogeneous user types. Thus, in this section, for brevity of notation, we drop the superscript $i$ in the notation of user types and denote the type of all users at a location $l$ by the triple $\Theta_l = (\Lambda_l, d_l, p_l)$.

%\vspace{-5pt}

\subsection{Contract Framework} \label{sec:optimal-contract-framework}

\vspace{-2pt}

We consider a contract game between three players: (i) a payoff-maximizing \emph{principal}, (ii) a revenue-maximizing administrator, and (iii) fraudulent users at different locations in a system.\footnote{To distinguish between a principal and an administrator, we can, for instance, interpret the principal as the government maximizing payoffs and the administrator as a police department within the government maximizing its own revenues.} In our contract game, a payoff-maximizing principal offers a contract, specified by a parameter $\alpha \in [0, 1]$, to a revenue-maximizing administrator, where the contract level determines the payment made by the principal to the administrator as a proportion of the total payoff it contributes to the system. In particular, the principal selects the contract level $\alpha \in [0, 1]$ to optimize its objective, given by the difference between the payoff accrued and the total payments it makes to the administrator, to which the administrator best responds by choosing a revenue-maximizing strategy $\sigmaa(\alpha)$, which can be computed via the following bi-level program: \vspace{-4pt}
\begin{maxi!}|s|[2]   
    {\substack{\sigmaa \in \Omega_R \\ y_l(\sigmaa) \in [0, 1], \forall l \in L}}                            
    { Q_R(\sigmaa) + \alpha P_R(\sigmaa),  \label{eq:obj-contract-alpha}}   
    {\label{eq:Eg00333}}             
    {}          
    \addConstraint{y_l(\sigmaa)}{\in \argmax_{y \in [0, 1]} U_l(\sigmaa, y), \quad \text{for all } l \in L. \label{eq:alpha-contract-con}}    
\end{maxi!}
In the upper level problem, the administrator selects a strategy $\sigmaa(\alpha)$ to maximize its total revenue to which users best respond by maximizing utilities in the lower-level problem. Note here that, in our contract game, since the principal gives $\alpha$ fraction of the total payoff accrued to the administrator, the administrator's objective of maximizing revenues corresponds to Equation~\eqref{eq:obj-contract-alpha}, which represents the sum of the fines collected by the administrator from allocating its security resources and the payment it receives from the principal for its payoff contribution. The corresponding objective of the principal, given a contract level $\alpha$ and the administrator strategy $\sigmaa(\alpha)$, as given by the solution of Problem~\eqref{eq:obj-contract-alpha}-\eqref{eq:alpha-contract-con}, is thus given by $(1-\alpha) P_R(\sigmaa(\alpha))$.

An equilibrium of our contract game is specified by a triple $(\alpha^*, \sigmaa(\alpha^*), (y_l(\sigmaa(\alpha^*)))_{l \in L})$, where $(\sigmaa(\alpha^*), (y_l(\sigmaa(\alpha^*)))_{l \in L})$ represent the solutions of the bi-level Program~\eqref{eq:obj-contract-alpha}-\eqref{eq:alpha-contract-con} given parameter $\alpha^*$, which the principal selects to maximizes its payoff $(1-\alpha) P_R(\sigmaa(\alpha))$. In the following, we first note that our earlier developed algorithmic approaches and theoretical guarantees in the revenue and payoff maximization settings naturally apply in studying the administrator and user strategies, i.e., the solution of Problem~\eqref{eq:obj-contract-alpha}-\eqref{eq:alpha-contract-con}, given any contract $\alpha \in [0, 1]$. Then, we present a dense-sampling approach to compute a near-optimal solution to the principal’s problem of selecting a contract $\alpha \in [0, 1]$ that maximizes its objective $(1-\alpha) P_R(\sigmaa(\alpha))$.

%naturally carry forward in studying administrator strategies that achieve a constant factor approximation ratio to the 

%study the strategies of the administrator and users, given any contract parameter $\alpha \in [0, 1]$. 

%To study the equilibrium strategies in this contract game, we proceed in two steps. First, given any contract parameter $\alpha \in [0, 1]$, 

\vspace{-5pt}

\paragraph{Administrator and User Strategies:}

In studying the equilibrium strategies of the administrator and users, we first note that solving Problem~\eqref{eq:obj-contract-alpha}-\eqref{eq:alpha-contract-con} is, in general, NP-hard, which follows from an analogous reduction to that in the proof of Theorem~\ref{thm:npHardness-swm-fm} (see Appendix~\ref{apdx:eq-strategies-admin-users-contract}). Thus, we develop an algorithm, which we henceforth refer to as \emph{Contract-Greedy}, akin to the greedy algorithms developed in the earlier studied revenue and payoff maximization settings, to compute an administrator strategy in our contract game %with strong approximation guarantees to the optimal solution of the bi-level Program~\eqref{eq:obj-contract-alpha}-\eqref{eq:alpha-contract-con}, 
for any contract $\alpha \in [0, 1]$. For instance, in the setting with homogeneous user types, \emph{Contract-Greedy} (see Algorithm~\ref{alg:GreedyContractDeterministic} in Appendix~\ref{apdx:eq-strategies-admin-users-contract}) is akin to Algorithm~\ref{alg:GreedyFraudminimizationDeterministic} in the payoff maximization setting other than in the process of sorting locations, as the administrator in our contract game, maximizes a linear combination of the revenue and payoff objectives. Moreover, since the administrator maximizes a linear combination of the revenue and payoff objectives, our earlier developed approximation ratio and resource augmentation guarantees also extend to this contract game. Thus, as many of the ideas developed in earlier sections apply in studying optimal administrator strategies in our contract game, for brevity, we defer the details of the algorithm, \emph{Contract-Greedy}, and its guarantees to Appendix~\ref{apdx:eq-strategies-admin-users-contract}.

\vspace{-2pt}

\paragraph{Principal's Strategy:}

Thus far, we have studied the strategies of the administrator and users given a contract $\alpha$. We now consider the principal's problem of selecting a contract $\alpha \in [0, 1]$ to maximize its objective $(1-\alpha) P_R(\sigmaa(\alpha))$. To this end, since obtaining an exact functional form of the payoff $P_R(\sigmaa(\alpha))$ for all contracts $\alpha \in [0, 1]$ is challenging, we present a \emph{dense-sampling} method to compute a near-optimal solution to the principal's optimization problem. In particular, consider the solutions $\sigmaa(\alpha)$ of the bi-level Program~\eqref{eq:obj-contract-alpha}-\eqref{eq:alpha-contract-con} for $\alpha$ taken from a finite set $\A_s = \{ 0, s, 2s, \ldots, 1 \}$ for some step-size $s\in (0, 1)$. We then evaluate the total payoff of each of the solutions $\sigmaa(\alpha)$ for $\alpha \in \A_s$ and return the value $\alpha_s^*$ from this discrete set that results in the highest objective for the principal, i.e., $(1-\alpha_s^*) P_R(\sigmaa(\alpha_s^*)) \geq (1-\alpha) P_R(\sigmaa(\alpha))$ for all $\alpha \in \A_s$.\footnote{Methods beyond dense sampling, e.g., gradient-based methods~\cite{solodov2007explicit}, can also be used; %to compute the parameter $\alpha$; %that maximizes the principal's payoff;
however, we use dense sampling, as it is computationally tractable when optimizing over a single variable and achieves the % near-optimality 
guarantee in Theorem~\ref{thm:dense-sampling-near-optimal}.} 

We now show that applying the above dense-sampling procedure approximately maximizes the principal's objective across all contract parameters $\alpha \in [0, 1]$, when the administrator strategy, given any parameter $\alpha$, corresponds to the optimal solution of the bi-level Program~\eqref{eq:obj-contract-alpha}-\eqref{eq:alpha-contract-con}.

%\begin{restatable}[Near-Optimality of Dense Sampling]{theorem}{nearOptDense}
%\label{thm:dense-sampling-near-optimal}
%Let $\alpha^* \in [0, 1]$ denote the optimal contract and $\alpha_{s}^* \in \A_s$ be the contract computed through dense-sampling. Further, given any contract $\alpha$, let $\sigmaa(\alpha)$ be the solution of the bi-level Program~\eqref{eq:obj-contract-alpha}-\eqref{eq:alpha-contract-con}. Then, for a step-size $s \leq \frac{\epsilon}{\sum_l v_l}$, the loss in the principal's payoff through dense sampling is bounded by $\epsilon$, i.e., $(1-\alpha^*)W_R(\sigmaa(\alpha^*)) \leq (1-\alpha_s^*)W_R(\sigmaa(\alpha_s^*)) + \epsilon$.
%\end{restatable}

\vspace{-1pt}

%\begin{comment}
\begin{theorem} [Near-Optimality of Dense Sampling] \label{thm:dense-sampling-near-optimal} 
Let $\alpha^* \in [0, 1]$ be the principal's optimal contract and $\alpha_{s}^* \in \A_s$ be the contract computed through dense-sampling. Further, given any $\alpha$, let $\sigmaa(\alpha)$ be the solution of Problem~\eqref{eq:obj-contract-alpha}-\eqref{eq:alpha-contract-con}. Then, for a step-size $s \leq \frac{\epsilon}{\sum_l p_l}$, the loss in the principal's objective through dense sampling is bounded by $\epsilon$, i.e., $(1-\alpha^*)P_R(\sigmaa(\alpha^*)) \leq (1-\alpha_s^*)P_R(\sigmaa(\alpha_s^*)) + \epsilon$.
\end{theorem}
%\end{comment}

The challenge in establishing Theorem~\ref{thm:dense-sampling-near-optimal} is that the payoff $P_R(\sigmaa(\alpha))$ is, in general, not continuous in $\alpha$ (e.g., the payoff function is discontinuous at $\sigma_l = \frac{d_l}{d_l+k}$ for any location $l$, as depicted in Figure~\ref{fig:social-welfare-revenue-best-response-plots}). Thus, the key idea in proving this result involves showing that the payoff $P_R(\sigmaa(\alpha))$ is monotonically (non)-decreasing in $\alpha$. For a complete proof of Theorem~\ref{thm:dense-sampling-near-optimal}, see Appendix~\ref{apdx:dense-sampling-theory}. 

The near-optimality of dense sampling, more generally, applies beyond administrator strategies corresponding to the solution of Problem~\eqref{eq:obj-contract-alpha}-\eqref{eq:alpha-contract-con} and, in particular, holds for any strategy $\Tilde{\sigmaa}(\alpha)$ such that the payoff $P_R(\Tilde{\sigmaa}(\alpha))$ is monotonically non-decreasing in $\alpha$ (see Appendix~\ref{apdx:key-lemmas-dense-sampling}). For instance, this monotonicity condition is satisfied by the solution computed using \emph{Contract-Greedy} under a correlation assumption on the payoffs $p_l$ and the payoff bang-per-buck ratios $\frac{p_l(d_l+k)}{d_l}$ (see Appendix~\ref{apdx:key-lemmas-dense-sampling}). Moreover, we observe from our experiments in Section~\ref{sec:numerical-experiments-optimal-contracts} that the payoff corresponding to the strategies computed using \emph{Contract-Greedy} is non-decreasing in $\alpha$. %Finally, methods beyond dense sampling, e.g., gradient-based methods~\cite{solodov2007explicit}, can be used to compute the parameter $\alpha$ that maximizes the principal's payoff; however, we use dense sampling, as it is computationally tractable when optimizing over a single variable and achieves the near-optimality guarantee in Theorem~\ref{thm:dense-sampling-near-optimal}.

\subsection{Numerical Experiments} \label{sec:numerical-experiments-optimal-contracts}

This section studies our contract game through numerical experiments based on a case study of queue jumping in IPT services (see Section~\ref{sec:examples-pertinent}) in Mumbai, India. In the following, we present model parameters for our experiment and results demonstrating the variation in the payoff accrued by a revenue-maximizing administrator for different model parameters as the contract level $\alpha$ is varied.

%Furthermore, while setting high fines is a natural mechanism to deter fraud in the presence of a welfare-maximizing administrator, our experiments demonstrate that setting high fines can often be detrimental to the system’s welfare in the presence of a revenue-maximizing administrator, which highlights the value of low to moderate fines, as often happens in practice, in mitigating fraud by deterring an administrator from solely maximizing revenues.

\textbf{\emph{Model Parameters:}} We consider a problem instance with $L = 448$ locations, representing the locations to hail the IPT service in Mumbai~\cite{mmrda-website}. We assume that each location $l$ has one type, i.e., $|\I| = 1$, where the number of fraudulent users $\Lambda_l$ are exponentially distributed with rate $80$, i.e., $\Lambda_l \sim Exp(80)$ for all $l$, and the benefits $d_l$ from engaging in fraud are exponentially distributed with rate $20$, i.e., $d_l \sim Exp(20)$ for all $l$. Moreover, we vary the number of resources $R \in \{ 1, 2, \ldots, 30\}$, the fine $k \in \{ 50, 100, \ldots, 500 \}$, and consider a payoff function given by $p_l = \Lambda_l (d_l)^x$ (see Section~\ref{subsec:game-params}) for $x$ lying in the range $\{ 1, 1.25, 1.5, 1.75, 2\}$ for all locations $l$. For a detailed overview of the motivations behind the choice and calibration of our model parameters, we refer to Appendix~\ref{apdx:additional-numerics}.

%Moreover, in our experiments, we vary the number of resources $R$ to lie in the range between one and thirty, with increments of one, i.e., $R \in \{ 1, 2, \ldots, 30\}$, and fines $k$ to lie in the range between 50 and 500, with increments of 50, i.e., $k \in \{ 50, 100, \ldots, 500 \}$. Finally, we consider a welfare function $v_l$ of the administrator to belong to the family of functions given by $v_l = \Lambda_l (d_l)^x$ for $x$ lying in the range $\{ 1, 1.25, 1.5, 1.75, 2\}$. For a detailed overview of the motivations behind the choice and calibration of our model parameters, we refer to Appendix~\ref{apdx:additional-numerics}.

%in the range $\{ 1, 2, \ldots, 30\}$, the fine $k$ varies in the range $\{ 50, 100, \ldots, 500 \}$, and the value function $v_l = \Lambda_l (d_l)^x$ for $x$ lying in the range $\{ 1, 1.25, 1.5, 1.75, 2\}$.

%\subsubsection{Results}

\textbf{\emph{Results:}} Figure~\ref{fig:social-welfare-proportion-vary-rk} depicts the variation in the payoff achieved by the allocation corresponding to \emph{Contract-Greedy} as the contract $\alpha$ is varied as a fraction of the payoff achieved using Algorithm~\ref{alg:GreedyFraudminimizationDeterministic} in the payoff maximization setting for different model parameters. In the left of Figure~\ref{fig:social-welfare-proportion-vary-rk}, we fix the fine to $k = 500$ and the number of resources $R = 10$ and vary the exponent $x$ of the payoff function $p_l = \Lambda_l (d_l)^x$. Analogously, in the center of Figure~\ref{fig:social-welfare-proportion-vary-rk}, we fix the payoff function $p_l = \Lambda_l (d_l)^{1.25}$ and the fine $k = 500$, with each curve corresponding to a different number of resources $R$. Finally, the right of Figure~\ref{fig:social-welfare-proportion-vary-rk} depicts curves for different fines, where the payoff function $p_l = \Lambda_l (d_l)^{1.25}$ and the number of resources $R = 15$. For the results in Figure~\ref{fig:social-welfare-proportion-vary-rk}, we consider the contract $\alpha$ to be chosen from a discrete set between zero and one with $0.05$ increments, i.e., the step-size $s = 0.05$.

%Figure~\ref{fig:social-welfare-proportion-vary-rk} depicts the variation in the welfare achieved by the allocation strategy corresponding to Algorithm~\ref{alg:GreedyContractDeterministic} as the contract level $\alpha$ is varied in our contract game as a fraction of the welfare achieved by the allocation strategy corresponding to Algorithm~\ref{alg:GreedyFraudminimizationDeterministic} in the welfare maximization setting for different model parameters. In particular, in the left of Figure~\ref{fig:social-welfare-proportion-vary-rk}, we fix the fine to $k = 500$ and the number of resources $R = 10$ and each curve corresponds to a different valuation function of the form $v_l = \Lambda_l (d_l)^x$ for $x$ lying in the range $\{ 1, 1.25, 1.5, 1.75, 2\}$. Analogously, in the center of Figure~\ref{fig:social-welfare-proportion-vary-rk}, we fix the valuation function $v_l = \Lambda_l (d_l)^{1.25}$ and the fine $k = 500$, with each curve corresponding to a different number of resources $R$. Finally, the right of Figure~\ref{fig:social-welfare-proportion-vary-rk} depicts curves for different fine levels, where the valuation function is $v_l = \Lambda_l (d_l)^{1.25}$ and the number of resources is $R = 15$. For the results in Figure~\ref{fig:social-welfare-proportion-vary-rk}, we consider the contract level $\alpha$ to be chosen from a discrete set of values between zero and one with $0.05$ increments, i.e., the step-size $s = 0.05$. 

From Figure~\ref{fig:social-welfare-proportion-vary-rk}, we first observe that regardless of the model parameters, the payoff corresponding to an administrator strategy computed using \emph{Contract-Greedy} is monotonically non-decreasing in the contract $\alpha$. Such a monotonic relation aligns with the proof of Theorem~\ref{thm:dense-sampling-near-optimal} and is natural as a higher contract $\alpha$ implies that the administrator is compensated more for the payoff it contributes. Further, we note from Figure~\ref{fig:social-welfare-proportion-vary-rk} that the revenue-maximizing solution, corresponding to $\alpha = 0$, achieves only a small fraction of the payoff achieved using Algorithm~\ref{alg:GreedyFraudminimizationDeterministic} in the payoff maximization setting, suggesting that the administrator's revenue and payoff maximization goals can often be at odds. However, for most tested parameters, an appropriately chosen contract $\alpha$ can recover a majority of the payoff achieved by Algorithm~\ref{alg:GreedyFraudminimizationDeterministic}. For instance, when $p_l = \Lambda_l (d_l)^{1.25}$, a contract of $\alpha = 0.5$ maximizes the principal's objective and achieves about $86\%$ of Algorithm~\ref{alg:GreedyFraudminimizationDeterministic}'s payoff.

Next, we note from Figure~\ref{fig:social-welfare-proportion-vary-rk} (left) that as we increase the exponent $x$ of the payoff function $p_l = \Lambda_l (d_l)^x$, the fraction of the payoff achieved by the allocation computed using \emph{Contract-Greedy} to that achieved by Algorithm~\ref{alg:GreedyFraudminimizationDeterministic} increases for each contract $\alpha$. Such a relation naturally follows as the payoff term in the administrator's Objective~\eqref{eq:obj-contract-alpha} increasingly dominates the revenues from the collected fines at each location with an increase in the exponent of the payoff function. Consequently, from the left of Figure~\ref{fig:social-welfare-proportion-vary-rk}, our results, for the studied payoff functions with an exponent $x>1$, demonstrate that using even small values of the contract $\alpha$ can recover most of the payoffs achieved using Algorithm~\ref{alg:GreedyFraudminimizationDeterministic}, thus bridging the gap between the payoff and revenue-maximizing outcomes.

Further, the fraction of the payoff achieved by the strategy computed using \emph{Contract-Greedy} to that achieved by Algorithm~\ref{alg:GreedyFraudminimizationDeterministic} in the payoff maximization setting for any contract $\alpha$ (i) remains nearly constant regardless of the number of resources (center of Figure~\ref{fig:social-welfare-proportion-vary-rk}) and (ii) increases for lower fines (right of Figure~\ref{fig:social-welfare-proportion-vary-rk}). Such a result holds as while varying the number of resources $R$ does not influence any property of the locations, e.g., a location's type $\Theta_l = (\Lambda_l, d_l, v_l)$ is independent of $R$, changing fines impacts the threshold $\frac{d_l}{d_l+k}$ at each location $l$ at which the revenue and payoff functions have a jump discontinuity (see Figure~\ref{fig:social-welfare-revenue-best-response-plots}). Consequently, while our greedy-like algorithms are not influenced by a change in the number of resources other than that the algorithms either terminate sooner or later depending on the number of resources, a change in the fine influences our algorithms' outcomes as the locations are sorted in a (potentially) different order for each fine $k$.

%We also note that the fraction of the welfare achieved by the administrator's strategy computed using \emph{Contract-Greedy} to that achieved by Algorithm~\ref{alg:GreedyFraudminimizationDeterministic} in the welfare maximization setting for any contract $\alpha$ (i) remains nearly constant regardless of the number of resources (center of Figure~\ref{fig:social-welfare-proportion-vary-rk}) and (ii) increases for lower fines (right of Figure~\ref{fig:social-welfare-proportion-vary-rk}). Such a result holds as while varying the number of resources does not influence any property of the locations, e.g., a location's type $\Theta_l = (\Lambda_l, d_l, v_l)$ is independent of the number of resources, changing fines impacts the threshold $\frac{d_l}{d_l+k}$ at each location $l$ at which the revenue and welfare functions have a jump discontinuity (see Figure~\ref{fig:social-welfare-revenue-best-response-plots}). Consequently, while our proposed greedy-like algorithms are not influenced by a change in the number of resources other than that the algorithms either terminate sooner or later depending on the number of available resources, a change in the fine influences the outcome of our algorithms as the locations are sorted in a (potentially) different order for each fine $k$.

Moreover, Figure~\ref{fig:social-welfare-proportion-vary-rk} (right) implies that at lower fines, the allocation computed using \emph{Contract-Greedy} recovers a higher proportion of the payoff compared to that achieved using Algorithm~\ref{alg:GreedyFraudminimizationDeterministic} in the payoff maximization setting. Such a result holds, as at lower fines, an administrator maximizing Objective~\eqref{eq:obj-contract-alpha} is more likely to optimize payoffs as the compensation it receives from the principal outweighs its low fine collections. Thus, Figure~\ref{fig:social-welfare-proportion-vary-rk} (right) highlights that setting high fines can be detrimental to the payoffs in the presence of a revenue-maximizing administrator. Such a result thus highlights the value of setting low to moderate fines, as often happens in practice, in deterring an administrator from solely maximizing revenues through the collected fines and instead incorporating payoff maximization in its objective even at low contract levels $\alpha$.

Overall, our results present several sensitivity relations that elucidate the impact of the payoff function $p_l$, number of resources $R$, and fine $k$ on the payoff achieved in the system using \emph{Contract-Greedy} for different contract parameters $\alpha$. Moreover, our results highlight the effectiveness of contracts in bridging the gap between the payoff and revenue maximization administrator objectives. For a further discussion and analysis of the results in Figure~\ref{fig:social-welfare-proportion-vary-rk}, we refer to Appendix~\ref{apdx:additional-numerical-experiments}.

\begin{figure}
    \centering %\hspace{-50pt}
\begin{subfigure}[t] {0.3\columnwidth}
    % This file was created with tikzplotlib v0.10.1.
\begin{tikzpicture}

\definecolor{crimson2143940}{RGB}{214,39,40}
\definecolor{darkgray176}{RGB}{176,176,176}
\definecolor{darkorange25512714}{RGB}{255,127,14}
\definecolor{forestgreen4416044}{RGB}{44,160,44}
\definecolor{mediumpurple148103189}{RGB}{148,103,189}
\definecolor{steelblue31119180}{RGB}{31,119,180}

\begin{axis}[
width=2.1in,
height=1.5in,
tick align=outside,
tick pos=left,
x grid style={darkgray176},
xlabel={Contract ($\alpha$)},
xlabel style= {font = \scriptsize, yshift=0.24cm},
xmin=-0.05, xmax=1.05,
xtick style={color=black, font = \tiny, yshift = 0.08cm},
%x tick label style= {font = \footnotesize},
y grid style={darkgray176},
ylabel = Payoff Fraction %\\ \vspace{5pt} 
$\frac{P_R(\sigmaa_{\alpha}^A)}{P_R(\sigmaa_{A}^*)} $,
title = {Variation in Payoff Functions},
title style = {font = \scriptsize, yshift=-0.27cm},
x tick label style={font=\scriptsize, yshift = 0.2cm},
axis background/.style={fill=white},
%ylabel=Fraction of Optimal \\ Welfare,
ylabel style={align=center, font = \scriptsize},
ymin=0.0299464809200509, ymax=1.02,
%ytick style={color=black, font = \footnotesize},
ytick = {0.0, 0.2, 0.4, 0.6, 0.8, 1.0},
y tick label style={
        /pgf/number format/.cd,
        fixed,
        fixed zerofill,
        precision=1,
        /tikz/.cd,
        font = \scriptsize
    },
legend style={
  fill opacity=0.8,
  draw opacity=1,
  text opacity=1,
  at={(0.46,0.75)},
  anchor=north west,
  draw=white!80!black,
  inner xsep=-0.2pt, inner ysep=-2pt,
  font = {\tiny\arraycolsep=2pt}
}
]
\addplot [line width = 1pt, steelblue31119180]
table {%
0 0.0761395056381438
0.05 0.0761395056381438
0.1 0.0761395056381438
0.15 0.0761395056381438
0.2 0.0761395056381438
0.25 0.0761395056381438
0.3 0.076144550835148
0.35 0.0761472077690979
0.4 0.0762833781042246
0.45 0.07642447531704
0.5 0.07642447531704
0.55 0.07642447531704
0.6 0.07642447531704
0.65 0.07642447531704
0.7 0.0764669557916713
0.75 0.0764722804462088
0.8 0.0764722804462088
0.85 0.0764722804462088
0.9 0.0764722804462088
0.95 0.0764722804462088
1 0.745712053080742
};
\addlegendentry{$x = 1$}
\addplot [line width = 1pt, darkorange25512714]
table {%
0 0.0839588985100597
0.05 0.0839588985100597
0.1 0.083964859560875
0.15 0.084333672518481
0.2 0.084333672518481
0.25 0.0843898391610377
0.3 0.123164930449549
0.35 0.298068764040212
0.4 0.604362835347428
0.45 0.768933462750956
0.5 0.858432131042332
0.55 0.943673926539339
0.6 0.972596750588812
0.65 0.990833787337509
0.7 0.995751886389841
0.75 0.997512424184067
0.8 0.998367364281925
0.85 0.999049512857504
0.9 0.999434358554778
0.95 0.999922000019503
1 0.999923761411651
};
\addlegendentry{$x = 1.25$}
\addplot [line width = 1pt, forestgreen4416044]
table {%
0 0.0908134331440615
0.05 0.0912718145216197
0.1 0.142089404004312
0.15 0.572244104260492
0.2 0.827966278063461
0.25 0.905344593181836
0.3 0.969161168595391
0.35 0.989247751351073
0.4 0.995027803250518
0.45 0.997930380811236
0.5 0.99888757777457
0.55 0.999380558879964
0.6 0.999554881829299
0.65 0.999719625336196
0.7 0.999893577954916
0.75 0.999902906118537
0.8 0.999902906118537
0.85 1
0.9 1
0.95 1
1 1
};
\addlegendentry{$x = 1.5$}
\addplot [line width = 1pt, crimson2143940]
table {%
0 0.0969524550488067
0.05 0.488933943537867
0.1 0.904379039659431
0.15 0.973543466418649
0.2 0.994428922477184
0.25 0.99858327443815
0.3 0.999912727709841
0.35 0.999976138207795
0.4 0.999976138207795
0.45 1
0.5 1
0.55 1
0.6 1
0.65 1
0.7 1
0.75 1
0.8 1
0.85 1
0.9 1
0.95 1
1 1
};
\addlegendentry{$x = 1.75$}
\addplot [line width = 1pt, mediumpurple148103189]
table {%
0 0.102212008328592
0.05 0.943042699543996
0.1 0.99135800124862
0.15 0.997622341573469
0.2 1
0.25 1
0.3 1
0.35 1
0.4 1
0.45 1
0.5 1
0.55 1
0.6 1
0.65 1
0.7 1
0.75 1
0.8 1
0.85 1
0.9 1
0.95 1
1 1
};
\addlegendentry{$x = 2$}
\end{axis}

\end{tikzpicture}
\end{subfigure} \hspace{20pt} 
\begin{subfigure}[t] {0.3\columnwidth} 
    % This file was created with tikzplotlib v0.10.1.
\begin{tikzpicture}

\definecolor{crimson2143940}{RGB}{214,39,40}
\definecolor{darkgray176}{RGB}{176,176,176}
\definecolor{darkorange25512714}{RGB}{255,127,14}
\definecolor{forestgreen4416044}{RGB}{44,160,44}
\definecolor{mediumpurple148103189}{RGB}{148,103,189}
\definecolor{steelblue31119180}{RGB}{31,119,180}

\begin{axis}[
width=2.1in,
height=1.5in,
tick align=outside,
tick pos=left,
x grid style={darkgray176},
xlabel={Contract ($\alpha$)},
xlabel style= {font = \scriptsize, yshift=0.24cm},
x tick label style={font=\scriptsize, yshift = 0.2cm},
xtick style={color=black, font = \tiny, yshift = 0.08cm},
xmin=-0.05, xmax=1.05,
%xtick style={color=black, font = \footnotesize},
%x tick label style= {font = \footnotesize},
y grid style={darkgray176},
title = {Variation in Number of Resources},
title style = {font = \scriptsize, yshift=-0.27cm},
ymin=0.0380961166310489, ymax=1.02,
%ytick style={color=black, font = \footnotesize},
y tick label style={
        /pgf/number format/.cd,
        fixed,
        fixed zerofill,
        precision=1,
        /tikz/.cd,
        font = \scriptsize
    },
legend style={
  fill opacity=0.8,
  draw opacity=1,
  text opacity=1,
  at={(0.55,0.76)},
  anchor=north west,
  draw=white!80!black,
  inner xsep=0pt, inner ysep=-1.5pt,
  font = {\tiny\arraycolsep=2pt}
}
]
\addplot [line width = 1pt, steelblue31119180]
table {%
0 0.0876746723860656
0.05 0.0876912706473767
0.1 0.0878789306220462
0.15 0.0880772144976402
0.2 0.0882793657097864
0.25 0.0883759759334167
0.3 0.140734010178451
0.35 0.333397195600466
0.4 0.616369766246201
0.45 0.800795655743932
0.5 0.897720911754775
0.55 0.960977473660398
0.6 0.98353605362291
0.65 0.995465605089147
0.7 0.998271844638581
0.75 0.999194888427226
0.8 0.999194888427226
0.85 0.999695549213939
0.9 0.999695549213939
0.95 1
1 1
};
\addlegendentry{$R = 5$}
\addplot [line width = 1pt, darkorange25512714]
table {%
0 0.0839588985100597
0.05 0.0839588985100597
0.1 0.083964859560875
0.15 0.084333672518481
0.2 0.084333672518481
0.25 0.0843898391610377
0.3 0.123164930449549
0.35 0.298068764040212
0.4 0.604362835347428
0.45 0.768933462750956
0.5 0.858432131042332
0.55 0.943673926539339
0.6 0.972596750588812
0.65 0.990833787337509
0.7 0.995751886389841
0.75 0.997512424184067
0.8 0.998367364281925
0.85 0.999049512857504
0.9 0.999434358554778
0.95 0.999922000019503
1 0.999923761411651
};
\addlegendentry{$R = 10$}
\addplot [line width = 1pt, forestgreen4416044]
table {%
0 0.0839454021349645
0.05 0.0839454021349645
0.1 0.0839571460221351
0.15 0.083957146022135
0.2 0.083957146022135
0.25 0.083957146022135
0.3 0.119292800026624
0.35 0.294846977063576
0.4 0.59826954137749
0.45 0.759821518917097
0.5 0.854543030986416
0.55 0.937245961151214
0.6 0.972244164456357
0.65 0.988414105397947
0.7 0.993584616623345
0.75 0.996214030209123
0.8 0.997434099303078
0.85 0.998559869475965
0.9 0.999116903865483
0.95 0.999722632846589
1 0.999802589206477
};
\addlegendentry{$R = 15$}
\addplot [line width = 1pt, crimson2143940]
table {%
0 0.0840089081278611
0.05 0.0840089081278611
0.1 0.0840089081278611
0.15 0.0840089081278611
0.2 0.0840089081278611
0.25 0.0840089081278611
0.3 0.120018258984578
0.35 0.294475193551885
0.4 0.597609071170664
0.45 0.759192984114778
0.5 0.854308600614215
0.55 0.936807037711187
0.6 0.972060343163231
0.65 0.988092600689493
0.7 0.993311014668923
0.75 0.996109103132544
0.8 0.997359390883083
0.85 0.998522705211775
0.9 0.999083779337896
0.95 0.999695912397683
1 0.999775166140761
};
\addlegendentry{$R = 20$}
\addplot [line width = 1pt, mediumpurple148103189]
table {%
0 0.0840089081278611
0.05 0.0840089081278611
0.1 0.0840089081278611
0.15 0.0840089081278611
0.2 0.0840089081278611
0.25 0.0840089081278611
0.3 0.120018258984578
0.35 0.294475193551885
0.4 0.597609071170664
0.45 0.759192984114778
0.5 0.854308600614215
0.55 0.936807037711187
0.6 0.972060343163231
0.65 0.988092600689493
0.7 0.993311014668923
0.75 0.996109103132544
0.8 0.997359390883083
0.85 0.998522705211775
0.9 0.999083779337896
0.95 0.999695912397683
1 0.999775166140761
};
\addlegendentry{$R = 30$}
\end{axis}

\end{tikzpicture}
\end{subfigure} %\hspace{10pt}
\begin{subfigure}[t] {0.3\columnwidth}
    % This file was created with tikzplotlib v0.10.1.
\begin{tikzpicture}

\definecolor{crimson2143940}{RGB}{214,39,40}
\definecolor{darkgray176}{RGB}{176,176,176}
\definecolor{darkorange25512714}{RGB}{255,127,14}
\definecolor{forestgreen4416044}{RGB}{44,160,44}
\definecolor{mediumpurple148103189}{RGB}{148,103,189}
\definecolor{steelblue31119180}{RGB}{31,119,180}

\begin{axis}[
width=2.1in,
height=1.5in,
tick align=outside,
tick pos=left,
x grid style={darkgray176},
xlabel={Contract ($\alpha$)},
xlabel style= {font = \scriptsize, yshift=0.24cm},
xmin=-0.05, xmax=1.05,
xtick style={color=black, font = \scriptsize},
x tick label style={font=\scriptsize, yshift = 0.2cm},
xtick style={color=black, font = \tiny, yshift = 0.08cm},
%x tick label style= {font = \footnotesize, yshift = 0.07cm},
y grid style={darkgray176},
title = {Variation in Fine Levels},
title style = {font = \scriptsize, yshift=-0.27cm},
ymin=0.0391237050485077, ymax=1.02,
%ytick style={color=black, font = \footnotesize, yshift = 0.07cm},
y tick label style={
        /pgf/number format/.cd,
        fixed,
        fixed zerofill,
        precision=1,
        /tikz/.cd,
        font = \scriptsize
    },
legend style={
  fill opacity=0.8,
  draw opacity=1,
  text opacity=1,
  at={(0.53,0.78)},
  anchor=north west,
  draw=white!80!black,
  inner xsep=0pt, inner ysep=-1.5pt,
  font = {\tiny\arraycolsep=2pt}
}
]
\addplot [line width = 1pt, steelblue31119180]
table {%
0 0.383981875028984
0.05 0.434835076985053
0.1 0.456733004881273
0.15 0.47566595214692
0.2 0.48841374860277
0.25 0.494250182433593
0.3 0.521926182668142
0.35 0.640940264280219
0.4 0.825858001887382
0.45 0.911909193211684
0.5 0.936079225650587
0.55 0.98647230469254
0.6 1
0.65 1
0.7 1
0.75 1
0.8 1
0.85 1
0.9 1
0.95 1
1 1
};
\addlegendentry{$k = 50$}
\addplot [line width = 1pt, darkorange25512714]
table {%
0 0.228126012028327
0.05 0.228126012028327
0.1 0.229567456992579
0.15 0.230019950404685
0.2 0.233595336188618
0.25 0.233744096052169
0.3 0.269505116011905
0.35 0.410345737179787
0.4 0.663368733358696
0.45 0.827238816531979
0.5 0.902151525814269
0.55 0.961056447144435
0.6 0.987624598210578
0.65 0.99764040742609
0.7 0.999459676959618
0.75 0.999459676959618
0.8 0.999459676959618
0.85 0.999937970539026
0.9 1
0.95 1
1 1
};
\addlegendentry{$k = 150$}
\addplot [line width = 1pt, forestgreen4416044]
table {%
0 0.149606211639067
0.05 0.151290463043025
0.1 0.153136570823481
0.15 0.153179621379093
0.2 0.153203521258589
0.25 0.15473944302698
0.3 0.189640471066976
0.35 0.345957666368713
0.4 0.626398283430321
0.45 0.787191776213746
0.5 0.871478053999744
0.55 0.952614548051142
0.6 0.979257705228733
0.65 0.991777881041976
0.7 0.996327880394975
0.75 0.997948288804707
0.8 0.998652870729374
0.85 0.999051019340505
0.9 0.999474475910003
0.95 0.999998903664458
1 0.999998903664458
};
\addlegendentry{$k = 250$}
\addplot [line width = 1pt, crimson2143940]
table {%
0 0.114763212707139
0.05 0.114768814876231
0.1 0.114768814876231
0.15 0.114768814876232
0.2 0.114769459858759
0.25 0.114769459858759
0.3 0.149867115599807
0.35 0.311212863068258
0.4 0.610553136946898
0.45 0.771011118528984
0.5 0.86474036942945
0.55 0.945244349763855
0.6 0.975459590805187
0.65 0.990655281928307
0.7 0.995121394507654
0.75 0.997265512120745
0.8 0.99812511346114
0.85 0.999066823808743
0.9 0.999441609576393
0.95 0.999916636212638
1 0.999916636212638
};
\addlegendentry{$k = 350$}
\addplot [line width = 1pt, mediumpurple148103189]
table {%
0 0.0839454021349645
0.05 0.0839454021349645
0.1 0.0839571460221351
0.15 0.083957146022135
0.2 0.083957146022135
0.25 0.083957146022135
0.3 0.119292800026624
0.35 0.294846977063576
0.4 0.59826954137749
0.45 0.759821518917097
0.5 0.854543030986416
0.55 0.937245961151214
0.6 0.972244164456357
0.65 0.988414105397947
0.7 0.993584616623345
0.75 0.996214030209123
0.8 0.997434099303078
0.85 0.998559869475965
0.9 0.999116903865483
0.95 0.999722632846589
1 0.999802589206477
};
\addlegendentry{$k = 500$}
\end{axis}

\end{tikzpicture}
\end{subfigure}
\vspace{-25pt}
    \caption{\small \sf Variation in the fraction of the payoffs achieved by the strategy $\sigmaa_{\alpha}^A$ computed using \emph{Contract-Greedy} to that achieved by the strategy $\sigmaa_{A}^*$ corresponding to Algorithm~\ref{alg:GreedyFraudminimizationDeterministic} as the contract level $\alpha$ is varied for different payoff functions $p_l = \Lambda_l (d_l)^x$ (left), number of resources $R$ (center), and fines $k$ (right).}
    \label{fig:social-welfare-proportion-vary-rk}
\end{figure}
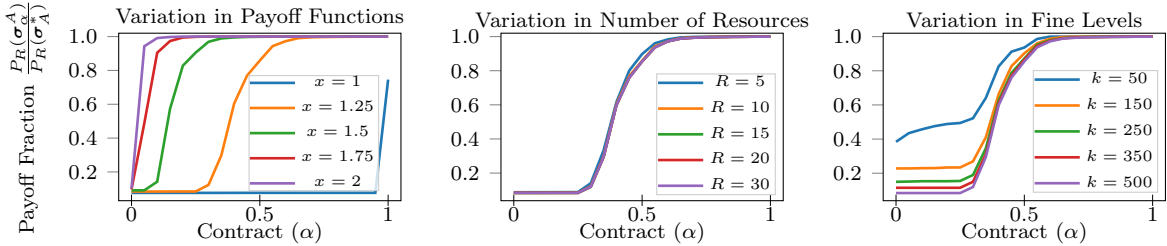

\section{Extension II: Additional Constraints on Allocating Resources} \label{sec:additional-constraints}

This section generalizes our model and algorithms to the setting when the administrator has additional constraints over allocating security resources beyond a resource budget constraint. Of particular interest are lower or upper bound constraints on the amount of security resources that can feasibly be allocated to given subsets of locations, i.e., the number of security resources allocated to given subsets of locations are constrained to lie between a specified floor and ceiling. We study this class of constraints as lower and upper bound constraints have theoretical significance, having been studied in domains spanning optimization theory~\cite{edmonds2003submodular} to matching theory~\cite{bckm-2013}, and arise naturally in security applications, e.g., due to fairness concerns over allocating security resources across locations or ensuring particular locations (or subsets of locations) receive a minimum security coverage.

In this section, we first formally introduce the administrator's payoff maximization problem under the additional lower and upper bound constraints over the allocation of security resources in Section~\ref{subsec:additional-constraints-model}. Then, in Section~\ref{subsec:greedy-arbitrarily-bad}, we show that a natural generalization of the greedy algorithm in the payoff-maximization setting (i.e., Algorithm~\ref{alg:GreedyWelMaxProb}) that respects these additional constraints, which we refer to as \emph{Constrained-Greedy}, can achieve an arbitrarily bad approximation ratios even with homogeneous user types. Thus, in Section~\ref{subsec:greedy-good-perf-heirarchy}, we identify a class of constraints, namely that of a \emph{hierarchy}, which have been widely studied in the literature on resource allocation~\cite{bckm-2013} and matching markets~\citep{kamada2018stability,bilemi-2004,gul2019market,jalota2022matchingtransfersdistributionalconstraints} and occurs in many applications, under which Constrained-Greedy achieves constant factor approximation ratios when augmented with a few additional resources. Finally, in Section~\ref{subsec:applications-hierarchical}, we highlight the implications and practical applicability of our theoretical guarantees through several real-world applications. While we focus on payoff maximization in this section, our algorithmic ideas and guarantees naturally generalize to the revenue maximization objective.

\subsection{Model for Security Game under Additional Constraints} \label{subsec:additional-constraints-model}

This section introduces the administrator's payoff maximization problem in the setting with additional constraints on allocating security resources. To this end, we first formally describe the lower and upper bound constraints on the amount of security resources the administrator can allocate to different subsets of locations. In particular, we let $\H$ denote the collection of all sets of locations over which the administrator has constraints over allocating security resources. Furthermore, for each set of locations $S \in \H$, let $\Bar{\lambda}_S \geq \underline{\lambda}_S$ be non-negative real-valued numbers and $\Bar{\llambda} = (\Bar{\lambda}_S)_{S \in \H}, \underline{\llambda} = (\underline{\lambda}_S)_{S \in \H}$ be the vector of upper and lower bound quotas on the allocation of security resources. Then, the additional upper and lower bound constraints of the administrator are defined by: $\underline{\lambda}_S \leq \sum_{l \in S} \sigma_l \leq \Bar{\lambda}_S$ for all subsets of locations $S \in \H$. That is, for each set $S \in \H$, the administrator must allocate between $[\underline{\lambda}_S, \Bar{\lambda}_S]$ security resources to the location set $S$.

Under the above-defined class of constraints, we obtain the following feasibility set for the administrator's mixed strategy vector:
\begin{align*}
    \Omega_R^{(\H, \underline{\llambda}, \Bar{\llambda})} = \left\{ \sigmaa = (\sigma_l)_{l \in L}: \sigma_l \in [0, 1] \text{ for all } l \in L,  \sum_{l \in L} \sigma_l \leq R, \text{ and } \underline{\lambda}_S \leq \sum_{l \in S} \sigma_l \leq \Bar{\lambda}_S, \forall S \in \H \right\},
\end{align*}
where the subscript $R$ represents the number of security resources available to the administrator while the superscript $(\H, \underline{\llambda}, \Bar{\llambda})$ represents the collection of sets of locations over which the administrator has constraints over allocating security resources and the corresponding upper and lower bound quotas on the allocation of security resources to each of these subsets of locations.

Then, the administrator's payoff-maximizing strategy with additional lower and upper bound constraints can be computed via a bi-level program analogous to Problem~\eqref{eq:admin-obj-fraud}-\eqref{eq:bi-level-con-fraud}, where the feasible set of strategies $\sigmaa \in \Omega_R^{(\H, \underline{\llambda}, \Bar{\llambda})}$ (rather than $\sigmaa \in \Omega_R$ in the setting without additional constraints). Defining $P_l(\sigma_l)$ as the payoff function when $\sigma_l$ resources are allocated to location $l$ (e.g., see right of Figure~\ref{fig:social-welfare-revenue-best-response-plots} or left of Figure~\ref{fig:helper-exp_welfare_mcua}), note that the administrator's payoff maximization bi-level program with additional constraints can be reformulated as the following problem:
\begin{maxi!}|s|[2]   
    {\sigmaa}                   
    { P_R(\sigmaa) = \sum_{l \in L} P_l(\sigma_l),  \label{eq:welfare-obj}}   
    {\label{eq:Eg001}}             
    {}          
    \addConstraint{\sum_{l \in L} \sigma_l}{\leq R, \label{eq:resource-constraint}}    
    \addConstraint{\sigma_l}{\in [0, 1], \quad \forall l \in L \label{eq:probability-feasibility}} 
    \addConstraint{\underline{\lambda}_S}{\leq \sum_{l \in S} \sigma_l \leq \Bar{\lambda}_S, \quad \forall S \in \H \label{eq:ub-lb-con}} 
\end{maxi!}
where Objective~\eqref{eq:welfare-obj} corresponds to maximizing the payoff across all locations and Constraints~\eqref{eq:resource-constraint}-\eqref{eq:ub-lb-con} correspond to the feasible set of administrator strategies, as specified by $\Omega_R^{(\H, \underline{\llambda}, \Bar{\llambda})}$. Recall from Section~\ref{subsec:payoffs} that the payoff function $P_l(\sigma_l) = \sum_{i \in \I} \left( p_l^i - (1 - \sigma_l) y_l^i(\sigma_l) p_l^i \right)$, where, with slight abuse of notation, we denote the user best-response function only as a function of $\sigma_l$ rather than the entire mixed-strategy vector $\sigmaa$, as all users at a given location best-respond solely based on the probability of allocating a security resource to that location (e.g., see Equation~\eqref{eq:best-response-users}).

\subsection{Low Approximation Ratio Under Arbitrary Constraints} \label{subsec:greedy-arbitrarily-bad}

This section presents an example demonstrating the challenge of attaining a good approximation ratio to the administrator's payoff maximization Problem~\eqref{eq:welfare-obj}-\eqref{eq:ub-lb-con} with additional lower and upper bound constraints on allocating security resources. In particular, we present an example of a constraint set to show that a natural generalization of the greedy algorithm in the payoff-maximization setting (i.e., Algorithm~\ref{alg:GreedyWelMaxProb}) that respects these additional constraints, which we henceforth refer to as Constrained-Greedy, can achieve an arbitrarily bad approximation ratio of $O(\frac{1}{|L|-1})$, where $|L|$ is the number of locations, even in the setting with homogeneous user types. %Since we consider an example with homogeneous user types in this section, for brevity of notation, we drop the superscript $i$ in the notation of user types.

%This section presents examples that demonstrate the challenge of attaining a good approximation ratio to the administrator's payoff maximization Problem~\eqref{eq:welfare-obj}-\eqref{eq:ub-lb-con} in the presence of additional lower and upper bound constraints on allocating security resources. In particular, we present two examples of constraint classes to show that a natural generalization of the greedy algorithm in the payoff-maximization setting (i.e., Algorithm~\ref{alg:GreedyWelMaxProb}) that respects these additional constraints, which we henceforth refer to as Constrained-Greedy, no longer achieves a half approximation in the presence of these additional constraints even in the setting when user types are homogeneous. Since we consider examples with homogeneous user types in this section, for expositional simplicity and brevity of notation, we drop the superscript $i$ in the notation of user types and denote the type of all users at a given location $l$ by the triple $\Theta_l = (\Lambda_l, d_l, p_l)$.

%In particular, the following example 

%Our first example shows that in the presence of additional upper bound constraints, in the worst-case, Constrained-Greedy can have an arbitrarily bad approximation ratio of $O(\frac{1}{|L|-1})$, where $|L|$ is the number of locations.

\begin{example} [Constrained-Greedy can have an Arbitrarily Bad Approximation Ratio] \label{eg:greedy-bad-performance}
Consider a setting with $|L|$ locations and $|L|-1$ additional upper bound constraints, i.e., $|\H| = |L|-1$, where each upper bound constraint is parametrized by $i$ and corresponds to the set $S_i = \{ |L|, i \}$ for all $i \in [|L|-1]$. In other words, the upper bound constraints are such that location $|L|$ belongs to all the constraint sets $S \in \H$, while all other locations belong to exactly one constraint set. To further specify the additional constraints, we let the lower bound quota $\underline{\lambda}_{S_i} = 0$ and the upper bound quota $\Bar{\lambda}_{S_i} = 0.5$ for all $i \in [|L|-1]$. Furthermore, we assume that $\frac{d_l}{d_l+k} = 0.5$ for all locations $l$ and that $p_l = p_{l'}$ for all $l, l' \in [|L|-1]$ and that $p_{|L|} = p_1+\epsilon$ for some small constant $0<\epsilon<p_1$. Finally, we let the resource budget $R \geq \sum_{l \in [|L|-1]} \frac{d_l}{d_l+k}$.

For this instance, we first evaluate the solution computed using Constrained-Greedy, which, akin to Algorithm~\ref{alg:GreedyWelMaxProb}, computes two allocation strategies: (i) a greedy allocation that allocates resources to locations in the descending order of their bang-per-buck ratios while respecting the allocation constraints and (ii) an allocation that maximizes the payoff from spending on a single location. For the greedy allocation corresponding to Constrained-Greedy, note by construction for the above-defined instance that location $|L|$ has the highest payoff bang-per-buck ratio and thus will be allocated resources first. However, on allocating $\sigma_{|L|} = 0.5$ to location $|L|$, the greedy allocation has exactly met the upper bound constraints on the allocation of resources to each set $S_i \in \H$ and hence, no more resources can be allocated to other locations. Consequently, the greedy allocation computed in step one of Constrained-Greedy corresponds to a cumulative payoff of $p_{|L|} = p_1+\epsilon$. Analogously, it is straightforward to see that the allocation that maximizes the payoff from spending on a single location achieves a cumulative payoff of $p_{|L|} = p_1+\epsilon$ as well. Since Constrained-Greedy outputs the allocation that maximizes the payoff between the above two allocations, it follows that the cumulative payoff achieved by Constrained-Greedy on this instance is $p_{|L|} = p_1+\epsilon$. 

On the other hand, the payoff maximizing policy corresponds to allocating resources to all locations other than $|L|$, i.e., allocate $\sigma_l = 0.5$ to all locations $l \in [|L|-1]$ and $\sigma_l = 0$ to location $|L|$. Consequently, the payoff corresponding to the payoff maximizing allocation is $(|L|-1) p_{1}$.

Thus, the above analysis suggests that the approximation ratio of Constrained-Greedy is $\frac{p_{1}+\epsilon}{(|L|-1) p_{1}} = \frac{1}{|L|-1} + \frac{\epsilon}{(|L|-1) p_{1}}$. Finally, since $\epsilon>0$ is a small constant, we have that our proposed greedy algorithm achieves an approximation ratio of $O(\frac{1}{|L|-1})$. \qed
\end{example}

Example~\ref{eg:greedy-bad-performance} highlights the limitations of applying a natural generalization of our proposed greedy algorithm (i.e., Algorithm~\ref{alg:GreedyWelMaxProb}) under additional lower and upper bound constraints and establishes that a greedy based approach no longer achieves a half approximation with these constraints. The primary reason for the generalization of our proposed greedy algorithm (i.e., Algorithm~\ref{alg:GreedyWelMaxProb}) attaining an arbitrarily bad approximation ratio in Example~\ref{eg:greedy-bad-performance} is that this algorithm allocates strictly fewer resources compared to the payoff-maximizing allocation. In fact, this algorithm allocates $O(\frac{1}{|L|-1})$ fraction of the resources allocated under the payoff-maximizing allocation due to the intersecting nature of the constraints, as $S_i \cap S_{i'} = \{ |L|\}$ for all sets $S_i, S_{i'} \in \H$, where $i \neq i'$ in Example~\ref{eg:greedy-bad-performance}. Thus, in the next section (Section~\ref{subsec:greedy-good-perf-heirarchy}), we introduce a class of constraints, that of a \emph{hierarchy}, where the constraint sets do not intersect, under which Constrained-Greedy allocates the same number of resources as the optimal allocation (see the proof of Lemma~\ref{lem:parametrization-simple-lem2} in Appendix~\ref{apdx:pf-lem-resource-same}), which enables us to obtain constant factor approximation ratios for this class of constraints.

\subsection{Approximation Guarantees Under Hierarchical Constraints} \label{subsec:greedy-good-perf-heirarchy}

This section studies the administrator's payoff maximization Problem~\eqref{eq:welfare-obj}-\eqref{eq:ub-lb-con} in the setting when the additional Constraints~\eqref{eq:ub-lb-con} correspond to a hierarchy, a notion we introduce in Section~\ref{subsec:hierarchy-constraint-notion}. While our Constrained-Greedy algorithm, formally introduced for hierarchical constraints in Section~\ref{subsec:constrained-greedy}, is not guaranteed to achieve a half-approximation guarantee even under hierarchical constraints (Section~\ref{subsec:no-half-apx-eg}), we show that it achieves constant factor approximation ratios if augmented with a few additional resources (Section~\ref{subsec:apx-ratio-hierarchy}). In Section~\ref{subsec:apx-ratio-hierarchy}, we also identify a condition on the lower bound quotas, which guarantees that Constrained-Greedy achieves a half approximation guarantee even when it does not have access to additional security resources.

\subsubsection{Hierarchical Constraints} \label{subsec:hierarchy-constraint-notion}

A \emph{hierarchy} corresponds to a family of subsets such that any two members of this family of subsets are either disjoint or one is a subset of the other, as is formalized by the following definition.

\begin{definition} [Hierarchy] \label{def:hierarchy}
A family (of constraints) $\H$ is a hierarchy if for every pair of elements $S, S' \in \H$, either $S \subseteq S'$ or $S' \subseteq S$ or $S \cap S' = \emptyset$. 
\end{definition}

To elucidate our approximation ratio guarantees and the Constrained-Greedy algorithm we develop, we also present an alternate characterization of hierarchical constraints. In particular, we characterize a constraint family $\H$ that corresponds to a hierarchy as comprising of a collection of sets $\L_1, \L_2, \ldots, \L_t$ with some depth $t \leq |L|$, where each $\L_i$ corresponds to a different layer. Here, $\L_1$ corresponds to the bottom-most layer and consists of all sets $S \in \H$ such that for all sets $S' \in \H \backslash S$, either $S' \cap S = \emptyset$ or $S \subseteq S'$. Moreover, $\L_t$ corresponds to the top-most layer and consists of all sets $S \in \H$ such that for all sets $S' \in \H \backslash S$, it holds that either $S' \cap S = \emptyset$ or $S' \subseteq S$. The layers $\L_2, \ldots, \L_{t-1}$ are all intermediate layers. We note that the layers are such that all sets in a layer $\L_i$ are a subset of exactly one set in each of the layers $\L_{j}$ for every $j>i$.

Finally, without loss of generality, we define the layers such that each location $l$ corresponds to at least one set of constraints in each layer, i.e., for each location $l$ and each layer $\L_i$, there exists at least one set $S \in \L_i$, such that $l \in S$. Note that if there is no constraint set $S \in \H$ corresponding to some layer $\L_i$ that consists of location $l$, we can duplicate the associated constraint from the previous layers. Adding these constraints to the respective layers adds redundant constraints without influencing the optimality or feasibility of the payoff maximization Problem~\eqref{eq:welfare-obj}-\eqref{eq:ub-lb-con}. For a pictorial depiction of the layers $\L_1, \L_2, \ldots, \L_t$ corresponding to a constraint hierarchy, see Figure~\ref{fig:hierarchy-layers}.

\begin{figure}[tbh!]
    \centering
    \includegraphics[width=0.55\linewidth]{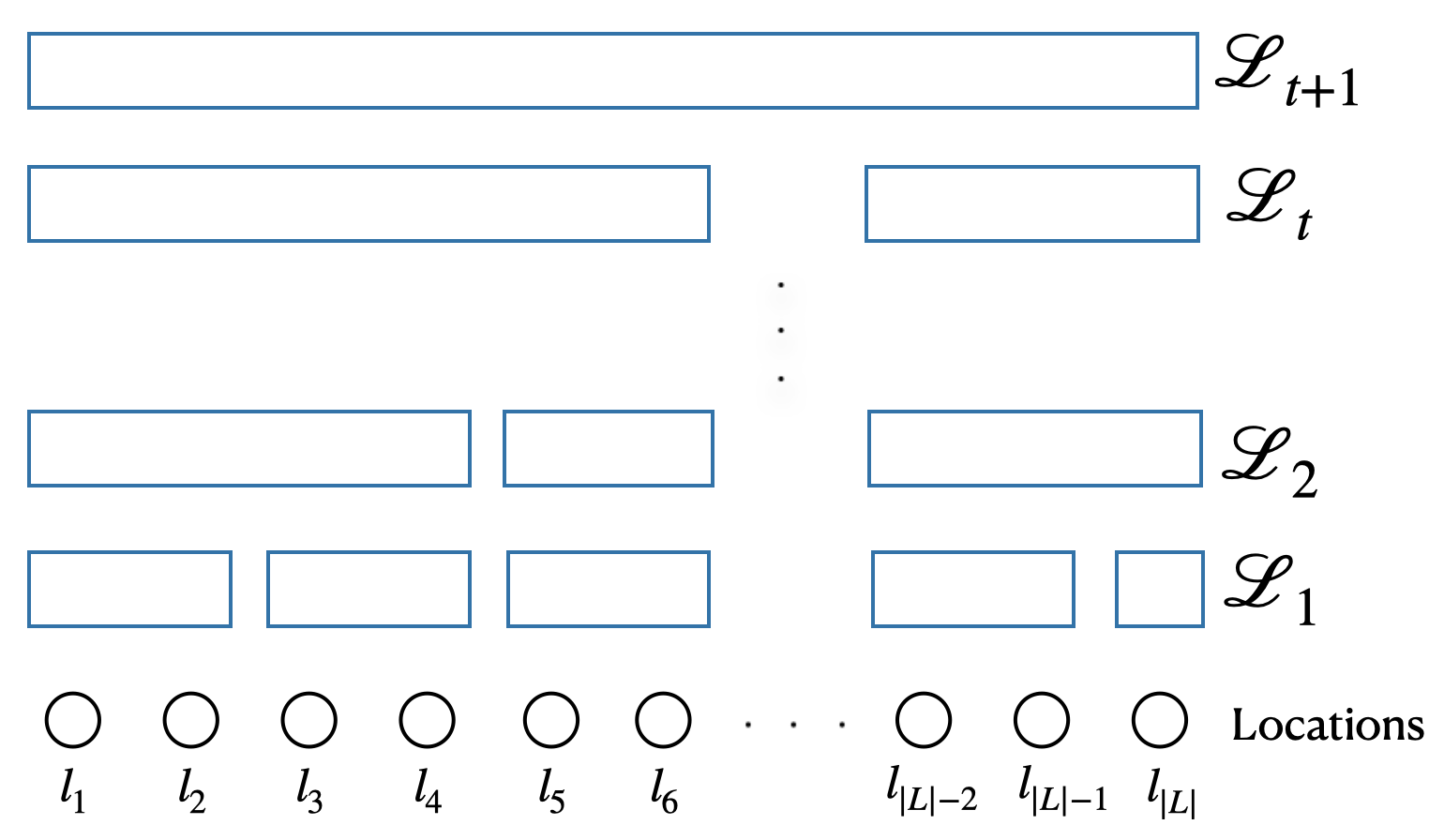}
    \vspace{-12pt}
    \caption{{\small \sf Depiction of a constraint hierarchy with layers $\L_1, ..., \L_t$ and a layer $\L_{t+1}$ depicting the resource budget constraint across all locations. Each of the circles corresponds to a location, while each rectangle represents a constraint, where the circles below each rectangle correspond to the locations over which the administrator has an upper and/or lower bound constraint. Note that each location belongs to one constraint (i.e., rectangle) in each layer, while a rectangle in a layer $\L_i$ is a subset of exactly one set in any layer $\L_j$ with $j>i$ under a hierarchy.
    }} 
    \label{fig:hierarchy-layers}
\end{figure} %\vspace{-10pt}

%Moreover, we define $\L_0$ as the constraints corresponding to each location, i.e., $\sigma_l \in [0, 1]$, and $\L_{t+1}$ as the constraints corresponding to the resource budget, i.e., $\sum_{l \in L} \sigma_l \leq R$.

\subsubsection{Constrained-Greedy Algorithm} \label{subsec:constrained-greedy}

This section presents a generalization of the greedy algorithm in the payoff-maximization setting (i.e., Algorithm~\ref{alg:GreedyWelMaxProb}) that respects the additional (hierarchical) constraints on allocating security resources. For clarity of exposition, we present this algorithm for the setting without lower bound constraints (i.e., the lower bound quotas are all zero) but note that the algorithm can be readily modified to accommodate lower bound constraints with non-zero quotas (see Appendix~\ref{apdx:lb-con-extension-pf}). 

The algorithm, which we henceforth refer to as \emph{Constrained-Greedy}, proceeds in two stages. In the first stage, Constrained-Greedy proceeds analogously to step one of Algorithm~\ref{alg:GreedyWelMaxProb} and computes an allocation strategy $\Hat{\sigmaa}$ while ensuring that the additional allocation constraints are satisfied. In particular, in this stage, resources are allocated to locations in the descending order of the slopes of the segments of the MCUA of the payoff functions while ensuring the additional allocation constraints are satisfied. Then, under the strategy $\Hat{\sigmaa}$, we compute the aggregate resource budget corresponding to each set $S \in \L_1$, where the aggregate resource budget is given by $\Hat{R}_S = \sum_{l \in S} \Hat{\sigma}_l$ for each set $S \in \L_1$. Given the aggregate resource budget corresponding to each set $S \in \L_1$, in stage two of Constrained-Greedy, we apply Algorithm~\ref{alg:GreedyWelMaxProb} as a sub-routine for each set $S \in \L_1$ to compute the final allocations for each location. This process is formally presented in Algorithm~\ref{alg:ConstrainedGreedy}.

%given the aggregate resource budget corresponding to each set $S \in \L_1$

%In particular, the algorithm proceeds in two stages. In the first stage, we proceed analogously to step one of Algorithm~\ref{alg:GreedyWelMaxProb} while ensuring that the additional resource allocation constraints are not violated. In particular, we allocate resources in the descending order of the slopes of the segments of the MCUA of the payoff function while ensuring the additional allocation constraints are satisfied. Then, we compute the total amount of resources allocated to each set of locations $S \in \L_1$. Given the aggregate allocations to each set $S \in \L_1$, in stage two, we apply Algorithm~\ref{alg:GreedyWelMaxProb} as a sub-routine to compute the final allocations to each location. This process is formally presented in Algorithm~\ref{alg:ConstrainedGreedy}.

%noting that the algorithm Constrained-Greedy proceeds analogously to Algorithm~\ref{alg:GreedyWelMaxProb} until the point at which 

\begin{algorithm}
\footnotesize
\SetAlgoLined
\SetKwInOut{Input}{Input}\SetKwInOut{Output}{Output}
\Input{Total Resources $R$, Constraint set $\H$ with Quotas $\underline{\llambda}, \Bar{\llambda}$, User Types $\Theta_l^i = (\Lambda_l^i, d_l^i, p_l^i)$ for all locations $l$ and types $i$}
\Output{Resource Allocation Strategy $\sigmaa_A^*$}
\textbf{Stage 1: Compute Aggregate Resource Budget for each Set $S \in \L_1$:} \\
Define affordability threshold $t_l \leftarrow \min \{ R, (\Bar{\lambda}_S)_{S: l \in S}, \max_i \frac{d_l^i}{d_l^i+k}\}$ for all locations $l$ \;
Generate MCUA of the payoff function in range $[0, t_l]$ for each location $l$ \;
$\Tilde{\S} \leftarrow $ Ordered list of segments $s$ across all locations of this MCUA in descending order of slopes $c_s$  \;
Initialize allocation strategy $\Hat{\sigmaa} \leftarrow \mathbf{0}$ \;
\For{\text{segment $s \in \Tilde{\S}$}}{
$\Hat{\sigma}_{l_s} \leftarrow \Hat{\sigma}_{l_s} + \min\{ R, (\Bar{\lambda}_S)_{S: l_s \in S}, x_s\}$ ; \texttt{\footnotesize \sf Allocate feasible amount of resources to location $l_s$} \;
$R \leftarrow R -  \min\{R, x_s\}$; \quad \texttt{\footnotesize \sf Update amount of remaining resources} \;
$\Bar{\lambda}_S \leftarrow \Bar{\lambda}_S -  \min\{ R, (\Bar{\lambda}_S)_{S: l_s \in S}, x_s\}$ if $l_s \in S$; \quad \texttt{\footnotesize \sf Update upper bound resource constraint} \;
  }
$\Hat{R}_S \leftarrow \sum_{l \in S} \Hat{\sigma}_l$ for all $S \in \L_1$ ; \texttt{\footnotesize \sf Compute the aggregate allocations for each set $S \in \L_1$} \;
\textbf{Stage 2: Run Algorithm~\ref{alg:GreedyWelMaxProb} Sub-Routine for each set $S \in \L_1$:} \\
\For{$S \in \L_1$}{
$(\sigma^*_{A, l})_{l \in S} \leftarrow $ Solution computed using Algorithm~\ref{alg:GreedyWelMaxProb} for the location subset $S$ given $\Hat{R}_S$ resources \;
}
\Return Allocation $\sigmaa_A^*$ \;
\caption{\footnotesize Constrained-Greedy for Administrator's Heterogeneous Payoff Maximization Objective}
\label{alg:ConstrainedGreedy}
\end{algorithm}

A few comments about Algorithm~\ref{alg:ConstrainedGreedy} are in order. First, the allocation $\Hat{\sigmaa}$ computed in stage one of Algorithm~\ref{alg:ConstrainedGreedy} is an intermediate allocation required to generate the aggregate allocations $\Hat{R}_S$ for each set $S \in \L_1$. Next, rather than directly outputting the allocation $\Hat{\sigmaa}$, we compute the final allocation $\sigmaa_A^*$ by running Algorithm~\ref{alg:GreedyWelMaxProb} as a sub-routine in stage two of Constrained-Greedy to obtain our desired approximation guarantees, as, for a subset of locations $S \in \L_1$, Algorithm~\ref{alg:GreedyWelMaxProb} achieves a half-approximation given a resource budget $\Hat{R}_S$. We reiterate that the two-stage procedure is crucial in obtaining our desired approximation guarantees in Section~\ref{subsec:apx-ratio-hierarchy} and recall that just optimizing the MCUA of the payoff functions (as in stage one of Constrained-Greedy) is insufficient in obtaining a half-approximation guarantee even without additional constraints (e.g., see Theorem~\ref{thm:greedy-half-approx-rev-max}).

%Note that the allocation for each subset $S$ corresponding to the allocation $\Hat{\sigmaa}$ may not give the desired half approximation.

%We henceforth refer to the above algorithm as Constrained-Greedy.

%This result suggests the limitations of applying the greedy algorithm for general constraint structures.

%\subsection{No Half Approximation for Greedy Algorithm for Well-Behaved Constraint Structures} \label{subsec:no-half-approx-well-behaved}

%Our proposed greedy procedure attained a low approximation ratio under general constraint structures due to the intersection between different constraints, as $S_i \cap S_{i'} = \{ |L|\}$ for all sets $S_i, S_{i'} \in \H$, where $i \neq i'$ in Example~\ref{eg:greedy-bad-performance}. Thus, in this section, we introduce a class of constraints under which constraint sets do not intersect. Even for this class of constraints, we show that our proposed greedy algorithm does not achieve a half approximation ratio.

%\begin{proposition}[Greedy Algorithm Achieves Lower than Half Approximation for Hierarchical Constraints] \label{prop:greedy-less-half-approx-hierarchy}
%Suppose that the constraint structure $\H$ is hierarchical. Then, there exists an instance with three locations and one upper bound constraint such that our proposed greedy algorithm achieves an approximation ratio of strictly smaller than $0.5$.
%\end{proposition}

\subsubsection{No Half-Approximation Guarantee under Hierarchical Constraints} \label{subsec:no-half-apx-eg}

This section shows that, in general, Constrained-Greedy does not achieve a half-approximation to the solution of Problem~\eqref{eq:welfare-obj}-\eqref{eq:ub-lb-con} even under hierarchical constraints (and when user types are homogeneous), as is elucidated through the following example. %Since we consider an example with homogeneous user types in this section, we drop the superscript $i$ in the notation of user types.

\begin{example}[Constrained Greedy Achieves $<0.5$ Approximation Under Hierarchical Constraints] \label{eg:no-half-approx-hierarchical}
Consider a setting with five locations $L = \{ L_1, L_2, L_3, L_4, L_5 \}$ and one upper bound constraint on the set $S_1 = \{ L_1, L_2\}$, where the total resource budget $R = 0.3$ and the upper bound quota $\Bar{\lambda}_{S_1} = 0.2$ for set $S_1$ (and a lower bound quota $\underline{\lambda}_{S_1} = 0$). Furthermore, consider the following location parameters:
\begin{itemize}
    \item Payoffs: $p_1 = 1$, $p_2 = 1.099$, $p_3 = 0.999$, $p_4 = 0.87$, $p_5 = 1.1$
    \item Allocation Thresholds: $\frac{d_1}{d_1+k} = 0.1$, $\frac{d_2}{d_2+k} = 0.11$, $\frac{d_3}{d_3+k} = 0.101$, $\frac{d_4}{d_4+k} = 0.089$, $\frac{d_5}{d_5+k} = 0.3$.
\end{itemize}
Note that these locations are ordered in the descending order of their payoff bang-per-buck ratios. We now compute the payoff corresponding to Constrained-Greedy.

Since the locations are ordered in the descending order of their payoff bang-per-buck ratios, the resulting allocation under stage one of Constrained-Greedy that respects the allocation constraints is $\Hat{\sigmaa} = (0.1, 0.1, 0.1, 0, 0, 0)$. Moreover, it is also straightforward to check that the allocation corresponding to stage two of Constrained-Greedy coincides with $\Hat{\sigmaa}$ and corresponds to a total payoff of $p_1 + 0.1 p_2 + 0.1 p_3 = 1.2098$. Note here that Constrained-Greedy only allocates $0.1$ units of resources to location two as allocating any more to location two would violate the upper bound constraint on the set $S_1$. Moreover, Constrained-Greedy only allocates $0.1$ units of resources to location three as allocating any more would violate the overall resource budget constraint.

%Furthermore, the allocation corresponding to step two of Constrained-Greedy, which involves spending all resources at a single location, yields a payoff of at most $1.1$.

Next, the payoff maximizing allocation is $\sigmaa^* = (0, 0.11, 0.101, 0.089, 0)$ with a payoff of $p_2+p_3+p_4 = 2.968$. Consequently, the approximation ratio of Constrained-Greedy is $\frac{1.2098}{2.968} = 0.4076 < 0.5$, which establishes that Constrained-Greedy achieves an approximation ratio of strictly smaller than $0.5$ even under hierarchical constraints. \qed
\end{example}

%The above example highlights that even under hierarchical constraint structures, Constrained-Greedy is not guaranteed to obtain an approximation ratio of at least $0.5$, which was guaranteed for Algorithm~\ref{alg:GreedyWelMaxProb} in the setting without additional constraints.

\subsubsection{Approximation Guarantees with Additional Resources} \label{subsec:apx-ratio-hierarchy}

While Constrained-Greedy is not guaranteed to achieve a half approximation to the solution to Problem~\eqref{eq:welfare-obj}-\eqref{eq:ub-lb-con} even under hierarchical constraints (see Example~\ref{eg:no-half-approx-hierarchical} in Section~\ref{subsec:no-half-apx-eg}), we now establish that augmenting Constrained-Greedy with a few more resources and relaxing some of the upper bound quotas of the additional constraints can recover the desired half approximation guarantee. Moreover, we identify a natural condition on the lower bound quotas under which Constrained-Greedy achieves the desired half-approximation guarantee without additional resources or a relaxation to any upper bound quotas of the additional constraints.

To present our results, we first introduce some notation. In particular, we define $\sigmaa_{R+|\L_1|}^A$ as the solution corresponding to Constrained-Greedy with $R+|\L_1|$ resources, where each of the upper bound quotas in layer $\L_1$ are relaxed by one unit, each of the upper bound quotas in layer $\L_2$ are relaxed based on the number of sets $S \in \L_1$ that belong to each set $S' \in \L_2$, and so on until the layer $\L_t$. Moreover, we define $\sigmaa_{R+|\L_2|}^A$ be the solution corresponding to Constrained-Greedy with $R+|\L_2|$ resources where at most $|\L_2|$ of the upper bound quotas in layer $\L_1$ are relaxed by one unit, each of the upper bound quotas in layer $\L_2$ are relaxed by one unit, each of the upper bound quotas in layer $\L_3$ are relaxed based on the number of sets $S \in \L_2$ that belong to each set $S' \in \L_3$, and so on until the layer $\L_t$. Note here that since layer $\L_2$ is a higher layer in the hierarchy compared to layer $\L_1$ (see Figure~\ref{fig:hierarchy-layers}) that $|\L_2| \leq |\L_1|$, i.e., fewer additional resources are required and the constraints are relaxed by smaller amounts under the allocation $\sigmaa_{R+|\L_2|}^A$ compared to $\sigmaa_{R+|\L_1|}^A$.

%We begin by presenting the main result of this section, which establishes that 

We now establish the approximation ratios of the allocations $\sigmaa_{R+|\L_1|}^A$ and $\sigmaa_{R+|\L_2|}^A$ corresponding to Constrained-Greedy when all lower bound quotas are zero.

%the approximation ratio guarantees of the greedy algorithm when it is given additional resources.
\vspace{-5pt}
\begin{theorem}[Approximation Guarantees of Constrained-Greedy Under Resource Augmentation] \label{thm:apx-greedy-resource-augmentation-constraints}
Suppose the constraint structure $\H$ is a hierarchy and $\L_1 \cup \ldots \cup \L_t$ are the corresponding constraint layers of this hierarchy such that each location $l$ belongs to exactly one constraint in each layer and the lower bound quotas are such that $\underline{\lambda}_S = 0$ for all $S \in \H$. Further, let $\sigmaa^*$ be the constrained payoff maximizing allocation that solves Problem~\eqref{eq:welfare-obj}-\eqref{eq:ub-lb-con} with $R$ resources, let $\sigmaa_{R+|\L_1|}^A$ and $\sigmaa_{R+|\L_2|}^A$ be the solutions corresponding to Constrained-Greedy as defined above. Then: \vspace{-8pt}
\begin{itemize}
    \item The total payoff under allocation $\sigmaa_{R+|\L_1|}^A$ is at least that under the constrained payoff maximizing allocation, i.e., $P_{R+|\L_1|}(\sigmaa_{R+|\L_1|}^A) \geq P_{R}(\sigmaa^*)$, \vspace{-8pt}
    \item The total payoff under allocation $\sigmaa_{R+|\L_2|}^A$ is at least half of that under the constrained payoff maximizing allocation, i.e., $P_{R+|\L_2|}(\sigmaa_{R+|\L_2|}^A) \geq \frac{1}{2} P_{R}(\sigmaa^*)$.
\end{itemize}
\end{theorem}

For a proof of Theorem~\ref{thm:apx-greedy-resource-augmentation-constraints}, see Appendix~\ref{apdx:pf-thm-additional-constraints-apx}. Theorem~\ref{thm:apx-greedy-resource-augmentation-constraints} establishes that augmenting Constrained-Greedy with additional resources enables it to obtain a half approximation (or even a one approximation guarantee) that was not possible in the setting without additional resources (see Example~\ref{eg:no-half-approx-hierarchical} in Section~\ref{subsec:no-half-apx-eg}). The allocation $\sigmaa_{R+|\L_1|}^A$ achieves a one-approximation while the allocation $\sigmaa_{R+|\L_2|}^A$ achieves a half-approximation as layer $\L_2$ is a higher layer in the hierarchy compared to $\L_1$ (see Figure~\ref{fig:hierarchy-layers}); hence, fewer additional resources are required and the constraints are relaxed by smaller amounts under the allocation $\sigmaa_{R+|\L_2|}^A$ compared to $\sigmaa_{R+|\L_1|}^A$. While $|\L_1|$ and $|\L_2|$ can be large in the worst case, for many practical problems of interest, $|\L_1|$ and $|\L_2|$ are small, as elucidated through examples in Section~\ref{subsec:applications-hierarchical}, which highlight the practical significance of the results obtained in Theorem~\ref{thm:apx-greedy-resource-augmentation-constraints}. Finally, we note that the administrator only requires $|\L_1| \max_i \frac{d_l^i}{d_l^i+k}$ ($|\L_2| \max_i \frac{d_l^i}{d_l^i+k}$) additional resources to establish the guarantee in Theorem~\ref{thm:apx-greedy-resource-augmentation-constraints} along with the associated relaxations to the upper bound quotas corresponding to the additional constraints; however, we present the result with $\L_1$ ($|\L_2|$) resources for ease of exposition.

While Theorem~\ref{thm:apx-greedy-resource-augmentation-constraints} established approximation ratios in the setting when Constrained-Greedy has additional resources and some of the upper bound quotas are relaxed (and the lower bound quotas are zero), the following proposition identifies a natural condition on the lower bound quotas under which Constrained-Greedy obtains a half-approximation to the payoff maximization allocation without any additional resources or a relaxation of the upper bound quotas.
\vspace{-2pt}
\begin{proposition}[Approximation Guarantee of Constrained-Greedy With Lower Bound Quotas] \label{prop:apx-greedy-lb-quotas}
Suppose the constraint structure $\H$ is a hierarchy and $\L_1 \cup \ldots \cup \L_t$ are the corresponding constraint layers of this heriarchy such that each location $l$ belongs to exactly one constraint in each layer. Moreover, suppose that for all sets $S \in \L_1$ it holds that $\underline{\lambda}_S \geq 2$ and Problem~\eqref{eq:welfare-obj}-\eqref{eq:ub-lb-con} has a feasible solution. Furthermore, let $\sigmaa^*$ be the constrained payoff maximizing allocation that solves Problem~\eqref{eq:welfare-obj}-\eqref{eq:ub-lb-con} and $\sigmaa_A^*$ be the allocation computed using the Constrained-Greedy algorithm. Then, $\sigmaa_A^*$ achieves at least half the payoff as $\sigmaa^*$, i.e., $P_R(\sigmaa_A^*) \geq \frac{1}{2} P_R(\sigmaa^*)$.
\end{proposition}

The proof of Proposition~\ref{prop:apx-greedy-lb-quotas} is analogous to that of Theorem~\ref{thm:apx-greedy-resource-augmentation-constraints} with a few additional caveats, which we highlight in Appendix~\ref{apdx:pf-prop-lb-quotas}. Proposition~\ref{prop:apx-greedy-lb-quotas} highlights that if the administrator has a requirement to meet a minimum coverage of security resources across all sets $S \in \L_1$, a natural condition in applications where the number of locations belonging to each set $S$ is large, Constrained-Greedy achieves at least half of the optimal administrator payoff. We note that unlike Theorem~\ref{thm:apx-greedy-resource-augmentation-constraints}, no relaxations to the upper bound resource quotas or the total number of available resources is required to establish Proposition~\ref{prop:apx-greedy-lb-quotas}. Such a result holds as a lower bound quota of two guarantees that the greedy solution corresponding to maximizing the MCUA of the payoff functions achieves at least the payoff as the solution that maximizes the payoff corresponding to spending all available resources at a single location (see Appendix~\ref{apdx:pf-prop-lb-quotas}), a condition not guaranteed to hold in the setting without lower bound quotas of at least two for each of the sets $S \in \L_1$.

%administrator can allocate resources to at least one of the locations $l$ in each set $S$ with probability $\max_i \frac{d_l^i}{d_l^i+k}$, a condition not guaranteed to hold in the setting without lower bound quotas of at least one for each of the sets $S \in \L_1$.

\subsection{Implications of Approximation Guarantees under Hierarchical Constraints} \label{subsec:applications-hierarchical}

This section highlights the implications of our approximation guarantees in Section~\ref{subsec:greedy-good-perf-heirarchy} (see Theorem~\ref{thm:apx-greedy-resource-augmentation-constraints}) through two real-world applications, one corresponding to environmental compliance and another corresponding to monitoring crime/traffic violations. These application cases highlight the practical viability of Constrained-Greedy in achieving a good approximation to the administrator's payoff maximization Problem~\eqref{eq:welfare-obj}-\eqref{eq:ub-lb-con} without requiring many additional resources.

\begin{example} [Environmental Compliance] \label{eg:env-compliance}
Consider a federal agency, e.g., the United States Environmental Protection Agency (USEPA), that enforces water quality standards and has limited inspection officers (i.e., security resources) to monitor the compliance behavior of water treatment facilities across the country. The absence of constraints on the allocation of security resources may result in an imbalanced allocation of inspections where some regions of the country have significantly more inspections than others. Thus, federal agencies often have quotas on the allocation of their inspection officers to achieve a more balanced allocation of inspections. For instance, the USEPA may seek to ensure that the total inspection resources allocated to each state remain between certain lower and upper bounds. Moreover, within each state, there may be similar quotas imposed on the inspection resources that can be allocated within each county.

Notice that such a setting with constraints on the allocation of inspections by state and county corresponds to a two-layer hierarchy (with layers $\L_1$ and $\L_2$), as all states and counties are disjoint from one another, with each county belonging to exactly one state. Here, each set $S$ belonging to the layer $\L_1$ corresponds to the set of water treatment facilities belonging to a given county. Moreover, each set $S$ in the layer $\L_2$ corresponds to the set of water treatment facilities belonging to a given state. In this setting, Theorem~\ref{thm:apx-greedy-resource-augmentation-constraints} implies that hiring one additional inspection officer for each state, a relatively small amount given that thousands of environmental inspections are performed each year~\cite{hino2018machine}, and relaxing the upper bound quota for at most one county in each state by one resource would result in an outcome where Constrained-Greedy achieves at least half the optimal payoff as obtained in the setting without relaxing any of the resource constraints.
\end{example}

\begin{example} [Crime and Traffic Monitoring] \label{eg:crime-traffic-monitoring}
Consider a police agency seeking to allocate its police officers to monitor crime or traffic violations in a city. To monitor crime/traffic violations, police agencies often divide a city into disjoint \emph{beats}~\cite{shixiang-atlanta-2022} and typically have quotas for allocating police officers within each beat. Such a setting corresponds to a single layer hierarchy $\L_1$ as all beats are disjoint; thus, we can consider the layer $\L_1$, where each set $S \in \L_1$ corresponds to the set of locations susceptible to crime/traffic violations in a given beat. In this setting the resource budget constraint can be interpreted as corresponding to layer $\L_2$; thus, Theorem~\ref{thm:apx-greedy-resource-augmentation-constraints} implies that hiring one additional police officer that can be allocated to one of the beats (i.e., by relaxing the upper bound quota of at most one of the beats by one) would result in an outcome where Constrained-Greedy achieves at least half the optimal payoff in the setting without relaxing any of the resource constraints. In addition, Proposition~\ref{prop:apx-greedy-lb-quotas} implies that if there is a lower bound constraint of allocating at least two police officers to each beat, then without relaxing any of the resource constraints, Constrained-Greedy achieves at least half of the optimal payoff.
%Moreover, we note that hiring one police officer resource for each beat would result in an outcome where the Constrained-Greedy algorithm achieves a welfare that is more than the optimum welfare without the additional resources.
\end{example}

%\ifarxiv

\section{Conclusion and Future Work} \label{sec:conclusion-future-work}

In this work, we studied the problem of policing fraud as a security game between an administrator and users. Motivated by several real-world settings where fraud is prevalent (see Section~\ref{sec:examples-pertinent}), we introduced a model of a security game wherein the administrator can deploy a budget of security resources across locations and levy fines against users found engaging in fraud. We studied our security game under both payoff and revenue maximization administrator objectives. In both settings, we showed that, in general, the problem of computing the administrator's optimal resource allocation strategy is NP-hard, and we developed greedy-like algorithms for both administrator objectives with associated approximation ratio and resource augmentation guarantees. Moreover, in the setting when user types are homogeneous, we showed that the administrator's revenue maximization problem can be solved in polynomial time and developed a PTAS for the administrator's payoff maximization problem. % that, when augmented with $\epsilon$ additional resources, achieves a $1-\epsilon$ approximation to the administrator's optimal payoff.

Next, we presented numerical experiments based on a real-world case study of parking enforcement at Stanford University's campus. Our results demonstrated that our proposed algorithms outperform the status-quo parking enforcement mechanism and, in particular, can increase earnings from parking permit purchases for the university by over \$300,000 (a 2\% increase) annually. Finally, we studied several model extensions, including incorporating contracts into our framework to bridge the gap between the administrator's payoff and revenue-maximizing outcomes and generalizing our model to incorporate additional constraints beyond a resource budget constraint.

There are several future research directions. First, it would be interesting to investigate whether an FPTAS exists for the setting of homogeneous payoff-maximization setting and whether a PTAS exists for the heterogeneous revenue and payoff-maximization settings when the number of user types is not a constant. Next, it would be interesting to study an online learning setting where the administrator learns users' types through repeated interactions. Moreover, it would be worthwhile to explore more properties of our contract framework and further investigate the algorithm design with additional constraints on allocating security resources. For further directions of future research, see Appendix~\ref{apdx:model-extensions}.

\section*{Acknowledgements}

This work was supported in part by the Stanford Interdisciplinary Graduate Fellowship (SIGF) and the NSF CCF grant 1918549. We also thank Vincent Bergado from Stanford University's Department of Public Safety for his support in providing the data for the computational study in Section~\ref{sec:numerical-experiments-parking-enforcement}.

\bibliographystyle{unsrt}
\bibliography{main}

\begin{thebibliography}{10}

\bibitem{bjerre2023playing}
Andreas Bjerre-Nielsen, Lykke~Sterll Christensen, Mikkel~H{o}st Gandil, and
  Hans~Henrik Sievertsen.
\newblock {Playing the system: address manipulation and access to schools}.
\newblock Papers 2305.18949, arXiv.org, May 2023.

\bibitem{RePEc:hal:psewpa:halshs-03828729}
Philippe Askenazy, Thomas Breda, Flavien Moreau, and Vladimir Pecheu.
\newblock {Do French companies under-report their workforce at 49 employees to
  get around the law?}
\newblock PSE Working Papers halshs-03828729, HAL, March 2022.

\bibitem{estornell2021incentivizing}
Andrew Estornell, Sanmay Das, and Yevgeniy Vorobeychik.
\newblock Incentivizing truthfulness through audits in strategic
  classification.
\newblock {\em Proceedings of the AAAI Conference on Artificial Intelligence},
  35(6):5347--5354, May 2021.

\bibitem{estornell2023incentivizing}
Andrew Estornell, Yatong Chen, Sanmay Das, Yang Liu, and Yevgeniy Vorobeychik.
\newblock Incentivizing recourse through auditing in strategic classification.
\newblock In Edith Elkind, editor, {\em Proceedings of the Thirty-Second
  International Joint Conference on Artificial Intelligence, {IJCAI-23}}, pages
  400--408. International Joint Conferences on Artificial Intelligence
  Organization, 8 2023.
\newblock Main Track.

\bibitem{perez2023fraud}
Eduardo Perez-Richet and Vasiliki Skreta.
\newblock Fraud-proof non-market allocation mechanisms.
\newblock Technical report, Working paper, 2023.

\bibitem{lundy-taylor-etal-2019}
Taylor Lundy, Alexander Wei, Hu~Fu, Scott~Duke Kominers, and Kevin
  Leyton-Brown.
\newblock Allocation for social good: Auditing mechanisms for utility
  maximization.
\newblock In {\em Proceedings of the 2019 ACM Conference on Economics and
  Computation}, EC '19, page 785–803, New York, NY, USA, 2019. Association
  for Computing Machinery.

\bibitem{bolton2018}
Lee Bolton.
\newblock Manipulation of the waitlist priority of the organ allocation system
  through the escalation of medical therapies.
\newblock {\em OPTN/UNOS EthicsCommittee}, 2018.

\bibitem{mcmichael2022stealing}
Benjamin~J McMichael.
\newblock Stealing organs?
\newblock {\em Ind. LJ}, 97:135, 2022.

\bibitem{tambe2011security}
Milind Tambe.
\newblock {\em Security and game theory: algorithms, deployed systems, lessons
  learned}.
\newblock Cambridge university press, 2011.

\bibitem{heyes2000implementing}
Anthony Heyes.
\newblock Implementing environmental regulation: enforcement and compliance.
\newblock {\em Journal of regulatory economics}, 17(2):107--129, 2000.

\bibitem{cervero2000informal}
Robert Cervero.
\newblock {\em Informal transport in the developing world}.
\newblock UN-HABITAT, 2000.

\bibitem{CERVERO2007445}
Robert Cervero and Aaron Golub.
\newblock Informal transport: A global perspective.
\newblock {\em Transport Policy}, 14(6):445--457, 2007.

\bibitem{southAfrica-minibus}
Andreas Neumann, Daniel Röder, and Johan~W. Joubert.
\newblock Toward a simulation of minibuses in south africa.
\newblock {\em Journal of Transport and Land Use}, 8(1):137--154, 2015.

\bibitem{SALOMON1985259}
Ilan Salomon and Lionel Silman.
\newblock Scheduled bus and sherut taxi operation in israel.
\newblock {\em Transportation Research Part B: Methodological}, 19(3):259--264,
  1985.

\bibitem{zwick2017analysis}
F~Zwick.
\newblock {\em Analysis and simulation of Santiago de Chile’s colectivo
  system}.
\newblock PhD thesis, Bachelor’s thesis. TU Berlin, 2017.

\bibitem{ipt-india}
Ravi Gadapalli.
\newblock Role of intermediate public transport in indian cities.
\newblock {\em Economic and Political Weekly}, 51(9):46--49, 2016.

\bibitem{Subramanya_2012}
Rupa Subramanya.
\newblock Economics journal: Stupid taxi rules in mumbai.
\newblock \url{https://www.wsj.com/articles/BL-IRTB-16187}, Jul 2012.
\newblock [Accessed: 2024-01-14].

\bibitem{queue-jumping-website}
Chris Wood.
\newblock The art of queue-jumping on hong kong trams.
\newblock
  \url{https://www.scmp.com/magazines/post-magazine/short-reads/article/2106273/art-queue-jumping-hong-kong-trams},
  2017.
\newblock [Accessed: 2024-02-09].

\bibitem{skipping-the-queue-2012}
Gad Allon and Eran Hanany.
\newblock Cutting in line: Social norms in queues.
\newblock {\em Management Science}, 58(3):493--506, 2012.

\bibitem{traffic-stats-nyc}
Lawrence Berezin.
\newblock The mind-blowing truth about nyc parking ticket stats.
\newblock
  \url{https://newyorkparkingticket.com/nyc-parking-ticket-statistics/#:~:text=In%20fiscal%20year%202022%2C%2015%2C486%2C730,about%204.7%20million%20issued%20violations.},
  2023.
\newblock [Accessed: 2024-06-29].

\bibitem{illegal-parking-traffic}
Shubhankar Gautam.
\newblock How illegal parking on roads causes traffic jams.
\newblock
  \url{https://blog.getmyparking.com/2018/11/13/how-illegal-parking-on-roads-causes-traffic-jams/},
  2018.
\newblock [Accessed: 2024-06-29].

\bibitem{fudenberg1991game}
Drew Fudenberg and Jean Tirole.
\newblock {\em Game theory}.
\newblock MIT press, 1991.

\bibitem{sinha2018stackelberg}
Arunesh Sinha, Fei Fang, Bo~An, Christopher Kiekintveld, and Milind Tambe.
\newblock Stackelberg security games: Looking beyond a decade of success.
\newblock In {\em Proceedings of the Twenty-Seventh International Joint
  Conference on Artificial Intelligence, {IJCAI-18}}, pages 5494--5501.
  International Joint Conferences on Artificial Intelligence Organization, 7
  2018.

\bibitem{pita2008deployed}
James Pita, Manish Jain, Janusz Marecki, Fernando Ord\'{o}\~{n}ez, Christopher
  Portway, Milind Tambe, Craig Western, Praveen Paruchuri, and Sarit Kraus.
\newblock Deployed armor protection: the application of a game theoretic model
  for security at the los angeles international airport.
\newblock In {\em Proceedings of the 7th International Joint Conference on
  Autonomous Agents and Multiagent Systems: Industrial Track}, AAMAS '08, page
  125–132, Richland, SC, 2008. International Foundation for Autonomous Agents
  and Multiagent Systems.

\bibitem{shieh2012protect}
Eric Shieh, Bo~An, Rong Yang, Milind Tambe, Craig Baldwin, Joseph DiRenzo, Ben
  Maule, and Garrett Meyer.
\newblock Protect: an application of computational game theory for the security
  of the ports of the united states.
\newblock In {\em Proceedings of the Twenty-Sixth AAAI Conference on Artificial
  Intelligence}, AAAI'12, page 2173–2179. AAAI Press, 2012.

\bibitem{brown2006defending}
Gerald Brown, Matthew Carlyle, Javier Salmer{\'o}n, and Kevin Wood.
\newblock Defending critical infrastructure.
\newblock {\em Interfaces}, 36(6):530--544, 2006.

\bibitem{yin2012trusts}
Zhengyu Yin, Albert~Xin Jiang, Milind Tambe, Christopher Kiekintveld, Kevin
  Leyton-Brown, Tuomas Sandholm, and John~P Sullivan.
\newblock Trusts: Scheduling randomized patrols for fare inspection in transit
  systems using game theory.
\newblock {\em AI magazine}, 33(4):59--59, 2012.

\bibitem{Xu_Bondi_Fang_Perrault_Wang_Tambe_2021}
Lily Xu, Elizabeth Bondi, Fei Fang, Andrew Perrault, Kai Wang, and Milind
  Tambe.
\newblock Dual-mandate patrols: Multi-armed bandits for green security.
\newblock {\em Proceedings of the AAAI Conference on Artificial Intelligence},
  35(17):14974--14982, May 2021.

\bibitem{yang2014adaptive}
Rong Yang, Benjamin Ford, Milind Tambe, and Andrew Lemieux.
\newblock Adaptive resource allocation for wildlife protection against illegal
  poachers.
\newblock In {\em Proceedings of the 2014 International Conference on
  Autonomous Agents and Multi-Agent Systems}, AAMAS '14, page 453–460,
  Richland, SC, 2014. International Foundation for Autonomous Agents and
  Multiagent Systems.

\bibitem{fang2015security}
Fei Fang, Peter Stone, and Milind Tambe.
\newblock When security games go green: designing defender strategies to
  prevent poaching and illegal fishing.
\newblock In {\em Proceedings of the 24th International Conference on
  Artificial Intelligence}, IJCAI'15, page 2589–2595. AAAI Press, 2015.

\bibitem{fang2016deploying}
Fei Fang, Thanh~H. Nguyen, Rob Pickles, Wai~Y. Lam, Gopalasamy~R. Clements,
  Bo~An, Amandeep Singh, Milind Tambe, and Andrew Lemieux.
\newblock Deploying paws: field optimization of the protection assistant for
  wildlife security.
\newblock In {\em Proceedings of the Thirtieth AAAI Conference on Artificial
  Intelligence}, AAAI'16, page 3966–3973. AAAI Press, 2016.

\bibitem{paruchuri2008playing}
Praveen Paruchuri, Jonathan~P Pearce, Janusz Marecki, Milind Tambe, Fernando
  Ordonez, and Sarit Kraus.
\newblock Playing games for security: An efficient exact algorithm for solving
  bayesian stackelberg games.
\newblock In {\em Proceedings of the 7th international joint conference on
  Autonomous agents and multiagent systems-Volume 2}, pages 895--902, 2008.

\bibitem{korzhyk2010complexity}
Dmytro Korzhyk, Vincent Conitzer, and Ronald Parr.
\newblock Complexity of computing optimal stackelberg strategies in security
  resource allocation games.
\newblock In {\em Proceedings of the AAAI Conference on Artificial
  Intelligence}, volume~24, pages 805--810, 2010.

\bibitem{kjtjft-2009}
Christopher Kiekintveld, Manish Jain, Jason Tsai, James Pita, Fernando
  Ord\'{o}\~{n}ez, and Milind Tambe.
\newblock Computing optimal randomized resource allocations for massive
  security games.
\newblock In {\em Proceedings of The 8th International Conference on Autonomous
  Agents and Multiagent Systems - Volume 1}, AAMAS '09, page 689–696,
  Richland, SC, 2009. International Foundation for Autonomous Agents and
  Multiagent Systems.

\bibitem{li2016catcher}
Yuqian Li, Vincent Conitzer, and Dmytro Korzhyk.
\newblock Catcher-evader games.
\newblock {\em CoRR}, abs/1602.01896, 2016.

\bibitem{becker1968crime}
Gary~S Becker.
\newblock Crime and punishment: An economic approach.
\newblock {\em Journal of political economy}, 76(2):169--217, 1968.

\bibitem{AVENHAUS20021947}
Rudolf Avenhaus, Bernhard {Von Stengel}, and Shmuel Zamir.
\newblock Chapter 51 inspection games.
\newblock volume~3 of {\em Handbook of Game Theory with Economic Applications},
  pages 1947--1987. Elsevier, 2002.

\bibitem{HARRINGTON198829}
Winston Harrington.
\newblock Enforcement leverage when penalties are restricted.
\newblock {\em Journal of Public Economics}, 37(1):29--53, 1988.

\bibitem{FRIESEN200372}
Lana Friesen.
\newblock Targeting enforcement to improve compliance with environmental
  regulations.
\newblock {\em Journal of Environmental Economics and Management},
  46(1):72--85, 2003.

\bibitem{EconomicsofEnvironmentalComplianceandEnforcement}
Wayne Gray and Ronald Shadbegian.
\newblock Economics of environmental compliance and enforcement, 06 2021.

\bibitem{su12051874}
Xu~Jing, Yao Guanxin, and Dai Panqian.
\newblock Quality decision-making behavior of bodies participating in the
  agri-foods e-supply chain.
\newblock {\em Sustainability}, 12(5), 2020.

\bibitem{SASAKI2014107}
Yasuo Sasaki.
\newblock Optimal choices of fare collection systems for public
  transportations: Barrier versus barrier-free.
\newblock {\em Transportation Research Part B: Methodological}, 60:107--114,
  2014.

\bibitem{NOURINEJAD201733}
Mehdi Nourinejad and Matthew~J. Roorda.
\newblock Parking enforcement policies for commercial vehicles.
\newblock {\em Transportation Research Part A: Policy and Practice},
  102:33--50, 2017.
\newblock SI: Freight Behavior Research.

\bibitem{che2024predictive}
Yeon-Koo Che, Jinwoo Kim, and Konrad Mierendorff.
\newblock Predictive enforcement.
\newblock {\em CoRR}, abs/2405.04764, 2024.

\bibitem{timing-inspections-predictable-2023}
Ian Ball and Jan Knoepfle.
\newblock Should the timing of inspections be predictable?
\newblock In {\em Proceedings of the 24th ACM Conference on Economics and
  Computation}, EC '23, page 206, New York, NY, USA, 2023. Association for
  Computing Machinery.

\bibitem{telle2009threat}
Kjetil Telle.
\newblock The threat of regulatory environmental inspection: impact on plant
  performance.
\newblock {\em Journal of Regulatory Economics}, 35:154--178, 2009.

\bibitem{calel2023policing}
Raphael Calel, Antoine Dech{\^e}zlepretre, and Frank Venmans.
\newblock Policing carbon markets.
\newblock 2023.

\bibitem{hochba1997approximation}
Dorit~S Hochba.
\newblock Approximation algorithms for np-hard problems.
\newblock {\em ACM Sigact News}, 28(2):40--52, 1997.

\bibitem{williamson2011design}
David~P Williamson and David~B Shmoys.
\newblock {\em The design of approximation algorithms}.
\newblock Cambridge university press, 2011.

\bibitem{kempe2003maximizing}
David Kempe, Jon Kleinberg, and \'{E}va Tardos.
\newblock Maximizing the spread of influence through a social network.
\newblock In {\em Proceedings of the Ninth ACM SIGKDD International Conference
  on Knowledge Discovery and Data Mining}, KDD '03, page 137–146, New York,
  NY, USA, 2003. Association for Computing Machinery.

\bibitem{Roughgarden_2021}
Tim Roughgarden.
\newblock {\em Resource Augmentation}, page 72–92.
\newblock Cambridge University Press, 2021.

\bibitem{auctions-vs-negotiations}
Jeremy Bulow and Paul Klemperer.
\newblock Auctions versus negotiations.
\newblock {\em The American Economic Review}, 86(1):180--194, 1996.

\bibitem{akbarpour2020just}
Mohammad Akbarpour, Suraj Malladi, and Amin Saberi.
\newblock {Just a Few Seeds More: Value of Network Information for Diffusion}.
\newblock Research Papers 3678, Stanford University, Graduate School of
  Business, April 2018.

\bibitem{akbarpout-matching-spatial}
Mohammad Akbarpour, Yeganeh Alimohammadi, Shengwu Li, and Amin Saberi.
\newblock The value of excess supply in spatial matching markets.
\newblock In {\em Proceedings of the 23rd ACM Conference on Economics and
  Computation}, EC '22, page~62, New York, NY, USA, 2022. Association for
  Computing Machinery.

\bibitem{siam-bgg-2020}
Moshe Babaioff, Kira Goldner, and Yannai~A. Gonczarowski.
\newblock {\em Bulow-Klemperer-Style Results for Welfare Maximization in
  Two-Sided Markets}, pages 2452--2471.
\newblock 2020.

\bibitem{hartline2009simple}
Jason~D. Hartline and Tim Roughgarden.
\newblock Simple versus optimal mechanisms.
\newblock In {\em Proceedings of the 10th ACM Conference on Electronic
  Commerce}, EC '09, page 225–234, New York, NY, USA, 2009. Association for
  Computing Machinery.

\bibitem{HART2017313}
Sergiu Hart and Noam Nisan.
\newblock Approximate revenue maximization with multiple items.
\newblock {\em Journal of Economic Theory}, 172:313--347, 2017.

\bibitem{hart1986theory}
Oliver Hart and Bengt Holmström.
\newblock {\em The theory of contracts}, page 71–156.
\newblock Econometric Society Monographs. Cambridge University Press, 1987.

\bibitem{bolton2004contract}
Patrick Bolton and Mathias Dewatripont.
\newblock {\em Contract theory}.
\newblock MIT press, 2004.

\bibitem{holmstrom1979moral}
Bengt Holmstr{\"o}m.
\newblock Moral hazard and observability.
\newblock {\em The Bell journal of economics}, pages 74--91, 1979.

\bibitem{grossman1992analysis}
Sanford~J Grossman and Oliver~D Hart.
\newblock An analysis of the principal-agent problem.
\newblock In {\em Foundations of Insurance Economics: Readings in Economics and
  Finance}, pages 302--340. Springer, 1992.

\bibitem{dutting2022combinatorial}
Paul Dütting, Tomer Ezra, Michal Feldman, and Thomas Kesselheim.
\newblock Combinatorial contracts.
\newblock In {\em 2021 IEEE 62nd Annual Symposium on Foundations of Computer
  Science (FOCS)}, pages 815--826, 2022.

\bibitem{dutting2024combinatorial}
Paul Dutting, Michal Feldman, and Yoav Gal~Tzur.
\newblock Combinatorial contracts beyond gross substitutes.
\newblock In {\em Proceedings of the 2024 Annual ACM-SIAM Symposium on Discrete
  Algorithms (SODA)}, pages 92--108. SIAM, 2024.

\bibitem{comb-agency-babaioff-2006}
Moshe Babaioff, Michal Feldman, and Noam Nisan.
\newblock Combinatorial agency.
\newblock In {\em Proceedings of the 7th ACM Conference on Electronic
  Commerce}, EC '06, page 18–28, New York, NY, USA, 2006. Association for
  Computing Machinery.

\bibitem{dutting-2023-multiagent}
Paul D\"{u}tting, Tomer Ezra, Michal Feldman, and Thomas Kesselheim.
\newblock Multi-agent contracts.
\newblock In {\em Proceedings of the 55th Annual ACM Symposium on Theory of
  Computing}, STOC 2023, page 1311–1324, New York, NY, USA, 2023. Association
  for Computing Machinery.

\bibitem{prensky_johnson_2023}
Matthew Prensky and Rob Johnson.
\newblock It's a trap! small towns across us use traffic tickets to collect big
  money from drivers.
\newblock
  \url{https://www.usatoday.com/story/opinion/2023/12/26/police-speeding-traffic-tickets-revenue-civil-rights/71970613007/},
  Dec 2023.
\newblock [Accessed: 2023-12-31].

\bibitem{lai2006knapsack}
Katherine Lai and M~Goemans.
\newblock The knapsack problem and fully polynomial time approximation schemes
  (fptas).
\newblock {\em Retrieved November}, 3:2012, 2006.

\bibitem{solodov2007explicit}
Mikhail Solodov.
\newblock An explicit descent method for bilevel convex optimization.
\newblock {\em Journal of Convex Analysis}, 14(2):227, 2007.

\bibitem{mmrda-website}
Mumbai metropolitan region development authority.
\newblock \url{https://mmrda.maharashtra.gov.in}, 2024.
\newblock [Accessed: 2024-01-30].

\bibitem{edmonds2003submodular}
Jack Edmonds.
\newblock Submodular functions, matroids, and certain polyhedra.
\newblock In {\em Combinatorial Optimization—Eureka, You Shrink! Papers
  Dedicated to Jack Edmonds 5th International Workshop Aussois, France, March
  5--9, 2001 Revised Papers}, pages 11--26. Springer, 2003.

\bibitem{bckm-2013}
Eric Budish, Yeon-Koo Che, Fuhito Kojima, and Paul Milgrom.
\newblock Designing random allocation mechanisms: Theory and applications.
\newblock {\em American Economic Review}, 103(2):585–623, April 2013.

\bibitem{kamada2018stability}
Yuichiro Kamada and Fuhito Kojima.
\newblock Stability and strategy-proofness for matching with constraints: A
  necessary and sufficient condition.
\newblock {\em Theoretical Economics}, 13(2):761--793, 2018.

\bibitem{bilemi-2004}
Meir Bing, Daniel Lehmann, and Paul Milgrom.
\newblock Presentation and structure of substitutes valuations.
\newblock In {\em Proceedings of the 5th ACM Conference on Electronic
  Commerce}, EC '04, page 238–239, New York, NY, USA, 2004. Association for
  Computing Machinery.

\bibitem{gul2019market}
Faruk~R. Gul, Wolfgang Pesendorfer, and Mu~Zhang.
\newblock {Market Design and Walrasian Equilibrium}.
\newblock Working Papers 2020-38, Princeton University. Economics Department.,
  June 2020.

\bibitem{jalota2022matchingtransfersdistributionalconstraints}
Devansh Jalota, Michael Ostrovsky, and Marco Pavone.
\newblock Matching with transfers under distributional constraints.
\newblock {\em CoRR}, abs/2202.05232, 2022.

\bibitem{hino2018machine}
Miyuki Hino, Elinor Benami, and Nina Brooks.
\newblock Machine learning for environmental monitoring.
\newblock {\em Nature Sustainability}, 1(10):583--588, 2018.

\bibitem{shixiang-atlanta-2022}
Shixiang Zhu, He~Wang, and Yao Xie.
\newblock Data-driven optimization for atlanta police-zone design.
\newblock {\em INFORMS Journal on Applied Analytics}, 52(5):412--432, 2022.

\bibitem{conitzer2006computing}
Vincent Conitzer and Tuomas Sandholm.
\newblock Computing the optimal strategy to commit to.
\newblock In {\em Proceedings of the 7th ACM Conference on Electronic
  Commerce}, EC '06, page 82–90, New York, NY, USA, 2006. Association for
  Computing Machinery.

\bibitem{economist-website}
The Economist.
\newblock Looking for the african middle class? head to the bus park.
\newblock
  \url{https://www.economist.com/middle-east-and-africa/2023/05/18/looking-for-the-african-middle-class-head-to-the-bus-park},
  2017.
\newblock [Accessed: 2023-05-21].

\bibitem{traffic-fines-website}
Team Acko.
\newblock New traffic rules in mumbai: List of rules and fines for violations.
\newblock \url{https://www.acko.com/traffic-rules/traffic-fines-in-mumbai/},
  2023.
\newblock [Accessed: 2024-02-03].

\end{thebibliography}

%require additional resources as compared to 0-1 knapsack

%implies ptas - Then, defining $\epsilon = \delta = \frac{1}{m+1}$, the running time of the algorithm is $O(m |L|^m \frac{1}{\delta} |L|) = O(\frac{1}{\epsilon^2} |L|^{\frac{1}{\epsilon}})$.

\appendix

\section{Additional Discussion on Modeling Assumptions} \label{apdx:model-assumptions-discussion} %\label{apdx:additional-comments}

%\subsection{Discussion of Modeling Assumptions} \label{apdx:model-assumptions-discussion}

In this section, we present some additional discussions on our modeling assumptions for the security game that we study. 

%First, in modeling the payoffs of the users and the administrator, we assume that users are utility maximizers, a standard assumption in the security games literature~\cite{yin2012trusts}, and study the administrator's resource allocation strategies under both welfare and revenue maximization objectives. A revenue maximization objective aligns with a selfish administrator seeking to maximize the overall fine collections across the different locations, a modeling assumption used in prior work on security games~\cite{yin2012trusts} with a close resemblance to practice. For instance, in road traffic scenarios, police often place speed traps where users are likely to violate the speed limit even though other locations may be more accident-prone~\cite {prensky_johnson_2023}. Moreover, a welfare maximization objective is a natural choice for an administrator seeking to ensure its security resources target the locations most susceptible to fraud.

First, as in prior literature on security games~\cite{tambe2011security}, we model our problem as a Stackelberg game, wherein users at each location observe and best respond to the administrator's mixed strategy~\cite{conitzer2006computing}. Such an assumption on the game structure is reasonable for several applications of interest to this paper. For instance, in the case of queue jumping in intermediate public transport services, users typically use the IPT service at the same location every day; thus, users can estimate the likelihood of being caught by a security resource (e.g., a police officer) on a given day. We note that similar assumptions on the game structure have been considered even in related examples in the literature, e.g., the scheduling of randomized patrols for inspecting whether users have paid their fares in transit systems~\cite{yin2012trusts}. While we assume that users observe the probability of a security resource being allocated to their respective locations precisely, exploring alternative models of equilibrium formation when such perfect observability is not possible is an exciting direction for future research.

To model the losses to users from committing fraud, we assume that if the administrator allocates a security resource to a given location, then the security resource will collect fines from all defaulting users, resulting in a loss through fines to a defaulting user of $\sigma_l k$ at location $l$. Such an assumption on the efficacy of a security resource in detecting the presence of a fraudulent user at the location to which it is allocated is ubiquitous in the literature on security games~\cite{tambe2011security} and natural, as, for instance, modern day technology, such as license plate recognition cameras, have enabled police officers to accurately issue parking citations to defaulting users in a timely manner.

%Next, in this work, we assume that if the administrator allocates a security resource to a given location, then the security resource will collect fines from any defaulting users and prevent any fraudulent or illegal activities from occurring at that location. Such an assumption on the efficacy of a security resource in preventing a security breach at the location to which it is allocated is ubiquitous in the literature on security games~\cite{tambe2011security} and also natural, as, for instance, the presence of a police officer is likely to deter users from engaging in queue jumping in IPT services. 

Moreover, we assume that the fine levied on defaulting users is a fixed (potentially small) constant. While setting high fines is likely to deter fraud, in most practical applications, fines for engaging in fraudulent or illegal activities are not arbitrarily high, e.g., traffic light or speeding violations cost about \$100 in California. Furthermore, we note that our numerical experiments in Section~\ref{sec:numerical-experiments-optimal-contracts} and Appendix~\ref{apdx:additional-numerics} highlight the benefits of setting low to moderate fines in practice as it deters an administrator from solely maximizing revenues through the collected fines.

In addition, in this work, we assume that fines are the only penalty for users engaging in fraudulent or illegal activities. Such an assumption is particularly applicable for parking violations and queue jumping in IPT services, as the penalty imposed on defaulting users is likely to only come in the form of a fine. However, for other application settings, different forms of punishment may be pertinent depending on the magnitude of the fraud~\cite{becker1968crime}, e.g., users may be sent to prison for certain crimes they commit. While we do not account for other forms of punishment beyond levying fines on defaulting users, there is scope to generalize the model presented in this work to account for different forms of punishment, which we defer to future research. 

Finally, while we consider allocating security resources, e.g., police officers, to monitor and police susceptible nodes or locations in the system, which is typical in the literature on security games~\cite{tambe2011security}, policing can be performed through many other security mechanisms. Most notably, modern security technology, such as security cameras, is frequently used to monitor certain fraudulent or illegal activities, e.g., store theft. Yet we note that while such technology has proven effective in some applications, e.g., in detecting theft at grocery stores, such technology is typically used to aid security officers in patrolling areas, e.g., license plate recognition cameras are used to enable efficient and accurate detection of parking violations. 

Moreover, while security cameras or license plate recognition cameras are often effective in detecting and fining users for road speeding or parking violations, as user records can be determined using license plate information, such technology is likely far more ineffective in monitoring queue jumping in IPT services. In particular, using cameras to identify and fine users that skip queues is difficult, particularly in developing nation cities where IPT services are prevalent, as this would require advanced facial recognition software and a database of records of the entire population in the city. Moreover, given the commotion during the hailing of IPT services and that users often steal public property such as security cameras, a frequent occurrence in developing nation contexts, camera technology on its own may not be the most effective in policing queue jumping in IPT services. Consequently, the traditional mechanism of deploying security resources, such as police officers, as is typically the case for monitoring many traffic and parking violations, despite the prevalence of camera technology, is likely the most effective mechanism in policing queue jumping in IPT services. Yet, we note that modern security technology, such as security cameras, has a role in augmenting traditional enforcement mechanisms, e.g., deploying police officers, in policing and monitoring fraudulent or illegal activities in the applications we consider. To that end, incorporating features of modern-day technology, such as security cameras, and how they can augment policing through traditional mechanisms into our studied framework is a worthwhile direction for future research.

%is unlikely to be as effective for many examples of interest in this paper, e.g., the transportation and healthcare examples presented in Section~\ref{sec:examples-pertinent}. For instance, in the case of healthcare, human intervention is often necessary to conduct thorough audits to ensure that medications have not been over-prescribed or that patients do not have to stay longer than necessary at a hospital. 

%no such identification of users is possible and thus requires in-person policing to fine users when they attempt to jump queues. Moreover, camera technology is not as sophisticated in most developing nation cities where IPT services operate, which further diminishes their effectiveness in monitoring queue-following, and the traditional mechanism of monitoring most traffic violations is through deploying police resources. 

%Yet, we note that incorporating features of modern-day technology, such as security cameras, into our studied framework is a worthwhile direction for future research.

%Furthermore, while, for expositional simplicity, we assume that the fine is constant across locations, we note that our model and results can be naturally extended to the setting when fines vary across locations (see Appendix~\ref{apdx:model-extensions}). 

%Mention other kinds of policing Appendix~\ref{apdx:policing-modern-tech}

\section{Proofs}

\subsection{Proof of Theorem~\ref{thm:greedy-opt-rev-max-deterministic}} \label{apdx:pf-thm-rev-max-deterministic}

We prove Theorem~\ref{thm:greedy-opt-rev-max-deterministic} using two intermediate lemmas. Our first intermediate lemma shows that the optimal solution of the revenue maximization bi-level Program~\eqref{eq:admin-obj-revenue}-\eqref{eq:bi-level-con-revenue} can be computed using a linear program, as is elucidated by the following lemma.

\begin{lemma} [Linear Programming Characterization of Revenue Maximization Problem] \label{lem:lp-rev-max-deterministic}
The optimal solution of the revenue maximization bi-level Program~\eqref{eq:admin-obj-revenue}-\eqref{eq:bi-level-con-revenue} can be computed through the solution of the following linear program:
\begin{maxi!}|s|[2]   
    {\sigmaa \in \mathbb{R}^{|L|}}                            
    { \sum_{l \in L} \sigma_l k \Lambda_l,  \label{eq:admin-obj-revenue-linear}}   
    {\label{eq:Eg003}}             
    {}          
    \addConstraint{\sigma_l}{\in [0, \frac{d_l}{d_l+k}], \quad \text{for all } l \in L, \label{eq:linear-rev-max-sigma}}
    \addConstraint{\sum_{l \in L} \sigma_l}{\leq R. \label{eq:linear-rev-max-resource-con}}
\end{maxi!}
\end{lemma}

For a proof of Lemma~\ref{lem:lp-rev-max-deterministic}, see Appendix~\ref{apdx:pf-lp-rev-max-deterministic}. Lemma~\ref{lem:lp-rev-max-deterministic} establishes that computing an optimal revenue maximizing strategy of the administrator can be done efficiently in polynomial time. Our next intermediate lemma states that Algorithm~\ref{alg:GreedyRevenueMaximization} achieves the optimal solution of the linear Program~\eqref{eq:admin-obj-revenue-linear}-\eqref{eq:linear-rev-max-resource-con}.

\begin{lemma} [Algorithm~\ref{alg:GreedyRevenueMaximization}] \label{lem:opt-algo-is-greedy-rev-max}
The allocation strategy corresponding to Algorithm~\ref{alg:GreedyRevenueMaximization} achieves an optimal solution to the linear Program~\eqref{eq:admin-obj-revenue-linear}-\eqref{eq:linear-rev-max-resource-con}.
%The solution of this linear program corresponds to spending $\frac{d_l}{d_l+k}$ resources on locations $l \in L$ in the descending order of $\Lambda_l$ up to the administrator's resource budget $R$.
\end{lemma}

For a proof of Lemma~\ref{lem:opt-algo-is-greedy-rev-max}, see Appendix~\ref{apdx:pf-lem-opt-algo-greedy}. Lemma~\ref{lem:opt-algo-is-greedy-rev-max} implies that the linear Program~\eqref{eq:admin-obj-revenue-linear}-\eqref{eq:linear-rev-max-resource-con} can be solved in $O(|L| \log(|L|))$ time as the time taken to sort the locations is $O(|L| \log(|L|))$ and that taken to iterate through the locations in linear in $|L|$.

Notice that jointly Lemmas~\ref{lem:lp-rev-max-deterministic} and~\ref{lem:opt-algo-is-greedy-rev-max} imply Theorem~\ref{thm:greedy-opt-rev-max-deterministic}.

\subsubsection{Proof of Lemma~\ref{lem:lp-rev-max-deterministic}} \label{apdx:pf-lp-rev-max-deterministic}

To prove this claim, we first present the linear program that can be used to compute the solution of the bi-level Program~\eqref{eq:admin-obj-revenue}-\eqref{eq:bi-level-con-revenue}. To construct this linear program, we note that any optimal solution to the bi-level Program~\eqref{eq:admin-obj-revenue}-\eqref{eq:bi-level-con-revenue} must satisfy $\sigma_l \in [0, \frac{d_l}{d_l+k}]$ for all locations $l \in L$. Note that if this relation is not satisfied and $\sigma_l > \frac{d_l}{d_l+k}$ for some location $l$, then we can consider an alternative allocation strategy $\Tilde{\sigmaa}$ where $\Tilde{\sigma}_l = \min \{ \sigma_l, \frac{d_l}{d_l+k} \}$ for all locations $l$. We note that $\Tilde{\sigmaa}$ is clearly feasible, i.e., $\Tilde{\sigmaa} \in \Omega_R$, as the strategy $\sigmaa$ is feasible. Next, by the best-response function of users in the revenue maximization setting given by Equation~\eqref{eq:best-response-users-rev-max}, we note that $y_l(\Tilde{\sigmaa}) = 1$ for all locations $l$, while $y_l(\sigmaa) = 0$ if $\sigma_l > \frac{d_l}{d_l+k}$ and $y_l(\sigmaa) = 1$ otherwise. As a consequence, we obtain for any strategy $\sigmaa$ that:
\begin{align*}
    Q_R(\sigmaa) &= \sum_{l \in L} \sigma_l y_l(\sigmaa) k \Lambda_l = \sum_{l: \sigma_l > \frac{d_l}{d_l+k}} \sigma_l y_l(\sigmaa) k \Lambda_l +  \sum_{l: \sigma_l \leq \frac{d_l}{d_l+k}} \sigma_l y_l(\sigmaa) k \Lambda_l, \\
    &\stackrel{(a)}{=} \sum_{l: \sigma_l \leq \frac{d_l}{d_l+k}} \sigma_l k \Lambda_l \stackrel{(b)}{=}  \sum_{l : \sigma_l \leq \frac{d_l}{d_l+k}} \Tilde{\sigma}_l k \Lambda_l, \\
    &\leq \sum_{l \in L} \Tilde{\sigma}_l k \Lambda_l, \\
    &\stackrel{(c)}{=} Q_R(\Tilde{\sigmaa})
\end{align*}
where (a) follows as $y_l(\sigmaa) = 0$ if $\sigma_l > \frac{d_l}{d_l+k}$ and $y_l(\sigmaa) = 1$ otherwise, (b) follows as $\Tilde{\sigma}_l = \sigma_l$ when $\sigma_l \leq \frac{d_l}{d_l+k}$, and (c) follows as $y_l(\Tilde{\sigmaa}) = 1$ for all locations $l$. 

The above analysis implies that there exists an revenue-maximizing allocation such that $\Tilde{\sigma}_l \in [0, \frac{d_l}{d_l+k}]$ for all locations $l$, where it holds by Equation~\eqref{eq:best-response-users-rev-max} that $y_l(\Tilde{\sigmaa}) = 1$ for all locations $l$. Thus, imposing the restriction that $\sigma_l \in [0, \frac{d_l}{d_l+k}]$ and that $y_l(\Tilde{\sigmaa}) = 1$ for all locations $l$, we have that the bi-level Program~\eqref{eq:admin-obj-revenue}-\eqref{eq:bi-level-con-revenue} can be reformulated as the following linear program:
\begin{maxi!}|s|[2]   
    {\sigmaa \in \mathbb{R}^{|L|}}                            
    { Q_R(\sigmaa) = \sum_{l \in L} \sigma_l k \Lambda_l,  \label{eq:admin-obj-revenue-linear2}}   
    {\label{eq:Eg004}}             
    {}          
    \addConstraint{\sigma_l}{\in [0, \frac{d_l}{d_l+k}], \quad \text{for all } l \in L, \label{eq:linear-rev-max-sigma2}}
    \addConstraint{\sum_{l \in L} \sigma_l}{\leq R, \label{eq:linear-rev-max-resource-con2}}
\end{maxi!}
where Constraints~\eqref{eq:linear-rev-max-sigma2} captures the upper bound on the administrator strategy required for optimality, Constraint~\eqref{eq:linear-rev-max-resource-con2} represents the resource constraint of the administrator, and Objective~\eqref{eq:admin-obj-revenue-linear2} represents the revenue maximization objective of the administrator when $y_l(\Tilde{\sigmaa}) = 1$, which follows from Equation~\eqref{eq:best-response-users-rev-max} under Constraint~\eqref{eq:linear-rev-max-resource-con2}. This establishes our claim.

\subsubsection{Proof of Lemma~\ref{lem:opt-algo-is-greedy-rev-max}} \label{apdx:pf-lem-opt-algo-greedy}

%Without loss of generality, we order the locations $1, \ldots, L$ in descending order of $\Lambda_l$ and consider a greedy algorithm that allocates $\frac{d_l}{d_l+k}$ for all locations $l$ (other than possibly the last location) up until its budget is exhausted. 

For simplicity of exposition, we prove this result in the setting when the values of $\Lambda_l$ are all distinct, i.e., for any two locations $l, l' \in L$, it holds that $\Lambda_l \neq \Lambda_{l'}$. We note that our analysis can be extended naturally to the setting when there are ties.

To prove this claim, let $\sigmaa^G = (\sigma_l^G)_{l \in L}$ denote the allocations of the greedy algorithm and $\sigmaa^* = (\sigma_l^*)_{l \in L}$ denote the revenue-maximizing allocation. Further, suppose for contradiction that the greedy algorithm is not optimal, i.e., $\sum_{l \in L} \sigma_l^G \Lambda_l < \sum_{l \in L} \sigma_l^* \Lambda_l$, where we have dropped the fine $k$ from the objective of the linear Program~\eqref{eq:admin-obj-revenue-linear2}-\eqref{eq:linear-rev-max-resource-con2}. Note that this implies that there is some location $l$ such that $\sigma_l^G \neq \sigma_l^*$, which consequently implies by the nature of the greedy algorithm that there is a location $\Tilde{l}$ such that $\sigma_{\Tilde{l}}^G > \sigma_{\Tilde{l}}^{*}$. However, by the feasibility constraint that $\sum_{l \in L} \sigma_l \leq R$ and the optimality of the revenue-maximizing solution it must hold that there is some location $l'$ such that $\sigma_{l'}^G < \sigma_{l'}^{*}$.

To derive our desired contradiction, we now construct another feasible strategy with a strictly higher objective than the strategy $\sigmaa^*$. In particular, consider $\epsilon>0$ as some small positive constant such that $\sigma_{\Tilde{l}}^{*} + \epsilon \leq \sigma_{\Tilde{l}}^{G}$ and $\sigma_{l'}^{*} \geq \sigma_{l'}^{G} + \epsilon$. Then, define the strategy $\sigmaa' = (\sigma_l')_{l \in L}$, where $\sigma_l' = \sigma_l^*$ for all $l \neq l', \Tilde{l}$, $\sigma_{l'} = \sigma_{l'}^* - \epsilon$ and $\sigma_{\Tilde{l}}' = \sigma_{\Tilde{l}}^{*} + \epsilon$. Note that this new strategy is feasible. In particular, it is straightforward to check that $\sigma_l' \in [0, \frac{d_l}{d_l+k}]$ for all locations $l$ and that the resource constraint is satisfied as:
\begin{align*}
    \sum_{l} \sigma_l' = \sum_{l \neq \{ l', \Tilde{l}\} } \sigma_l' + \sigma_{l'}' + \sigma_{\Tilde{l}}' = \sum_{l \neq \{ l', \Tilde{l}\} } \sigma_l^* + \sigma_{l'}^*-\epsilon + \sigma_{\Tilde{l}}^*+\epsilon = \sum_{l} \sigma_l^* \leq R,
\end{align*}
where the final inequality follows by the feasibility of $\sigmaa^*$.

Next, we show that the revenue of the new strategy $\sigmaa'$ is greater than that of $\sigmaa^*$. To see this note that:
\begin{align*}
    Q_R(\sigmaa') &= \sum_{l \in L} \sigma_l' \Lambda_l, \\
    &= \sum_{l \neq \{ l', \Tilde{l}\} } \sigma_l^* \Lambda_l + (\sigma_{l'}^*-\epsilon) \Lambda_{l'} + (\sigma_{\Tilde{l}}^* + \epsilon) \Lambda_{\Tilde{l}}, \\
    &= \sum_{l \in L} \sigma_l^* \Lambda_l + \epsilon (\Lambda_{\Tilde{l}} - \Lambda_{l'}), \\
    &> \sum_{l \in L} \sigma_l^* \Lambda_l = Q_R(\sigmaa^*),
\end{align*}
where the inequality follows as $\Lambda_{l'} < \Lambda_{\Tilde{l}}$ due to the nature of the greedy algorithm which allocates resources to locations in the descending order of the values of $\Lambda_l$. The above relationship implies that the strategy $\sigmaa'$ achieves a greater revenue compared to $\sigmaa^*$, a contradiction. Thus, it follows that Algorithm~\ref{alg:GreedyRevenueMaximization} cannot have a strictly lower objective than the optimal algorithm, implying its optimality, which establishes our claim.

\subsection{Proof of Proposition~\ref{prop:opt-mixed-strategy-solution}} \label{apdx:pf-thm-opt-mixed-strategy-soln}

Given the best-response function of the users at a given location $l$ in Equation~\eqref{eq:best-response-users}, we prove Proposition~\ref{prop:opt-mixed-strategy-solution} using two intermediate lemmas. Our first lemma establishes that, without loss of generality, it suffices to restrict our attention to administrator strategies where the total allocation of security personnel $\sigma_l$ at any location $l$ does not exceed $\frac{d_l}{d_l+k}$, i.e., establishing the first point in the statement of Proposition~\ref{prop:opt-mixed-strategy-solution}.

\begin{lemma} [Upper Bound on Administrator's Mixed Strategy Vector] \label{lem:ub-helper-mixed-strategy}
Suppose that users at each location best-respond by solving Problem~\eqref{eq:userOpt-eachLoc}. Then, there exists an administrator strategy $\Tilde{\sigmaa}^*$ that satisfies $\Tilde{\sigma}_l^* \leq \frac{d_l}{d_l+k}$ for all locations $l$ and is a solution to the Problem~\eqref{eq:admin-obj-fraud}-\eqref{eq:bi-level-con-fraud}.
%Suppose $(\sigmaa^*, \y^*)$ is a fraud minimization equilibrium, i.e., it is an optimal solution to Problem~\eqref{eq:admin-obj-fraud}-\eqref{eq:bi-level-con-fraud}. 
%For any strategy $\sigmaa$ adopted by the administrator suppose that users best-respond by solving Problem~\eqref{eq:userOpt-eachLoc} at each location $l$. Then, there exists an optimal solution to Problem~\eqref{eq:admin-obj-fraud}-\eqref{eq:bi-level-con-fraud} that satisfies $\sigma_l \leq \frac{d_l}{d_l+k}$ for all locations $l$.
\end{lemma}

For a proof of Lemma~\ref{lem:ub-helper-mixed-strategy}, see Appendix~\ref{apdx:pf-lem-ub-helper-mixed-strategy}. Lemma~\ref{lem:ub-helper-mixed-strategy} establishes that allocating more than a certain threshold of resources, as specified by the fraction $\frac{d_l}{d_l+k}$, at any location $l$ will not improve the administrator's payoff. Such a result holds as, by Equation~\eqref{eq:best-response-users}, allocating $\sigma_l = \frac{d_l}{d_l+k}$ is enough to deter users from engaging in fraudulent behavior at a given location $l$; hence, allocating any more resources than this threshold at a given location will not result in a further improvement in the administrator's payoff (see right of Figure~\ref{fig:social-welfare-revenue-best-response-plots}).

Our second intermediate lemma establishes the existence of an optimal administrator strategy $\Tilde{\sigmaa}^*$ such that for at most one location $l'$ $\Tilde{\sigma}_{l'}^* \in \left( 0, \frac{d_{l'}}{d_{l'}+k} \right)$, thereby establishing the second condition in Proposition~\ref{prop:opt-mixed-strategy-solution}.

\begin{lemma} [Structure of Optimal Mixed Strategy Solution] \label{lem:structure-opt-mixed-strategy}
Suppose that users at each location best-respond by solving Problem~\eqref{eq:userOpt-eachLoc}. Then, there exists an administrator strategy $\Tilde{\sigmaa}^*$ that is a solution to Problem~\eqref{eq:admin-obj-fraud}-\eqref{eq:bi-level-con-fraud} such that for at most one location $l'$ $\Tilde{\sigma}_{l'}^* \in \left( 0, \frac{d_{l'}}{d_{l'}+k} \right)$. For all other locations $l$, it holds that either $\Tilde{\sigma}_l^* = 0$ or $\Tilde{\sigma}_l^* = \frac{d_l}{d_l+k}$.
\end{lemma}

For a proof of Lemma~\ref{lem:structure-opt-mixed-strategy}, see Appendix~\ref{apdx:pf-lem-structure-opt-mixed-strategy}. Using the structure of the optimal administrator strategy established in Lemmas~\ref{lem:ub-helper-mixed-strategy} and~\ref{lem:structure-opt-mixed-strategy}, it is straightforward to see that the optimal administrator payoff satisfies the relation in the statement of Proposition~\ref{prop:opt-mixed-strategy-solution}. To see this, let $L_1$, $L_2$, and $l'$ be as defined in the statement of Proposition~\ref{prop:opt-mixed-strategy-solution}. Then, we have that:
\begin{align*}
    P_R^-(\Tilde{\sigmaa}^*) &= \sum_{l \in L} (1 - \Tilde{\sigma}_l^*) y_l(\Tilde{\sigmaa}^*) p_l, \\
    &= \sum_{l \in L_1} (1 - \Tilde{\sigma}^*_l) y_l(\Tilde{\sigmaa}^*) p_l + \sum_{l \in L_2} (1 - \Tilde{\sigma}_l^*) y_l(\Tilde{\sigmaa}^*) p_l + (1 - \Tilde{\sigma}_{l'}^*) y_{l'}(\Tilde{\sigmaa}^*) p_{l'}, \\
    &\stackrel{(a)}{=} \sum_{l \in L_2} (1 - \Tilde{\sigma}_l^*) y_l(\Tilde{\sigmaa}^*) p_l + (1 - \Tilde{\sigma}_{l'}^*) y_{l'}(\Tilde{\sigmaa}^*) p_{l'}, \\
    &\stackrel{(b)}{=} \sum_{l \in L_2} p_l + (1 - \Tilde{\sigma}_{l'}^*) p_{l'},
\end{align*}
where (a) follows from the fact that $y_l(\Tilde{\sigmaa}^*) = 0$ for $l \in L_1$ by Equation~\eqref{eq:best-response-users} as $\Tilde{\sigma}_l^* = \frac{d_l}{d_l+k}$ for all $l \in L_1$ and (b) follows as $y_l(\Tilde{\sigmaa}^*) = 1$ for $l \in L \backslash L_1$ and the fact that $\Tilde{\sigma}_l^* = 0$ for all $l \in L_2$ and that $\Tilde{\sigma}_{l'}^* \in \left(0, \frac{d_{l'}}{d_{l'}+k} \right)$. 

Using the above inequality, we have the following relation for the payoff under the strategy $\Tilde{\sigmaa}^*$:
\begin{align}
    P_R(\Tilde{\sigmaa}^*) &= \sum_{l \in L} p_l - P_R^-(\Tilde{\sigmaa}^*), \\
    &= \sum_{l \in L} p_l - \left( \sum_{l \in L_2} p_l + (1 - \Tilde{\sigma}_{l'}^*) p_{l'} \right), \\
    &= \sum_{l \in L_1} p_l + \Tilde{\sigma}_{l'}^* p_{l'},
\end{align}
which establishes the desired relation for the administrator's payoff under the optimal strategy $\Tilde{\sigmaa}^*$, thereby proving our claim.

%$\Tilde{\sigmaa}^*$ is given by: $F(\Tilde{\sigmaa}^*, \y(\Tilde{\sigmaa}^*)) = \sum_{l \in L_1} d_l \Lambda_l + (1-\Tilde{\sigma}_{l'}^*) \Lambda_{l'} d_{l'}$.

\subsubsection{Proof of Lemma~\ref{lem:ub-helper-mixed-strategy}} \label{apdx:pf-lem-ub-helper-mixed-strategy}

Suppose that $\sigmaa^*$ is a payoff maximizing strategy of the administrator and $\y(\sigmaa^*)$ is the corresponding best-response of users given by Equation~\eqref{eq:best-response-users}, i.e., $(\sigmaa^*, \y(\sigmaa^*))$ is an optimal solution to Problem~\eqref{eq:admin-obj-fraud}-\eqref{eq:bi-level-con-fraud}. Then, we define a new feasible strategy $\Tilde{\sigmaa}^*$ that satisfies $\Tilde{\sigma}_l^* \leq \frac{d_l}{d_l+k}$ for all locations $l$ and show that it achieves the same payoff as the strategy $\sigmaa^*$ using users' best response function given by Equation~\eqref{eq:best-response-users}. In particular, define $\Tilde{\sigmaa}^*$ such that $\Tilde{\sigma}_l^* = \min \{ \sigma_l^*, \frac{d_l}{d_l+k} \}$ for all locations $l$. 

It is straightforward to check that $\Tilde{\sigmaa}^*$ is a feasible strategy, i.e., $\Tilde{\sigmaa}^* \in \Omega_R$. To see this, note that $\Tilde{\sigma}_l^* = \min \{ \sigma_l^*, \frac{d_l}{d_l+k} \} \geq 0$ for all locations $l$ as $\sigma_l^* \geq 0$ by its feasibility to Problem~\eqref{eq:admin-obj-fraud}-\eqref{eq:bi-level-con-fraud}. Furthermore, we have that $\Tilde{\sigma}_l^* = \min \{ \sigma_l^*, \frac{d_l}{d_l+k} \} \leq 1$ as $\frac{d_l}{d_l+k} \leq 1$ for all locations $l$. Finally, since it holds that $\sum_{l \in L} \sigma_l^* \leq R$ by the feasibility of $\sigmaa^*$ to Problem~\eqref{eq:admin-obj-fraud}-\eqref{eq:bi-level-con-fraud}, it follows that
\begin{align*}
    \sum_{l \in L} \Tilde{\sigma}_l^* = \sum_{l \in L} \min \left\{ \sigma_l^*, \frac{d_l}{d_l+k} \right\} \leq \sum_{l \in L} \sigma_l^* \leq R,
\end{align*}
which establishes the feasibility of strategy $\Tilde{\sigmaa}^*$ for Problem~\eqref{eq:admin-obj-fraud}-\eqref{eq:bi-level-con-fraud}. 

Next, to see that $\Tilde{\sigmaa}^*$ achieves the same payoff for the administrator as $\sigmaa^*$, we first note that the best-response strategy of users at each location is unchanged under $\sigmaa^*$ and $\Tilde{\sigmaa}^*$, i.e., $\y(\sigmaa^*) = \y(\Tilde{\sigmaa}^*)$. To see this, we show that $y_l(\sigmaa^*) = y_l(\Tilde{\sigmaa}^*)$ for all locations $l$ by considering two cases: (i) $\sigma_l^* < \frac{d_l}{d_l+k}$ and (ii) $\sigma_l^* \geq \frac{d_l}{d_l+k}$. Note that in the first case, $\Tilde{\sigma}_l^* = \sigma_l^*$ and thus it must be that $y_l(\sigmaa^*) = y_l(\Tilde{\sigmaa}^*) = 1$. On the other hand, in the second case, note that $y_l(\sigmaa^*) = 0$ and that $\Tilde{\sigma}_l^* = \frac{d_l}{d_l+k}$, from which it follows by Equation~\eqref{eq:best-response-users} that $y_l(\Tilde{\sigmaa}^*) = 0$ as well. Having shown that $\y(\sigmaa^*) = \y(\Tilde{\sigmaa}^*)$, we now obtain that:
\begin{align*}
    P_R^-(\sigmaa^*) &= \sum_{l \in L} (1-\sigma_l^*) y_l^*(\sigmaa^*) p_l, \\
    &\stackrel{(a)}{=} \sum_{l \in L: y_l(\sigmaa^*) = 0} (1-\sigma_l^*) y_l^*(\sigmaa^*) p_l + \sum_{l \in L: y_l(\sigmaa^*) = 1} (1-\sigma_l^*) y_l^*(\sigmaa^*) p_l, \\
    &\stackrel{(b)}{=} \sum_{l \in L: y_l(\Tilde{\sigmaa}^*) = 1} (1-\Tilde{\sigma_l}^*) y_l^*(\Tilde{\sigmaa}^*)p_l, \\
    &\stackrel{(c)}{=} P_R^-(\Tilde{\sigmaa}^*)
\end{align*}
where (a) follows by splitting the sum, (b) follows from the fact that $\Tilde{\sigma}_l^* = \sigma_l^*$ for all $l$ where $y_l(\sigmaa^*) = 1$, as in this regime $\sigma_l^* < \frac{d_l}{d_l+k}$ by Equation~\eqref{eq:best-response-users}, and (c) follows from the fact that $\y(\sigmaa^*) = \y(\Tilde{\sigmaa}^*)$.

Finally, by the optimality of $\sigmaa^*$ and from the above obtained relations for the new administrator strategy $\Tilde{\sigmaa}^*$, it holds $P_R^-(\Tilde{\sigmaa}^*) = P_R^-(\sigmaa^*) \leq P_R^-(\sigmaa)$, for all feasible $\sigmaa \in \Omega_R$. Thus, we have established our claim that there exists an administrator strategy $\Tilde{\sigmaa}^*$ that satisfies $\Tilde{\sigma}_l^* \leq \frac{d_l}{d_l+k}$ for all locations $l$ and is a solution to Problem~\eqref{eq:admin-obj-fraud}-\eqref{eq:bi-level-con-fraud}.
%\end{proof}

\subsubsection{Proof of Lemma~\ref{lem:structure-opt-mixed-strategy}} \label{apdx:pf-lem-structure-opt-mixed-strategy}

%\begin{proof}
To prove this claim, we show that any administrator strategy $\sigmaa^*$ can be transformed into another strategy $\Tilde{\sigmaa}^*$ satisfying the condition in the statement of the lemma with at most the same payoff as that corresponding to $\sigmaa^*$.

To see this, consider an optimal administrator strategy $\sigmaa^*$ such that the strategies $(\sigmaa^*, \y(\sigmaa^*))$ are an optimal solution to Problem~\eqref{eq:admin-obj-fraud}-\eqref{eq:bi-level-con-fraud}. Further, suppose that there are at least two locations $l_1$ and $l_2$ such that $\sigma_{l_1}^* \in \left( 0, \frac{d_{l_1}}{d_{l_1}+k} \right)$ and $\sigma_{l_2}^* \in \left( 0, \frac{d_{l_2}}{d_{l_2}+k} \right)$. In particular, let $L' \subseteq L$ be the set of locations such that for any $l' \in L'$ it holds that $\sigma_{l'}^* \in \left( 0, \frac{d_{l'}}{d_{l'}+k} \right)$. Then, we order the locations in the set $L'$ in descending order of the payoffs $p_l$ and without loss of generality number these locations $1, \ldots, |L'|$. Then, we construct $\Tilde{\sigmaa}^*$ through a water-filling approach. To elucidate this approach, we begin by transferring the mass $\sigma_{|L'|}^*$ corresponding the location with the smallest payoff $p_l$ in the set $L'$ to the location with the highest payoff $p_l$ in the set $L'$. In particular, we transfer $\min \{ \sigma_{|L'|}^*, \frac{d_1}{d_1+k} - \sigma_{1}^* \}$ from location $|L'|$ to location $1$. If location $|L'|$'s mass is exhausted before filling up location one to $\frac{d_1}{d_1+k}$, then we transfer further mass of location $|L'-1|$ to location one. On the other hand, if location $1$ is filled up to $\frac{d_1}{d_1+k}$ with location $|L'|$, then the remaining mass from location $|L'|$ is transferred to location two in an analogous way. We then repeat this process of transferring mass from locations with lower payoffs $p_l$ to those with higher payoffs $p_l$ in the set $L'$ until there is at most one location remaining such that $\Tilde{\sigma}_{l'}^* \in \left( 0, \frac{d_{l'}}{d_{l'}+k} \right)$, where $\Tilde{\sigmaa}^*$ is the administrator strategy constructed by the above procedure.

Having constructed the strategy $\Tilde{\sigmaa}^*$ in the above manner, it is then straightforward to see that $P_R^-(\Tilde{\sigmaa}^*) \leq P_R^-(\sigmaa^*)$ as:
\begin{align*}
    P_R^-(\sigmaa^*) &= \sum_{l \in L} (1-\sigma_l^*) y_l(\sigmaa^*) p_l, \\
    &= \sum_{l \in L'} (1-\sigma_l^*) y_l(\sigmaa^*) p_l + \sum_{l \in L \backslash L'} (1-\sigma_l^*) y_l(\sigmaa^*) p_l , \\
    &\stackrel{(a)}{=} \sum_{l \in L'} (1-\sigma_l^*) y_l(\sigmaa^*) p_l + \sum_{l \in L \backslash L'} (1-\Tilde{\sigma}_l^*) y_l(\Tilde{\sigmaa}^*) p_l , \\
    &\stackrel{(b)}{\geq} \sum_{l \in L'} (1-\Tilde{\sigma}_l^*) y_l(\sigmaa^*) p_l + \sum_{l \in L \backslash L'} (1-\Tilde{\sigma}_l^*) y_l(\Tilde{\sigmaa}^*) p_l , \\
    &\stackrel{(c)}{\geq} \sum_{l \in L'} (1-\Tilde{\sigma}_l^*) y_l(\Tilde{\sigmaa}^*) p_l + \sum_{l \in L \backslash L'} (1-\Tilde{\sigma}_l^*) y_l(\Tilde{\sigmaa}^*) p_l , \\
    &= P_R^-(\Tilde{\sigmaa}^*),
\end{align*}
where (a) follows as $\sigma_l^* = \Tilde{\sigma}_l^*$ for all $l \in L \backslash L'$, (b) follows by our algorithmic procedure of constructing $\Tilde{\sigmaa}^*$ as we transfer mass from locations with lower payoffs $p_l$ to those with higher payoffs $p_l$ in the set $L'$, and (c) follows as $y_l(\Tilde{\sigmaa}^*) \leq 1 = y_l(\sigmaa^*)$ for all $l \in L'$.

The above analysis establishes our claim that there exists an administrator strategy $\Tilde{\sigmaa}^*$ that is a solution to Problem~\eqref{eq:admin-obj-fraud}-\eqref{eq:bi-level-con-fraud} such that for at most one location $l'$, it holds that $\Tilde{\sigma}_{l'}^* \in \left( 0, \frac{d_{l'}}{d_{l'}+k} \right)$, which establishes our claim.

\subsection{Proof of Theorem~\ref{thm:npHardness-swm-fm}} \label{apdx:pf-np-hardness-swm-fm}

%\subsection{Reduction from Partition} \label{subsec:partition-reduction-swm-fm}

We prove this result through a reduction from an instance of the partition problem. A partition instance consists of a sequence of numbers $a_1, \ldots, a_n$ with $\sum_{l \in [n]} a_l = A$ and involves the task of deciding whether there is some subset $S_1$ of numbers such that $\sum_{l \in S_1} a_l = \frac{A}{2}$. Without loss of generality, we consider a partition instance where $\max_{l \in [n]} a_l \leq \frac{A}{2}$. Note that if $\max_{l \in [n]} a_l > \frac{A}{2}$, then clearly, there is no subset $S_1$ of numbers such that $\sum_{l \in S_1} a_l = \frac{A}{2}$, i.e., such an instance of partition can be solved in polynomial time.

We now construct an instance of the payoff maximization problem. In particular, we consider an instance with $n+1$ locations, where the first $n$ locations correspond to each number of the partition instance, where we define $p_l = a_l$ and let $\frac{d_l}{d_l+k} = \frac{a_l}{A}$ for all locations $l \in [n]$. Furthermore, we consider a location $n+1$ such that $p_{n+1} = \max_{l \in [n]} p_l + \epsilon$ and $\frac{d_{n+1}}{d_{n+1}+k} = \frac{1}{2} + \delta$, where $\epsilon>0$ and $\delta>0$ can be interpreted as small constants. We choose $\epsilon$ such that $\max_{l \in [n]} a_l + \epsilon < A$, which is well defined as we consider partition instances where $\max_{l \in [n]} a_l \leq \frac{A}{2}$. Finally, we let the total resource budget $R = \frac{1}{2}$. Given this instance of the payoff maximization problem, we now show that a sequence of numbers correspond to a ``Yes'' instance of the partition problem if and only if the objective of Problem~\eqref{eq:admin-obj-fraud}-\eqref{eq:bi-level-con-fraud} for above defined instance is at most $\frac{A}{2} + v_{n+1}$, i.e., the total payoff is at least $\frac{A}{2}$.

($\implies$:) We first suppose that we have a ``Yes'' instance of the partition problem, i.e., there exists a set $S$ of numbers such that $\sum_{l \in S} a_l = \frac{A}{2}$. In this case, we present a method to construct a feasible resource allocation strategy $\sigmaa$ of the administrator that achieves $P_R(\sigmaa) \geq \frac{A}{2}$. In particular, we consider an allocation $\sigmaa$ for the above defined instance of the payoff maximization problem, where $\sigma_l = \frac{d_l}{d_l+k}$ units of resources are allocated to all locations $l \in S$ and $\sigma_l = 0$ for all $l \in [n+1] \backslash S$. Such an allocation is feasible for this instance of the payoff maximization problem as:
\begin{align*}
    \sum_{l \in [n+1]} \sigma_l = \sum_{l \in S} \frac{d_l}{d_l+k} = \sum_{l \in S} \frac{a_l}{A} = \frac{1}{A} \sum_{l \in S} a_l = \frac{1}{A} \times \frac{A}{2} = \frac{1}{2} = R,
\end{align*}
i.e., the resource constraint is satisfied.

Furthermore, the objective of Problem~\eqref{eq:admin-obj-fraud}-\eqref{eq:bi-level-con-fraud} under the strategy $\sigmaa$ is given by 
\begin{align*}
    P_R^-(\sigmaa) &= \sum_{l \in [n+1]} (1 - \sigma_l) y_l(\sigmaa) p_l, \\
    &\stackrel{(a)}{=} \sum_{l \in S} (1 - \frac{d_l}{d_l+k}) y_l(\sigmaa) p_l + \sum_{l \in [n+1] \backslash S} (1 - 0) y_l(\sigmaa) p_l, \\
    &\stackrel{(b)}{=} \sum_{l \in [n] \backslash S} p_l + p_{n+1}, \\
    &\stackrel{(c)}{=} \frac{A}{2} + p_{n+1},
\end{align*}
where (a) follows as $\sigma_l = \frac{d_l}{d_l+k}$ for all $l \in S$ and $\sigma_l = 0$ for all $l \in [n+1] \backslash S$, (b) follows as $y_l(\sigmaa) = 0$ when $\sigma_l = \frac{d_l}{d_l+k}$ and $y_l(\sigmaa) = 1$ when $\sigma_l = 0$ by the best-response of users in Equation~\eqref{eq:best-response-users}, and (c) follows as $\sum_{l \in [n] \backslash S} p_l = \frac{A}{2}$. The above relation implies that the payoff of the allocation $\sigmaa$ is given by
\begin{align*}
    P_R(\sigmaa) = \sum_{l \in [n+1]} p_l -  P_R^-(\sigmaa) = A + p_{n+1} - \left( \frac{A}{2} + p_{n+1} \right) = \frac{A}{2}.
\end{align*}
The above relation implies that $P_R(\sigmaa) \geq \frac{A}{2}$, and thus we have shown that a ``Yes'' instance of the partition problem implies that the administrator's optimal payoff for the above defined instance is at least $\frac{A}{2}$, which establishes the forward direction of our claim.

($\impliedby$:) Next, suppose that the optimal resource allocation strategy $\sigmaa$ of the administrator is such that $P_R(\sigmaa) \geq \frac{A}{2}$, i.e., $P_R^-(\sigmaa) \leq \frac{A}{2} + p_{n+1}$ where $\sigmaa \in \Omega_R$ satisfies $\sum_{l \in [n+1]} \sigma_l \leq R = \frac{1}{2}$. From the structure of the optimal administrator strategy established in Proposition~\ref{prop:opt-mixed-strategy-solution}, we note that the objective of Problem~\eqref{eq:admin-obj-fraud}-\eqref{eq:bi-level-con-fraud} is given by:
\begin{align} \label{eq:helper-wr-minus-relation}
    P_R^-(\sigmaa) = \sum_{l \in L_2} p_l + (1-\sigma_{l'} ) p_{l'},
\end{align}
for some location $l'$ and a set of locations $L_2$, where it holds that $\sum_{l \in [n+1] \backslash L_2} \sigma_l \leq \frac{1}{2}$. Then, we show that we have a ``Yes'' instance of the partition problem by considering two cases: (i) $l' = \emptyset$ and (ii) $l' \neq \emptyset$. 

\textbf{Case (i):} In the setting when $l' = \emptyset$, without loss of generality, for ease of notation, we define $L_2 \cup \{ n+1\}$ as the set of locations for which $\sigma_l = 0$ and $L_1$ as the set of locations for which $\sigma_l = \frac{d_l}{d_l+k}$. Note here that as $R = \frac{1}{2} < \frac{1}{2}+\delta = \frac{d_{n+1}}{d_{n+1}+k}$ it follows for location $n+1$ that $\sigma_{n+1} = 0$ as $l' = \emptyset$. Next, for our payoff maximization instance, we note by assumption that the following two inequalities are satisfied: (i) $\sum_{l \in L_2 \cup \{ n+1\}} p_l \leq \frac{A}{2} + p_{n+1}$ and (ii) $\sum_{l \in L_1} \frac{a_l}{A} = \sum_{l \in L_1} \frac{d_l}{d_l+k} \leq \frac{1}{2}$. From the first inequality, we have that
\begin{align*}
    \sum_{l \in L_1} a_l = \sum_{l \in L_1} p_l = \sum_{l \in L} p_l - \sum_{l \in L_2 \cup \{ n+1\}} p_l \geq \frac{A}{2},
\end{align*}
where the inequality follows as $\sum_{l \in [n+1]} a_l = A + p_{n+1}$ and the fact that $\sum_{l \in L_2 \cup \{ n+1\}} p_l \leq \frac{A}{2} + p_{n+1}$. Combining this derived inequality with the resource constraint that $\sum_{l \in L_1} \frac{a_l}{A} \leq \frac{1}{2}$, it follows that $\sum_{l \in L_1} a_l = \frac{A}{2}$, i.e., we have a ``Yes'' instance of the partition problem.

\textbf{Case (ii):} In the setting when $l' \neq \emptyset$, we first show that it must follow that $l' = \{ n+1\}$. At a high-level, this result follows as location $n+1$ has the highest payoff of $p_l$ among all the locations. We now proceed by contradiction to establish this claim by supposing that $l' \neq \{ n+1\}$. Next, note that by our constructed instance that $\frac{d_{n+1}}{d_{n+1}+k} = \frac{1}{2} + \delta > R$ and thus the location $\{n+1\} \in L_2$. 

To derive our contradiction, we consider an allocation $\Tilde{\sigmaa}$ where $\Tilde{\sigma}_{l'} = 0$ and $\Tilde{\sigma}_{n+1} = \sigma_{l'}$. Note that $\sigma_{l'} \in \left( 0, \frac{d_{l^{'}}}{d_{l^{'}}+k}\right)$ and thus it also holds that $\sigma_{l'} \in \left( 0, \frac{d_{n+1}}{d_{n+1}+k}\right)$, as $\frac{d_{n+1}}{d_{n+1}+k} > \frac{1}{2} \geq \frac{\max_{l \in [n]} a_l}{A} = \frac{d_{l^{'}}}{d_{l^{'}}+k}$. Clearly the new allocation $\Tilde{\sigmaa}$ is feasible as $\sigmaa$ is feasible. Next, we observe that
\begin{align*}
    P_R^-(\sigmaa) &\stackrel{(a)}{=} \sum_{l \in L_2} p_l + (1-\sigma_{l'} ) p_{l'} = \sum_{l \in L_2 \backslash \{ n+1\}} p_l + p_{n+1} + (1-\sigma_{l'} ) p_{l'}, \\
    &\stackrel{(b)}{>} \sum_{l \in L_2 \cup \{ l'\}} p_l + (1-\sigma_{l'} ) p_{n+1}, \\
    &\stackrel{(c)}{=} \sum_{l \in L_2 \cup \{ l'\}} p_l + (1-\Tilde{\sigma}_{n+1} ) p_{n+1}, \\
    &\stackrel{(d)}{=} P_R^-(\Tilde{\sigmaa}),
\end{align*}
where (a) follows from Equation~\eqref{eq:helper-wr-minus-relation}, (b) follows as $p_{n+1} = \max_{l \in [n]} v_l + \epsilon > p_{l'}$, (c) follows as $\Tilde{\sigma}_{n+1} = \sigma_{l'}$, and (d) follows by the definition of $P_R^-(\Tilde{\sigmaa})$ and the fact that the only change in the allocation between $\sigmaa$ and $\Tilde{\sigmaa}$ is for locations $l'$ and $n+1$. The above relations imply that the allocation strategy $\Tilde{\sigmaa}$ achieves a lower objective for Problem~\eqref{eq:admin-obj-fraud}-\eqref{eq:bi-level-con-fraud} than the strategy $\sigmaa$, a contradiction. Thus, it follows that if $l' \neq \emptyset$, then it must be that $l' = \{ n+1\}$ for our above defined instance of the payoff maximization problem.

%In the setting when $l' \neq \emptyset$, we note that it must be the case that $l' = \{ n+1\}$. Such a relation holds as, by our constructed instance, $v_{n+1} = \max_{l \in [n]} v_l + \epsilon$ and the fact that $\frac{d_{n+1}}{d_{n+1}+k} = \frac{1}{2} + \delta > R$. 

%In particular, since it holds that $l^{'} \neq \emptyset$, there is some location $l^{''}$ such that $\sigma_{l^{''}} \in \left( 0, \frac{d_{l^{''}}}{d_{l^{''}}+k}\right)$. Note that $y_{l^{''}}(\sigmaa) = 1$ and thus the accrued welfare is $\sigma_{l^{''}} v_{l^{''}}$

%Hence, if $l' \neq \emptyset$, then $l'$ must correspond to location $n+1$ as it has the highest value of $v_l$ among all the locations. In particular, it holds that $\sigma_{n+1} \in \left( 0, \frac{d_{n+1}}{d_{n+1}+k}\right)$. 

Next, recall that $L_2$ are the set of locations for which $\sigma_l = 0$ and $L_1$ are the set of locations for which $\sigma_l = \frac{d_l}{d_l+k}$. Then, for our payoff maximization instance, we have that the following two inequalities hold:
\begin{align}
    &\sum_{l \in L_2} p_l + (1-\sigma_{n+1}) p_{n+1} \leq \frac{A}{2} + p_{n+1}, \label{eq:helper-1-} \\
    &\sum_{l \in L_2} \frac{d_l}{d_l+k} + \sigma_{n+1} \leq \frac{1}{2}. \label{eq:helper-2-}
\end{align}
The first inequality represents the fact that the optimal objective of the Problem~\eqref{eq:admin-obj-fraud}-\eqref{eq:bi-level-con-fraud} is at most $\frac{A}{2} + p_{n+1}$ and the second inequality implies that the total resource spending does not exceed the available resources.

Then, noting that $\sum_{l \in [n+1]} p_l = A + p_{n+1}$, we have from the Equation~\eqref{eq:helper-1-} that:
\begin{align*}
    \sum_{l \in L_1} p_l + \sigma_{n+1} p_{n+1} = \sum_{l \in [n+1]} p_l - \left( \sum_{l \in L_2} p_l + (1-\sigma_{n+1}) p_{n+1} \right) \geq A + p_{n+1} - \left( \frac{A}{2} + p_{n+1} \right) = \frac{A}{2}.
\end{align*}
Substituting $p_l = a_l$ and $p_{n+1} = \max_{l \in [n]} a_l + \epsilon$ in the above inequality, we obtain that:
\begin{align} \label{eq:modified-objective-helper-}
    \sum_{l \in L_1} a_l + \sigma_{n+1} (\max_{l \in [n]} a_l + \epsilon) \geq \frac{A}{2}.
\end{align}
Moreover, substituting $\frac{d_l}{d_l+k} = \frac{a_l}{A}$ in Equation~\eqref{eq:helper-2-}, we obtain that:
\begin{align} \label{eq:mnodified-resource-helper-}
    \sum_{l \in L_1} a_l + \sigma_{n+1} A \leq \frac{A}{2}.
\end{align}
Combining Equations~\eqref{eq:modified-objective-helper-} and~\eqref{eq:mnodified-resource-helper-}, we obtain that:
\begin{align*}
    \sigma_{n+1} (\max_{l \in [n]} a_l + \epsilon) \geq \sigma_{n+1} A.
\end{align*}
Note that this inequality can only be satisfied if either $\sigma_{n+1} = 0$ or $\max_{l \in [n]} a_l + \epsilon \geq A = \sum_{l \in [n]} a_l$. Note that the latter inequality cannot be satisfied as $\epsilon$ has been chosen to be a small constant such that $\max_{l \in [n]} a_l + \epsilon < A$, where, recall that $\max_{l \in [n]} a_l \leq \frac{A}{2}$. As a result, the above relation implies $\sigma_{n+1} = 0$, which, from Equations~\eqref{eq:modified-objective-helper-} and~\eqref{eq:mnodified-resource-helper-}, implies the following two relations:
\begin{align*}
    \sum_{l \in L_1} a_l \geq \frac{A}{2} \text{ and } \sum_{l \in L_1} a_l \leq \frac{A}{2},
\end{align*}
which together imply that $\sum_{l \in L_1} a_l = \frac{A}{2}$. Thus, again, we have a ``Yes'' instance of partition, which establishes our claim.

\subsection{Proof of Theorem~\ref{thm:greedy-half-approx-fraud-min}} \label{apdx:pf-greedy-half-approx-welfare}

\subsubsection{Proof Overview} \label{apdx:sketch-pf-thm-greedy-half-approx-welfare}

We first define a linear program resembling a fractional knapsack like optimization and leverage Proposition~\ref{prop:opt-mixed-strategy-solution} to show that the optimal administrator payoff is upper bounded by the optimal objective of this linear program. Next, since the optimal objective of the fractional knapsack problem satisfies $\sum_{l \in S} p_l + x_{\Tilde{l}} p_{\Tilde{l}}$ for some subset of locations $S$ and a location $\Tilde{l}$ with $x_{\Tilde{l}} \leq \frac{d_{\Tilde{l}}}{d_{\Tilde{l}}+k}$, to establish the desired half approximation, we show that (i) $P_R(\sigmaa') \geq x_{\Tilde{l}} p_{\Tilde{l}}$ and (ii) $P_R(\Tilde{\sigmaa}) \geq \sum_{l \in S} p_l$, where the allocations $\sigmaa'$ and $\Tilde{\sigmaa}$ are as defined in Algorithm~\ref{alg:GreedyFraudminimizationDeterministic}. Notice that the proof of claim (i) is by construction as $\sigmaa'$ by definition is chosen to maximize the payoff from spending on a single location. The proof of claim (ii) relies on the fact that the greedy algorithm in Step 1 on Algorithm~\ref{alg:GreedyFraudminimizationDeterministic} precisely corresponds to the greedy algorithm to optimize the knapsack linear program. Finally, combining the results of claims (i) and (ii), the desired result follows.

\subsubsection{Complete Proof of Theorem~\ref{thm:greedy-half-approx-fraud-min}}

To prove this claim, we first recall from Proposition~\ref{prop:opt-mixed-strategy-solution} that the payoff maximizing objective of the administrator at the solution $\sigmaa^*$ is given by $P_R(\sigmaa^*) = \sum_{l \in L_1} p_l + \sigma_{\Tilde{l}}^* p_{\Tilde{l}}$ for some subset of locations $L_1$ where $\sigma_l^* = \frac{d_l}{d_l+k}$ for all $l \in L_1$ and a location $\Tilde{l}$ such that $\sum_{l \in L_1} \frac{d_l}{d_l+k} + \sigma_{\Tilde{l}}^* \leq R$. Furthermore, we consider the following knapsack optimization problem:
\begin{align} \label{eq:fractional-knapsack-ub}
    G_R' = \max_{0 \leq x_l \leq 1} \sum_{l \in L} x_l p_l \text{ s.t. } \sum_{l \in L} t_l x_l \leq R.
\end{align}
Then, to prove this claim, we proceed in two steps. First, we show that the optimal payoff-maximizing objective of the administrator is no more than $G_R'$, the optimal objective corresponding to Problem~\eqref{eq:fractional-knapsack-ub}. Then, we use this result and the property of the optimal payoff-maximizing objective of the administrator from Proposition~\ref{prop:opt-mixed-strategy-solution} to establish the half approximation guarantee of Algorithm~\ref{alg:GreedyFraudminimizationDeterministic}.

\paragraph{Proof of $P_R(\sigmaa^*) \leq G_R'$:} To prove this claim, first note by the structure of the optimal solution established in Proposition~\ref{prop:opt-mixed-strategy-solution} that $\sigma^*_l \geq 0$ and $\sigma^*_l \leq \frac{d_l}{d_l+k}$ for all locations $l$, and by the feasibility of $\sigmaa^*$ that $\sum_{l \in L} \sigma^*_l \leq R$. Then, we transform $\sigmaa^*$ to a feasible solution of Problem~\eqref{eq:fractional-knapsack-ub} by defining $x_l = \frac{1}{t_l} \sigma_l^*$ for all $l$. Note that $\x = (x_l)_{l \in L}$ is a feasible solution to Problem~\eqref{eq:fractional-knapsack-ub} as, by definition of $\x$ it holds that $0 \leq x_l \leq 1$, and the resource constraint is satisfied:
\begin{align*}
    \sum_{l \in L} t_l x_l = \sum_{l \in L} t_l \frac{1}{t_l} \sigma_l^* = \sum_{l \in L} \sigma_l^* \leq R.
\end{align*}
Next, to show that $G_R' \geq P_R(\sigmaa^*)$, we have:
\begin{align*}
    G_R' &\stackrel{(a)}{\geq} \sum_{l \in L} x_l p_l \stackrel{(b)}{=} \sum_{l \in L_1} \frac{1}{t_l} \sigma_l^* p_l + \frac{1}{t_{\Tilde{l}}} \sigma_{\Tilde{l}}^* p_{\Tilde{l}} \stackrel{(c)}{\geq} \sum_{l \in L_1} p_l + \sigma_{\Tilde{l}}^* p_{\Tilde{l}} \stackrel{(d)}{=} P_R(\sigmaa^*),
\end{align*}
where (a) follows from the fact that $\x = (x_l)_{l \in L}$ is one of the feasible solutions to Problem~\eqref{eq:fractional-knapsack-ub}, (b) follows from the fact that $x_l = \frac{1}{t_l} \sigma^*_l$ and that $\sigma^*_l = 0$ for $l \notin L_1 \cup \{ \Tilde{l} \}$, (c) follows from the fact that $\sigma_{l}^* = \frac{d_l}{d_l+k} = t_l$ for all $l \in L_1$ and that $\frac{1}{t_{\Tilde{l}}} \geq 1$, and (d) follows by the definition of $P_R(\sigmaa^*)$. Thus, we have shown that $G_R' \geq P_R(\sigmaa^*)$.

\paragraph{Proving the Approximation Ratio:} Finally, to complete the proof of our claim, notice that as Problem~\eqref{eq:fractional-knapsack-ub} is a fractional knapsack problem that $G_R' = \sum_{l \in S} p_l + x_{l'}^* p_{l'}$ for some subset $S$ and allocation $\x^*$, where resources are allocated to locations in the descending orders of the affordable bang-per-buck ratios. Next, note by Algorithm~\ref{alg:GreedyFraudminimizationDeterministic} that 
\begin{align} \label{eq:helper-greedy-lb}
    P_R(\sigmaa^*_A) \geq P_R(\Tilde{\sigmaa}) \geq \sum_{l \in S} p_l,
\end{align}
where recall that $\Tilde{\sigmaa}$ is the allocation corresponding to the greedy algorithm in Step 1 of Algorithm~\ref{alg:GreedyFraudminimizationDeterministic} and the second inequality follows as the greedy algorithm selects at least the first $S$ locations in the descending order of their affordable bang-per-buck values. Furthermore, by Step 2 of Algorithm~\ref{alg:GreedyFraudminimizationDeterministic} and the construction of the allocation $\sigmaa'$, it follows that 
\begin{align} \label{eq:helper-single-location-lb}
    P_R(\sigmaa^*_A) \geq P_R(\sigmaa') \geq x_{l'}^* p_{l'}.
\end{align}
Then, summing Equations~\eqref{eq:helper-greedy-lb} and~\eqref{eq:helper-single-location-lb}, we obtain that:
\begin{align*}
    2 P_R(\sigmaa^*_A) \geq \sum_{l \in S} p_l + x_{l'}^* p_{l'} = G_R' \geq P_R(\sigmaa^*),
\end{align*}
which establishes our claim that $P_R(\sigmaa^*_A) \geq \frac{1}{2} P_R(\sigmaa^*)$, i.e., Algorithm~\ref{alg:GreedyFraudminimizationDeterministic} is a half approximation to the optimal solution to the bi-level Program~\eqref{eq:admin-obj-fraud}-\eqref{eq:bi-level-con-fraud}.

\subsection{Proof of Theorem~\ref{thm:greedy-resource-augmentation-frauad-min}} \label{apdx:pf-resource-augmentation-welfare-max}

To prove this result, we again consider the fractional knapsack Problem~\eqref{eq:fractional-knapsack-ub} defined in the proof of Theorem~\ref{thm:greedy-half-approx-fraud-min} and use the fact that the objective of the fractional knapsack Problem~\eqref{eq:fractional-knapsack-ub} upper bounds the optimal administrator payoff, given by $P_R(\sigmaa^*)$. In particular, to establish the desired resource augmentation guarantee, we first recall from the optimal solution of the fractional knapsack problem that $G_R' = \sum_{l \in S} p_l + x_{l'}^* p_{l'}$ for some subset $S$ and allocation $\x^*$, where resources are allocated to locations in the descending orders of their affordable bang-per-buck ratios. Next, given $R+1$ resources, we define $\Tilde{\sigmaa}^{R+1}$ as the allocation corresponding to Step 1 of Algorithm~\ref{alg:GreedyFraudminimizationDeterministic} and let $\sigmaa_A^{R+1}$ be the allocation corresponding to Algorithm~\ref{alg:GreedyFraudminimizationDeterministic} with $R+1$ resources. We now show that $P_{R+1}(\sigmaa_A^{R+1}) \geq P_R(\sigmaa^*)$ to prove our claim.

To show this, we note by our greedy procedure in Step 1 of Algorithm~\ref{alg:GreedyFraudminimizationDeterministic} that $\Tilde{\sigmaa}^{R+1} \geq \Tilde{\sigmaa}$ and that $\Tilde{\sigma}_{l'}^{R+1} = \frac{d_{l'}}{d_{l'}+k}$. Such a result holds as the additional resource ensures that the administrator will allocate $ \frac{d_{l'}}{d_{l'}+k} \leq 1$ units of resources to location $l'$, as location $l'$ has the highest affordable bang-per-buck ratio among the remaining locations by definition (as $G_R' = \sum_{l \in S} p_l + x_{l'}^* p_{l'}$ for some subset $S$ and allocation $\x^*$, where resources are allocated to locations in the descending orders of the bang-per-buck). Since $\Tilde{\sigma}^{R+1}_{l'} = \frac{d_{l'}}{d_{l'}+k}$, note by the best-response Problem~\eqref{eq:best-response-users} of users that $y_{l'}(\Tilde{\sigmaa}^{R+1}) = 0$. Consequently, as $\Tilde{\sigmaa}^{R+1}$ corresponds to an allocation where resources are allocated to locations in the descending order of the affordable bang-per-buck ratios and, in particular, corresponds to an allocation where $\frac{d_{l'}}{d_{l'}+k} \leq 1$ fraction of resources are allocated to location $l'$, it follows that 
\begin{align*}
    P_{R+1}(\sigmaa_A^{R+1}) \stackrel{(a)}{\geq} P_R(\Tilde{\sigmaa}^{R+1}) \stackrel{(b)}{\geq} \sum_{l \in S} p_l + p_{l'} \stackrel{(c)}{\geq} \sum_{l \in S} p_l + x_{l'}^* p_{l'} = G_R' \stackrel{(d)}{\geq} P_R(\sigmaa^*)
\end{align*}
where (a) follows from the fact that $\sigmaa_A^{R+1}$ corresponds to the outcome with a greater payoff from steps 1 and 2 of Algorithm~\ref{alg:GreedyFraudminimizationDeterministic}, (b) follows as $\Tilde{\sigmaa}^{R+1}$ corresponds to an allocation of resources in the descending orders of the affordable bang-per-buck ratios when an administrator has an additional resource that it can spend on location $l'$ and that $y_{l'}(\Tilde{\sigmaa}^{R+1}) = 0$. Furthermore, (c) follows from the fact that $x_{l'}^* \in (0, \frac{d_{l'}}{d_{l'}+k})$ and (d) follows from our relation derived in proof of Theorem~\ref{thm:greedy-half-approx-fraud-min} that the optimal objective of the fractional knapsack Problem~\eqref{eq:fractional-knapsack-ub} is at least the payoff-maximizing objective. Thus, we have shown that $P_{R+1}(\sigmaa_A^{R+1}) \geq P_R(\sigmaa^*)$, which establishes our claim.

\subsection{Proof of Theorem~\ref{thm:ptas}} \label{apdx:pf-thm-ptas}

To prove this claim, we first recall from Proposition~\ref{prop:opt-mixed-strategy-solution} that the administrator's optimal payoff is given by: $P_R(\sigmaa^*) = \sum_{l \in L_1} p_l + \sigma_{l'}^* p_{l'}$ for some set of locations $L_1$ for which $\sigma_l^* = \frac{d_l}{d_l+k}$ for all $l \in L_1$ and at most one location $l'$ with $\sigma_{l'}^* = R - \sum_{l \in L_1} \frac{d_l}{d_l+k} \leq \frac{d_{l'}}{d_{l'}+k}$ and $\sigma_{l'}^* < \frac{d_{l'}}{d_{l'}+k}$. We proceed by analysing two cases: (i) $|L_1| \leq m$ and (ii) $|L_1| > m$.

%Thus, we can further write the optimal welfare as $W_R(\sigmaa^*) = \sum_{l \in L_1} v_l + (R - \sum_{l \in L_1} \frac{d_l}{d_l+k}) v_{l'}$.

\paragraph{Case (i) - $|L_1| \leq m$:} In this case, observe that given $\delta$ additional resources there is one feasible pair $(S, \sigma_{l'}^{\delta})$, where $S = L_1$ and $\sigma_{l'}^{\delta} = \min \{ \sigma_{l'}^* + \delta, \frac{d_{l'}}{d_{l'}+k} \}$, which is one of the pairs considered in the brute-force step of Algorithm~\ref{alg:PTASWelMaxHom}. Let $\Tilde{\sigmaa}$ be allocation corresponding to applying the greedy procedure in step two of Algorithm~\ref{alg:PTASWelMaxHom} that extends the allocation corresponding to this pair $(S, \sigma_{l'}^{\delta})$. Then, the payoff of the allocation $\Tilde{\sigmaa}$ satisfies
\begin{align*}
    P_R(\Tilde{\sigmaa}) \geq \sum_{l \in L_1} p_l +  \min \left\{ \sigma_{l'}^* + \delta, \frac{d_{l'}}{d_{l'}+k} \right\} p_{l'} \geq P_R(\sigmaa^*),
\end{align*}
where the first inequality follows as the payoff corresponding to $\Tilde{\sigmaa}$ is at least that of the allocation to the pair $(S, \sigma_{l'}^{\delta})$ and the final inequality follows as $\sigma_{l'}^* < \frac{d_{l'}}{d_{l'}+k}$. Finally, noting that the allocation $\sigmaa'$ corresponds to the maximum payoff of Algorithm~\ref{alg:PTASWelMaxHom} across all feasible pairs, it follows that $P_R(\sigmaa') \geq P_R(\Tilde{\sigmaa}) \geq P_R(\sigmaa^*)$, which establishes our claim for $|L_1| \leq m$.

\begin{comment}
which corresponds to allocating $\frac{d_l}{d_l+k}$ to all locations $l \in L_1$ and allocating at least $\sigma_{l'}' = \min \{ \sigma_{l'}^* + \delta, \frac{d_{l'}}{d_{l'}+k} \}$

allocation corresponding to allocating $\sigma_l' = \frac{d_l}{d_l+k}$ to all locations $l \in L_1$ and allocating at least $\sigma_{l'}' = \min \{ \sigma_{l'}^* + \delta, \frac{d_{l'}}{d_{l'}+k} \}$ to location $l'$ given $\delta$ additional resources. The payoff corresponding to this allocation satisfies
\begin{align*}
    P_R(\sigmaa') \geq \sum_{l \in L_1} p_l +  \min \{ \sigma_{l'}^* + \delta, t_{l'} \} p_{l'} Q_R(\sigmaa^*),
\end{align*}
which implies that the developed algorithmic procedure achieves at least the optimal payoff if $|L_1| \leq m$ given that the algorithm has access to $\delta$ additional resources.
\end{comment}

%We first consider the case when the set $L_1$ corresponding to the optimal allocation has cardinality at most $m$, i.e., $|L_1| \leq m$. 

\paragraph{Case (ii) - $|L_1| > m$:} In the case when $|L_1| \geq m+1$, consider a set $(S, \sigma_{l'}^{\delta})$ as follows. In particular, let $S = \{ l_1, \ldots, l_m \}$ be the $m$ locations under the optimal allocation with $\sigma_l^* = \frac{d_l}{d_l+k}$ corresponding to the $m$ highest payoffs, i.e., for each $l \in S$ and $\Tilde{l} \in L_1 \backslash S$ it holds that $p_l \geq p_{\Tilde{l}}$. Moreover, let $\sigma_{l'}^{\delta} = \min \{ \sigma_{l'}^* + \delta, \frac{d_{l'}}{d_{l'}+k} \}$ be the allocation to location $l'$. Then, observe that since all feasible pairs are considered in the brute-force step of Algorithm~\ref{alg:PTASWelMaxHom}, $(S, \sigma_{l'}^{\delta})$ corresponds to one such pair. Let $\Tilde{\sigmaa}$ be allocation corresponding to applying the greedy procedure in step two of Algorithm~\ref{alg:PTASWelMaxHom} that extends the allocation corresponding to this pair $(S, \sigma_{l'}^{\delta})$ to the remaining location. We will now show that the allocation $\Tilde{\sigmaa}$ achieves the desired approximation guarantee. 

%We will now show for this pair $(S, \sigma_{l'}^{\delta})$ that applying the greedy procedure second stage of Algorithm~\ref{alg:PTASWelMaxHom}

%in the algorithm's second stage will give the desired approximation guarantee. For this subset, we will denote the allocation corresponding to the greedy procedure via the allocation $\Tilde{\sigmaa}$. Note that the allocation corresponding to $\sigmaa'$ will result in a at least the same payoff as that corresponding to $\Tilde{\sigmaa}$.

%We will now show for this subset $H$ with $\sigma_l' = \frac{d_l}{d_l+k}$ for all $l \in H$ and the above defined allocation for $l'$ that applying the greedy procedure in the algorithm's second stage will give the desired approximation guarantee. %For this subset, we will denote the allocation corresponding to the greedy procedure via the allocation $\Tilde{\sigmaa}$.

To establish the desired approximation guarantee for the allocation $\Tilde{\sigmaa}$, we first define $L_2 = \{ l_{m+1}, \ldots, l_x \}$ to be the location set with $\sigma_{l}^* = \frac{d_l}{d_l+k}$ for all $l \in L_1 \backslash S$ under the optimal allocation, where the locations are ordered in descending order of their affordable bang-per-buck ratios. Furthermore, let $\Bar{l}$ be the first location in the set $L_2$ such that $\Tilde{\sigma}_{\Bar{l}} < \sigma_{\Bar{l}}^*$ and let $\Tilde{\sigma}^G$ denote the total number of resources that are not allocated to any location by Algorithm~\ref{alg:PTASWelMaxHom}. %Next, note that the greedy algorithm in the second stage does not allocate resources to all locations other than at most one location $l''$ which is allocated $\Tilde{\sigma}_{l''} < \frac{d_{l''}}{d_{l''}+k}$. 
Then, consider the set of locations $S_1$ that are allocated resources under the greedy procedure in step two of Algorithm~\ref{alg:PTASWelMaxHom} with at least the bang-per-buck of location $\Bar{l}$ but not including location $\Bar{l}$ and let $S_2$ be the set of locations allocated resources by the greedy procedure in step two of Algorithm~\ref{alg:PTASWelMaxHom} with a strictly lower bang-per-buck ratio than $\Bar{l}$ and includes the location $\Bar{l}$. Then, by construction of the greedy procedure in step two of Algorithm~\ref{alg:PTASWelMaxHom} and the fact that $\Tilde{\sigma}_{\Bar{l}} < \sigma_{\Bar{l}}^*$, note that $\sum_{l \in S_2} \Tilde{\sigma}_l + \Tilde{\sigma}^G < \frac{d_{\Bar{l}}}{d_{\Bar{l}}+k}$. Note that if this inequality were not true then the greedy procedure in step two of Algorithm~\ref{alg:PTASWelMaxHom} would allocate at least $\sigma_{\Bar{l}}^* = \frac{d_{\Bar{l}}}{d_{\Bar{l}}+k}$ to location $\Bar{l}$, a contradiction.

Next, given the above definitions of the sets $S_1$ and $S_2$ and the location $\Bar{l}$, we obtain the following lower bound on the payoff of the allocation $\Tilde{\sigmaa}$:
\begin{align} \label{eq:helper-greedy-lb-ptas-approx}
    P_R(\Tilde{\sigmaa}) \geq \sum_{l = 1}^m p_l + \sigma_{l'}^* p_{l'} + \sum_{l = m+1}^{\Bar{l}-1} p_l + \sum_{l \in S_1} p_{l} + \sum_{l \in S_2} p_{l}.
\end{align}

Note further that the locations picked in the greedy step of Algorithm~\ref{alg:PTASWelMaxHom} in set $S_1$ are such that they satisfy $\frac{p_l}{\frac{d_l}{d_l+k}} \geq \frac{p_{\Bar{l}}}{\frac{d_{\Bar{l}}}{d_{\Bar{l}}+k}}$ for all $l \in S_1$, as the greedy procedure allocates resources to locations in the descending order of their bang-per-buck ratios. Consequently, it holds that:
\begin{align} \label{eq:helper-lb-ptas-approx}
    \sum_{l \in S_1} p_{l} \geq \frac{p_{\Bar{l}}}{\frac{d_{\Bar{l}}}{d_{\Bar{l}}+k}} \left(R + \delta - \Tilde{\sigma}^G - \sum_{l \in S_2} \Tilde{\sigma}_l - \sum_{l = 1}^{{\Bar{l}}-1} \Tilde{\sigma}_l - \Tilde{\sigma}_{l'} \right).
\end{align}
Then, from Equations~\eqref{eq:helper-greedy-lb-ptas-approx} and~\eqref{eq:helper-lb-ptas-approx}, we obtain the following upper bound on the optimal payoff:
\begin{align*}
    P_R(\sigmaa^*) &= \sum_{l \in L_1} p_l + \sigma_{l'}^* p_{l'} = \sum_{l = 1}^m p_l + \sigma_{l'}^* p_{l'} + \sum_{l = m+1}^{\Bar{l}-1} p_l + \sum_{l = \Bar{l}}^x p_l, \\
    &\stackrel{(a)}{\leq} \sum_{l = 1}^m p_l + \sigma_{l'}^* p_{l'} + \sum_{l = m+1}^{\Bar{l}-1} p_l + \left( R - \sum_{l = 1}^{\Bar{l}-1} \sigma_l^* - \sigma_{l'}^* \right) \frac{p_{\Bar{l}}}{\frac{d_{\Bar{l}}}{d_{\Bar{l}}+k}}, \\
    &\stackrel{(b)}{=} \sum_{l = 1}^m p_l + \sigma_{l'}^* p_{l'} + \sum_{l = m+1}^{\Bar{l}-1} p_l + \sum_{l \in S_1} p_l - \sum_{l \in S_1} p_l + \left( R - \sum_{l = 1}^{\Bar{l}-1} \sigma_l^* - \sigma_{l'}^* \right) \frac{p_{\Bar{l}}}{\frac{d_{\Bar{l}}}{d_{\Bar{l}}+k}}, \\
    &\stackrel{(c)}{\leq} \sum_{l = 1}^m p_l + \sigma_{l'}^* p_{l'} + \sum_{l = m+1}^{\Bar{l}-1} p_l + \sum_{l \in S_1} p_l - \frac{p_{\Bar{l}}}{\frac{d_{\Bar{l}}}{d_{\Bar{l}}+k}} \left(R + \delta - \Tilde{\sigma}^{G} - \sum_{l \in S_2} \Tilde{\sigma}_l - \sum_{l = 1}^{\Bar{l}-1} \Tilde{\sigma}_l - \Tilde{\sigma}_{l'} \right) \\ 
    &+ \left( R - \sum_{l = 1}^{\Bar{l}-1} \sigma_l^* - \sigma_{l'}^* \right) \frac{p_{\Bar{l}}}{\frac{d_{\Bar{l}}}{d_{\Bar{l}}+k}}, \\
    &\stackrel{(d)}{\leq} \sum_{l = 1}^m p_l + \sigma_{l'}^* p_{l'} + \sum_{l = m+1}^{\Bar{l}-1} p_l + \sum_{l \in S_1} p_l + \frac{p_{\Bar{l}}}{\frac{d_{\Bar{l}}}{d_{\Bar{l}}+k}}  \left( \Tilde{\sigma}^{G} + \sum_{l \in S_2} \Tilde{\sigma}_l  \right), \\
    &\stackrel{(e)}{\leq} \sum_{l = 1}^m p_l + \sigma_{l'}^* p_{l'} + \sum_{l = m+1}^{\Bar{l}-1} p_l + \sum_{l \in S_1} p_l + p_{\Bar{l}}, \\
    &\stackrel{(f)}{\leq} P_R(\Tilde{\sigmaa}) + p_{\Bar{l}}, %\\
    %&\stackrel{(g)}{\leq} P_R(\sigmaa') + p_{\Bar{l}}
\end{align*}
where (a) follows from the fact that all locations $\{ l_{\Bar{l}}, \ldots, l_{x}\}$ have at most the bang-per-buck ratio as location ${\Bar{l}}$, (b) follows from adding an subtracting the term $\sum_{l \in S_1} p_l$, (c) follows from Equation~\eqref{eq:helper-lb-ptas-approx}, (d) follows from the fact that $\Tilde{\sigma}_l = \sigma_l^*$ for all $l \in [\Bar{l}-1]$ for our selected subset in the brute-force step and the definition of $\Bar{l}$ along with the fact that $\Tilde{\sigma}_{l'} \in [\sigma_{l'}^*, \sigma_{l'}^* + \delta]$. Moreover, (e) follows from the fact that $\Tilde{\sigma}^{G} + \sum_{l \in S_2} \Tilde{\sigma}_l < \frac{d_{\Bar{l}}}{d_{\Bar{l}}+k}$ from our earlier analysis and (f) follows from Equation~\eqref{eq:helper-greedy-lb-ptas-approx}. 

Then, observe that $p_{\Bar{l}} \leq \frac{P_R(\sigmaa^*)}{m+1}$, as the first $m$ locations in the set $S$ all have a higher payoff than $p_{\Bar{l}}$. Consequently, rearranging the above derived inequality $P_R(\sigmaa^*) \leq P_R(\Tilde{\sigmaa}) + p_{\Bar{l}}$, it follows that $\left(1 - \frac{1}{m+1} \right) P_R(\sigmaa^*) \leq P_R(\Tilde{\sigmaa})$. Finally, noting that the allocation $\sigmaa'$ corresponds to the maximum payoff of Algorithm~\ref{alg:PTASWelMaxHom} across all feasible pairs, it follows that $P_R(\sigmaa') \geq P_R(\Tilde{\sigmaa})$. Consequently, the obtained approximation guarantee for the allocation $\Tilde{\sigmaa}$ given by $\left(1 - \frac{1}{m+1} \right) P_R(\sigmaa^*) \leq P_R(\Tilde{\sigmaa})$ implies our desired approximation guarantee that $\left(1 - \frac{1}{m+1} \right) P_R(\sigmaa^*) \leq P_R(\sigmaa')$, which establishes our claim.

\subsection{Proof of Theorem~\ref{thm:npHardness-erm}} \label{apdx:pf-np-hardness-erm}

We prove Theorem~\ref{thm:npHardness-erm} through a reduction from an instance of the partition problem. Recall that a partition instance consists of a sequence of numbers $a_1, \ldots, a_n$ with $\sum_{l \in [n]} a_l = A$ and involves the task of deciding whether there is some subset $S_1$ of numbers such that $\sum_{l \in S_1} a_l = \frac{A}{2}$. 

In the following, we construct an instance of the heterogeneous revenue maximization problem (HRMP), henceforth referred to as HRMP for brevity, based on a partition instance and present four intermediate lemmas to complete the proof of Theorem~\ref{thm:npHardness-erm}. We then present the detailed proofs of these lemmas in later subsections.

To this end, we begin by constructing an instance of the HRMP problem with with two types, i.e., where $|\I| = 2$, and $n$ locations, where each number $a_l$ corresponds to a given location. In this setting, for ease of exposition, we drop the fine $k$ from the revenue maximization Objective~\eqref{eq:admin-obj-revenue} of the administrator as this is a uniform constant that applies to all locations $l$ and types $i$. Then, we define the following four quantities: (i) $\frac{d_l^1}{d_l^1+k} = \frac{1}{2} \frac{a_l}{A}$, (ii) $\frac{d_l^2}{d_l^2+k} = \frac{a_l}{A}$, (iii) $\Lambda_l^1 = A \left( 2 \frac{\max_{l' \in [n]} a_{l'}}{a_l} - 1 \right)$, and (iv) $\Lambda_l^2 = A \left(1+ \frac{2 \max_{l' \in [n]} a_{l'}}{a_l} \right)$ for all $l \in [n]$. Moreover, we define the resource budget as $R = \frac{3}{4}$.

Given this instance of HRMP, we claim that we have a ``Yes'' instance of partition if and only if there is a feasible allocation of resources to locations with total revenue that is at least $\frac{A}{2} + 2 n \times \max_{l \in [n]} a_l$ for the above defined HRMP instance.

To establish this claim, we first prove its forward direction that if we have a ``Yes'' instance of partition, then there exists a feasible allocation of resources to locations with total revenue that is at least $\frac{A}{2} + 2 n \times \max_{l \in [n]} a_l$ for the above defined HRMP instance.

\begin{lemma} [Forward Direction of NP-Hardness of HRMP] \label{lem:forward-partition-to-ermp}
Consider a partition instance with a sequence of numbers $a_1, \ldots, a_n$ with $\sum_{l \in [n]} a_l = A$, such that there exists a subset $S_1$ of numbers such that $\sum_{l \in S_1} a_l = \frac{A}{2}$. Then, there exists a feasible allocation of resources to locations with total expected revenue that is at least $\frac{A}{2} + 2 n \times \max_{l \in [n]} a_l$ for the above defined HRMP instance.
\end{lemma}

To prove this claim, we present a method to construct a feasible administrator strategy $\sigmaa$ for the above defined HRMP instance and show that the revenue of this administrator strategy satisfies $Q_R(\sigmaa) \geq \frac{A}{2} + 2 n \times \max_{l \in [n]} a_l$. For a proof of Lemma~\ref{lem:forward-partition-to-ermp}, see Appendix~\ref{apdx:pf-lem-forward-partition-to-ermp}.

Next, to prove the reverse direction of this claim, we suppose that the optimal resource allocation strategy corresponding to the above defined HRMP instance is such that the revenue $Q_R(\sigmaa) \geq \frac{A}{2} + 2 n \times \max_{l \in [n]} a_l$ and the total resources $\sum_{l \in [n]} \sigma_l \leq \frac{3}{4} = R$. To prove that such a setting corresponds to a ``Yes'' instance of partition, we use three intermediate lemmas. Our first two lemmas establish certain structural properties of the administrator's revenue maximizing solution. In particular, we first establish a lower bound on the resources allocated to each location under the revenue maximizing strategy of the administrator, as is elucidated through the following lemma.

\begin{lemma} [Lower Bound on Administrator Allocation Strategy] \label{lem:back-direction-lower-bound-admin-allocation}
Suppose that there is some feasible administrator strategy $\sigmaa$ that satisfies $Q_R(\sigmaa) \geq \frac{A}{2} + 2 n \times \max_{l \in [n]} a_l$ for $R = \frac{3}{4}$. Then, there exists a revenue maximizing administrator strategy $\Tilde{\sigmaa}^*$ for the above defined HRMP instance such that $\Tilde{\sigma}_l^* \geq \frac{d_l^1}{d_l^1+k}$ for all locations $l$.
\end{lemma}

For a proof of Lemma~\ref{lem:back-direction-lower-bound-admin-allocation}, see Appendix~\ref{apdx:pf-back-direction-lb-admin}. Given the result obtained in Lemma~\ref{lem:back-direction-lower-bound-admin-allocation}, for the remainder of this proof, we focus on revenue maximizing administrator strategies $\sigmaa$ that satisfy $\sigma_l \geq \frac{d_l^1}{d_l^1+k}$ for all locations $l$ for our above defined HRMP instance.

%($\impliedby$:) Next, suppose that the optimal resource allocation strategy corresponding to the above defined ERMP instance is such that the expected revenue $Q_R(\sigmaa) \geq \frac{A}{2} + 2 n \times \max_{l \in [n]} a_l$ and the total resources $\sum_{l \in [n]} \sigma_l \leq \frac{3}{4} = R$. We now show that such an instance corresponds to a ``Yes'' instance of partition. To do so, we proceed in two steps. We first show that in the setting when $R = \frac{3}{4}$ there exists an optimal administrator strategy $\Tilde{\sigmaa}^*$ wherein the $\Tilde{\sigma}_l^* \geq \frac{d_l^1}{d_l^1+k}$ for all locations $l$. Next, we use this structure on the optimal solution of the constructed instance to show that our ERMP instance that satisfies $Q_R(\sigmaa) \geq \frac{A}{2} + 2 n \times \max_{l \in [n]} a_l$ corresponds to a ``Yes'' instance of partition.

Next, following a similar analysis in the proof of Proposition~\ref{prop:opt-mixed-strategy-solution}, we can establish the following property on the optimal solution of Problem~\eqref{eq:admin-obj-revenue}-\eqref{eq:bi-level-con-revenue}.

%show that there exists a solution to the ERMP problem such that there is at most one location $\Tilde{l}$ for which $x_{\Tilde{l}} \in (0, \frac{d_{\Tilde{l}}^2}{d_{\Tilde{l}}^2 + k} - \frac{d_{\Tilde{l}}^1}{d_{\Tilde{l}}^1 + k}) = (0, \frac{1}{2} \frac{a_l}{A})$, as is elucidated through the following lemma.

\begin{lemma} [Structure of Revenue-Maximizing Solution] \label{lem:reverse-direction-opt-sol-ermp}
Suppose that users at each location best-respond based on Equation~\eqref{eq:best-response-users-rev-max-prob}. Then, there exists a revenue maximizing strategy $\Tilde{\sigmaa}^*$ of the administrator corresponding to the solution of Problem~\eqref{eq:admin-obj-revenue}-\eqref{eq:bi-level-con-revenue} such that:
\begin{itemize}
    \item There exists a set of locations $L_1$ such that $\Tilde{\sigma}_l^* = \frac{d_l^{i_l}}{d_l^{i_l}+k}$ for all $l \in L_1$ for some type $i_l \in \I$ that can be location specific, a set of locations $L_2$ such that $\Tilde{\sigma}_l^* = 0$ for all $l \in L_2$, and at most one location $l'$ such that either $\Tilde{\sigma}_{l'}^* \in (0, \frac{d_{l'}^{1}}{d_{l'}^{1}+k})$ or $\Tilde{\sigma}_{l'}^* \in (\frac{d_{l'}^{i-1}}{d_{l'}^{i-1}+k}, \frac{d_{l'}^{i}}{d_{l'}^{i}+k})$ for some $i>1$. Here, $L_1$, $L_2$, and $l'$ are all disjoint and $L_1 \cup L_2 \cup \{ l'\} = L$.
\end{itemize}
\end{lemma}

Since the analysis of the above result follows similar arguments to that in the proof of Proposition\ref{prop:opt-mixed-strategy-solution}, we omit it for brevity. %Then, following Lemma~\ref{lem:back-direction-lower-bound-admin-allocation}, we have from the result in Lemma~\ref{lem:reverse-direction-opt-sol-ermp} that there exists a solution to the ERMP problem such that there is at most one location $\Tilde{l}$ for which $x_{\Tilde{l}} \in (0, \frac{d_{\Tilde{l}}^2}{d_{\Tilde{l}}^2 + k} - \frac{d_{\Tilde{l}}^1}{d_{\Tilde{l}}^1 + k}) = (0, \frac{1}{2} \frac{a_l}{A})$. Such a result follows as 

Finally, we use the structure of the optimal (revenue maximizing) solution of the above defined HRMP instance established in Lemmas~\ref{lem:back-direction-lower-bound-admin-allocation} and~\ref{lem:reverse-direction-opt-sol-ermp} to show that our HRMP instance with $Q_R(\sigmaa) \geq \frac{A}{2} + 2 n \times \max_{l \in [n]} a_l$ corresponds to a ``Yes'' instance of partition.

\begin{lemma} [Reverse Direction of NP-Hardness of HRMP] \label{lem:reverse-np-hardness-ermp}
Suppose that there is some feasible administrator strategy $\sigmaa$ that satisfies $Q_R(\sigmaa) \geq \frac{A}{2} + 2 n \times \max_{l \in [n]} a_l$ for $R = \frac{3}{4}$. Then, for the above defined HRMP instance, there exists a subset $S_1$ of numbers such that $\sum_{l \in S_1} a_l = \frac{A}{2}$, i.e., we have a ``Yes'' instance of partition.
\end{lemma}

For a proof of Lemma~\ref{lem:reverse-np-hardness-ermp}, see Appendix~\ref{apdx:pf-reverse-np-hardness-ermp}. Note that Lemmas~\ref{lem:forward-partition-to-ermp}-\ref{lem:reverse-np-hardness-ermp} prove Theorem~\ref{thm:npHardness-erm}. %, as we have established both directions of our claim.

\subsubsection{Proof of Lemma~\ref{lem:forward-partition-to-ermp}} \label{apdx:pf-lem-forward-partition-to-ermp}

We first suppose that we have a ``Yes'' instance of partition, i.e., there exists a set of $S$ numbers such that $\sum_{l \in S} a_l = \frac{A}{2}$. Then, consider an administrator strategy $\sigmaa$ such that $\sigma_l = \frac{a_l}{A}$ for all $l \in S$ and $\sigma_l = \frac{1}{2} \frac{a_l}{A}$ for all $l \in [n] \backslash S$. We first note that $\sigmaa$ is feasible, i.e., $\sigmaa \in \Omega_R$, as $\sigma_l \in [0, 1]$ by construction and the resource constraint is satisfied as:
\begin{align*}
    \sum_{l \in [n]} \sigma_l &= \sum_{l \in S} \sigma_l + \sum_{l \in S \backslash [n]} \sigma_l = \sum_{l \in S} \frac{a_l}{A} + \sum_{l \in [n] \backslash S} \frac{1}{2} \frac{a_l}{A}, \\
    &\stackrel{(a)}{=} \frac{1}{A} \frac{A}{2} + \frac{1}{2A} \frac{A}{2} = \frac{3}{4} = R,
\end{align*}
where (a) follows as $\sum_{l \in S} a_l = \frac{A}{2}$ and $\sum_{l \in [n] \backslash S} a_l = \frac{A}{2}$ as we have a ``Yes'' instance of partition.

Next, we obtain the following revenue objective for the administrator when it plays the above defined strategy $\sigmaa$:
\begin{align*}
    Q_R(\sigmaa) &= \sum_{i \in \I} \sum_{l \in L} \sigma_l y_l^i(\sigmaa) \Lambda_l^i \stackrel{(a)}{=} \sum_{l \in S} \Lambda_l^2 \frac{d_l^2}{d_l^2+k} + \sum_{l \in [n] \backslash S} (\Lambda_l^1 + \Lambda_l^2) \frac{d_l^1}{d_l^1+k}, \\
    &\stackrel{(b)}{=}  \sum_{l \in S} A \left(1 \! + \! \frac{2 \max_{l' \in [n]} a_{l'}}{a_l} \right)  \! \! \frac{a_l}{A} + \! \! \sum_{l \in S \backslash [n]} \! \! \! \left( A \left( 2 \frac{\max_{l' \in [n]} a_{l'}}{a_l} \! - \! 1 \right) \! + \! A \left(1 \! + \! \frac{2 \max_{l' \in [n]} a_{l'}}{a_l} \right) \right) \frac{1}{2} \frac{a_l}{A} , \\
    &= \sum_{l \in S} \left( a_l + 2 \max_{l' \in [n]} a_{l'} \right) + \sum_{l \in S \backslash [n]} 2 \max_{l' \in [n]} a_{l'}, \\
    &= \frac{A}{2} + 2 \sum_{l \in [n]} \max_{l' \in [n]} a_{l'}, \\
    &= \frac{A}{2} + 2n \max_{l' \in [n]} a_{l'},
\end{align*}
where we have dropped the fine $k$ from the objective for ease of exposition, (a) follows by the definition of $y_l^i(\sigmaa)$ as given in Equation~\eqref{eq:best-response-users-rev-max-prob}, and (b) follows by plugging in the values of the respective quantities in the above defined instance. The above analysis implies that $\sigmaa$ is both feasible and achieves a total revenue of $\frac{A}{2} + 2 n \times \max_{l \in [n]} a_l$, which establishes the forward direction of our claim.

\subsubsection{Proof of Lemma~\ref{lem:back-direction-lower-bound-admin-allocation}} \label{apdx:pf-back-direction-lb-admin}

%In other words, there exists an optimal solution to the ERMP instance such that the resource spending on all locations is at least enough for th

%\paragraph{Structure of Optimal Solution:}
%For any optimal solution $\sigmaa^*$ to the ERMP instance, we will construct a 
To prove this claim, we first show that in the setting when $R = \frac{3}{4}$ any optimal administrator strategy $\Tilde{\sigmaa}^*$ satisfies $\Tilde{\sigma}_l^* \geq \frac{d_l^1}{d_l^1+k}$ for all locations $l$.
To this end, suppose for contradiction that there is an optimal solution $\sigmaa^*$ to the HRMP instance such that there is some set of locations $S'$ for which $\sigma_{l'}^* < \frac{d_{l'}^1}{d_{l'}^1+k}$ for all $l' \in S'$. Furthermore, define $\gamma = \sum_{l \in L} \min \{ \sigma_l^*, \frac{d_l^1}{d_l^1+k} \}$ to represent the total spending across locations to cover the first type $i = 1$. In this case we note that $\gamma < \frac{1}{2}$, as 
\begin{align*}
    \gamma &= \sum_{l \in L} \min \{ \sigma_l^*, \frac{d_l^1}{d_l^1+k} \} \stackrel{(a)}{<} \sum_{l \in L} \frac{d_l^1}{d_l^1+k} \stackrel{(b)}{=} \sum_{l \in L} \frac{1}{2} \frac{a_l}{A} \stackrel{(c)}{=} \frac{1}{2},
\end{align*}
where (a) follows as there is some set of location $S'$ for which $\sigma_{l'}^* < \frac{d_{l'}^1}{d_{l'}^1+k}$ for all $l' \in S'$ by our contradiction assumption, (b) follows by substituting $\frac{d_l^1}{d_l^1+k} = \frac{1}{2} \frac{a_l}{A}$ for our defined HRMP instance, and (c) follows as $\sum_{l \in [n]} a_l = A$.

We now define an allocation strategy $\Tilde{\sigmaa}$ with a strictly higher revenue than $\sigmaa^*$ to derive our desired contradiction. To do so, we first define $\delta = \sum_{l \in L} \max \{ 0, \sigma_l^* - \frac{d_l^1}{d_l^1+k} \}$ as the cumulative resources allocated to locations beyond $\frac{d_l^1}{d_l^1+k}$. In this case, we note that $\delta \geq \frac{1}{2} - \gamma$, as if this were not the case then $\Tilde{\sigmaa}^*$ would not be optimal (as the amount of resources $\frac{1}{2} - \gamma$ could be spent on locations in the set $S'$ for which $\sigma_{l'}^* < \frac{d_{l'}^1}{d_{l'}^1+k}$ for all $l' \in S'$ to obtain a higher revenue for the administrator). Then, we construct $\Tilde{\sigmaa}$ by transferring $\frac{1}{2} - \gamma$ of mass from locations $l \in [n] \backslash S'$ with $\sigma_l^* > \frac{d_l^1}{d_l^1+k}$ to locations in the set $S'$ with $\sigma_l^* < \frac{d_l^1}{d_l^1+k}$, such that $\Tilde{\sigma}_l \geq \frac{d_l^1}{d_l^1+k}$ for all locations $l$ and the total resource spending $\sum_{l \in L} \Tilde{\sigma}_l = \sum_{l \in L} \sigma_l^*$.

We now show that such a defined strategy $\Tilde{\sigmaa}$ achieves a strictly higher revenue than $\sigmaa^*$. To see this, first note that the maximum additional gain in the revenue from spending $\sigma_l' = \frac{d_l^2}{d_l^2+k}$ as compared to $\sigma_l' = \frac{d_l^1}{d_l^1+k}$ is given by $\Lambda_l^2 \frac{d_l^2}{d_l^2+k} - (\Lambda_l^1 + \Lambda_l^2) \frac{d_l^1}{d_l^1+k} = a_l$ for all locations $l$ (see left of Figure~\ref{fig:concave-upper-approximation}). Furthermore, the spending required to gain this amount $a_l$ is given by $\frac{d_l^2}{d_l^2+k} - \frac{d_l^1}{d_l^1+k} = \frac{1}{2}\frac{a_l}{A}$. In other words, for the additional spending of $\frac{1}{2} - \gamma$ on locations in the set $[n] \backslash S'$, the maximum gain in the administrator's revenue is given by $(\frac{1}{2} - \gamma) \max_{l \in [n]} a_l$, regardless of which location(s) this additional amount $\frac{1}{2} - \gamma$ is spent on. However, spending the same amount on locations in the set $S'$ with $\sigma_{l'}^* < \frac{d_{l'}^1}{d_{l'}^1+k}$ for all $l' \in S'$, will result in an increase in the administrator revenue by
\begin{align*}
    \sum_{l \in S'} (\Lambda_l^1 + \Lambda_l^2) (\frac{d_l^1}{d_l^1+k} - \sigma_l^*) &\stackrel{(a)}{=}  \sum_{l \in S'} 4 \frac{A}{a_l} \max_{l \in [n]} a_l (\frac{d_l^1}{d_l^1+k} - \sigma_l^*) \stackrel{(b)}{\geq}  4 \max_{l \in [n]} a_l \sum_{l \in S'} (\frac{d_l^1}{d_l^1+k} - \sigma_l^*), \\
    &\stackrel{(c)}{=} 4 \max_{l \in [n]} a_l (\frac{1}{2} - \gamma) > (\frac{1}{2} - \gamma) \max_{l \in [n]} a_l,
\end{align*}
where (a) follows by substituting the relations for $\Lambda_l^1$ and $\Lambda_l^2$ for our above defined HRMP instance, (b) follows $A \geq a_l$ for all locations $l$ as $\sum_l a_l = A$, and (c) follows as $\sum_{l \in S'} (\frac{d_l^1}{d_l^1+k} - \sigma_l^*) = \frac{1}{2} - \gamma$ by our definition of $\gamma$. Thus, we have obtained that the increase in the administrator's revenue corresponding to shifting $\frac{1}{2} - \gamma$ of resources from locations $l \in [n] \backslash S'$ to locations $l \in S'$ (as in our allocation $\Tilde{\sigmaa}$) outweighs the maximum loss in the administrator objective from shifting these resources from $\sigmaa^*$ to $\Tilde{\sigmaa}$, given by $(\frac{1}{2} - \gamma) \max_{l \in [n]} a_l$. As a result, this relation implies that $\sigmaa^*$ is not optimal, a contradiction. Consequently, any optimal solution to the above defined HRMP instance must be such that $\Tilde{\sigma}_l^* \geq \frac{d_l^1}{d_l^1+k}$ for all locations $l$, which establishes our claim.

\begin{comment}
from a given location for spending i

First, we show that there is always an optimal solution that allocates to the first types of all locations. To see this, consider a solution $(\sigma_l^*)_{l \in [L]}$ such that there are some locations $l$ such that $\sigma_l^* < \frac{d_l^1}{d_l^1+k}$. Then, it must hold that there are some locations with $\sigma_l^* > \frac{d_l^1}{d_l^1+k}$. Then, we show that shifting the $\sigma^*$ mass to the locations with  $\sigma_l^* < \frac{d_l^1}{d_l^1+k}$ will only result in at least a higher objective. In particular, consider shifting a mass of $y = \frac{L}{2} \frac{a_l}{A} - \sum_{\text{spending on type 1}} \sigma_l^*$ from type to locations to the remaining type one locations. Then, we get the following change in objective:
\begin{align*}
    y \times 2 \times \max_{l} a_l \geq y \times \max_{l} a_l,
\end{align*}
which is clearly greater than the corresponding spending of $y$ on type two. Thus, it is clear that spending on type 1 will at least result in as high an objective as that for type 2.
\end{comment}

\subsubsection{Proof of Lemma~\ref{lem:reverse-np-hardness-ermp}} \label{apdx:pf-reverse-np-hardness-ermp}

We now use the structure of the optimal solution established in Lemmas~\ref{lem:back-direction-lower-bound-admin-allocation} and~\ref{lem:reverse-direction-opt-sol-ermp} to show that our HRMP instance with $Q_R(\sigmaa) \geq \frac{A}{2} + 2 n \times \max_{l \in [n]} a_l$ corresponds to a ``Yes'' instance of partition. To do so, we first note that the total revenue when spending exactly $\frac{d_l^1}{d_l^1+k}$ at each location $l$ is given by: $\sum_{l \in [n]} (\Lambda_l^1 + \Lambda_l^2) \frac{d_l^1}{d_l^1+k} = 2n \max_{l' \in [n]} a_{l'}$. Note that the total such spending on locations is $\sum_{l \in [n]}\frac{d_l^1}{d_l^1+k} = \frac{1}{2}$; hence, the total additional spending beyond $\frac{d_l^1}{d_l^1+k}$ across all locations $l$ must not be more than $\frac{1}{4}$ as the total available resources $R = \frac{3}{4}$. In particular, we have that $\sum_{l \in [n]} x_l \leq \frac{1}{4}$, where $\sigma_l = \frac{d_l}{d_l+k} + x_l$ for all locations $l$. 

Next, from Lemma~\ref{lem:reverse-direction-opt-sol-ermp}, we have that there exists a solution to the HRMP problem such that there is at most one location $\Tilde{l}$ for which $x_{\Tilde{l}} \in (0, \frac{d_{\Tilde{l}}^2}{d_{\Tilde{l}}^2 + k} - \frac{d_{\Tilde{l}}^1}{d_{\Tilde{l}}^1 + k}) = (0, \frac{1}{2} \frac{a_{\Tilde{l}}}{A})$. Such a result holds as from Lemma~\ref{lem:back-direction-lower-bound-admin-allocation} we have that $\sigma_l \geq \frac{d_l^1}{d_l^1+k}$ for all locations $l$ and that from Lemma~\ref{lem:reverse-direction-opt-sol-ermp}, there exists at most one location $\Tilde{l}$ for which $\sigma_l \in (\frac{d_{\Tilde{l}}^1}{d_{\Tilde{l}}^1 + k}, \frac{d_{\Tilde{l}}^2}{d_{\Tilde{l}}^2 + k})$. As a consequence, the spending of the administrator beyond $\frac{d_l^1}{d_l^1+k}$ must satisfy $\sum_{l \in S''} \frac{1}{2} \frac{a_l}{A} + x_{\Tilde{l}} \leq \frac{1}{4}$ for some location set $S''$ and some location $\Tilde{l}$, and the revenue of the administrator must satisfy:
\begin{align*}
     Q_R(\sigmaa) &= \sum_{l \in S''} \Lambda_l^2 \frac{d_l^2}{d_l^2+k} + \Lambda_{\Tilde{l}}^2 \left( x_{\Tilde{l}} + \frac{d_{\Tilde{l}}^1}{d_{\Tilde{l}}^1+k} \right) + \sum_{l \in [n] \backslash (S'' \cup \{ \Tilde{l}\})} (\Lambda_l^1 + \Lambda_l^2) \frac{d_l^1}{d_l^1+k}, \\
     &\stackrel{(a)}{=} \sum_{l \in S''} \Lambda_l^2 \frac{d_l^2}{d_l^2+k} - \sum_{l \in S''} (\Lambda_l^1 + \Lambda_l^2) \frac{d_l^1}{d_l^1+k} + \Lambda_{\Tilde{l}}^2 \left( x_{\Tilde{l}} + \frac{d_{\Tilde{l}}^1}{d_{\Tilde{l}}^1+k} \right) \\
     &- \frac{d_{\Tilde{l}}^1}{d_{\Tilde{l}}^1+k} ( \Lambda_{\Tilde{l}}^1 + \Lambda_{\Tilde{l}}^2) + \sum_{l \in [n]} (\Lambda_l^1 + \Lambda_l^2) \frac{d_l^1}{d_l^1+k}, \\
     &\stackrel{(b)}{=}\sum_{l \in S''} a_l + \left( x_{\Tilde{l}} + \frac{d_{\Tilde{l}}^1}{d_{\Tilde{l}}^1+k} \right) \Lambda_{\Tilde{l}}^2 - \frac{d_{\Tilde{l}}^1}{d_{\Tilde{l}}^1+k} ( \Lambda_{\Tilde{l}}^1 + \Lambda_{\Tilde{l}}^2) + 2n \max_{l' \in [n]} a_{l'},
\end{align*}
where (a) follows from adding and subtracting $\sum_{l \in S'' \cup \Tilde{l}} (\Lambda_l^1 + \Lambda_l^2) \frac{d_l^1}{d_l^1+k}$, and (b) follows by substituting the relations for $\Lambda_l^1$, $\Lambda_l^2$, and $\frac{d_l^i}{d_l^i+k}$ for our above defined HRMP instance. 

Then, since $Q_R(\sigmaa) \geq \frac{A}{2} + 2n \max_{l' \in [n]} a_{l'}$, the above equalities implies that, we have the following relation is satisfied:
\begin{align} \label{eq:helper-ermp-obj}
    \sum_{l \in S''} a_l + \left( x_{\Tilde{l}} + \frac{d_{\Tilde{l}}^1}{d_{\Tilde{l}}^1+k} \right) \Lambda_{\Tilde{l}}^2 - \frac{d_{\Tilde{l}}^1}{d_{\Tilde{l}}^1+k} (\Lambda_{\Tilde{l}}^1 + \Lambda_{\Tilde{l}}^2) \geq \frac{A}{2}.
\end{align}

We now show that we have a ``Yes'' instance of the partition problem by showing that $\Tilde{l} = \emptyset$.

To show this, we proceed by contradiction and suppose that there is some location $\Tilde{l}$ such that $x_{\Tilde{l}} \in (0, \frac{1}{2} \frac{a_l}{A})$. In this setting, we first note that the resource constraint $\frac{1}{4} \geq \sum_{l \in S'' \cup \{ \Tilde{l} \}} x_l = \sum_{l \in S''} \frac{1}{2} \frac{a_l}{A} + x_{\Tilde{l}}$ can be rewritten as:
\begin{align}
    &\sum_{l \in S''} \frac{1}{2} \frac{a_l}{A} + x_{\Tilde{l}} \leq \frac{1}{4}, \nonumber \\
    \implies &\sum_{l \in S''} a_l + 2A x_{\Tilde{l}} \leq \frac{A}{2} \label{eq:helper-23-resource-con}
\end{align}

Furthermore, from Equation~\eqref{eq:helper-ermp-obj}, we have that:
\begin{align}
    \frac{A}{2} &\leq \sum_{l \in S''} a_l + \left( x_{\Tilde{l}} + \frac{d_{\Tilde{l}}^1}{d_{\Tilde{l}}^1+k} \right) \Lambda_{\Tilde{l}}^2 -  \frac{d_{\Tilde{l}}^1}{d_{\Tilde{l}}^1+k} (\Lambda_{\Tilde{l}}^1 + \Lambda_{\Tilde{l}}^2), \nonumber \\
    &< \sum_{l \in S''} a_l + \frac{x_{\Tilde{l}}}{\frac{d_{\Tilde{l}}^2}{d_{\Tilde{l}}^2+k} - \frac{d_{\Tilde{l}}^1}{d_{\Tilde{l}}^1+k}} a_{\Tilde{l}}, \label{eq:helper-ermp-obj-22}
    %\frac{d_{\Tilde{l}}^2}{d_{\Tilde{l}}^2+k} q_{\Tilde{l}}^2 \Lambda_{\Tilde{l}}^2 -  \frac{d_{\Tilde{l}}^1}{d_{\Tilde{l}}^1+k} (q_{\Tilde{l}}^1 \Lambda_{\Tilde{l}}^1 + q_{\Tilde{l}}^2 \Lambda_{\Tilde{l}}^2), \\
    %&= \sum_{l \in S' \cup \{ \Tilde{l} \}} a_l
\end{align}
where the second inequality follows as the total gain from spending $\frac{d_{\Tilde{l}}^2}{d_{\Tilde{l}}^2+k}$ compared to $\frac{d_{\Tilde{l}}^1}{d_{\Tilde{l}}^1+k}$ is given by $\frac{d_{\Tilde{l}}^2}{d_{\Tilde{l}}^2+k} \Lambda_{\Tilde{l}}^2 -  \frac{d_{\Tilde{l}}^1}{d_{\Tilde{l}}^1+k} (\Lambda_{\Tilde{l}}^1 + \Lambda_{\Tilde{l}}^2) = a_{\Tilde{l}}$ and the fact that $\frac{x_{\Tilde{l}}}{\frac{d_{\Tilde{l}}^2}{d_{\Tilde{l}}^2+k} - \frac{d_{\Tilde{l}}^1}{d_{\Tilde{l}}^1+k}} a_{\Tilde{l}}$ represents a strict upper bound on the expected revenue that the administrator can gain for any $x_{\Tilde{l}} \in (0, \frac{d_{\Tilde{l}}^2}{d_{\Tilde{l}}^2+k} - \frac{d_{\Tilde{l}}^1}{d_{\Tilde{l}}^1+k})$. For more intuition on the derivation of the second inequality, we refer to Figure~\ref{fig:helper-plot-revenue-and-upper-bound}, which provides a geometric intuition for this claim through a depiction of the revenue function (in blue) and its corresponding upper bound (in orange).

\begin{figure}[tbh!]
    \centering
    \includegraphics[width=0.6\linewidth]{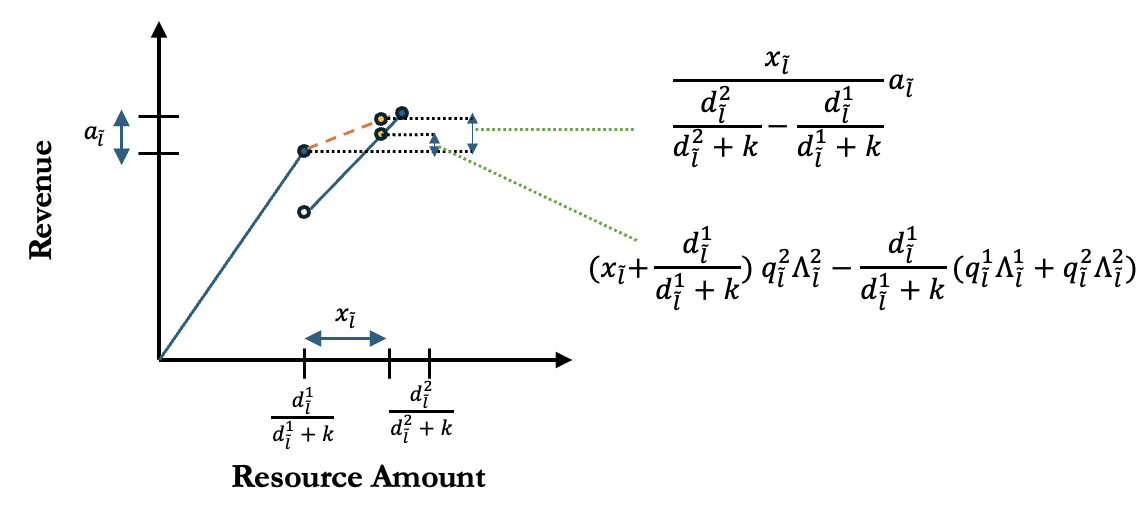}
    \vspace{-10pt}
    \caption{{\small \sf Depiction of the revenue of the above defined HRMP instance with the amount of resources spent at location $\Tilde{l}$ in the range $\left[0, \frac{d_{\Tilde{l}}^2}{d_{\Tilde{l}}^2+k} \right]$. The curve in blue represents the revenue as a function of the number of resources allocated to location $\Tilde{l}$, while the dashed line in orange represents the upper bound on the revenue function. Moreover, the two points marked in yellow represent the revenue (bottom yellow point) and its corresponding upper bound (top yellow point) corresponding to a resource allocation of $\frac{d_{\Tilde{l}}^1}{d_{\Tilde{l}}^1+k} + x_{\Tilde{l}}$ at location $\Tilde{l}$.
    }}
    \label{fig:helper-plot-revenue-and-upper-bound}
\end{figure} %\vspace{-10pt}

%due to the geometry of the situation (can show by picture). Note that the above inequality is strict as $x_{\Tilde{l}} \in (0, \frac{1}{2} \frac{a_l}{A})$.
%equality if $x_{\Tilde{l}} = 0$ or $x_{\Tilde{l}} = 1$. 

From Equations~\eqref{eq:helper-23-resource-con} and~\eqref{eq:helper-ermp-obj-22}, we obtain that:
\begin{align*}
    \sum_{l \in S''} a_l + \frac{x_{\Tilde{l}}}{\frac{d_l^2}{d_l^2+k} - \frac{d_l^1}{d_l^1+k}} a_{\Tilde{l}} \stackrel{(a)}{=} \sum_{l \in S''} a_l + \frac{x_{\Tilde{l}}}{\frac{1}{2} \frac{a_{\Tilde{l}}}{A}} a_{\Tilde{l}} = \sum_{l \in S''} a_l + 2A x_{\Tilde{l}} \stackrel{(b)}{\leq} \frac{A}{2} \stackrel{(c)}{<} \sum_{l \in S''} a_l + \frac{x_{\Tilde{l}}}{\frac{d_{\Tilde{l}}^2}{d_{\Tilde{l}}^2+k} - \frac{d_{\Tilde{l}}^1}{d_{\Tilde{l}}^1+k}} a_{\Tilde{l}},
\end{align*}
for any $x_{\Tilde{l}} \in (0, \frac{1}{2} \frac{a_l}{A})$, a contradiction. Here, (a) follows by substituting for the values of $\frac{d_l^2}{d_l^2+k}$ and $\frac{d_l^1}{d_l^1+k}$ for our above defined HRMP instance, (b) follows by Equation~\eqref{eq:helper-23-resource-con}, and (c) follows by Equation~\eqref{eq:helper-ermp-obj-22}. Thus, we have that $\Tilde{l} \neq \emptyset$ is not possible. 

Finally, since $\Tilde{l} = \emptyset$, we have that Equation~\eqref{eq:helper-ermp-obj} reduces to:
\begin{align} \label{eq:helper-1-revenue-obj}
    \sum_{l \in S''} a_l \geq \frac{A}{2}.
\end{align}
Moreover, we have by the resource constraint that for these locations $S''$ that $\frac{1}{4} \geq \sum_{l \in S''} x_l =  \frac{d_{\Tilde{l}}^2}{d_{\Tilde{l}}^2 + k} - \frac{d_{\Tilde{l}}^1}{d_{\Tilde{l}}^1 + k} = \sum_{l \in S''} \frac{1}{2} \frac{a_l}{A}$. This inequality for the resource constraint implies the following relation:
\begin{align} \label{eq:helper-2-allocation}
    \sum_{l \in S''} a_l \leq \frac{A}{2}.
\end{align}
Together, Equations~\eqref{eq:helper-1-revenue-obj} and~\eqref{eq:helper-2-allocation} imply that there exists a set $S''$ such that $ \sum_{l \in S''} a_l = \frac{A}{2}$, i.e., we have a ``Yes'' instance of partition. This establishes our claim.

\subsection{Proof of Theorem~\ref{thm:greedy-half-approx-rev-max}} \label{apdx:greedy-half-approx-rev-max}

To prove this claim, we begin by first introducing some notation. In particular, let $Q_l(\cdot)$ denote the revenue function at each location as a function of the amount of allocated resources $\sigma_l$ to that location (e.g., see left of Figure~\ref{fig:concave-upper-approximation}). Moreover, we denote $\Hat{Q}_l(\cdot)$ as the MCUA of the revenue function at each location $l$. We then have the following problem of maximizing the MCUA of the revenue function:
\begin{maxi!}|s|[2]   
    {\sigmaa}                   
    { \sum_{l \in L} \Hat{Q}_l(\sigma_l),  \label{eq:revenue-obj-mcua}}   
    {\label{eq:Eg001}}             
    {}          
    \addConstraint{\sum_{l \in L} \sigma_l}{\leq R, \label{eq:resource-constraint-mcua}}    
    \addConstraint{\sigma_l}{\in [0, 1], \quad \forall l \in L, \label{eq:probability-feasibility-mcua}}
\end{maxi!}
which rather than maximizing the sum of the revenue functions across all locations corresponds to maximizing the sum of the MCUA of the revenue functions across locations subject to the constraint that $\sigmaa \in \Omega_R$, as given by the Constraints~\eqref{eq:resource-constraint-mcua}-\eqref{eq:probability-feasibility-mcua}.

Furthermore, for convenience of analysis, we introduce Algorithm~\ref{alg:GreedyMCUA-OPT}, which is analogous to the greedy allocation rule in step one of Algorithm~\ref{alg:GreedyRevMaxProb} other than that the resources are allocated until either all the resources are exhausted or there are no further segments remaining to iterate over. Recall that step one of Algorithm~\ref{alg:GreedyRevMaxProb} terminates when a segment's resource requirement exceeds the amount of available resources.

\begin{algorithm}
\footnotesize
\SetAlgoLined
\SetKwInOut{Input}{Input}\SetKwInOut{Output}{Output}
\Input{Total Resource capacity $R$, Segment set $\S = \{(l_s, c_s, x_s)\}_{s \in \S}$ of MCUA of revenue function}
\Output{Resource Allocation Strategy $\Tilde{\sigmaa}$}
Define affordability threshold $t_l \leftarrow \min \left\{ R, \max_i \frac{d_l^i}{d_l^i+k}\right\}$ for all locations $l$ \;
Generate MCUA of the revenue function in range $[0, t_l]$ for each location $l$ \;
$\Tilde{\S} \leftarrow $ Ordered list of segments $s$ across all locations of this MCUA in descending order of slopes $c_s$  \;
Initialize available resources $\Tilde{R} \leftarrow R$ \;
Initialize allocation strategy $\Tilde{\sigmaa} \leftarrow \mathbf{0}$ \;
\For{\text{segment $s \in \Tilde{\S}$}}{
$\Tilde{\sigma}_{l_s} \leftarrow \Tilde{\sigma}_{l_s} + \min\{\Tilde{R}, x_s\}$ ; \texttt{\footnotesize \sf Allocate the minimum of the available resources and $x_{s}$ to location $l_s$} \;
        $\Tilde{R} \leftarrow \Tilde{R} -  \min\{\Tilde{R}, x_s\}$; \quad \texttt{\footnotesize \sf Update amount of remaining resources} \;
  }
\caption{\footnotesize Greedy Algorithm for MCUA of Revenue Function}
\label{alg:GreedyMCUA-OPT}
\end{algorithm}

We now establish Theorem~\ref{thm:greedy-half-approx-rev-max} using three intermediate lemmas. Our first lemma shows that the allocation computed via Algorithm~\ref{alg:GreedyMCUA-OPT} optimizes the MCUA of the revenue function, as is elucidated by the following lemma.

\begin{lemma}[MCUA is optimized by Algorithm~\ref{alg:GreedyMCUA-OPT}] \label{lem:greedy-opt-for-mcua}
The allocation $\Tilde{\sigmaa}$ computed using Algorithm~\ref{alg:GreedyMCUA-OPT} is an optimal solution to the problem of maximizing the MCUA of the revenue function, i.e., $\Tilde{\sigmaa}$ is an optimal solution to Problem~\eqref{eq:revenue-obj-mcua}-\eqref{eq:probability-feasibility-mcua}.
\end{lemma}

Our next lemma establishes a key structural property that the allocation corresponding to Algorithm~\ref{alg:GreedyMCUA-OPT} is such that the associated MCUA coincides with the original revenue function at all but at most one location.

%greedy allocation in step one of Algorithm~\ref{alg:GreedyRevMaxProb} results in an outcome that coincides with the original revenue function at all but at most one location.

\begin{lemma} [Structure of Greedy Allocation] \label{lem:greedy-algorithm-structure}
The allocation $\Tilde{\sigmaa}$ computed via Algorithm~\ref{alg:GreedyMCUA-OPT} is such that it does not coincide with the original revenue function at at most one location, i.e., there is at most one location $l'$ such that $\Hat{Q}_{l'}(\Tilde{\sigma}_{l'}) \neq Q_{l'}(\Tilde{\sigma}_{l'})$.
\end{lemma}

Our final lemma establishes that the MCUA of the revenue function coincides with the original revenue function at $\sigma_l = 0$ and at the allocation $\sigma_l = \sigma_l^{\max}$ at which the revenue function is maximized.

%at the endpoints of the domain, i.e., at $\sigma_l = 0$ and $\sigma_l = t_l$.

\begin{lemma} [MCUA Coincides with Revenue Function at Endpoints] \label{lem:mcua-coincides-revenue-het}
Let $\Hat{Q}_l(\sigma_l)$ be the MCUA of the revenue function $Q_l(\sigma_l)$. Then, the MCUA of the revenue function coincides with the revenue function at the allocations $\sigma_l = 0$ and $\sigma_l = \sigma_l^{\max}$, i.e., $\Hat{Q}_l(0) = Q_l(0)$ and $\Hat{Q}_l(\sigma_l^{\max}) = Q_l(\sigma_l^{\max})$.
\end{lemma}

Finally, we leverage Lemmas~\ref{lem:greedy-opt-for-mcua},~\ref{lem:greedy-algorithm-structure}, and~\ref{lem:mcua-coincides-revenue-het} to establish Theorem~\ref{thm:greedy-half-approx-rev-max}, as is elucidated by the following corollary.

\begin{corollary} [1/2 Approximation for Heterogeneous Revenue Maximization] \label{cor:half-approx-het-cont}
Suppose $\sigmaa^*$ is the optimal enforcement strategy corresponding to the solution of Problem~\eqref{eq:admin-obj-revenue}-\eqref{eq:bi-level-con-revenue} and $\sigmaa_A^*$ is the allocation corresponding to Algorithm~\ref{alg:GreedyRevMaxProb}.
Then, $Q_R(\sigmaa_A^*) \geq \frac{1}{2} Q_R(\sigmaa^*)$.
\end{corollary}

\subsubsection{Proof of Lemma~\ref{lem:greedy-opt-for-mcua}}

To prove this claim, let $\Tilde{\sigmaa}$ be the solution corresponding to Algorithm~\ref{alg:GreedyMCUA-OPT} and let $\Tilde{\sigmaa}^*$ be the allocation that maximizes the MCUA of the revenue function. Furthermore, suppose for contradiction that Algorithm~\ref{alg:GreedyMCUA-OPT} is not optimal, i.e., $\sum_{l \in L} \Hat{Q}_l(\Tilde{\sigma}_l) < \sum_{l \in L} \Hat{Q}_l(\Tilde{\sigma}_l^*)$. Such a relation implies by the (strict) monotonicity of the MCUA of the revenue function at each location that there exists a location $l'$ such that $\Tilde{\sigma}_{l'} < \sigma_{l'}^*$. However, we note that the total amount of resources allocated by Algorithm~\ref{alg:GreedyMCUA-OPT} is given by $\min \left\{ R, \sum_{l \in L} \max_{i} \frac{d_l^i}{d_l^i+k} \right\}$. However, note that the optimal allocation also cannot allocate any more resources than this amount to guarantee feasibility, i.e., it must be the case that $\sum_{l \in L} \Tilde{\sigma}_l = \sum_{l \in L} \Tilde{\sigma}_l^*$, which implies that there exists some location $\Tilde{l}$ such that $\Tilde{\sigma}_{\Tilde{l}} > \Tilde{\sigma}_{\Tilde{l}}^*$.

To derive our desired contradiction, we now construct another feasible strategy with a strictly higher objective than the strategy $\Tilde{\sigmaa}^*$. In particular, consider $\epsilon>0$ as some small positive constant such that $\Tilde{\sigma}_{\Tilde{l}}^{*} + \epsilon \leq \Tilde{\sigma}_{\Tilde{l}}$ and $\Tilde{\sigma}_{l'}^{*} \geq \Tilde{\sigma}_{l'} + \epsilon$. Then, define the strategy $\sigmaa' = (\sigma_l')_{l \in L}$, where $\sigma_l' = \Tilde{\sigma}_l^*$ for all $l \neq l', \Tilde{l}$, $\sigma_{l'} = \Tilde{\sigma}_{l'}^* - \epsilon$ and $\sigma_{\Tilde{l}}' = \Tilde{\sigma}_{\Tilde{l}}^{*} + \epsilon$. Note that this new strategy is feasible. In particular, it is straightforward to check that $\sigma_l' \in \left[0, \max_i \frac{d_l^i}{d_l^i+k} \right]$ for all locations $l$ and that the resource constraint is satisfied as:
\begin{align*}
    \sum_{l} \sigma_l' = \sum_{l \neq \{ l', \Tilde{l}\} } \sigma_l' + \sigma_{l'}' + \sigma_{\Tilde{l}}' = \sum_{l \neq \{ l', \Tilde{l}\} } \Tilde{\sigma}_l^* + \Tilde{\sigma}_{l'}^*-\epsilon + \Tilde{\sigma}_{\Tilde{l}}^*+\epsilon = \sum_{l} \Tilde{\sigma}_l^* \leq R,
\end{align*}
where the final inequality follows by the feasibility of $\Tilde{\sigmaa}^*$. 

Next, we show that the MCUA Objective~\eqref{eq:revenue-obj-mcua} of this new strategy $\sigmaa'$ is greater than that of $\Tilde{\sigmaa}^*$. To see this, note that:
\begin{align*}
    \Hat{Q}_R(\sigmaa') &= \sum_{l \in L} \Hat{Q}_l(\sigma_l), \\
    &= \sum_{l \neq \{ l', \Tilde{l}\} } \Hat{Q}_l(\Tilde{\sigma}_l^*) + \Hat{Q}_{l'}(\Tilde{\sigma}_{l'}^*-\epsilon) + \Hat{Q}_{\Tilde{l}}(\sigma_{\Tilde{l}}^* + \epsilon), \\
    &\stackrel{(a)}{\geq} \sum_{l \neq \{ l', \Tilde{l}\} } \Hat{Q}_l(\Tilde{\sigma}_l^*) + \Hat{Q}_{l'}(\Tilde{\sigma}_{l'}^*) - \partial \Hat{Q}_{l'}(\Tilde{\sigma}_{l'}^*-\epsilon) \epsilon + \Hat{Q}_{\Tilde{l}}(\sigma_{\Tilde{l}}^*) + \partial \Hat{Q}_{\Tilde{l}}(\sigma_{\Tilde{l}}^* + \epsilon) \epsilon, \\
    &= \sum_{l \in L} \Hat{Q}_l(\Tilde{\sigma}_l^*) + \epsilon (\partial \Hat{Q}_{\Tilde{l}}(\Tilde{\sigma}_{\Tilde{l}}^* + \epsilon) - \partial \Hat{Q}_{l'}(\Tilde{\sigma}_{l'}^*-\epsilon)), \\
    &\stackrel{(b)}{\geq} \sum_{l \in L} \Hat{Q}_l(\Tilde{\sigma}_l^*) + \epsilon (\partial \Hat{Q}_{\Tilde{l}}(\Tilde{\sigma}_{\Tilde{l}}^* + \epsilon) - \partial \Hat{Q}_{l'}(\Tilde{\sigma}_{l'})), \\
    &\stackrel{(c)}{\geq} \sum_{l \in L} \Hat{Q}_l(\Tilde{\sigma}_l^*) = \Hat{Q}_R(\Tilde{\sigmaa}^*),
\end{align*}
where (a) follows by the concavity of the MCUA of the revenue function and the fact that the MCUA of the revenue function is super-differentiable at all points in the domain $\sigma_l \in [0, t_l]$. Note that any concave function on a non-empty convex set is super-differentiable at each interior point and the super-differentiability at the points $\sigma_l = 0$ and $\sigma_l = t_l$ follow as an immediate consequence of the fact that an affine function upper bounds the welfare function due to the discreteness of the distribution. Next, inequality (b) follows by the fact that $\Hat{Q}_{l'}$ is a monotonically increasing concave function and hence its derivative (and minimum sub-gradient) is monotonically non-increasing and (c) follows from the fact that the minimum sub-gradient $\partial \Hat{Q}_{\Tilde{l}}(\Tilde{\sigma}_{\Tilde{l}}^* + \epsilon) \geq \partial \Hat{Q}_{l'}(\Tilde{\sigma}_{l'}))$ by the nature of the greedy algorithm that allocates resources to locations in the descending order of the slopes of the segments of the MCUA of the revenue function (as otherwise the greedy algorithm would have allocated more resources to location $l'$ instead of location $\Tilde{l}$). The above relationship implies that the allocation strategy $\sigmaa'$ is optimal. We can then repeat the above procedure of comparing the allocation $\sigmaa'$ to the greedy allocation corresponding to Algorithm~\ref{alg:GreedyMCUA-OPT} and construct a sequence of allocations that will eventually correspond exactly to the greedy allocation, which yields our desired contradiction. Consequently, Algorithm~\ref{alg:GreedyMCUA-OPT} maximizes the MCUA of the revenue function.

\subsubsection{Proof of Lemma~\ref{lem:greedy-algorithm-structure}}

%We now show that the structure of the allocation corresponding to Algorithm~\ref{alg:GreedyMCUA-OPT} is such that the greedy algorithm allocates to all locations but at most one on the original revenue function, i.e., it holds that $\Hat{Q}_R(\Tilde{\sigmaa}) = \sum_{l \neq l'} Q_l(\Tilde{\sigma}_l) + \Hat{Q}_{l'}(\Tilde{\sigma}_{l'})$, where $\Hat{Q}_{l'}(\Tilde{\sigma}_{l'}) > Q_{l'}(\Tilde{\sigma}_{l'})$ for at most one location $l'$.

Suppose for contradiction that there are at least two locations $l_1, l_2$ such that $\Hat{Q}_{l_1}(\Tilde{\sigma}_{l_1}) > Q_{l_1}(\Tilde{\sigma}_{l_1})$ and $\Hat{Q}_{l_2}(\Tilde{\sigma}_{l_2}) > Q_{l_2}(\Tilde{\sigma}_{l_2})$. Note by the definition of the MCUA that the points at which the MCUA and the original revenue function do not coincide correspond to line segments with a fixed slope. Consequently, since the greedy algorithm, i.e., Algorithm~\ref{alg:GreedyMCUA-OPT}, allocates some resources to both line segments corresponding to locations $l_1$ and $l_2$ it follows that the slope of these segments must be the same by the nature of the greedy algorithm. But, under our greedy algorithm, we iterate through each of the segments sequentially, which implies that there is at most one segment $s$ for which the allocation is strictly less than $x_s$. 
Finally, noting that if a segment is allocated either zero or $x_s$ that the resulting allocation lies on the original revenue function, it follows that there is at most one location for which $\Hat{Q}_{l}(\Tilde{\sigma}_{l}) > Q_{l}(\Tilde{\sigma}_{l})$, which establishes our claim.

%it suffices to consider an equivalent allocation wherein we first fully allocate to one of the line segments (if we had enough resources available to do so) and then allocated to the next segment for $l_1$ (depending on the amount of resources remaining). Thus, without loss of generality, we can consider the greedy algorithm where we fully allocate to segments before allocating to other segments not on our original welfare curve, in which case our claim holds.

\subsubsection{Proof of Lemma~\ref{lem:mcua-coincides-revenue-het}}

We first show that $\Hat{Q}_l(0) = Q_l(0)$. To see this, we proceed by contradiction and assume that $\Hat{Q}_l(0) > Q_l(0)$. 

To establish our desired contradiction, we first note two properties of the MCUA of the revenue function. First note that at the points where the MCUA of the revenue function does not coincide with the revenue function, the associated MCUA is affine. Note that if the MCUA was not affine at the points where it does not coincide with the revenue function, then this would violate the definition of the MCUA, which is a point-wise smallest concave function that upper bounds the revenue function. Next, note by the definition of the MCUA of the revenue function that the MCUA must coincide with the revenue function at some point. Note that if the MCUA did not coincide with the MCUA at at least one point in our compact convex domain, we could point-wise decrease the associated function by a constant, which would violate the fact that the function was an MCUA of the revenue function.

Consequently, let $\sigma_l^1$ be the first resource consumption value at which the MCUA of the revenue function coincides with the revenue function, i.e., $\Hat{Q}_l(\sigma_l^1) = Q_l(\sigma_l^1)$. Moreover, let $s_1$ be the slope of the segment from $\Hat{Q}_l(0)$ to the first $\sigma_l$ at which $\Hat{Q}_l(\sigma_l) = Q_l(\sigma_l)$. Since $\Hat{Q}_l(\sigma_l) > Q_l(\sigma_l)$ for all $\sigma_l \in [0, \sigma_l^1)$, i.e., there exists some $\epsilon>0$ such that $\Hat{Q}_l(\sigma_l) + \epsilon \geq Q_l(\sigma_l)$ for all $\sigma_l \in [0, \sigma_l^1)$, it follows that if we consider the segment from $\Hat{Q}_l(0)-\epsilon$ to $\Hat{Q}_l(\sigma_l^1)$, the corresponding curve is point-wise no greater than the original MCUA function and is strictly smaller over the domain $\sigma_l \in [0, \sigma_l^1)$. Moreover, the resulting curve preserves concavity, as the new segment is a slope that is strictly greater than $s_1$. %Note that constructing such a curve is possible as $s_1$ is bounded due to the discreteness of the distribution by Assumption~\ref{asmpn:ub-linear-function}. 
Consequently, we have constructed a new function that satisfies the required properties of an MCUA and is strictly smaller than the original MCUA function over the domain $\sigma_l \in [0, \sigma_l^1)$. Thus, we have obtained our desired contradiction that the original function was an MCUA, which implies that $\Hat{Q}_l(0) = Q_l(0)$.

We can apply an entirely analogous line of reasoning as above to establish that $\Hat{Q}_l(\sigma_l^{\max}) = Q_l(\sigma_l^{\max})$.

\subsubsection{Proof of Corollary~\ref{cor:half-approx-het-cont}}

To establish this claim, we first introduce some notation. In particular, let $\Tilde{\sigmaa}^*$ be the solution corresponding to maximizing the MCUA of the revenue function. Furthermore, let $\sigmaa^*$ be the solution corresponding to maximizing the revenue function. Then, it follows that:
\begin{align} \label{eq:mcua-ub-opt}
    \Hat{Q}_R(\Tilde{\sigmaa}^*) \stackrel{(a)}{\geq} \Hat{Q}_R(\sigmaa^*) \stackrel{(b)}{\geq} Q_R(\sigmaa^*),
\end{align}
where (a) follows from the optimality of $\Tilde{\sigmaa}^*$ for the problem of maximizing the MCUA of the revenue function and (b) follows from the fact that the MCUA of the revenue function is a point-wise upper bound on the revenue function at each location. Consequently, the above inequalities imply that the optimal objective of the MCUA of the revenue function is at least that of maximizing the revenue function directly.

In the remainder of this proof, we leverage Lemmas~\ref{lem:greedy-opt-for-mcua},~\ref{lem:greedy-algorithm-structure}, and~\ref{lem:mcua-coincides-revenue-het} to show that the allocation corresponding to Algorithm~\ref{alg:GreedyRevMaxProb} in fact upper bounds the optimal objective of the MCUA of the revenue function (up to a factor of two), which thus implies the desired half approximation guarantee by Equation~\eqref{eq:mcua-ub-opt}. Thus, in the remainder of this proof we show that $2 Q_R(\sigmaa_A^*) \geq \Hat{Q}_R(\Tilde{\sigmaa}^*)$.

To establish this inequality, we first note from Lemmas~\ref{lem:greedy-opt-for-mcua},~\ref{lem:greedy-algorithm-structure}, and~\ref{lem:mcua-coincides-revenue-het} that the greedy allocation corresponding to Algorithm~\ref{alg:GreedyMCUA-OPT} achieves an optimal MCUA objective of $\Hat{Q}_R(\Tilde{\sigmaa}^*) = \sum_{l \neq l'} Q_l(\Tilde{\sigma}_l^*) + \Hat{Q}_{l'}(\Tilde{\sigma}_{l'}^*)$, with at most one location $l'$ such that $\Hat{Q}_{l'}(\Tilde{\sigma}_{l'}^*) > Q_{l'}(\Tilde{\sigma}_{l'}^*)$. We now establish the following two inequalities: (i) $Q_R(\Tilde{\sigmaa}) \geq \sum_{l \neq l'} Q_l(\Tilde{\sigma}_l^*)$ and (ii) $Q_R(\sigmaa') \geq \Hat{Q}_{l'}(\Tilde{\sigma}_{l'}^*)$, where $\Tilde{\sigmaa}$ and $\sigmaa'$ are the allocations corresponding to steps one and two of Algorithm~\ref{alg:GreedyRevMaxProb}, respectively.

\paragraph{Inequality (i):} We recall from earlier that the solution maximizing the MCUA of the revenue function exactly corresponds to the solution in step one of Algorithm~\ref{alg:GreedyRevMaxProb} other than that resources are allocated until either all the resources are exhausted or there are no further segments remaining to iterate over. Consequently, it holds that $\Tilde{\sigma}^*_l = \Tilde{\sigma}_l$ for all $l \neq l'$ and for $l = l'$ it holds that $\Tilde{\sigma}^*_{l'} \geq \Tilde{\sigma}_{l'}$, where $\Tilde{\sigma}^*_{l'} \leq \Tilde{\sigma}_{l'} + x_{s'}$. Thus, it follows that $Q_R(\Tilde{\sigmaa}) = \sum_{l \in L \backslash \{ l' \}} Q_l(\Tilde{\sigma}_l^*)$, which implies our desired inequality.

\paragraph{Inequality (ii):} Furthermore, observe that $Q_R(\sigmaa') \geq Q_{l'}(\sigma_{l'}^{\max})$, where $\sigma_{l'}^{\max}$ corresponds to the allocation that maximizes the revenue at location $l'$ among the set of all affordable allocations, i.e., $\sigma_{l'}^{\max} = \argmax_{\sigma_{l'} \in [0, t_l]} Q_{l'}(\sigma_l)$. Since the MCUA of the revenue function is such that $\Hat{Q}_{l}(\sigma_l^{\max}) = Q_l(\sigma_l^{\max})$ (see Lemma~\ref{lem:mcua-coincides-revenue-het}) for all locations $l$ and is monotonically increasing, it follows that $Q_R(\sigmaa') \geq Q_{l'}(\sigma_{l'}^{\max}) = \Hat{Q}_{l}(\sigma_l^{\max}) \geq \Hat{Q}_{l'}(\Tilde{\sigma}_{l'}^*)$.

Together, inequalities (i) and (ii) imply that:
\begin{align} \label{eq:ub-mcua-greedy-algorithm-het}
    Q_R(\Tilde{\sigmaa}^*) + Q_R(\sigmaa') \geq \sum_{l \neq l'} Q_l(\Tilde{\sigma}_l^*) + \Hat{Q}_{l'}(\Tilde{\sigma}_{l'}^*) = \Hat{Q}_R(\Tilde{\sigmaa}^*).
\end{align}
Finally, since $\sigmaa_A^* = \argmax \{ Q_R(\Tilde{\sigmaa}^*), Q_R(\sigmaa') \}$ from step 3 of Algorithm~\ref{alg:GreedyRevMaxProb}, it follows from the above
inequality that
\begin{align*}
    2 Q_R(\sigmaa_A^*) \geq Q_R(\Tilde{\sigmaa}^*) + Q_R(\sigmaa') \stackrel{(a)}{\geq} \Hat{Q}_R(\Tilde{\sigmaa}^*) \stackrel{(a)}{\geq} Q_R(\sigmaa^*),
\end{align*}
where (a) follows by Equation~\eqref{eq:ub-mcua-greedy-algorithm-het} and (b) follows by Equation~\eqref{eq:mcua-ub-opt}. Thus, we have obtained our desired half approximation guarantee, which establishes our desired result.

\subsection{Proof of Theorem~\ref{thm:greedy-resource-augmentation-rev-max}} \label{apdx:pf-resource-aug-rev-max}

To prove this claim, we first recall that, by construction, the MCUA is an upper bound on the revenue function for all feasible administrator strategies, i.e., $\Hat{Q}_R(\sigmaa) \geq Q_R(\sigmaa)$ for all $\sigmaa \in \Omega_R$ (also see proof of Theorem~\ref{thm:greedy-half-approx-rev-max}). Next, to establish our resource augmentation guarantee with $R+1$ resources, we define $\Tilde{\sigmaa}^{R+1}$ as the allocation corresponding to Step 1 of Algorithm~\ref{alg:GreedyRevMaxProb} with $R+1$ resources and $\Tilde{\sigmaa}_A^{R+1}$ as the allocation corresponding to Algorithm~\ref{alg:GreedyRevMaxProb} with $R+1$ resources. Moreover, we let $\Tilde{\sigmaa}^*$ be the optimal solution corresponding to directly optimizing the MCUA of the expected revenue function as given by Algorithm~\ref{alg:GreedyMCUA-OPT}. We now show that $Q_R(\sigmaa^*) \leq Q_{R+1}(\Tilde{\sigmaa}_A^{R+1})$ to establish our claim.

To show that $Q_R(\sigmaa^*) \leq Q_{R+1}(\Tilde{\sigmaa}_A^{R+1})$, we establish four key relations: (i) $Q_R(\sigmaa^*) \leq \Hat{Q}_R(\Tilde{\sigmaa}^*)$, (ii) $\Hat{Q}_R(\Tilde{\sigmaa}^*) \leq \Hat{Q}_{R+1}(\Tilde{\sigmaa}^{R+1})$, (iii) $\Hat{Q}_{R+1}(\Tilde{\sigmaa}^{R+1}) = Q_{R+1}(\Tilde{\sigmaa}^{R+1})$, and (iv) $Q_{R+1}(\Tilde{\sigmaa}^{R+1}) \leq Q_{R+1}(\Tilde{\sigmaa}_A^{R+1})$.

\paragraph{Relation (i): $Q_R(\sigmaa^*) \leq \Hat{Q}_R(\Tilde{\sigmaa}^*)$:} By construction, the MCUA of the revenue function upper bounds the revenue for all allocation strategies $\sigmaa$. Hence, it follows that:
\begin{align*}
    Q_R(\sigmaa^*) \leq \Hat{Q}_R(\sigmaa^*) \leq \Hat{Q}_R(\Tilde{\sigmaa}^*),
\end{align*}
where the first inequality follows as the MCUA of the revenue function upper bounds the revenue at $\sigmaa^*$ and the latter inequality follows by the optimality of $\Tilde{\sigmaa}^*$ for the MCUA of the revenue function.

\paragraph{Relation (ii): $\Hat{Q}_R(\Tilde{\sigmaa}^*) \leq \Hat{Q}_{R+1}(\Tilde{\sigmaa}^{R+1})$:}

We recall from the proof of Theorem~\ref{thm:greedy-half-approx-rev-max} that the allocation maximizing the MCUA of the revenue function corresponds to the solution of Algorithm~\ref{alg:GreedyMCUA-OPT}. In other words, the allocation maximizing the MCUA of the revenue function corresponds to the solution in the step 1 of Algorithm~\ref{alg:GreedyRevMaxProb}, where resources are allocated until either all the resources are exhausted or there are no further segments are remaining to iterate over. Consequently, by the termination condition of step one of Algorithm~\ref{alg:GreedyMCUA-OPT}, there is at most one segment $s'$ and associated location $l_{s'}$ such that the solutions of step one of Algorithm~\ref{alg:GreedyRevMaxProb} and that of Algorithm~\ref{alg:GreedyMCUA-OPT} differ. In particular, it holds that $\Tilde{\sigma}^*_l = \Tilde{\sigma}_l$ for all $l \neq l_{s'}$ and for $l = l_{s'}$ it holds that $\Tilde{\sigma}^*_{l_{s'}} \geq \Tilde{\sigma}_{l_{s'}}$, where $\Tilde{\sigma}^*_{l_{s'}} \leq \Tilde{\sigma}_{l_{s'}} + x_{s'}$. Without loss of generality, we assume that $\Tilde{\sigma}^*_{l_{s'}} > \Tilde{\sigma}_{l_{s'}}$, as if this were not the case, then our desired inequality trivially holds.

Next, with $R+1$ resources, it must hold that the allocation $\Tilde{\sigmaa}^{R+1}$ corresponding to step one of Algorithm~\ref{alg:GreedyRevMaxProb} satisfies $\Tilde{\sigma}_l^{R+1} \geq \Tilde{\sigma}_l = \Tilde{\sigma}^*_l$ for all $l \neq l_{s'}$, as the greedy allocation with $R+1$ resources will at least allocate the same amount of resources to all locations as that with $R$ resources. Furthermore, since $x_{s'} \leq 1$ for all segments $s$, it also holds for $l = l_{s'}$ that the greedy allocation will completely allocate $x_{s'}$ to location $l_{s'}$; hence, it holds that $\Tilde{\sigma}_{l_{s'}}^{R+1} \geq \Tilde{\sigma}_{l_{s'}} + x_{s'} \geq \Tilde{\sigma}^*_{l_{s'}}$.

Thus, we have established for all locations $l$ that $\Tilde{\sigma}^{R+1}_l \geq \Tilde{\sigma}^*_l$. Finally, our desired inequality that $\Hat{Q}_R(\Tilde{\sigmaa}^*) \leq \Hat{Q}_{R+1}(\Tilde{\sigmaa}^{R+1})$ follows by the monotonicity of the MCUA of the revenue function in the resource spending at each location.

\paragraph{Relation (iii): $\Hat{Q}_{R+1}(\Tilde{\sigmaa}^{R+1}) = Q_{R+1}(\Tilde{\sigmaa}^{R+1})$:} We note that the allocation corresponding to the greedy procedure in Step 1 of Algorithm~\ref{alg:GreedyRevMaxProb} corresponds to a point where both the revenue function and its corresponding MCUA coincide. Consequently, it follows that $\Hat{Q}_{R+1}(\Tilde{\sigmaa}^{R+1}) = Q_{R+1}(\Tilde{\sigmaa}^{R+1})$.

%satisfies either $\Tilde{\sigma}^{R+1}_l = 0$ or $\Tilde{\sigma}^{R+1}_l = \frac{d_l^i}{d_l^i+k}$ for some type $i$ for all locations $l$. Consequently, the MCUA of the expected revenue function and the expected revenue function coincide at the allocation $\Tilde{\sigmaa}^{R+1}$; hence, it follows that $M(\Tilde{\sigmaa}^{R+1}) = Q_{R+1}(\Tilde{\sigmaa}^{R+1})$.

\paragraph{Relation (iv): $Q_{R+1}(\Tilde{\sigmaa}^{R+1}) \leq Q_{R+1}(\Tilde{\sigmaa}_A^{R+1})$:} This result directly follows, by definition, from the fact that $\Tilde{\sigmaa}_A^{R+1} = \argmax \{ Q_{R+1}(\Tilde{\sigmaa}), Q_{R+1}(\sigmaa') \}$. 

Finally, combining the above four relations that we have shown, we obtain that:
\begin{align*}
    Q_R(\sigmaa^*) \leq \Hat{Q}_{R}(\Tilde{\sigmaa}^*) \leq \Hat{Q}_{R+1}(\Tilde{\sigmaa}^{R+1}) = Q_{R+1}(\Tilde{\sigmaa}^{R+1}) \leq Q_{R+1}(\Tilde{\sigmaa}_A^{R+1}),
\end{align*}
which proves our claim that $Q_R(\sigmaa^*) \leq Q_{R+1}(\Tilde{\sigmaa}_A^{R+1})$, i.e., Algorithm~\ref{alg:GreedyRevMaxProb} with $R+1$ resources achieves a total revenue that is at least that corresponding to the optimal solution of Problem~\eqref{eq:admin-obj-revenue}-\eqref{eq:bi-level-con-revenue}. This establishes our claim.

\subsection{Proof of Theorem~\ref{thm:apx-greedy-resource-augmentation-constraints}} \label{apdx:pf-thm-additional-constraints-apx}

%To elucidate the key ideas, we focus on the setting with upper bound constraints and present the additional caveats required to extend this proof to the setting with lower bound constraints in Appendix~\ref{apdx:lb-con-extension-pf}. 

In the following, we prove Theorem~\ref{thm:apx-greedy-resource-augmentation-constraints} using four intermediate lemmas, which we elucidate below. After presenting the statements of these lemmas, we present their proofs in the remainder of this section.

To elucidate our lemmas, we begin by defining the following optimization problem, which rather than maximizing the payoff across all locations, maximizes the sum of the administrator's payoffs across all sets of locations in the level $\L_1$ of the hierarchy:
\begin{maxi!}|s|[2]   
    {(\Tilde{R}_S)_{S \in \L_1} \in \mathbbm{R}^{|\L_1|}_{\geq 0}}                   
    { \sum_{S \in \L_1} P_S(\Tilde{R}_S),  \label{eq:welfare-obj-decomp21}}   
    {\label{eq:Eg00121}}             
    {}          
    \addConstraint{\sum_{S \in \L_1} \Tilde{R}_S}{\leq R, \label{eq:resource-constraint-decomp21}} 
    \addConstraint{\Tilde{R}_S}{\leq \Bar{\lambda}_S, \quad \forall S \in \L_1 \label{eq:L1-con21}} 
    \addConstraint{\sum_{S' \subseteq S} \Tilde{R}_{S'}}{\leq \Bar{\lambda}_S, \quad \forall S \in \L_2 \cup \ldots \cup \L_t, \label{eq:probability-feasibility21}}
\end{maxi!}
where, for each set $S \in \L_1$, we define
\begin{maxi!}|s|[2]   
    {(\sigma_l)_{l\in S}}                   
    { \sum_{l \in S} P_l(\sigma_l),  \label{eq:welfare-obj-decomp2}}   
    {\label{eq:Eg001}}             
    {P_S(\Tilde{R}_S) = }          
    \addConstraint{\sum_{l\in S} \sigma_l}{\leq \Bar{\lambda}_S, \label{eq:resource-constraint-decomp2}} 
    \addConstraint{\sigma_l}{\in [0, 1]. \label{eq:probab-constraint-decomp2}} 
    %\addConstraint{2}{\leq \sum_{l \in S} \sigma_l \leq \Tilde{R}_S, \quad \forall S \in \L_1. \label{eq:L1-con2}} 
\end{maxi!}

Our first lemma establishes an equivalence between Problem~\eqref{eq:welfare-obj-decomp21}-\eqref{eq:probability-feasibility21} and the administrator's payoff maximization Problem~\eqref{eq:welfare-obj}-\eqref{eq:ub-lb-con} by establishing that the optimal payoff corresponding to both problems is equal.

\begin{lemma} \label{lem:opt-sol-reparametrization}
Suppose that $\sigmaa^*$ is the optimal solution to Problem~\eqref{eq:welfare-obj}-\eqref{eq:ub-lb-con} and $(R_S^*)_{S \in \L_1}$ is the optimal solution to Problem~\eqref{eq:welfare-obj-decomp21}-\eqref{eq:probability-feasibility21}. Then, the optimal payoff corresponding to the solutions of both optimization problems are equal, i.e., $\sum_{S \in \L_1} P_S(R_S^*) = P_R(\sigmaa^*)$.
\end{lemma}

Lemma~\ref{lem:opt-sol-reparametrization} establishes that, without loss of generality, it suffices to focus our attention on Problem~\eqref{eq:welfare-obj-decomp21}-\eqref{eq:probability-feasibility21} to establish our desired approximation ratio guarantee. We also note that analogous to Problem~\eqref{eq:welfare-obj-decomp21}-\eqref{eq:probability-feasibility21}, which corresponds to maximizing the administrator's payoffs, we can formulate an analogous problem of maximizing the MCUA of the payoff function as follows:
\begin{maxi!}|s|[2]   
    {(\Tilde{R}_S)_{S \in \L_1} \in \mathbbm{R}^{|\L_1|}_{\geq 0}}                   
    { \sum_{S \in \L_1} \Hat{P}_S(\Tilde{R}_S),  \label{eq:welfare-obj-decomp21-mcua}}   
    {\label{eq:Eg00121-mcua}}             
    {}          
    \addConstraint{\sum_{S \in \L_1} \Tilde{R}_S}{\leq R, \label{eq:resource-constraint-decomp21-mcua}} 
    \addConstraint{\Tilde{R}_S}{\leq \Bar{\lambda}_S, \quad \forall S \in \L_1 \label{eq:L1-con21-mcua}} 
    \addConstraint{\sum_{S' \subseteq S} \Tilde{R}_{S'}}{\leq \Bar{\lambda}_S, \quad \forall S \in \L_2 \cup \ldots \cup \L_t, \label{eq:probability-feasibility21-mcua}}
\end{maxi!}
where, for each set $S \in \L_1$, we define
\begin{maxi!}|s|[2]   
    {(\sigma_l)_{l\in S}}                   
    { \sum_{l \in S} \Hat{P}_l(\sigma_l),  \label{eq:welfare-obj-decomp2-mcua}}   
    {\label{eq:Eg001-mcua}}             
    {\Hat{P}_S(\Tilde{R}_S) = }          
    \addConstraint{\sum_{l\in S} \sigma_l}{\leq \Bar{\lambda}_S, \label{eq:resource-constraint-decomp2-mcua}} 
    \addConstraint{\sigma_l}{\in [0, 1], \label{eq:probab-constraint-decomp2-mcua}} 
\end{maxi!}
where we denote $\Hat{P}_l(\sigma_l)$ as the MCUA of the payoff function $P_l(\sigma_l)$ at each location $l$. 

Our next lemma establishes that maximizing the MCUA of the payoff function, i.e., solving Problem~\eqref{eq:welfare-obj-decomp21-mcua}-\eqref{eq:probability-feasibility21-mcua}, results in a higher objective compared to directly maximizing the payoff function, i.e., solving Problem~\eqref{eq:welfare-obj-decomp21}-\eqref{eq:probability-feasibility21}.

\begin{lemma} \label{lem:ub-greedy-reparametrization}
Suppose that $(R_S^*)_{S \in \L_1}$ is the optimal solution to Problem~\eqref{eq:welfare-obj-decomp21}-\eqref{eq:probability-feasibility21} and $(\Hat{R}_S^*)_{S \in \L_1}$ is the optimal solution to Problem~\eqref{eq:welfare-obj-decomp21-mcua}-\eqref{eq:probability-feasibility21-mcua}. Then, it holds that $\sum_{S \in \L_1} P_S(R_S^*) \leq  \sum_{S \in \L_1} \Hat{P}_S(\Hat{R}_S^*)$.
\end{lemma}

Lemma~\ref{lem:ub-greedy-reparametrization} establishes that it suffices to obtain the desired approximation ratio guarantees with respect to the objective of maximizing the MCUA of the payoff functions.

Our third lemma establishes a characterization of the optimal solution of Problem~\eqref{eq:welfare-obj-decomp2-mcua}-\eqref{eq:probab-constraint-decomp2-mcua}. In particular, Lemma~\ref{lem:parametrization-simple-lem1} establishes that for any allocation $\Tilde{R}_S$ corresponding to a given set $S \in \L_1$, a greedy algorithm that allocates resources to locations in set $S$ in the descending order of slopes of the MCUA of the payoff functions achieves the optimal solution to Problem~\eqref{eq:welfare-obj-decomp2-mcua}-\eqref{eq:probab-constraint-decomp2-mcua}.

\begin{lemma} \label{lem:parametrization-simple-lem1}
Suppose $\Tilde{R}_S$ corresponds to the allocation for a set $S \in \L_1$ and $(\sigma_l')_{l \in S}$ is the allocation computed using step one of Algorithm~\ref{alg:GreedyWelMaxProb} for the location subset $S$ given $\Tilde{R}_S$ resources. Then, it holds that $\Hat{P}_S(\Tilde{R}_S) = \sum_{l \in S} \Hat{P}_l(\sigma_l')$.
\end{lemma}

Lemma~\ref{lem:parametrization-simple-lem1} is akin to Lemma~\ref{lem:greedy-opt-for-mcua} and thus we omit its proof for brevity. 

Our final lemma establishes that Step 1 of the Constrained-Greedy algorithm computes the optimal allocation vector $(\Hat{R}_S^*)_{S \in \L_1}$ that corresponds to the solution to Problem~\eqref{eq:welfare-obj-decomp21-mcua}-\eqref{eq:probability-feasibility21-mcua}.

\begin{lemma} \label{lem:parametrization-simple-lem2}
Suppose that $(\Hat{R}_S^*)_{S \in \L_1}$ is the optimal solution to Problem~\eqref{eq:welfare-obj-decomp21-mcua}-\eqref{eq:probability-feasibility21-mcua}. Then, the resulting allocation computed in step 1 of the Constrained-Greedy algorithm achieves the optimal objective corresponding to the solution of Problem~\eqref{eq:welfare-obj-decomp21-mcua}-\eqref{eq:probability-feasibility21-mcua}.
\end{lemma}

We note that Lemma~\ref{lem:parametrization-simple-lem1} is analogous to several classical results in the literature which establish the optimality of greedy based algorithms when optimizing over a constraint structure that corresponds to a polymatroid (note that a hierarchy is a special case of a polymatroid constraint structure).

Finally, we combine Lemmas~\ref{lem:opt-sol-reparametrization}-\ref{lem:parametrization-simple-lem2} to complete the proof of Theorem~\ref{thm:apx-greedy-resource-augmentation-constraints} in Appendix~\ref{apdx:completing-pf-thm10-additional-constraints}.

\begin{comment}
Our final lemma establishes that 

\begin{lemma} \label{lem:half-approx-reparametrization}
Suppose that $(\Hat{R}_S^*)_{S \in \L_1}$ is the optimal solution to Problem~\eqref{eq:welfare-obj-decomp21-mcua}-\eqref{eq:probability-feasibility21-mcua}. Furthermore, let $\sigmaa_{R+|\L_1|}^A$ be the solution corresponding to the constrained greedy algorithm with $R+|\L_1|$ resources, and let $\sigmaa_{R+|\L_2|}^A$ be the solution corresponding to the constrained greedy algorithm with $R+|\L_2|$ resources. Then:
\begin{itemize}
    \item The total welfare under allocation $\sigmaa_{R+|\L_1|}^A$ is at least that under the constrained welfare maximizing allocation, i.e., $W_{R+|\L_1|}(\sigmaa_{R+|\L_1|}^A) \geq \sum_{S \in \L_1} \Hat{W}_S(\Hat{R}_S^*)$,
    \item The total welfare under allocation $\sigmaa_{R+|\L_2|}^A$ is at least half of that under the constrained welfare maximizing allocation, i.e., $W_{R+|\L_2|}(\sigmaa_{R+|\L_2|}^A) \geq \frac{1}{2} \sum_{S \in \L_1} \Hat{W}_S(\Hat{R}_S^*)$.
\end{itemize}
\end{lemma}   

We now establish Lemmas~\ref{lem:opt-sol-reparametrization},~\ref{lem:ub-greedy-reparametrization}, and~\ref{lem:half-approx-reparametrization}.
\end{comment}

\subsubsection{Proof of Lemma~\ref{lem:opt-sol-reparametrization}}

($\implies$:) Suppose that $\sigmaa^*$ is an optimal solution to Problem~\eqref{eq:welfare-obj}-\eqref{eq:ub-lb-con}. Then, we show that the optimal objective of Problem~\eqref{eq:welfare-obj-decomp21}-\eqref{eq:probability-feasibility21} is at least $P_R(\sigmaa^*)$. To see this, let $\Tilde{R}_S = \sum_{l \in S} \sigma_l^*$. Then, it clearly holds by the feasibility of $\sigmaa^*$ that $(\Tilde{R}_S)_{S \in \L_1}$ is feasible. To see this, note that:
\begin{itemize}
    \item $\sum_{S \in \L_1} \Tilde{R}_S = \sum_{S \in \L_1} \sum_{l \in S} \sigma_l^* = \sum_{l \in L} \sigma_l^* \leq R$, where the second equality follows by the definition of $\L_1$ for which it holds that each location $l$ belongs to exactly one set $S \in \L_1$.
    \item $\Tilde{R}_S = \sum_{l \in S} \sigma_l^* \leq \lambda_S,$ for all $S \in \L_1$
    \item $\sum_{S' \subseteq S} \Tilde{R}_{S'} = \sum_{S' \subseteq S} \sum_{l \in S'} \sigma_l^* = \sum_{l \in S} \sigma_l^* \leq \Bar{\lambda}_S$, for all $S \in \L_2 \cup \ldots \cup \L_t$, where the second equality follows by the definition of the layer $\L_1$ for which it holds that each location $l$ belongs to exactly one set $S' \in \L_1$.
\end{itemize}
Next, by the definition of $P_S(\Tilde{R}_S)$ in Equations~\eqref{eq:welfare-obj-decomp2}-\eqref{eq:probab-constraint-decomp2}, it follows that $P_S(\Tilde{R}_S) \geq \sum_{l \in S} P_l(\sigma_l^*)$. Thus, we have shown that:
\begin{align}
    P_R(\sigmaa^*) \stackrel{(a)}{=} \sum_{S \in \L_1} \sum_{l \in S} P_l(\sigma_l^*) \stackrel{(b)}{\leq} \sum_{S \in \L_1} P_S(\Tilde{R}_S) \stackrel{(c)}{\leq} \sum_{S \in \L_1} P_S(R_S^*),
\end{align}
where (a) follows by the definition of the layers $\L_1$ for which it holds that each location $l$ belongs to exactly one set $S \in \L_1$, (b) follows from our above established fact that $P_S(\Tilde{R}_S) \geq \sum_{l \in S} P_l(\sigma_l^*)$, and (c) follows by the optimality of $(R_S^*)_{S \in \L_1}$ for Problem~\eqref{eq:welfare-obj-decomp21}-\eqref{eq:probability-feasibility21}.

The proof of the reverse direction of this claim follows an almost analogous line of reasoning where we establish that the optimal solution of Problem~\eqref{eq:welfare-obj-decomp21}-\eqref{eq:probability-feasibility21} can be transformed to a feasible solution to Problem~\eqref{eq:welfare-obj}-\eqref{eq:probability-feasibility}. We omit the details of this claim for brevity.

\subsubsection{Proof of Lemma~\ref{lem:ub-greedy-reparametrization}}

Let $(\Bar{\sigma}_l^*)_{l \in S}$ be the optimal solution corresponding to Problem~\eqref{eq:welfare-obj-decomp2}-\eqref{eq:probab-constraint-decomp2} given the allocation $R_S^*$ for each set $S \in \L_1$, i.e., $P_S(R_S^*) = \sum_{l \in S} P_l(\Bar{\sigma}_l^*)$.

Then, to establish this claim, consider the following sequence of inequalities:
\begin{align*}
    \sum_{S \in \L_1} P_S(R_S^*) &= \sum_{S \in \L_1} \sum_{l \in S} P_l(\Bar{\sigma}_l^*) \stackrel{(a)}{\leq} \sum_{S \in \L_1} \sum_{l \in S} \Hat{P}_l(\Bar{\sigma}_l^*) \stackrel{(b)}{\leq} \sum_{S \in \L_1} \Hat{P}_S(R_S^*) \stackrel{(c)}{\leq} \sum_{S \in \L_1} \Hat{P}_S(\Hat{R}_S^*),
\end{align*}
where (a) follows as $\Hat{P}_l(\Bar{\sigma}_l^*) \geq P_l(\Bar{\sigma}_l^*)$ as $\Hat{P}_l(\cdot)$ is the MCUA of the payoff function $P_l(\cdot)$ and thus point-wise upper bounds the payoff function, (b) follows by the definition of $\Hat{P}_S(R_S^*)$, and (c) follows by the optimality of $(\Hat{R}_S^*)_{S \in \L_1}$ when maximizing the MCUA of the payoff function. This establishes our claim.

\subsubsection{Proof of Lemma~\ref{lem:parametrization-simple-lem2}} \label{apdx:pf-lem-resource-same}

To prove this claim, consider two solutions $(\Hat{R}_S)_{S \in \L_1}$ and $(\Hat{R}_S^*)_{S \in \L_1}$, where the first corresponds to the solution of step one of Constrained-Greedy and the latter corresponds to the optimal solution to Problem~\eqref{eq:welfare-obj-decomp21-mcua}-\eqref{eq:probability-feasibility21-mcua}.

Now, suppose for contradiction that $\sum_{S \in \L_1} \Hat{P}_S(\Hat{R}_S^*) > \sum_{S \in \L_1} \Hat{P}_S(\Hat{R}_S)$. This implies that there is some set $S'$ such that $\Hat{R}_{S'}^* \neq \Hat{R}_{S'}$. In particular, there must be some set $S'$ such that $\Hat{R}_{S'}^* > \Hat{R}_{S'}$ as the functions $\Hat{P}_S(\cdot)$ are monotonically non-decreasing in their argument for all $S \in \L_1$. Furthermore, given the hierarchical nature of the constraint set, where no two constraints in a given layer intersect with each other, we note that the total number of resources allocated under step one of Constrained-Greedy is equal to that under the solution to Problem~\eqref{eq:welfare-obj-decomp21-mcua}-\eqref{eq:probability-feasibility21-mcua}, i.e., $\sum_{S \in \L_1} \Hat{R}_S^* = \sum_{S \in \L_1} \Hat{R}_{S}$. Consequently, it holds for some set $S'' \in \L_1$ that $\Hat{R}_{S''}^* < \Hat{R}_{S''}$ since for some set $S' \in \L_1$ it holds that that $\Hat{R}_{S'} < \Hat{R}_{S'}^*$.

We will now construct another feasible strategy with a strictly higher objective for Problem~\eqref{eq:welfare-obj-decomp21-mcua}-\eqref{eq:probability-feasibility21-mcua} than that achieved under the allocation $(\Hat{R}_S^*)_{S \in \L_1}$. To do so, we first define an appropriate constant $\epsilon$ as follows. To define this constant, we first define $\epsilon_1 = \min_{|\Hat{R}_S^* - \Hat{R}_S|>0} |\Hat{R}_S^* - \Hat{R}_S|$, i.e., the minimum difference between the two solutions for which the solutions have some discrepancy. Furthermore, we define $\epsilon_{\min} = \min_{|\sum_{S' \subseteq S} \Hat{R}_{S'}^* - \Hat{R}_{S'}|>0} |\sum_{S' \subseteq S} \Hat{R}_{S'}^* - \Hat{R}_{S'}|$. Note here that for a set $S$, the equality $\sum_{S' \subseteq S} \Hat{R}_{S'}^* - \Hat{R}_{S'} = 0$ can only happen in one of two cases: (i) for all sets $S'$, it holds that $\Hat{R}_{S'}^* = \Hat{R}_{S'}$ or (ii) for some sets $S_1$ it holds that $\Hat{R}_{S_1}^* > \Hat{R}_{S_1}$, for some sets $S_2$ it holds that $\Hat{R}_{S_2}^* = \Hat{R}_{S_2}$, and for some sets $S_3$, it holds that $\Hat{R}_{S_3}^* < \Hat{R}_{S_3}$. Finally, we define $\epsilon = \min \{ \epsilon_1, \epsilon_{\min} \}$. Note by definition that $\epsilon > 0$.

Having defined $\epsilon$, to construct a feasible strategy with a strictly higher objective for Problem~\eqref{eq:welfare-obj-decomp21-mcua}-\eqref{eq:probability-feasibility21-mcua} than that achieved under the allocation $(\Hat{R}_S^*)_{S \in \L_1}$, we now select two sets $S', S''$. To define these two sets, first consider all sets $S$ in the layer $\L_2$ and check if there is any set that contains both a set $S'$ such that $\Hat{R}_{S'}^* > \Hat{R}_{S'}$ and a set $S''$ such that $\Hat{R}_{S''}^* < \Hat{R}_{S''}$. If there is not such set in $\L_1$, we proceed to layer $\L_3$ and so on. Note this process certainly terminates at layer $\L_{t+1}$ which corresponds to the overall resource constraint by our assumption that there exists some $S'$ such that $\Hat{R}_{S'}^* > \Hat{R}_{S'}$ and a set $S''$ such that $\Hat{R}_{S''}^* < \Hat{R}_{S''}$. Thus, define $\L_{i'}$ as the first layer at which the above-mentioned phenomena occurs.

Then, define allocation $(\Tilde{R}_S)_{S \in \L_1}$ such that $\Tilde{R}_S = \Hat{R}_S^*$ for all $S \in \L_1 \backslash (S' \cup S'')$, $\Tilde{R}_{S'} = \Hat{R}_{S'}^* - \epsilon$ and $\Tilde{R}_{S''} = \Hat{R}_{S''}^*+\epsilon$.
We first show that this new strategy is feasible:
\begin{align*}
    &\sum_{S \in \L_1} \Tilde{R}_S = \sum_{S \in \L_1 \backslash (S' \cup S'')} \Hat{R}_S^* + \Hat{R}_{S'}^* - \epsilon + \Hat{R}_{S''}^* + \epsilon = \sum_{S \in \L_1} \Hat{R}_S^* \leq R, \\
    &\Tilde{R}_S = \Hat{R}_S^* \leq \Bar{\lambda}_S, \quad \forall S \in \L_1 \backslash (S' \cup S''), \\
    &\Bar{\lambda}_{S'} \geq \Hat{R}_{S'}^* \geq \Hat{R}_{S'}^* - \epsilon = \Tilde{R}_{S'} \implies \Tilde{R}_{S'} \leq \Bar{\lambda}_{S'}, \\
    &\Tilde{R}_{S''} = \Hat{R}_{S''}^*+\epsilon \leq \Hat{R}_{S''} \leq \Bar{\lambda}_{S''} \implies \Tilde{R}_{S''} \leq \Bar{\lambda}_{S''},
\end{align*}
Moreover for all layers $i \geq i'$, it holds that:
\begin{align*}
    &\sum_{S \in S_1} \Tilde{R}_S = \sum_{S \in S_1 } \Hat{R}_S^* \leq \Bar{\lambda}_{S_1}, \forall S_1 \subseteq \L_{i'} \cup \ldots \cup \L_t, \text{ s.t. } S', S'' \notin S_1 \\
    &\sum_{S \in S_1} \Tilde{R}_S = \sum_{S \in S_1 \backslash (S' \cup S'') } \Hat{R}_S^* + \Hat{R}_{S'}^* - \epsilon + \Hat{R}_{S''}^* + \epsilon = \sum_{S \in S_1} \Hat{R}_S^* \leq \Bar{\lambda}_{S_1}, \forall S_1 \subseteq \L_{i'} \cup \ldots \cup \L_t, \text{ s.t. } S', S'' \in S_1,
\end{align*}
as either $S', S'' \in S_1$ or $S', S'' \notin S_1$ for all layers $i \geq i'$.

Next, for $i<i'$, it holds that:
\begin{align*}
    &\sum_{S \in S_1} \Tilde{R}_S = \sum_{S \in S_1 } \Hat{R}_S^* \leq \Bar{\lambda}_{S_1}, \forall S_1 \subseteq \L_{2} \cup \ldots \cup \L_{i'}, \text{ s.t. } S', S'' \notin S_1 \\
    &\sum_{S \in S_1} \Tilde{R}_S = \sum_{S \in S_1 \backslash S' } \Hat{R}_S^* + \Hat{R}_{S'}^* - \epsilon \leq \sum_{S \in S_1} \Hat{R}_S^* \leq \Bar{\lambda}_{S_1}, \forall S_1 \subseteq \L_{2} \cup \ldots \cup \L_{i'}, \text{ s.t. } S' \in S_1, S'' \notin S_1, \\
    &\sum_{S \in S_1} \Tilde{R}_S = \sum_{S \in S_1 \backslash S'' } \Hat{R}_S^* + \Hat{R}_{S''}^* + \epsilon = \sum_{S \in S_1} \Hat{R}_S^* + \epsilon \leq \sum_{S \in S_1} \Hat{R}_S \leq \Bar{\lambda}_{S_1}, \forall S_1 \subseteq \L_{2} \cup \ldots \cup \L_{i'}, \text{ s.t. } S' \notin S_1, S'' \in S_1,
\end{align*}
where the inequality that $\sum_{S \in S_1} \Hat{R}_S^* + \epsilon \leq \sum_{S \in S_1} \Hat{R}_S$ follows from our construction of $i'$ and also the definition of $\epsilon$. Thus, we have established that the allocation $(\Tilde{R}_S)_{S \in \L_1}$ is feasible.

Finally, we show that the objective of Problem~\eqref{eq:welfare-obj-decomp21-mcua}-\eqref{eq:probability-feasibility21-mcua} is greater under the allocation $(\Tilde{R}_S)_{S \in \L_1}$. To see this, first recall by Lemma~\ref{lem:parametrization-simple-lem1} that the objective $\Hat{P}_S(R_S) = \sum_{l \in S} \Hat{P}_l(\Hat{\sigma}_l)$ for any set $S$ and allocation $R_S$, where $(\Hat{\sigma}_l)_{l \in S}$ is given by the solution of a greedy procedure corresponding to the step one of Algorithm~\ref{alg:GreedyWelMaxProb} for the location subset $S$ given $R_S$ resources. Furthermore, let $c_{S'}$ be the maximum slope of the segments corresponding to set $S'$ that have been allocated resources under the optimal allocation but not allocated resources under step one of Constrained-Greedy and let $c_{S''}$ be the minimum slope of the segments corresponding to the set $S''$ that have been allocated resources under step one of Constrained-Greedy but not allocated resources under the optimal allocation. For simplicity of exposition, we assume that the slopes of all segments of the MCUA of the payoff function are distinct (and, in particular, that $c_{S''}>c_{S'}$) but note that the proof can readily be generalized to the setting when this condition is not necessarily met. Then, we have that:
\begin{align*}
    \sum_{S \in \L_1} \Hat{P}_S(\Tilde{R}_S) &= \sum_{S \in \L_1 \backslash (S' \cup S'')} \Hat{P}_S(\Hat{R}_S^*) + \Hat{P}_{S'}(\Hat{R}_{S'}^* -\epsilon) + \Hat{P}_{S''}(\Hat{R}_{S''}^* + \epsilon), \\
    &\stackrel{(a)}{\geq} \sum_{S \in \L_1 \backslash (S' \cup S'')} \Hat{P}_S(\Hat{R}_S^*) + \Hat{P}_{S'}(\Hat{R}_{S'}^*) - \epsilon c_{S'} + \Hat{P}_{S''}(\Hat{R}_{S''}^*) + \epsilon c_{S''}, \\
    &= \sum_{S \in \L_1} \Hat{P}_S(\Hat{R}_S^*) + \epsilon (c_{S''} - c_{S'}), \\
    &\stackrel{(b)}{>} \sum_{S \in \L_1} \Hat{P}_S(\Hat{R}_S^*),
\end{align*}
where (a) follows from the definition of $c_{S'}$ and $c_{S''}$ and (b) follows as $c_{S''}>c_{S'}$ due to the nature of step one of Constrained-Greedy which allocates resources in the descending order of the slopes of the segments. Thus, we have obtained that the objective under the new allocation $(\Tilde{R}_S)_{S \in \L_1}$ is greater than that under $(\Hat{R}_S^*)_{S \in \L_1}$, a contradiction. Hence, we have shown that the step one of Constrained-Greedy indeed computes the optimal allocations corresponding to the solution to Problem~\eqref{eq:welfare-obj-decomp21-mcua}-\eqref{eq:probability-feasibility21-mcua}, which establishes our claim.

\subsubsection{Completing Proof of Theorem~\ref{thm:apx-greedy-resource-augmentation-constraints}} \label{apdx:completing-pf-thm10-additional-constraints}

We establish the approximation ratio guarantee in the statement of Theorem~\ref{thm:apx-greedy-resource-augmentation-constraints} for the allocation $\sigmaa_{R+|\L_2|}^A$ and note that the approximation ratio guarantee for the allocation $\sigmaa_{R+|\L_1|}^A$ follows similarly.

To establish the approximation ratio guarantee for the allocation $\sigmaa_{R+|\L_2|}^A$, we begin by noting the structure of the solution of Constrained-Greedy. In particular, the solution of the first stage of Constrained-Greedy is such that for all sets $S \in \L_1$, one of the following three conditions holds:
\begin{enumerate}
    \item The total allocation to set $S$ is equal to the upper bound $\Bar{\lambda}_S$
    %\item The total allocation to set $S$ is strictly less than the upper bound $\Bar{\lambda}_S$ but for some set $S' \in \L_2 \cup \ldots \cup \L_t$, it holds that $\sum_{S \subseteq S'} \Hat{R}_S^* = \Bar{\lambda}_{S'}$
    \item The condition in case (i) does not hold and for each location $l \in S$, either $\Hat{\sigma}_l = 0$ or $\Hat{\sigma}_l = \frac{d_l^i}{d_l^i+k}$ for some $i \in \I$
    \item The condition in case (i) does not hold and for at most one location $l \in S$ it holds that $\Hat{\sigma}_l \neq 0$ and $\Hat{\sigma}_l \neq \frac{d_l^i}{d_l^i+k}$ for all $i \in \I$. For all other locations in the set $S$, it holds that either $\Hat{\sigma}_l = 0$ or $\Hat{\sigma}_l = \frac{d_l^i}{d_l^i+k}$ for some $i \in \I$.
\end{enumerate}
Note that there can be at most $|\L_2|$ sets $S$ belonging to case three. We now consider each of the above three cases in turn.

\paragraph{Case (i):} Suppose that under stage one of the Constrained-Greedy algorithm the total allocation to a set $S \in \L_1$ is equal to the upper bound $\Bar{\lambda}_S$. In this case, it follows from Lemma~\ref{lem:parametrization-simple-lem1} that $\Hat{P}_S(\Hat{R}_S^*) = \sum_{l \in S} \Hat{P}_l(\Hat{\sigma}_l) = \sum_{l \in S_1} P_l(\Hat{\sigma}_l) + \Hat{P}_{l'}(\Hat{\sigma}_{l'})$, where $S = S_1 \cup \{ l'\}$ and the second equality follows from the fact that the allocation in stage one will coincide with the original payoff function for all but at most one location $l'$ in each set $S \in \L_1$. Next, without relaxing any of the upper bound quotas corresponding to the constraint sets in case (i), we note that the allocation computed in stage two of Constrained-Greedy achieves at least half the objective $\Hat{P}_S(\Hat{R}_S^*)$, i.e., $P_S((\sigma^A_{R+|\L_2|, l})_{l \in S}) \geq \frac{1}{2} \Hat{P}_S(\Hat{R}_S^*)$, which follows similar arguments to our earlier obtained half approximation guarantees (e.g., see proof of Theorem~\ref{thm:greedy-half-approx-rev-max}).

%now show that the vector of allocations $(\sigma^A_{R+|\L_1|, l})_{l \in S}$ achieves a half-approximation

%Note also that under additional resources where the constraints in the above layers may be relaxed, the maximum that can be allocated to this set is $R_S$.

%\paragraph{Case (ii):} Consider the sets $S'$ such that $\sum_{S \subseteq S'} \Hat{R}_S^* = \Bar{\lambda}_{S'}$, where for sets $S_1$ it holds that $\Hat{R}_S^* = R_S$ for all $S \in S_1$ and $\Hat{R}_S^* < R_S$ for all $S \in S' \backslash S_1$ [Take the sets such that we have no intersection across these sets]. Then, we have that $\Hat{W}_{S' \backslash S_1} = \sum_{l \in S_2} v_l + \sum_{l' \in S_3} \sigma_{l'} c_{l'}$, where it must be that $|S_3| \leq |\L_2|$.

\paragraph{Case (ii):} Suppose that under step 1 of the Constrained-Greedy algorithm the total allocation to a set $S \in \L_1$ is strictly less than the upper bound $\Bar{\lambda}_S$ and that case (ii) does not hold. Moreover, suppose for each location $l \in S$ that either $\Hat{\sigma}_l = 0$ or $\Hat{\sigma}_l = \frac{d_l^i}{d_l^i+k}$ for some $i$. In this case, it holds for some location set $S_1 \subseteq S$ that $\Hat{P}_S(\Hat{R}_S^*) = \sum_{l \in S_1} P_l(\Hat{\sigma}_l)$. Note that with $R+|\L_2|$ resources, it holds that the greedy algorithm allocates at least $\Hat{R}_S^*$ to the locations in the set $S$. Consequently, it holds that $P_S((\sigma^A_{R+|\L_1|, l})_{l \in S}) \geq P_S((\Tilde{\sigma}^A_{R+|\L_1|, l})_{l \in S}) \geq P_S((\Tilde{\sigma}^A_{R, l})_{l \in S}) = \sum_{l \in S_1} P_l(\Hat{\sigma}_l) = \Hat{P}_S(\Hat{R}_S^*)$, where $(\Tilde{\sigma}^A_{R, l})_{l \in S}$ ($(\Tilde{\sigma}^A_{R+|\L_2|, l})_{l \in S}$) represents the allocation computed in step one of Algorithm~\ref{alg:GreedyWelMaxProb} in stage two of Constrained-Greedy for the set $S \in \L_1$ given $R$ resources ($R+|\L_2|$ resources and an associated relaxation of the upper bound quotas in the constraint hierarchy).

\paragraph{Case (iii):} Finally, suppose that under step 1 of the Constrained-Greedy algorithm neither of cases (i) or (ii) hold and for at most one location $l \in S$ it holds that $\Hat{\sigma}_l \neq 0$ and $\Hat{\sigma}_l \neq \frac{d_l^i}{d_l^i+k}$ for all $i \in \I$. For all other locations in the set $S$, it holds that either $\Hat{\sigma}_l = 0$ or $\Hat{\sigma}_l = \frac{d_l^i}{d_l^i+k}$ for some $i \in \I$. In this case, it holds that $\Hat{P}_S(\Hat{R}_S^*) = \sum_{l \in H_3^S} P_l(\Hat{\sigma}_l) + \Hat{P}_{l^{''}_S}(\Hat{\sigma}_{l^{''}_S})$, where $S = H_3^S \cup \{ l^{''}_S\}$ for some set $H_3^S$ and location $l^{''}_S$ specific to the set $S$.

Thus, denoting $\L_1^j$ as the set of locations corresponding to case $j$ above for $j \in \{ 1, 2, 3\}$, it follows that:
\begin{align*}
    \sum_{S \in \L_1} \Hat{P}_S(\Hat{R}_S^*) &= \sum_{S \in \L_1^1} \Hat{P}_S(\Hat{R}_S^*) +  \sum_{S \in \L_1^2} \Hat{P}_S(\Hat{R}_S^*) + \sum_{S \in \L_1^3} \Hat{P}_S(\Hat{R}_S^*), \\
    &\stackrel{(a)}{\leq} \sum_{S \in \L_1^1} 2 P_S((\sigma^A_{R+|\L_2|, l})_{l \in S}) + \sum_{S \in \L_1^2} P_S((\sigma^A_{R+|\L_2|, l})_{l \in S}) + \sum_{S \in \L_1^3} \left( \sum_{l \in H_3^S} P_l(\Hat{\sigma}_l) + \Hat{P}_{l''_S}(\Hat{\sigma}_{l''_S}) \right), %\\
    %&\leq 2 \sum_{S \in \L_1^1} W_S(\sigmaa_A^R) + \sum_{S \in \L_1^3} W_S(\sigmaa_A^R) + \sum_{S \in \L_1^2} W_S(\sigmaa_A^R) + \sum_{S \in \L_1^4} W_S(\sigmaa_A^R) + \sum_{l' \in S_3} \sigma_{l'} c_{l'} + \sigma_{l''} c_{l''}, \\
    %&\leq 2 \sum_{S \in \L_1^1} W_S(\sigmaa_A^R) + \sum_{S \in \L_1^3} W_S(\sigmaa_A^R) + \sum_{S \in \L_1^2} W_S(\sigmaa_A^R) + \sum_{S \in \L_1^4} W_S(\sigmaa_A^R) + \sum_{l' \in S_3} \sigma_{l'} \frac{v_{l'}}{\frac{d_{l'}}{d_{l'}+k}} + \sigma_{l''} \frac{v_{l''}}{\frac{d_{l''}}{d_{l''}+k}},
\end{align*}
where (a) follows from our analysis of the three cases above. Next, noting that if we relax the upper bound constraints in the set $l^{''}_S$ for all $S \in \L_1^3$ by one unit, where we recall that $|\L_1^3| \leq |\L_2|$, following arguments similar to our earlier established resource augmentation guarantees (e.g., see proof of Theorem~\ref{thm:greedy-resource-augmentation-rev-max}), the resulting allocation computed in stage two of Constrained-Greedy given $R+|\L_2|$ resources with the relaxation of the upper bound quotas will yield the following inequality:
\begin{align*}
    \sum_{S \in \L_1} \Hat{P}_S(\Hat{R}_S^*) &\leq 2 \sum_{S \in \L_1^1} P_S((\sigma^A_{R+|\L_2|, l})_{l \in S}) + \sum_{S \in \L_1^2} P_S((\sigma^A_{R+|\L_2|, l})_{l \in S}) + \sum_{S \in \L_1^3} \left( \sum_{l \in H_3^S} P_l(\Hat{\sigma}_l) + P_{l^{''}_S}(\Tilde{\sigma}_{l^{''}_S}) \right), \\
    &\leq 2 \sum_{S \in \L_1^1} P_S((\sigma^A_{R+|\L_2|, l})_{l \in S}) + \sum_{S \in \L_1^2} P_S((\sigma^A_{R+|\L_2|, l})_{l \in S}) + \sum_{S \in \L_1^3} P_S((\sigma^A_{R+|\L_2|, l})_{l \in S}), \\
    &\leq 2 \sum_{S \in \L_1} P_S((\sigma^A_{R+|\L_2|, l})_{l \in S})
\end{align*}
where the first inequality follows as $\Tilde{\sigma}_{l^{''}_S} \geq \Hat{\sigma}_{l''_S}$ due to the relaxation of the upper bound quotas and the second inequality follows as the allocation computed via Constrained-Greedy for the sets in $\L_1^3$ achieves at least the objective as $P_l(\Hat{\sigma}_l) + P_{l^{''}_S}(\Tilde{\sigma}_{l^{''}_S})$ for each set $S \in \L_1^3$. 

Finally, applying Lemma~\ref{lem:ub-greedy-reparametrization} establishes our desired half approximation guarantee as:
\begin{align*}
    \sum_{S \in \L_1} P_S(R_S^*) \leq \sum_{S \in \L_1} \Hat{P}_S(\Hat{R}_S^*) \leq 2 \sum_{S \in \L_1} P_S((\sigma^A_{R+|\L_2|, l})_{l \in S}),
\end{align*}
which establishes our claim. We note that we can apply a similar argument to the one above to establish the desired one-approximation guarantee for the allocation $\sigmaa_{R+|\L_1|}^A$ and we omit this for brevity.

\subsection{Extending Constrained-Greedy to Lower Bound Constraints} \label{apdx:lb-con-extension-pf}

We generalize of Constrained-Greedy presented in Algorithm~\ref{alg:ConstrainedGreedy} to the setting with lower bound constraints. In presenting this generalization, we assume that a feasible allocation that satisfies the constraints of Problem~\eqref{eq:welfare-obj}-\eqref{eq:ub-lb-con} exists. Then, in stage one of Constrained-Greedy, after computing the MCUA of the payoff functions and ordering the segments in the descending order of their slopes, we first ensure that we satisfy all the lower bound constraints as follows:
\begin{itemize}
    \item For all sets in $\L_1$, allocate resources greedily in the descending order of the slopes of the MCUA of the payoff functions to the respective sets of locations to exactly satisfy the lower bound constraint for each set $S \in \L_1$.
    \item Given the allocations in the first step to satisfy the lower bound for the layer $\L_1$, then allocate resources to satisfy lower bound constraint in layer $\L_2$ without violating any upper bound constraints in the process of doing so.
    \item Continue this process in the increasing depth of the layers from $\L_3$ to $\L_t$ until all lower bound constrains are satisfied.
\end{itemize}
Finally, once all lower bound constraints across all layers are satisfied, we proceed by applying stage one of Constrained-Greedy as presented in Algorithm~\ref{alg:ConstrainedGreedy}, given the already allocated resources to satisfy the lower bound constraints. We output the corresponding allocation $\Hat{\sigmaa}$ as the final allocation.

Thus, incorporating lower bound constraints involves two modifications to Algorithm~\ref{alg:ConstrainedGreedy}, which applies to upper bound constraints. First, we require an initial processing step to ensure the allocation computed in the first stage of Constrained-Greedy meets the lower bound constraints across layers. Next, the allocation computed using stage one is the final allocation, unlike the Constrained-Greedy algorithm with upper bound constraints (that has a second stage), as applying the sub-routine corresponding to Algorithm~\ref{alg:GreedyWelMaxProb} is only guaranteed to satisfy an upper bound resource constraint and thus may violate the lower bound constraints.

%Once we have computed the aggregate resource budgets in stage one for all sets $S \in \L_1$ corresponding to the allocation computed via the above-described process, we apply stage two of Algorithm~\ref{alg:ConstrainedGreedy} to compute the final allocations for all locations. Thus, in incorporating lower bound constraints, the main modification to Algorithm~\ref{alg:ConstrainedGreedy} involves an initial processing step to make sure the allocation computed in the first stage meets the lower bound constraints across layers. %It is straightforward to see that the final allocation computed in stage two 

\subsection{Key Ideas in Proof of Proposition~\ref{prop:apx-greedy-lb-quotas}} \label{apdx:pf-prop-lb-quotas}

The proof of Proposition~\ref{prop:apx-greedy-lb-quotas} follows similar arguments to that in the proof of Theorem~\ref{thm:apx-greedy-resource-augmentation-constraints}. In the following, we note some caveats required to prove Proposition~\ref{prop:apx-greedy-lb-quotas}.

To this end, first note that the corresponding analogues of Lemmas~\ref{lem:opt-sol-reparametrization},~\ref{lem:ub-greedy-reparametrization}, and~\ref{lem:parametrization-simple-lem1} naturally hold in the setting with lower bound constraints. Moreover, the corresponding analogue of Lemma~\ref{lem:parametrization-simple-lem2} also holds in the setting with lower bound constraints, with the additional caveat that we also need to show that the new solution $(\Tilde{R}_S)_{S \in \L_1}$ we construct also satisfies the lower bound constraints, which follows similar arguments to the satisfaction of the upper bound constraints in the proof of Lemma~\ref{lem:parametrization-simple-lem2}. We note that establishing the analogues of Lemmas~\ref{lem:opt-sol-reparametrization}-\ref{lem:parametrization-simple-lem2} applies for any class of lower bound constraints and does not rely on the assumption that the lower bound quotas for all sets $S \in \L_1$ are two.

Finally, we leverage the fact that the lower bound quotas are such that $\underline{\lambda}_S \geq 2$ for all $S \in \L_1$. To this end, first note that the optimal objective for each location set $S \in \L_1$ satisfies: $\Hat{P}_S(\Hat{R}_S^*) = \sum_{l \in S} \Hat{P}_l(\Hat{\sigma}_l) = \sum_{l \in H_S} P_l(\Hat{\sigma}_l) + \Hat{P}_{l'_S}(\Hat{\sigma}_{l'_S})$, where $S = H_S \cup \{ l'_S\}$ and the second equality follows from the fact that the allocation in stage one will coincide with the original payoff function for all but at most one location $l'_S$ in each set $S \in \L_1$.

Next, fix any set $S \in \L_1$. Note that the solution corresponding to maximizing the MCUA of the payoff function achieves an outcome with a payoff given by $P_S(\Hat{R}_S^*) \geq \sum_{l \in H_S} P_l(\Hat{\sigma}_l)$. 

Next, let $P_S^{*,1}$ correspond to the maximum payoff corresponding to spending all resources to a single location in a set $S \in \L_1$ and let $P_S^{*,R=1}$ represent optimal payoff corresponding to spending one unit of resources to a set $S \in \L_1$, where note that $P_S^{*,R=1} \geq P_S^{*,1} \geq \Hat{P}_{l'_S}(\Hat{\sigma}_{l'_S})$. Finally, since we have at least two resources being allocated to each $S \in \L_1$ since $\underline{\lambda}_S = 2$ for all $S \in \L_1$, it follows that $\sum_{l \in H_S} P_l(\Hat{\sigma}_l) \geq P_S^{*,R=1}$ by our earlier analysis for our resource augmentation results (e.g., see proof of Theorem~\ref{thm:greedy-half-approx-rev-max}).

Thus, from the above two derived relations we have that $P_S(\Hat{R}_S^*) \geq \sum_{l \in H_S} P_l(\Hat{\sigma}_l)$ and that $P_S(\Hat{R}_S^*) \geq P_S^{*,R=1} \geq \Hat{P}_{l'_S}(\Hat{\sigma}_{l'_S})$, which implies that $2 P_S(\Hat{R}_S^*) \geq \sum_{l \in H_S} P_l(\Hat{\sigma}_l) + \Hat{P}_{l'_S}(\Hat{\sigma}_{l'_S}) = \Hat{P}_S(\Hat{R}_S^*) \geq P_S(R_S^*)$ for each set $S \in \L_1$. Consequently, summing the above relation for all sets $S \in \L_1$, we obtain our desired approximation ratio guarantee that $2 \sum_{S \in \L_1} P_S(\Hat{R}_S^*) \geq \sum_{S \in \L_1} P_S(R_S^*)$, which establishes our claim.

\section{Payoff Maximization with Heterogeneous Users} \label{apdx:welfare-max-prob-setting}

In this section, we extend our results obtained in the setting with heterogeneous users under the revenue maximization objective to the payoff maximization objective. To this end, in the following, we present the associated greedy algorithm for the administrator's payoff maximization objective and its corresponding approximation ratio and resource augmentation guarantees. % (Section~\ref{apdx:greedy-exp-wel-max}).

In studying the payoff maximization problem of the administrator with heterogeneous user types, we first note that the problem of computing the administrator's payoff-maximizing strategy in the setting with heterogeneous users is NP-hard, which follows by our hardness result on solving for the administrator's payoff-maximizing strategy in the setting with homogeneous users (Theorem~\ref{thm:npHardness-swm-fm}). Thus, in this section, we introduce a greedy algorithm, described in Algorithm~\ref{alg:GreedyWelMaxProb}, to approximately solve for the administrator's payoff maximizing resource allocation strategy and highlight its corresponding approximation ratio and resource augmentation guarantees.

We begin by first describing our algorithmic approach to achieve an approximate solution to the administrator's payoff maximization problem, which proceeds as follows. First, rather than directly optimizing the payoff of the administrator, which as noted earlier is NP-hard to optimize, we maximize its corresponding monotone concave upper approximation (MCUA), which we depict in Figure~\ref{fig:helper-exp_welfare_mcua}. In particular, similar of maximizing the MCUA of the revenue function in the presence of heterogeneous user types (see Section~\ref{sec:near-opt-greedy-rev-max}), optimizing the MCUA of the payoff function can be reduced to solving a linear program whose solution boils down to a greedy-like procedure presented in Step 1 of Algorithm~\ref{alg:GreedyWelMaxProb} (see Theorem~\ref{thm:greedy-half-approx-wel-max-exp}). 

To elucidate this greedy-like procedure, we first define an affordability threshold $t_l = \min \left\{ R, \max_i \frac{d_l^i}{d_l^i+k}\right\}$ for each location $l$. Next, we define the MCUA of the payoff function over the range $[0, t_l]$ for each location $l$, which we note is a piece-wise linear function; thus, we define $\S$ as the set of all such piece-wise linear segments of the MCUA of the payoff function across all locations. We then characterize each segment $s \in \S$ by three parameters: (i) $l_s$, which represents the location corresponding to segment $s$, (ii) $c_s$, which corresponds to the slope of segment $s$, and (iii) $x_s$, which represents the horizontal width, i.e., resource requirement, of segment $s$. In particular, Figure~\ref{fig:helper-exp_welfare_mcua} (right) depicts the associated slopes $c_s$ and resource requirements $x_s$ corresponding to the segments for the MCUA of the payoff function for a given location $l$. We then order the segments in the set $\S$ in the descending order of the slopes of the MCUA of the payoff function and find the solution corresponding to a greedy algorithm that allocates at most $z_s$ to each segment in the descending order of the slopes of the MCUA of the payoff function.

Next, as with Algorithm~\ref{alg:GreedyRevMaxProb}, we compute an allocation corresponding to spending all the available resources at a single location that yields the highest payoff to the administrator. Finally, we return the best of the greedy allocation corresponding to optimizing the MCUA of the payoff function and the allocation corresponding to spending all the available resources at a single location that achieves a higher payoff for the administrator. This procedure is formally presented in Algorithm~\ref{alg:GreedyWelMaxProb}.

\begin{algorithm}
\footnotesize
\SetAlgoLined
\SetKwInOut{Input}{Input}\SetKwInOut{Output}{Output}
\Input{Total Resource capacity $R$, User Types $\Theta_l^i = (\Lambda_l^i, d_l^i, p_l^i)$ for all locations $l$ and types $i$}
\Output{Resource Allocation Strategy $\sigmaa_A^*$}
\textbf{Step 1: Greedy Allocation Based on Slopes of MCUA of Payoff Function:} \\
Define affordability threshold $t_l \leftarrow \min \{ R, \frac{d_l}{d_l+k}\}$ for all locations $l$ \;
Generate MCUA of the payoff function in range $[0, t_l]$ for each location $l$ \;
$\Tilde{\S} \leftarrow $ Ordered list of segments $s$ across all locations of this MCUA in descending order of slopes $c_s$  \;
Initialize allocation strategy $\Tilde{\sigmaa} \leftarrow \mathbf{0}$ \;
\For{\text{segment $s \in \Tilde{\S}$}}{
$\Tilde{\sigma}_{l_s} \leftarrow \Tilde{\sigma}_{l_s} + \min\{ R, x_s\}$ ; \texttt{\footnotesize \sf Allocate $x_{s}$ to location $l_s$} \;
$R \leftarrow R -  \min\{R, x_s\}$; \quad \texttt{\footnotesize \sf Update amount of remaining resources} \;
  }
\textbf{Step 2: Find Solution $\sigmaa'$ that maximizes payoff from spending on single location:} \\
$\sigmaa^l \leftarrow \argmax_{\sigmaa \in \Omega_R: \sigma_{l'} = 0, \forall l' \neq l} P_R(\sigmaa)$ for all locations $l$ \; 
$\sigmaa' \leftarrow \argmax_{l \in L} P_R(\sigmaa^l)$ \; 
\textbf{Step 3: Return Solution with a Higher Payoff:} \\
$\sigmaa^*_A \leftarrow \argmax \{ P_R(\Tilde{\sigmaa}), P_R(\sigmaa') \}$ \;
\caption{\footnotesize Greedy Algorithm for Administrator's Heterogeneous Payoff Maximization Objective}
\label{alg:GreedyWelMaxProb}
\end{algorithm}

\begin{figure}[tbh!]
    \centering
    \includegraphics[width=0.95\linewidth]{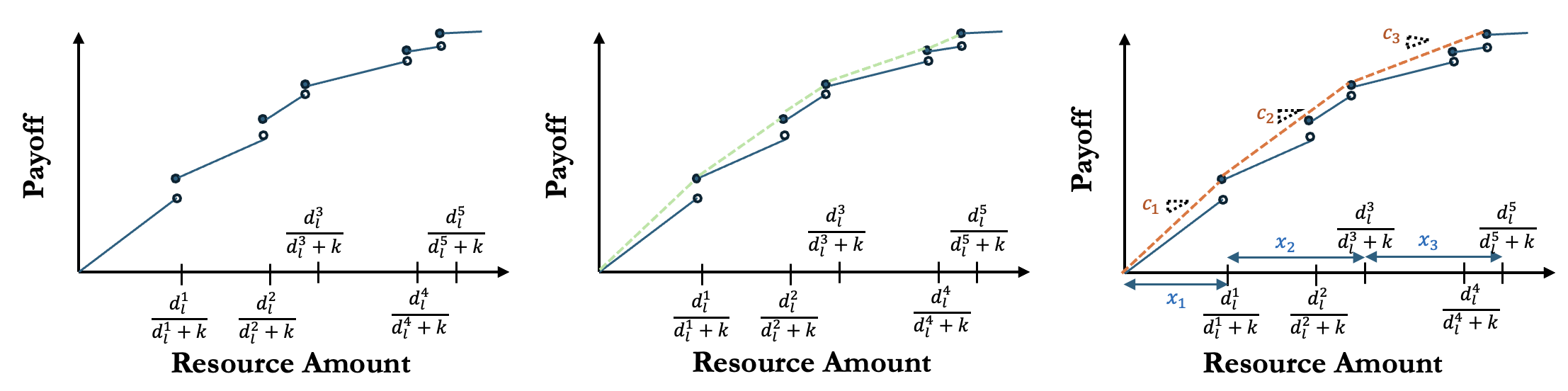}
    \vspace{-10pt}
    \caption{{\small \sf Depiction of the payoff as a function of the amount of resources allocated to location $l$ (left), its upper bound (center), and its corresponding monotone concave upper approximation (right) for a setting with five types, i.e., $|\I|  = 5$. Note that the payoff remains constant after a resource allocation that exceeds $\max_i \frac{d_l^i}{d_l^i+k}$. Further, for the MCUA of the payoff function, there are three segments $s$ for location $l$, with corresponding slopes $c_s$ of the segments and the corresponding width $x_s$ of the segments marked in the plot on the right. 
    }}
    \label{fig:helper-exp_welfare_mcua}
\end{figure} %\vspace{-10pt}

We now present the main results of this section, which establish the approximation guarantees of Algorithm~\ref{alg:GreedyWelMaxProb} to the optimal payoff of the administrator. Our first result establishes that Algorithm~\ref{alg:GreedyWelMaxProb} achieves at least half the payoff as that corresponding to the solution of Problem~\eqref{eq:admin-obj-fraud}-\eqref{eq:bi-level-con-fraud}.

\begin{theorem} [1/2 Approximation of Greedy Algorithm for Heterogeneous Payoff Maximization] \label{thm:greedy-half-approx-wel-max-exp}
Denote $\sigmaa^*_A$ as the solution corresponding to Algorithm~\ref{alg:GreedyWelMaxProb} and let $\sigmaa^*$ be the payoff maximizing allocation that solves Problem~\eqref{eq:admin-obj-fraud}-\eqref{eq:bi-level-con-fraud}. Then, $\sigmaa^*_A$ achieves at least half the payoff as compared to $\sigmaa^*$, i.e., $P_R(\sigmaa_A^*) \geq \frac{1}{2} P_R(\sigmaa^*)$.
\end{theorem}

Our next result establishes that if the administrator had just one additional resource, i.e., $R+1$ resources, then Algorithm~\ref{alg:GreedyWelMaxProb} will achieve at least the payoff as the payoff-maximizing solution of the NP-hard bi-level Program~\eqref{eq:admin-obj-fraud}-\eqref{eq:bi-level-con-fraud} with $R$ resources.

\begin{theorem} [Resource Augmentation Guarantee for Payoff Maximization] \label{thm:greedy-resource-augmentation-wel-max-exp}
Denote $\sigmaa^*_A$ as the solution corresponding to Algorithm~\ref{alg:GreedyWelMaxProb} with $R+1$ resources and let $\sigmaa^*$ be the payoff-maximizing allocation that solves Problem~\eqref{eq:admin-obj-fraud}-\eqref{eq:bi-level-con-fraud} with $R$ resources. Then, the payoff under the allocation $\sigmaa^*_A$ with $R+1$ resources is at least that corresponding to $\sigmaa^*$ with $R$ resources, i.e., $P_{R+1}(\sigmaa^*_A) \geq P_{R}(\sigmaa^*)$.
\end{theorem}

We note that the administrator only requires $R+\max_{l \in L, i\in \I} \frac{d_l^i}{d_l^i+k}$ resources to obtain the resource augmentation guarantee in Theorem~\ref{thm:greedy-half-approx-fraud-min}; however, we present the result with $R+1$ resources for ease of exposition. Theorems~\ref{thm:greedy-half-approx-wel-max-exp} and~\ref{thm:greedy-resource-augmentation-wel-max-exp} imply that Algorithm~\ref{alg:GreedyWelMaxProb} obtains the same approximation ratio and resource augmentation guarantees for the administrator's payoff maximization problem in the setting with heterogeneous users as Algorithm~\ref{alg:GreedyFraudminimizationDeterministic} achieved in the setting with homogeneous users (see Section~\ref{subsec:greedy-fraud-min}). Furthermore, we note that Algorithm~\ref{alg:GreedyWelMaxProb} in the payoff maximization setting follows a similar idea to Algorithm~\ref{alg:GreedyRevMaxProb} in the revenue maximization setting in terms of generating a monotone concave upper approximation to reduce the originally NP-hard bi-level optimization into a tractable linear program that can be solved with a greedy like procedure. Consequently, the proofs of Theorems~\ref{thm:greedy-half-approx-wel-max-exp} and~\ref{thm:greedy-resource-augmentation-wel-max-exp} follow similarly to the corresponding results in the revenue maximization setting with heterogeneous user types; hence, we omit these proofs for brevity.

Yet, we note that while Algorithm~\ref{alg:GreedyWelMaxProb} in the payoff maximization setting is akin to Algorithm~\ref{alg:GreedyRevMaxProb} in the revenue maximization setting, there is one key point of difference. In particular, unlike in Algorithm~\ref{alg:GreedyRevMaxProb}, we do not terminate the greedy procedure at the point in the algorithm when there is a segment in the ordered set $\Tilde{\S}$ such that the segment’s resource requirement $x_s$ exceeds the available resources. Such a termination of the greedy procedure in step one of Algorithm~\ref{alg:GreedyWelMaxProb} is not necessary in the payoff maximization setting, as, unlike the revenue function, the payoff function is monotonically (non)-decreasing in the amount of resources allocated to each location (see Figures~\ref{fig:concave-upper-approximation} and~\ref{fig:helper-exp_welfare_mcua}). Yet, we note from our analysis in the revenue maximization setting that the obtained approximation ratio and resource augmentation guarantees obtained in Theorems~\ref{thm:greedy-half-approx-wel-max-exp} and~\ref{thm:greedy-resource-augmentation-wel-max-exp} would continue to hold even in the setting when the greedy procedure in step 1 of Algorithm~\ref{alg:GreedyWelMaxProb} is terminated in a similar manner to Algorithm~\ref{alg:GreedyRevMaxProb}.

%the allocations in step 1 are truncated such that there is no such $\sigma_l \in (0, \frac{d_l^1}{d_l^1+k})$ or $\sigma_l \in (\frac{d_l^{i-1}}{d_l^{i-1}+k}, \frac{d_l^i}{d_l^i+k})$ for any $i>1$, we perform no such truncation in the allocations in step 1 of Algorithm~\ref{alg:GreedyWelMaxProb}. We note that no such truncation is necessary in the welfare maximization, as, unlike the expected revenue function, the expected welfare function is monotonically (non)-decreasing in the amount of resources allocated to each location (see Figures~\ref{fig:concave-upper-approximation} and~\ref{fig:helper-exp_welfare_mcua}). Yet, we note from our analysis in the revenue maximization setting that the obtained approximation ratio and resource augmentation guarantees obtained in Theorems~\ref{thm:greedy-half-approx-wel-max-exp} and~\ref{thm:greedy-resource-augmentation-wel-max-exp} would continue to hold even in the setting when the allocations in step 1 of Algorithm~\ref{alg:GreedyWelMaxProb} are truncated as in Algorithm~\ref{alg:GreedyRevMaxProb}.

Finally, we note that Theorems~\ref{thm:greedy-half-approx-wel-max-exp} and~\ref{thm:greedy-resource-augmentation-wel-max-exp} further highlight the benefit to administrators for recruiting one additional security personnel and applying a simple algorithm, i.e., Algorithm~\ref{alg:GreedyWelMaxProb}, rather than investing significant computational power and effort to compute the NP-hard payoff-maximizing strategy of the administrator.

\section{Model Extensions} \label{apdx:model-extensions-main}

In this section, we present several natural extensions to the model studied in this work.

\subsection{Extension to Continuous Set of User Types} \label{apdx:extension-continuous-distributions}

In this section, we present the additional notation and the regularity condition necessary to extend the results obtained in this work for a discrete set of user types to a continuous set of user types. For brevity, we present the regularity condition for the payoff maximization setting and note that an analogous condition can be developed for the revenue maximization setting as well.

%\paragraph{Notation:} 
To define our security game in the setting with a continuous set of user types at each location, we define the total mass of users at a location $l$ as $\Lambda_l$ where each user belongs to a \emph{type} $\Theta_l^i = (d_l^i, p_l^i)$, where $d_l^i$ corresponds to the benefit received by users of type $i$ who engage in fraud, and $p_l^i$ represents the payoff to the administrator for allocating a security resource to mitigate fraud at location $l$ for users of type $i$. The two dimensions of the users' type vectors at each location $l$ have joint distribution in the population that is known to the administrator.

Let $P_l(\cdot)$ be a (potentially) non-convex and discontinuous payoff function at a given location $l$ and let $t_l = \min\left\{ R, \max_i \frac{d_l^i}{d_l^i+k} \right\}$ be the maximum amount of resources the administrator can feasibly allocate to location $l$. Then, we make the following regularity assumption on the boundedness of the payoff function.

\begin{assumption} [Boundedness of Payoff Function] \label{asmpn:ub-linear-function}
There exists an affine function that upper bounds payoff function, i.e., there exists some bounded $r$ such that $P_l(0) + \sigma_l r \geq P_l(\sigma_l)$ for all $\sigma_l \in [0, t_l]$.
\end{assumption}

Assumption~\ref{asmpn:ub-linear-function} imposes a mild regularity condition on the administrator's payoff function and states that the administrator's payoff function can be upper bounded by an affine function with some bounded slope $r$. Note that Assumption~\ref{asmpn:ub-linear-function}, in particular, holds for discrete distributions with finite support. We note that Assumption~\ref{asmpn:ub-linear-function} is crucial in ensuring the existence of the super-gradients of the MCUA of the payoff function at all points in the domain $[0, t_l]$ and in ensuring that the MCUA of the payoff function coincides with the original payoff function at its endpoints even for continuous distributions, as is elucidated through the following lemma.

\begin{lemma} [MCUA Coincides with Payoff Function at Endpoints] \label{lem:mcua-coincides-welfare-het}
Let $\Hat{P}_l(\sigma_l)$ be the MCUA of the payoff function $P_l(\sigma_l)$ and that Assumption~\ref{asmpn:ub-linear-function} holds. Then, the MCUA of the payoff function coincides with the payoff function at its endpoints, i.e., $\Hat{P}_l(0) = P_l(0)$ and $\Hat{P}_l(t_l) = P_l(t_l)$.
\end{lemma}

We note that Lemma~\ref{lem:mcua-coincides-welfare-het} in the heterogeneous payoff maximization setting with continuously distributed user types is akin to Lemma~\ref{lem:mcua-coincides-revenue-het} in the heterogeneous revenue maximization setting with discrete user types. Moreover, since the proof of Lemma~\ref{lem:mcua-coincides-welfare-het} is almost entirely analogous to that of Lemma~\ref{lem:mcua-coincides-revenue-het}, we omit it for brevity.

Given the regularity condition in Assumption~\ref{asmpn:ub-linear-function}, which guarantees that the MCUA of the payoff function coincides with the payoff function at its endpoints (Lemma~\ref{lem:mcua-coincides-welfare-het}), we note that our results and analysis in the setting with discrete user types at each location can be naturally extended to the continuous user type setting with one caveat. In particular, rather than the MCUA of the payoff function corresponding to finitely many segments, as in the setting with discrete user types at each location, the MCUA of the payoff function in the continuous user type setting corresponds to an arbitrary monotonically increasing concave function, which has affine segments at the points at which the MCUA does not coincide with the payoff function. Then, our earlier developed algorithms that achieve the half approximation and resource augmentation guarantees, e.g., Algorithms~\ref{alg:GreedyRevMaxProb} and~\ref{alg:GreedyWelMaxProb}, can be modified such that resources are allocated to locations in the descending order of the minimum super-gradients of the MCUA of the payoff functions, which generalizes allocating resources to locations in the descending order of the slopes of the MCUA of the payoff functions in the setting with discrete types to the continuous user type setting. With the above modification to our earlier developed greedy-based algorithms, under Assumption~\ref{asmpn:ub-linear-function}, our obtained half approximation and resource augmentation guarantees hold even in the setting with a continuous set of user types with almost entirely analogous proofs to that in the discrete user type setting. For brevity, we omit presenting the extensions of the results obtained in the discrete user types setting to the continuous user types setting.

\subsection{Extension to Probabilistic Setting} \label{apdx:probabilistic-extension}

This section presents the additional notation required to extend our security game to the Bayesian (or probabilistic) setting, where user types are not deterministic but can be drawn from some distribution.

To model the probabilistic setting, we let $\I$ denote the set of user \emph{groups} and assume that a user type $\Theta_l^i$ at each location $l$ is drawn independently from some discrete probability distribution with finite support $(\Theta_l^{i,j})_{j \in \J} = (\Lambda_l^{i,j}, d_l^{i,j}, p_l^{i,j})_{j \in \J}$, where the support, defined by $\J$, satisfies $|\J| \in \mathbb{N}$. We define the probability that a user group $i$ at location $l$ has a type $\Theta_l^{i,j}$ as $q_l^{i,j}$ for all $j \in \J$. In other words, $\mathbb{P}(\Theta_l^i = \Theta_l^{i,j}) = q_l^{i,j}$ for all $j \in \J$, where $q_l^{i,j} \geq 0$ and $\sum_{j \in \J} q_l^{i,j} = 1$ for all locations $l \in L$ and user groups $i \in \I$. 

In this setting, in line with Bayesian Stackelberg games~\cite{pita2008deployed}, we assume that while users at each location know the realization of the type $j \in \J$ at that location, the administrator only knows the distribution of each user group's type. Consequently, denoting the best response of users in the group $i$ of type $j$ at location $l$ as $y_l^{i,j}(\sigmaa)$, given an administrator strategy $\sigmaa$, we formulate the following expected revenue maximization problem (ERMP) of the administrator: \vspace{-2pt}
\begin{maxi!}|s|[2]   
    {\substack{\sigmaa \in \Omega_R \\ y_l^i(\sigmaa) \in [0, 1], \forall l \in L, i \in \I}}                            
    { Q_R(\sigmaa) =  \sum_{j \in \J} \sum_{i \in \I} q_{l}^{i,j} \sum_{l \in L}  \sigma_l y_l^{i,j}(\sigmaa) k \Lambda_l^{i,j},  \label{eq:admin-obj-expected-revenue-prob}}   
    {\label{eq:Eg005-prob}}             
    {}          
    \addConstraint{y_l^{i,j}(\sigmaa)}{\in \argmax_{y \in [0, 1]} U_l^{i,j}(\sigmaa, y) = y[ (1-\sigma_l) d_l^{i,j} - \sigma_l k], \quad \text{for all } l \in L, i \in \I, j \in \J \label{eq:bi-level-con-expected-revenue-prob}}    
\end{maxi!}
where in upper-level problem, the administrator uses a strategy $\sigmaa$ that maximizes its expected revenue to which the users of each group at each location best respond by maximizing their utilities in the lower-level problem based on the realized type. We note that an analogous formulation for the administrator's expected payoff maximization problem can also be derived.

Note that the above ERMP exactly corresponds to the revenue maximization Problem~\eqref{eq:admin-obj-revenue}-\eqref{eq:bi-level-con-revenue} in the deterministic setting if we consider a type space $|\I| \times |\J|$, where the number of users belonging to a type $(i, j)$ at any location $l$ is given by $q_{l}^{i,j} \Lambda_l^{i,j}$. Analogously, the administrator's expected payoff maximization problem exactly corresponds to the payoff maximization Problem~\eqref{eq:admin-obj-fraud}-\eqref{eq:bi-level-con-fraud} in the deterministic setting if we consider a type space $|\I| \times |\J|$, where the administrator payoff corresponding to preventing fraud from users belonging to a type $(i, j)$ at any location $l$ is given by $q_{l}^{i,j} p_l^{i,j}$. Given this transformation from the probabilistic setting to the deterministic setting, any algorithms and corresponding guarantees derived in this work for the deterministic setting extend to solving the above-defined ERMP and the administrator's expected payoff maximization problem.

\subsection{Extension to Setting where Administrator can Optimize over Fines} \label{apdx:variability-in-fines-extension}

In this section, we consider the setting when the fines in the system are not fixed, and the administrator is tasked with the problem of not only computing a resource allocation strategy but also the fine from an interval $k \in [\underline{k}, \Bar{k}]$, for some constants $\underline{k}, \Bar{k}>0$. In both the payoff and revenue maximization settings, we show that the optimal strategy for the administrator is to always set the fine to the maximum allowable level, i.e., $k = \Bar{k}$. For brevity of notation and expositional simplicity, we prove these results in the setting with homogeneous user types and thus drop $i$ in the superscript of the notation of user types.

%administratornot only seeks to compute the optimal resource allocation strategy but also optimizes over the fines that it sets

%finds the optimal allocations of its resources but also determines the optimal fines from a set $k \in [\underline{k}, \Bar{k}]$. 

\subsubsection{Variability of Fines under Payoff Maximization Objective}

We first consider the administrator's payoff maximization objective. In particular, defining $P_l(\sigma_l, k)$ as the payoff function when $\sigma_l$ resources are allocated to location $l$ under a fine $k$, the administrator's payoff maximization Problem~\eqref{eq:admin-obj-fraud}-\eqref{eq:bi-level-con-fraud} can be formulated as the following optimization problem:
\begin{maxi!}|s|[2]   
    {\sigmaa, k \in [\underline{k}, \Bar{k}]}                   
    { P_R(\sigmaa, k) = \sum_{l \in L} P_l(\sigma_l, k),  \label{eq:welfare-obj2-fine-variable}}   
    {\label{eq:Eg001-fine-variable-payoff}}             
    {}          
    \addConstraint{\sum_{l \in L} \sigma_l}{\leq R, \label{eq:resource-constraint2-payoff-fine-variable}}    
    \addConstraint{\sigma_l}{\in [0, 1], \quad \forall l \in L \label{eq:probability-feasibility2-payoff-fine-variable}} 
\end{maxi!}
We now show that setting the fine $k = \Bar{k}$ is an optimal solution to the above problem.

To prove this claim, fix fines $k_1$ and $k_2$ and let $\sigmaa^{k_1}$ and $\sigmaa^{k_2}$ be the corresponding optimal resource allocation strategies that maximize the administrator's payoff under the fines $k_1$ and $k_2$, respectively. Furthermore, without loss of generality, let $k_1 < k_2$. Then, it is straightforward to see that the inequality $P_l(\sigma_l, k_1) \leq P_l(\sigma_l, k_2)$ holds for the payoff functions under the two fines for all $\sigma_l \in [0, 1]$. In particular, since $k_1<k_2$, we note that $P_l(\sigma_l, k_1) = P_l(\sigma_l, k_2)$ for $\sigma_l< \frac{d_{l}}{d_l + k_2}$ and for $\sigma_l \geq \frac{d_{l}}{d_l + k_1}$, but $P_l(\sigma_l, k_1) < P_l(\sigma_l, k_2)$ for $\sigma_l \in [\frac{d_{l}}{d_l + k_2}, \frac{d_{l}}{d_l + k_1})$. Consequently, it holds that $P_l(\sigma_l, k_1) \leq P_l(\sigma_l, k_2)$ for all $\sigma_l \in [0, 1]$.

%To see this, let $(\sigmaa_{k_1}, k_1)$ be the optimal solution to Problem~\eqref{eq:welfare-obj2-fine-variable}-\eqref{eq:probability-feasibility2-payoff-fine-variable}.

Then, it follows that:
\begin{align*}
    P_R(\sigmaa^{k_2}, k_2) \stackrel{(a)}{\geq} P_R(\sigmaa_{k_1}, k_2) = \sum_{l \in L} P_l(\sigma_l^{k_1}, k_2) \stackrel{(b)}{\geq} \sum_{l \in L} P_l(\sigma_l^{k_1}, k_1) = P_R(\sigmaa_{k_1}, k_1),
\end{align*}
where (a) follows by the optimality of $\sigmaa^{k_2}$ under the fine $k_2$ and (b) follows from our above analysis that $P_l(\sigma_l, k_1) \leq P_l(\sigma_l, k_2)$ for all $\sigma_l \in [0, 1]$. Consequently, our analysis establishes that the payoff is monotonically non-decreasing in the fine, which implies that setting the fine to $k = \Bar{k}$ is optimal in terms of maximizing the administrator's payoff. 

We also note that the above analysis even extends to the setting when fines are allowed to vary across locations.

\subsubsection{Variability of Fines Under Revenue Maximization Objective}

This section considers the administrator's revenue maximization objective. In this setting, we define $Q_l(\sigma_l, k)$ as the revenue function when $\sigma_l$ resources are allocated to location $l$ under a fine $k$. Then, the administrator's revenue maximization Problem~\eqref{eq:admin-obj-revenue}-\eqref{eq:bi-level-con-revenue} can be formulated as the following optimization problem:
\begin{maxi!}|s|[2]   
    {\sigmaa, k \in [\underline{k}, \Bar{k}]}                   
    { Q_R(\sigmaa, k) = \sum_{l \in L} Q_l(\sigma_l, k),  \label{eq:welfare-obj2-rev-fine-variability1}}   
    {\label{eq:Eg001-rev-fine-variability1}}             
    {}          
    \addConstraint{\sum_{l \in L} \sigma_l}{\leq R, \label{eq:resource-constraint2-rev-fine-variability1}}    
    \addConstraint{\sigma_l}{\in [0, 1], \quad \forall l \in L \label{eq:probability-feasibility2-rev-fine-variability1}} 
\end{maxi!}
Note from Lemma~\ref{lem:lp-rev-max-deterministic} that the above problem can be reformulated as:
\begin{maxi!}|s|[2]   
    {\sigmaa, k \in [\underline{k}, \Bar{k}]}                   
    { Q_R(\sigmaa, k) = \sum_{l \in L} \sigma_l k \Lambda_l,  \label{eq:welfare-obj2-rev-fine-variability}}   
    {\label{eq:Eg001-rev-fine-variability}}             
    {}          
    \addConstraint{\sum_{l \in L} \sigma_l}{\leq R, \label{eq:resource-constraint2-rev-fine-variability}}    
    \addConstraint{\sigma_l}{\in \left[0, \frac{d_l}{d_l+k} \right], \quad \forall l \in L \label{eq:probability-feasibility2-rev-fine-variability}} 
\end{maxi!}
We now show that setting the fine $k = \Bar{k}$ is an optimal solution to the above problem.

%We now consider two cases: (i) fine is kept the same across locations and (ii) fine can vary across locations.

%\subsubsection{Fine is the Same Across Locations}

To prove this claim, first note by the optimality of Algorithm~\ref{alg:GreedyRevenueMaximization} that the optimal administrator revenue for a fixed fine $k$ is given by:
\begin{align*}
    Q_R^*(k) = \frac{d_1}{d_1+k} k \Lambda_1 + \ldots + \frac{d_{l_k}}{d_{l_k}+k} k \Lambda_{l_k} + \left( R - \sum_{l \in [l_k]} \frac{d_l}{d_l+k} \right) k \Lambda_{l_k+1},
\end{align*}
where the locations are ordered in descending order of $\Lambda_l$ and $\sum_{l \in [l_k]} \frac{d_l}{d_l+k} < R$ but $\sum_{l \in [l_k+1]} \frac{d_l}{d_l+k} > R$ for some $l_k$. Note also that any change in the fines does not influence the ordering of the locations under the optimal allocation.

Now, we compute the derivative of the above objective with the fine to get:
\begin{align*}
    \frac{\partial Q_R^*(k)}{\partial k} &= \sum_{l \in [l_k]} \left( \frac{d_l}{d_l+k} \Lambda_l - \frac{k \Lambda_l d_l}{(d_l+k)^2} \right) + \Lambda_{l_k+1} \left( R - \sum_{l \in [l_k]} \frac{d_l}{d_l+k} \right) + \sum_{l \in [l_k]} \frac{k \Lambda_{l_k+1} d_l}{(d_l+k)^2}, \\
    &= R \Lambda_{l_k+1} + \sum_{l \in [l_k]} \left( \frac{d_l^2 \Lambda_l}{(d_l+k)^2} - \frac{d_l^2 \Lambda_{L_k+1}}{(d_l+k)^2}\right), \\
    &= R \Lambda_{l_k+1} + \sum_{l \in [l_k]} \frac{d_l^2}{(d_l+k)^2} (\Lambda_l - \Lambda_{l_k+1}), \\
    &\geq 0,
\end{align*}
where the final inequality follows as $\Lambda_l \geq \Lambda_{l_k+1}$ for all $l \in [l_k]$, as we ordered locations in the descending order of the $\Lambda_l$ values in Algorithm~\ref{alg:GreedyRevenueMaximization}.

The above analysis implies that the revenue is monotonically non-decreasing in the fine, which implies that setting the fine to $k = \Bar{k}$ is optimal in terms of maximizing the administrator's revenues. We note that a similar analysis can be performed even for the heterogeneous revenue maximization setting by applying the structure of the optimal solution of the administrator's revenue maximization Problem~\eqref{eq:admin-obj-revenue}-\eqref{eq:bi-level-con-revenue} established in Lemma~\ref{lem:reverse-direction-opt-sol-ermp}.

\subsubsection{Variability of Fines Under Revenue Maximization Objective where Fines can vary across Locations}

This section extends the result obtained in the previous section for the administrator's revenue maximization problem to the setting when the fines can vary across locations. In particular, consider fines $(k_l)_{l \in L}$, where for one location $l'$, we let the fines vary such that $k_{l'}^1<k_{l'}^2$ and let the fines of all other locations be fixed. Then, the optimal administrator revenue under the fine vector $\mathbf{k}^1 = (k_l^1)_{l\in L}$ is given by:
\begin{align*}
    Q_R^*(\textbf{k}^1) = \frac{d_1}{d_1+k_1^1} k_1^1 \Lambda_1 + \ldots + \frac{d_{l_{k_l^1}}}{d_{l_{k_l^1}}+k_{l_{k_l^1}}} k_{l_{k_l^1}}^1 \Lambda_{l_{k_l^1}} + \left( R - \sum_{l \in [l_{k_l^1}]} \frac{d_l}{d_l+k_l^1} \right) k_{l_{k_l^1}}^1 \Lambda_{l_{k_l^1}+1},
\end{align*}
for some $l_{k_l^1}$, where we order locations by their $k_l^1 \Lambda_l$ values (rather than just their $\Lambda_l$ values). Let $\sigmaa^{\textbf{k}^1}$ be the revenue-maximizing allocation given the fine vector $\textbf{k}^1$.

Now consider the setting corresponding to the fine vector $\textbf{k}^2 = (k_l^2)_{l\in L}$, which is such that for one location $l'$, $k_{l'}^1<k_{l'}^2$, and fines for all other locations are the same as under the vector $\textbf{k}^1$. In this case, we construct an allocation $\sigmaa'$ with at least the same objective as that under the fines $(k_l^1)_{l\in L}$. In particular, there are three cases to consider: (i) $l'>l_{k_l^1+1}$, (ii) $l' < l_{k_l^1+1}$, and (iii) $l'=l_{k_l^1+1}$.

\textbf{Case (i):} Let $\sigmaa' = \sigmaa^{k^1}$. This implies that $Q_R(\sigmaa^{\textbf{k}^2}, \textbf{k}^2) \stackrel{(a)}{\geq} Q_R(\sigmaa', \textbf{k}^2) \stackrel{(b)}{=} Q_R(\sigmaa^{\textbf{k}^1}, \textbf{k}^1)$, where (a) follows by the optimality of $\sigmaa^{\textbf{k}^2}$ under the fine $\textbf{k}^2$ and (b) follows as $l'>l_{k_l^1+1}$.

\textbf{Case (ii):} Let $\sigma_l' = \sigma_l^{\textbf{k}^1}$ for all $l \neq l'$ and let $\sigma_{l'}' = \frac{d_{l'}}{d_{l'} + k_{l'}^2} < \frac{d_{l'}}{d_{l'} + k_{l'}^1}$. In this case, we have that:
\begin{align*}
    Q_R(\sigmaa^{\textbf{k}^2}, \textbf{k}^2) &\geq Q_R(\sigmaa', \textbf{k}^2), \\
    &= \frac{d_1}{d_1+k_1^2} k_1^2 \Lambda_1 + \ldots + \frac{d_{l'}}{d_{l'}+k_{l'}^2} k_{l'}^2 \Lambda_{l'} + \ldots + \frac{d_{l_{k_l^2}}}{d_{l_{k_l^2}}+k_{l_{k_1^2}}} k_{l_{k_l^2}}^1 \Lambda_{l_{k_l^2}} \\ 
    &+ \left( R - \sum_{l \in [l_{k_l^2}]} \frac{d_l}{d_l+k_l^2} \right) k_{l_{k_l^2}}^1 \Lambda_{l_{k_l^2}+1}, \\
    &= \frac{d_1}{d_1+k_1^1} k_1^1 \Lambda_1 + \ldots + \frac{d_{l'}}{d_{l'}+k_{l'}^2} k_{l'}^2 \Lambda_{l'} + \ldots + \frac{d_{l_{k_l^1}}}{d_{l_{k_l^1}}+k_{l_{k_1^1}}} k_{l_{k_l^1}}^1 \Lambda_{l_{k_l^1}} \\ &+ \left( R - \sum_{l \in [l_{k_l^1}] \backslash l'} \frac{d_l}{d_l+k_l^1} - \frac{d_{l'}}{d_{l'}+k_{l'}^2} \right) k_{l_{k_l^1}}^1 \Lambda_{l_{k_l^1}}, \\
    &= Q_R(\sigmaa^{\textbf{k}^1}, \textbf{k}^1) + \frac{d_{l'}}{d_{l'}+k_{l'}^2} k_{l'}^2 \Lambda_{l'} - \frac{d_{l'}}{d_{l'}+k_{l'}^2}  k_{l_{k_l^1}}^1 \Lambda_{l_{k_l^1}} - \frac{d_{l'}}{d_{l'}+k_{l'}^1} k_{l'}^1 \Lambda_{l'} + \frac{d_{l'}}{d_{l'}+k_{l'}^1} k_{l_{k_l^1}}^1 \Lambda_{l_{k_l^1}}
\end{align*}
Next, we compute the derivative of the latter function $f = \frac{d_{l'}}{d_{l'}+k_{l'}^2} k_{l'}^2 \Lambda_{l'}$ on the right-hand-side of the above term to get:
\begin{align*}
    \frac{\partial f}{\partial k_{l'}^2} &= \frac{d_{l'}}{d_{l'}+k_{l'}^2} \Lambda_{l'} - \frac{k_{l'}^2 \Lambda_{l'} d_{l'}}{(d_{l'} + k_{l'}^2)^2} = \frac{d_{l'}^2 \Lambda_{l'}}{(d_{l'} + k_{l'}^2)^2}  \geq 0,
\end{align*}
which implies that the additional term on the right hand side of the above inequality is non-negative, i.e., $Q_R(\sigmaa^{\textbf{k}^2}, \textbf{k}^2) \geq Q_R(\sigmaa^{\textbf{k}^1}, \textbf{k}^1)$.

\textbf{Case (iii):} We consider two cases: (a) $R - \sum_{l \in [l_{k_l^1}]} \frac{d_l}{d_l+k_l^1} \leq \frac{d_{l'}}{d_{l'}+k_{l'}^2}$ and (b) $R - \sum_{l \in [l_{k_l^1}]} \frac{d_l}{d_l+k_l^1} > \frac{d_{l'}}{d_{l'}+k_{l'}^2}$. In case (a), we let $\sigmaa' = \sigmaa^{\textbf{k}^1}$, i.e., it clearly follows that $Q_R(\sigmaa^{\textbf{k}^2}, \textbf{k}^2) \geq Q_R(\sigmaa', \textbf{k}^2) \geq Q_R(\sigmaa^{\textbf{k}^1}, \textbf{k}^1)$, where the final inequality follows as $k_{l'}^2 > k_{l'}^1$. Thus, consider case (b), where we let $\sigmaa'$ be such that $\sigma_l' = \sigma_l^{\textbf{k}^1}$ for all $l \neq l'$ and let $\sigma_{l'} = \frac{d_{l'}}{d_{l'}+k_{l'}^2}$. In this case, we get that:
\begin{align*}
    Q_R(\sigmaa^{\textbf{k}^2}, \textbf{k}^2) &\geq Q_R(\sigmaa^{\textbf{k}^1}, \textbf{k}^1) + \frac{d_{l'}}{d_{l'}+k_{l'}^2} k_{l'}^2 \Lambda_{l'} - \left( R - \sum_{l \in [l_{k_l^1}]} \frac{d_l}{d_l+k_l^1} \right) k_{l'}^1 \Lambda_{l'}
\end{align*}
Without loss of generality, we consider the case when $R - \sum_{l \in [l_{k_l^1}]} \frac{d_l}{d_l+k_l^1} = \frac{d_{l'}}{d_{l'}+k_{l'}^1}$, as we are only interested in infinitesimal changes, as we would always be in case (a) otherwise for small enough change in the fine. Thus, the above term reduces to:
\begin{align*}
    Q_R(\sigmaa^{\textbf{k}^2}, \textbf{k}^2) &\geq Q_R(\sigmaa^{\textbf{k}^1}, \textbf{k}^1) + \frac{d_{l'}}{d_{l'}+k_{l'}^2} k_{l'}^2 \Lambda_{l'} - \frac{d_{l'}}{d_{l'}+k_{l'}^1} k_{l'}^1 \Lambda_{l'}.
\end{align*}
However, from our earlier analysis, we know that the latter term in the right hand side is non-negative as the derivative of the corresponding function is non-negative in the fine. Thus, we have shown that $Q_R(\sigmaa^{\textbf{k}^2}, \textbf{k}^2) \geq Q_R(\sigmaa^{\textbf{k}^1}, \textbf{k}^1)$, which establishes our claim.

Note that in all the above cases $\sigmaa'$ is feasible by construction as the fine is increased and thus the resource requirement for that location is lesser.

Thus, the above analysis implies that the revenue is monotonically non-decreasing in the fine for all locations. Consequently, the revenue optimal fines correspond to setting $k_l = \Bar{k}$ for all locations $l$.

%\textbf{Same case by case analysis as above can be extended to probabilistic revenue maximization setting - just use structure of optimal solution in Lemma 7 of arxiv paper; note also that the probabilistic welfare maximization with variable can be solved similarly to that in the deterministic setting too; Thus, we have that in all cases, setting the fines to the maximum value is optimal}

\section{Equilibrium Strategies of Administrator and Users in Contract Game} \label{apdx:eq-strategies-admin-users-contract}

%\subsubsection{Equilibria of Contract Game}

In this section, we study the strategies of the administrator and users in our contract game, given a contract parameter $\alpha \in [0, 1]$. To simplify exposition and elucidate the main ideas of this contract game, we focus on the setting with homogeneous users at each location, and, following ideas developed in Section~\ref{sec:probabilistic-setting}, note that our proposed framework and results can be naturally generalized to the setting with heteregeneous users as well.

%An equilibrium of our contract game is specified by a triple $(\alpha^*, \sigmaa(\alpha^*), (y_l(\sigmaa(\alpha^*)))_{l \in L})$, where $(\sigmaa(\alpha^*), (y_l(\sigmaa(\alpha^*)))_{l \in L})$ represent the solutions of the bi-level Program~\eqref{eq:obj-contract-alpha}-\eqref{eq:alpha-contract-con} given parameter $\alpha^*$, which the principal selects to maximizes its payoff $(1-\alpha) W_R(\sigmaa(\alpha))$. 

In studying the equilibrium strategies of the administrator and users, we first note that the problem of computing equilibria of this contract game is NP-hard, as solving the bi-level Program~\eqref{eq:obj-contract-alpha}-\eqref{eq:alpha-contract-con} is NP-hard, in general. The proof of this claim follows almost entirely analogously to the proof of Theorem~\ref{thm:npHardness-swm-fm}. In particular, as in the proof of Theorem~\ref{thm:npHardness-swm-fm}, we reduce from an instance of the partition problem and consider an instance of the contract game that is akin to the instance in the proof of Theorem~\ref{thm:npHardness-swm-fm} and where the number of fraudulent users $\Lambda_l = \delta$ for all locations $l$ is a small constant. In this setting, Objective~\eqref{eq:obj-contract-alpha} is dominated by the welfare term for a sufficiently large $\alpha \in [0, 1]$, e.g., $\alpha = 1$. Thus, on constructing an instance akin to that in the proof of Theorem~\ref{thm:npHardness-swm-fm} and where the number of fraudulent users $\Lambda$ at each location is a small constant, the remainder of the proof follows from an almost entirely analogous line of reasoning to the proof of Theorem~\ref{thm:npHardness-swm-fm}; thus, we omit the remaining proof details for brevity.

%, in general, which follows naturally from the fact that computing the welfare-maximizing strategy of the administrator is NP-hard (see Theorem~\ref{thm:npHardness-swm-fm}). Thus, in the following, we first present a computationally efficient algorithm, akin to the greedy-like algorithms developed in the earlier studied revenue and welfare maximization settings, to compute a strategy for the administrator with strong approximation guarantees to the optimal solution of the bi-level Program~\eqref{eq:obj-contract-alpha}-\eqref{eq:alpha-contract-con}, given any contract parameter $\alpha \in [0, 1]$. Then, we present a natural dense-sampling approach to compute a near-optimal solution to the principal's problem of selecting a contract level that maximizes its payoff $(1-\alpha) W_R(\sigmaa(\alpha))$.

%\paragraph{Greedy Algorithm to Solve Problem~\eqref{eq:obj-contract-alpha}-\eqref{eq:alpha-contract-con}:}

Given the hardness of solving Problem~\eqref{eq:obj-contract-alpha}-\eqref{eq:alpha-contract-con}, we now present a variant of a greedy algorithm to compute an administrator strategy with an approximation guarantee to the optimal solution of Problem~\eqref{eq:obj-contract-alpha}-\eqref{eq:alpha-contract-con} for any given contract parameter $\alpha \in [0, 1]$. To this end, we first note that the best-response strategy $y_l(\sigmaa, \alpha)$ of users at a given location $l$, corresponding to the solution of the lower-level Problem~\eqref{eq:alpha-contract-con}, given a contract parameter $\alpha$ and an administrator strategy $\sigmaa$, is
\begin{align} \label{eq:best-response-users-contract-game}
    y_l(\sigmaa, \alpha) = 
    \begin{cases}
        0, &\text{if } \sigma_l > \frac{d_l}{d_l+k} \text{ or } \left( \sigma_l = \frac{d_l}{d_l+k} \text{ and } (k \Lambda_l + \alpha p_l) \frac{d_l}{d_l+k} \leq \alpha p_l \right), \\ 
        1, &\text{otherwise}.
    \end{cases}
\end{align}
Notice that as with our earlier studied best-response function of users in the revenue and payoff maximization settings, when $\sigma_l = \frac{d_l}{d_l+k}$, any $y_l(\sigmaa, \alpha) \in [0, 1]$ is a best-response for users at location $l$. However, at the threshold $\sigma_l = \frac{d_l}{d_l+k}$, we let $y_l(\sigmaa, \alpha)$ take on the value zero or one depending on whether the administrator's revenue corresponding to location $l$, given by $\alpha p_l$, when $y_l(\sigmaa, \alpha) = 0$ is smaller than the administrator's revenue at location $l$, given by $(\Lambda_l + \alpha p_l) \frac{d_l}{d_l+k}$, when $y_l(\sigmaa, \alpha) = 1$ or not. Moreover, in the case when $\sigma_l = \frac{d_l}{d_l+k}$ and the administrator is indifferent between the outcomes corresponding to $y_l(\sigmaa, \alpha) = 1$ or $y_l(\sigmaa, \alpha) = 0$, i.e., $(k \Lambda_l + \alpha v_l) \frac{d_l}{d_l+k} = \alpha p_l$, we set $y_l(\sigmaa, \alpha) = 0$, as this maximizes the payoff of the principal. Thus, the best-response function of users in Equation~\eqref{eq:best-response-users-contract-game} is in alignment with the notion of strong Stackelberg equilibria~\cite{kjtjft-2009}.

Having presented the best-response function of users in our studied contract game, we now present Algorithm~\ref{alg:GreedyContractDeterministic}, \emph{Contract-Greedy}, which extends our greedy-like algorithmic approach for the welfare maximization setting (see Algorithm~\ref{alg:GreedyFraudminimizationDeterministic}) to compute an administrator strategy in our contract game for any contract $\alpha \in [0, 1]$. We note that Algorithm~\ref{alg:GreedyContractDeterministic} is entirely analogous to Algorithm~\ref{alg:GreedyFraudminimizationDeterministic} in the payoff maximization setting other than in the process of sorting locations, as the administrator in our contract game, maximizes a linear combination of the revenue and payoff objectives. 

\begin{algorithm}
\footnotesize
\SetAlgoLined
\SetKwInOut{Input}{Input}\SetKwInOut{Output}{Output}
\Input{Contract parameter $\alpha$, Resource capacity $R$, User types $\Theta_l = (\Lambda_l, d_l, p_l)$ for all locations $l$}
\textbf{Step 1: Find Greedy Solution $\Tilde{\sigmaa}$:} \\
Define affordability threshold $t_l \leftarrow \min \{ R, \frac{d_l}{d_l+k}\}$ for all locations $l$ \;
Define revenue $z_l$ corresponding to allocating $t_l$ resources to each location $l$, where $z_l = \max \{ \alpha p_l, (k \Lambda_l + \alpha p_l) \frac{d_l}{d_l+k} \}$ if $t_l = \frac{d_l}{d_l+k}$ and $z_l = t_l (k \Lambda_l + \alpha p_l)$ if $t_l < \frac{d_l}{d_l+k}$ \;
%Define $z_l \leftarrow \max \{ \alpha p_l, (k \Lambda_l + \alpha p_l) \frac{d_l}{d_l+k} \}$ for all locations $l$ \;
Order locations in descending order of $\frac{z_l}{t_l}$ \;
  \For{$l = 1, 2, ..., |L|$}{
      $\Tilde{\sigma}_l \leftarrow \min \{ R, \frac{d_l}{d_l+k} \}$ ; \texttt{\footnotesize \sf Allocate the minimum of the remaining resources and $\frac{d_l}{d_l+k}$ to location $l$} \;
      $R \leftarrow R -  \Tilde{\sigma}_l$; \quad \texttt{\footnotesize \sf Update amount of remaining resources} \;
  }
\textbf{Step 2: Find Solution $\sigmaa'$ that Maximizes Revenue from Spending on Single Location:} \\
$\sigmaa^l \leftarrow \argmax_{\sigmaa \in \Omega_R: \sigma_{l'} = 0, \forall l' \neq l} Q_R(\sigmaa) + \alpha P_R(\sigmaa)$ for all locations $l$ \; %\quad \texttt{\footnotesize \sf Compute the amount of resources $\sigmaa^l$ that maximizes the expected revenue from solely spending on a given location $l$} \;
$\sigmaa' \leftarrow \argmax_{l \in L} Q_R(\sigmaa^l) + \alpha P_R(\sigmaa^l)$ \;
\textbf{Step 3: Return Solution with a Higher Administrator Revenue:} \\
$\sigmaa^A_{\alpha} \leftarrow \argmax \{ Q_R(\Tilde{\sigmaa}) + \alpha P_R(\Tilde{\sigmaa}), Q_R(\sigmaa') + \alpha P_R(\sigmaa') \}$ \;
\caption{\footnotesize \emph{Contract-Greedy}}
\label{alg:GreedyContractDeterministic}
\end{algorithm}

We now show that Algorithm~\ref{alg:GreedyContractDeterministic} achieves at least half the total revenue for the administrator as the optimal solution to the bi-level Program~\eqref{eq:obj-contract-alpha}-\eqref{eq:alpha-contract-con} for any contract parameter $\alpha \in [0, 1]$. 

\begin{theorem} [1/2 Approximation of Greedy Algorithm for Contract Game] \label{thm:contract-game-eq}
For any given contract parameter $\alpha$, let $\sigmaa_{\alpha}^A$ be the solution corresponding to Algorithm~\ref{alg:GreedyContractDeterministic} and let $\sigmaa_{\alpha}^*$ be the optimal solution of the bi-level Program~\eqref{eq:obj-contract-alpha}-\eqref{eq:alpha-contract-con}. Then, $\sigmaa_{\alpha}^A$ achieves at least half the revenue for the administrator as the optimal solution $\sigmaa_{\alpha}^*$, i.e., $Q_R(\sigmaa_{\alpha}^A) + \alpha P_R(\sigmaa_{\alpha}^A) \geq \frac{1}{2} \left( Q_R(\sigmaa_{\alpha}^*) + \alpha P_R(\sigmaa_{\alpha}^*) \right)$.
\end{theorem}

The proof of Theorem~\ref{thm:contract-game-eq} follows an almost entirely analogous line of reasoning to that of Theorem~\ref{thm:greedy-half-approx-fraud-min}, and thus we omit it for brevity. Moreover, as in the earlier studied revenue and payoff maximization settings, we can also establish an analogous resource augmentation guarantee for Algorithm~\ref{alg:GreedyContractDeterministic}. Overall, our earlier developed algorithmic approaches and associated theoretical guarantees naturally carry forward to computing the administrator strategies in the contract game, as Objective~\eqref{eq:obj-contract-alpha} is a linear combination of our earlier studied revenue and payoff objectives.

\section{Theoretical Results on Dense Sampling} \label{apdx:dense-sampling-theory}

In this section, we study the near-optimality of our proposed dense sampling procedure presented in Section~\ref{sec:optimal-contract-framework}. In particular, we prove two key results, which establish that applying the above dense-sampling procedure approximately maximizes the principal’s objective across all contract parameters $\alpha \in [0, 1]$, where the loss in the principal's payoff depends on the chosen step-size $s$. Our first result is a restatement of Theorem~\ref{thm:dense-sampling-near-optimal}, which establishes the near-optimality of dense sampling in the setting when, given any contract parameter $\alpha$, the administrator strategy corresponds to the solution of the bi-level Program~\eqref{eq:obj-contract-alpha}-\eqref{eq:alpha-contract-con}.

\begin{theorem}[Near-Optimality of Dense Sampling under Optimal Administrator Strategy] \label{thm:dense-sampling-near-optimal-apdx-2} 
Let $\alpha^* \in [0, 1]$ be the principal's optimal contract and $\alpha_{s}^* \in \A_s$ be the contract computed through dense-sampling. Further, given any $\alpha$, let $\sigmaa(\alpha)$ be the solution of Problem~\eqref{eq:obj-contract-alpha}-\eqref{eq:alpha-contract-con}. Then, for a step-size $s \leq \frac{\epsilon}{\sum_l p_l}$, the loss in the principal's objective through dense sampling is bounded by $\epsilon$, i.e., $(1-\alpha^*)P_R(\sigmaa(\alpha^*)) \leq (1-\alpha_s^*)P_R(\sigmaa(\alpha_s^*)) + \epsilon$.
\end{theorem}  

%\begin{T1}
%Let $\alpha^* \in [0, 1]$ denote the optimal contract and $\alpha_{s}^* \in \A_s$ be the contract computed through dense-sampling. Further, given any contract $\alpha$, let $\sigmaa^*(\alpha)$ be the optimal solution of the bi-level Program~\eqref{eq:obj-contract-alpha}-\eqref{eq:alpha-contract-con}. Then, for a step-size $s \leq \frac{\epsilon}{\sum_l v_l}$, the loss in the principal's payoff through dense sampling is bounded by $\epsilon$, i.e., $(1-\alpha^*)W_R(\sigmaa^*(\alpha^*)) \leq (1-\alpha_s^*)W_R(\sigmaa^*(\alpha_s^*)) + \epsilon$.
%\end{T1}

%\nearOptDense*

\begin{comment}
\begin{theoremEnd}[end, restate]{thm} [Near-Optimality of Dense Sampling under Optimal Administrator Strategy] \label{thm:dense-sampling-near-optimal-apdx-2} 
Let $\alpha^* \in [0, 1]$ denote the optimal contract and $\alpha_{s}^* \in \A_s$ be the contract computed through dense-sampling. Further, given any contract $\alpha$, let $\sigmaa^*(\alpha)$ be the optimal solution of the bi-level Program~\eqref{eq:obj-contract-alpha}-\eqref{eq:alpha-contract-con}. Then, for a step-size $s \leq \frac{\epsilon}{\sum_l v_l}$, the loss in the principal's payoff through dense sampling is bounded by $\epsilon$, i.e., $(1-\alpha^*)W_R(\sigmaa^*(\alpha^*)) \leq (1-\alpha_s^*)W_R(\sigmaa^*(\alpha_s^*)) + \epsilon$.
\end{theoremEnd}    
\end{comment}

Our second result establishes the near-optimality of dense sampling in the setting when, given any contract parameter $\alpha$, the administrator strategy is computed using Algorithm~\ref{alg:GreedyContractDeterministic} under a correlation assumption on the payoffs $p_l$ and payoff bang-per-buck ratios $\frac{p_l(d_l+k)}{d_l}$ at each location $l$.

\begin{theorem} [Near-Optimality of Dense Sampling under Greedy Administrator Strategy] \label{thm:dense-sampling-near-optimal-apdx} 
Let $\alpha^* \in [0, 1]$ denote the optimal contract, $\alpha_{s}^* \in \A_s$ be the contract computed through dense-sampling, and suppose that the payoffs $p_l$ are positively correlated with the payoff bang-per-buck ratios $\frac{p_l(d_l+k)}{d_l}$ at each location, i.e., if $p_1 \leq \ldots \leq p_{|L|}$, then $\frac{p_{1}(d_1+k)}{d_1} \leq \frac{p_2(d_2+k)}{d_2} \leq \ldots \leq \frac{p_{|L|}(d_{|L|}+k)}{d_{|L|}}$. Further, given $R \geq 1$ resources, suppose that, given any contract $\alpha$, the administrator strategy $\sigmaa_{\alpha}^A$ is computed using Algorithm~\ref{alg:GreedyContractDeterministic}. Then, for a step-size $s \leq \frac{\epsilon}{\sum_l p_l}$, the loss in the principal's objective through dense sampling is bounded by $\epsilon$, i.e., $(1-\alpha^*)P_R(\sigmaa_{\alpha^*}^A) \leq (1-\alpha_s^*)P_R(\sigmaa_{\alpha_s^*}^A) + \epsilon$.
\end{theorem}

In the remainder of this section, we first present the key lemmas necessary to prove the above results (Appendix~\ref{apdx:key-lemmas-dense-sampling}), following which we present the proofs of these lemmas in subsequent sections (Appendices~\ref{apdx:lem-pf-optimality-of-monotonicity}-\ref{apdx:pf-lem-monotonicity-greedy-welfare-alpha}).

\subsection{Key Lemmas} \label{apdx:key-lemmas-dense-sampling}

To establish Theorems~\ref{thm:dense-sampling-near-optimal} and~\ref{thm:dense-sampling-near-optimal-apdx}, we follow a two-step procedure. In particular, as a first step, we show that if the administrator strategy $\sigmaa(\alpha)$ for any contract parameter $\alpha$ is such that the payoff function $P_R(\sigmaa(\alpha))$ is monotonically non-decreasing in $\alpha$, then our dense-sampling procedure approximately maximizes the principal’s objective across all contract parameters $\alpha \in [0, 1]$, as is elucidated by the following lemma.

\begin{lemma} [Monotonicity of Payoff Function Implies Near-Optimality] \label{lem:optimality-of-monotonicity}
Let $\sigmaa(\alpha)$ be an administrator strategy, given a contract parameter $\alpha$, such that the payoff function is monotonically non-decreasing in $\alpha$, i.e., if $\alpha_1<\alpha_2$, $P_R(\sigmaa(\alpha_2)) \geq P_R(\sigmaa(\alpha_1))$. Furthermore, denote $\alpha^* \in [0, 1]$ as the optimal contract and let $\alpha_{s}^* \in \A_s$ be the contract computed through dense-sampling. Then, given $R\geq 1$ resources, for a step-size $s \leq \frac{\epsilon}{\sum_l p_l}$, the loss in the principal's objective is bounded by $\epsilon$, i.e., $(1-\alpha^*)P_R(\sigmaa_{\alpha^*}^A) \leq (1-\alpha_s^*)P_R(\sigmaa_{\alpha_s^*}^A) + \epsilon$.
\end{lemma}

For a proof of Lemma~\ref{lem:optimality-of-monotonicity}, see Appendix~\ref{apdx:lem-pf-optimality-of-monotonicity}. Lemma~\ref{lem:optimality-of-monotonicity} establishes that if the algorithm or procedure of computing the administrator's strategy given a contract parameter $\alpha$ satisfies a natural monotonicity property, then dense sampling is near optimal for an appropriately chosen step size of the discretized set $\A_s$. Note that such a monotonicity property is natural as increasing the contract parameter $\alpha$ is synonymous with providing the administrator a higher compensation for its contribution to the payoff of the system.

Our next two results establish that this monotonicity of the payoff function in the parameter $\alpha$ is indeed satisfied for the optimal administrator strategy corresponding to the solution of the bi-level Program~\eqref{eq:obj-contract-alpha}-\eqref{eq:alpha-contract-con} as well as the administrator strategy that is computed using Algorithm~\ref{alg:GreedyContractDeterministic} under a correlation assumption on the payoffs $p_l$ and payoff bang-per-buck ratios $\frac{p_l(d_l+k)}{d_l}$ at each location $l$, as is elucidated by the following two lemmas.

\begin{lemma} [Monotonicity of Optimal Payoff in Contract $\alpha$] \label{lem:monotonicity-optimal-welfare-alpha}
Let $\sigmaa^*(\alpha)$ be the optimal solution of the bi-level Program~\eqref{eq:obj-contract-alpha}-\eqref{eq:alpha-contract-con}, given a contract parameter $\alpha$. Then, the optimal payoff is monotonically non-decreasing in $\alpha$, i.e., for any two contract parameters $\alpha_1, \alpha_2 \in [0, 1]$ with $\alpha_1<\alpha_2$, it holds that $P_R(\sigmaa^*(\alpha_2)) \geq P_R(\sigmaa^*(\alpha_1))$.
\end{lemma}

\begin{lemma} [Monotonicity of Payoff of Contract-Greedy in Contract $\alpha$] \label{lem:monotonicity-greedy-welfare-alpha}
Let $\sigmaa_{\alpha}^A$ be the allocation computed using Algorithm~\ref{alg:GreedyContractDeterministic} given a contract parameter $\alpha$. Moreover, suppose that the payoffs $p_l$ are positively correlated with the payoff bang-per-buck ratios $\frac{p_l(d_l+k)}{d_l}$ at each location, i.e., if $p_1 \leq \ldots \leq p_{|L|}$, then $\frac{p_{1}(d_1+k)}{d_1} \leq \frac{p_2(d_2+k)}{d_2} \leq \ldots \leq \frac{p_{|L|}(d_{|L|}+k)}{d_{|L|}}$. Then, given $R \geq 1$ resources, the payoff of the solutions corresponding to Algorithm~\ref{alg:GreedyContractDeterministic} is monotonically non-decreasing in $\alpha$, i.e., for any two contract parameters $\alpha_1, \alpha_2 \in [0, 1]$ with $\alpha_1<\alpha_2$, it holds that $P_R(\sigmaa^A_{\alpha_2}) \geq P_R(\sigmaa^A_{\alpha_1})$.
\end{lemma}

We note here that there are several settings when the correlation condition in the statement of Lemma~\ref{lem:monotonicity-optimal-welfare-alpha} is satisfied. For instance, this condition is trivially satisfied in settings when the benefits $d_l$ are fixed regardless of the location, i.e., if $d_l = d_{l'}$ for all $l, l' \in L$. Such a setting may be pertinent in several applications of interest where the benefits accrued from engaging in fraud may not vary across locations and represents a fixed constant that is accrued if a user engages in fraud. %In other instances, a natural formulation for the valuation of an administrator at any location is $v_l = \Lambda_l (d_l)^x$ for some exponent $x$. In this case, we note that the correlation condition in the statement of Lemma~\ref{lem:monotonicity-optimal-welfare-alpha} holds if the number of fraudulent users is constant across locations, i.e., $\Lambda_l = \Lambda_{l'}$ for all $l, l' \in L$, and the benefits satisfy $d_1 \leq d_2 \leq \ldots \leq d_{|L|}$ as $\frac{v_l(d_l+k)}{d_l} = \Lambda_l d_l^{x-1} (d_l+k) = v_l \left( 1+ \frac{k}{d_l} \right)$. 
%Further, we note that while, for many instances, the correlation condition may not be exactly met, under reasonable valuation functions of the form $v_l = \Lambda_l (d_l)^x$, there is a strong positive correlation between the valuations and the benefits. 
Furthermore, while, for many instances, the above correlation condition may not be exactly met, in our conducted numerical experiments (e.g., see Figure~\ref{fig:social-welfare-proportion-vary-rk}), we observe that the payoff function corresponding to the administrator strategies computed using Algorithm~\ref{alg:GreedyContractDeterministic} given contract parameters $\alpha$ is monotonically increasing in $\alpha$ (for the discrete set of values $\A_s$ where the step size is $s=0.05$).

For proofs of Lemmas~\ref{lem:monotonicity-optimal-welfare-alpha} and~\ref{lem:monotonicity-greedy-welfare-alpha}, see Appendices~\ref{apdx:pf-lem-monotonicity-optimal-welfare} and~\ref{apdx:pf-lem-monotonicity-greedy-welfare-alpha}, respectively. Note that Lemmas~\ref{lem:optimality-of-monotonicity} and~\ref{lem:monotonicity-optimal-welfare-alpha} jointly establish Theorem~\ref{thm:dense-sampling-near-optimal}, while Lemmas~\ref{lem:optimality-of-monotonicity} and~\ref{lem:monotonicity-greedy-welfare-alpha} jointly establish Theorem~\ref{thm:dense-sampling-near-optimal-apdx}.

\subsection{Proof of Lemma~\ref{lem:optimality-of-monotonicity}} \label{apdx:lem-pf-optimality-of-monotonicity}

Let $\alpha^*_s \in \A_s$ be the optimal solution maximizing the principal's objective for $\alpha$ lying in the discrete set of values in the set $\A_s = \{ 0, s, 2s, \ldots, 1 \}$. Furthermore, let $\sigmaa(\alpha)$ be an administrator strategy, given a contract parameter $\alpha$, such that the payoff function is monotonically non-decreasing in $\alpha$, i.e., $P_R(\sigmaa(\alpha_2)) \geq P_R(\sigmaa(\alpha_1))$ if $\alpha_1<\alpha_2$.

Then, it follows by the optimality of $\alpha^*_s$ in the set $\A_s$ that:
\begin{align} \label{eq:helper-monotonicity-implies-near-optimality-1}
    (1-\alpha^*_s) P_R(\sigmaa(\alpha^*_s)) \geq (1-\alpha) P_R(\sigmaa(\alpha)), \quad \text{for all } \alpha \in \A_s.
\end{align}
Next, define $\alpha^* \in [0, 1]$ as the contract parameter that maximizes the administrator's payoff, i.e., it holds that
\begin{align*}
    (1-\alpha^*)P_R(\sigmaa(\alpha^*)) \geq (1-\alpha) P_R(\sigmaa(\alpha)), \quad \text{for all } \alpha \in [0, 1].
\end{align*}
Next, without loss of generality, we suppose that $\alpha^* \notin \A_s$, as otherwise $\alpha_s^* = \alpha^*$ by the optimality of $\alpha_s^*$ in the discrete set $\A_s$ of $\alpha$ values. Thus, it holds that there are two neighboring contract parameters $\alpha_{s_1}, \alpha_{s_2}$, such that $\alpha_{s_1} < \alpha^* < \alpha_{s_2}$ and $\alpha_{s_2} - \alpha_{s_1} \leq s$. Consequently, it holds that $\alpha^* - \alpha_{s_1} \leq s$ and $\alpha_{s_2} - \alpha^* \leq s$. Moreover, by the monotonicity of the payoff in $\alpha$, we have that:
\begin{align*}
    (1-\alpha^*)P_R(\sigmaa(\alpha^*) &\stackrel{(a)}{\leq} (1-\alpha^*)P_R(\sigmaa(\alpha_{s_2})), \\
    &= (1-\alpha_{s_2})P_R(\sigmaa(\alpha_{s_2})) + (\alpha_{s_2} - \alpha^*)P_R(\sigmaa(\alpha_{s_2})), \\
    &\stackrel{(b)}{\leq} (1-\alpha^*_s) P_R(\sigmaa(\alpha^*_s)) + s \sum_{l \in L} p_l, \\
    &\stackrel{(c)}{\leq} (1-\alpha^*_s) P_R(\sigmaa(\alpha^*_s)) + \epsilon,
\end{align*}
where (a) follows by the monotonocity of the payoff in $\alpha$ and thus the fact that $P_R(\sigmaa(\alpha_{s_2})) \geq P_R(\sigmaa(\alpha^*)$, as $\alpha_{s_2} > \alpha^*$, (b) follows by Equation~\eqref{eq:helper-monotonicity-implies-near-optimality-1}, the fact that $\alpha_{s_2} - \alpha^* \leq s$, and that the maximum achievable payoff can never exceed $\sum_{l \in L} p_l$, and (c) follows as $s \leq \frac{\epsilon}{\sum_l p_l}$. The final inequality above implies our desired sub-optimailty result that the loss in the optimal payoff from our dense sampling approach is bounded by at most $\epsilon$ for a step-size $s \leq \frac{\epsilon}{\sum_l p_l}$.

\subsection{Proof of Lemma~\ref{lem:monotonicity-optimal-welfare-alpha}} \label{apdx:pf-lem-monotonicity-optimal-welfare}

Let $\alpha_1 < \alpha_2$ be two contract parameters and let the optimal solution of the administrator's bi-level Program~\eqref{eq:obj-contract-alpha}-\eqref{eq:alpha-contract-con} given $\alpha_1$ be $\sigmaa^*(\alpha_1)$ and given $\alpha_2$ be $\sigmaa^*(\alpha_2)$. Then, by the optimality of $\sigmaa^*(\alpha_1)$ for Problem~\eqref{eq:obj-contract-alpha}-\eqref{eq:alpha-contract-con}, given contract parameter $\alpha_1$, we have that:
\begin{align} \label{eq:helper-monotonicity-optimal-1}
    Q_R(\sigmaa^*(\alpha_1)) + \alpha_1 P_R(\sigmaa^*(\alpha_1)) \geq Q_R(\sigmaa^*(\alpha_2)) + \alpha_1 P_R(\sigmaa^*(\alpha_2))
\end{align}
and by the optimality of $\sigmaa^*(\alpha_2)$ for Problem~\eqref{eq:obj-contract-alpha}-\eqref{eq:alpha-contract-con}, given contract parameter $\alpha_2$, we have that:
\begin{align} \label{eq:helper-monotonicity-optimal-2}
    Q_R(\sigmaa^*(\alpha_2)) + \alpha_2 P_R(\sigmaa^*(\alpha_2)) \geq Q_R(\sigmaa^*(\alpha_1)) + \alpha_2 P_R(\sigmaa^*(\alpha_1)).
\end{align}
Next, summing Equations~\eqref{eq:helper-monotonicity-optimal-1} and~\eqref{eq:helper-monotonicity-optimal-2}, we obtain:
\begin{align*}
     \alpha_1 P_R(\sigmaa^*(\alpha_1)) + \alpha_2 P_R(\sigmaa^*(\alpha_2)) \geq \alpha_1 P_R(\sigmaa^*(\alpha_2)) + \alpha_2 P_R(\sigmaa^*(\alpha_1)).
\end{align*}
The above inequality implies that:
\begin{align*}
    (\alpha_2 - \alpha_1)(P_R(\sigmaa^*(\alpha_2)) - P_R(\sigmaa^*(\alpha_1))) \geq 0.
\end{align*}
Finally, since $\alpha_2 > \alpha_1$, the above inequality can only hold if $P_R(\sigmaa^*(\alpha_2)) \geq P_R(\sigmaa^*(\alpha_1))$, i.e., the payoff function corresponding to the optimal solution of Problem~\eqref{eq:obj-contract-alpha}-\eqref{eq:alpha-contract-con} is monotone in the contract parameter $\alpha$, which establishes our desired result.

\subsection{Proof of Lemma~\ref{lem:monotonicity-greedy-welfare-alpha}} \label{apdx:pf-lem-monotonicity-greedy-welfare-alpha}

Let $\alpha_1 < \alpha_2$ be two contract parameters, $\Tilde{\sigmaa}^{\alpha_1}, \Tilde{\sigmaa}^{\alpha_2}$ be the the solutions corresponding to step one of Algorithm~\ref{alg:GreedyContractDeterministic}, and $\sigmaa^{\alpha_1'}, \sigmaa^{\alpha_2'}$ be the the solutions corresponding to step two of Algorithm~\ref{alg:GreedyContractDeterministic}. In the following we show that: (i) $P_R(\sigmaa^{\alpha_1'}) \leq P_R(\sigmaa^{\alpha_2'})$ and (ii) $P_R(\Tilde{\sigmaa}^{\alpha_1}) \leq P_R(\Tilde{\sigmaa}^{\alpha_2})$. Note that proving both these claims establishes our desired result as $\sigmaa_{\alpha}^A = \argmax \{ P_R(\Tilde{\sigmaa}^{\alpha}), P_R(\sigmaa^{\alpha'}) \}$, and hence, from the above two claims, it is straightforward to see that $P_R(\sigmaa_{\alpha_1}^A) \leq P_R(\sigmaa_{\alpha_2}^A)$. Thus, in the remainder of this proof, we establish claims (i) and (ii).

\paragraph{Proof of Claim (i):} 

Notice that in step two of Algorithm~\ref{alg:GreedyContractDeterministic}, only one location is allocated resources by the administrator. In particular, at the contract level $\alpha_1$, suppose $l_{\alpha_1}$ is the location to which the administrator allocates resources. In this case, note that either (a) $l_{\alpha_1} = \argmax_{l \in L} \alpha_1 p_l$ or (b)  $l_{\alpha_1} = \argmax_{l \in L} (k \Lambda_l + \alpha_1 p_l) \frac{d_l}{d_l+k}$ as $R \geq 1$, depending on whether $\max_{l \in L} \alpha_1 p_l$ or $\max_{l \in L} (k \Lambda_l + \alpha_1 p_l) \frac{d_l}{d_l+k}$ is greater. Analogously, we can define $l_{\alpha_2}$

\textbf{Case (a):} We first suppose that case (a) holds for the parameter $\alpha_1$. We now show that $l_{\alpha_2} = \argmax_{l \in L} \alpha_2 p_l$, and consequently that $l_{\alpha_1} = l_{\alpha_2}$, which subsequently establishes our claim that $P_R(\sigmaa^{\alpha_1'}) \leq P_R(\sigmaa^{\alpha_2'})$.

%We now show that $l_{\alpha_1} = l_{\alpha_2}$ to establish our claim. 
To see this, first note by the definition of $l_{\alpha_1}$ that $\alpha_1 p_{l_{\alpha_1}} \geq \alpha_1 p_l$ for all $l$, i.e., $p_{l_{\alpha_1}} \geq p_l$ for all $l$. Thus, if $l_{\alpha_2} = \argmax_{l \in L} \alpha_2 p_l$, it follows as $l_{\alpha_2} = l_{\alpha_1}$.

Moreover, it holds by the definition of $l_{\alpha_1}$ that $\alpha_1 p_{l_{\alpha_1}} \geq (k \Lambda_{l_{\alpha_1}} + \alpha_1 p_{l_{\alpha_1}}) \frac{d_{l_{\alpha_1}}}{d_{l_{\alpha_1}}+k}$. Consequently, at $\alpha_2$, it holds for any location $l$ that
\begin{align*}
    (k \Lambda_{l} + \alpha_2 p_{l}) \frac{d_{l}}{d_{l}+k} &= (k \Lambda_{l} + \alpha_1 p_{l}) \frac{d_{l}}{d_{l}+k} + (\alpha_2 - \alpha_1) p_{l} \frac{d_{l}}{d_{l}+k}, \\
    &\stackrel{(a)}{\leq} \alpha_1 p_{l_{\alpha_1}} + (\alpha_2 - \alpha_1) p_{l}, \\
    &\stackrel{(b)}{\leq} \alpha_2 p_{l_{\alpha_1}}
\end{align*}
where (a) follows by the definition of $l_{\alpha_1}$ and the fact that $\frac{d_l}{d_l+k} \leq 1$ and (b) follows from the fact that $p_{l_{\alpha_1}} \geq p_l$ for all locations $l$, as noted above. Thus, we have that $l_{\alpha_2} = \argmax_{l \in L} \alpha_2 p_l$, as, in particular, $\alpha_2 p_{l_{\alpha_1}} \geq (k \Lambda_{l} + \alpha_2 p_{l}) \frac{d_{l}}{d_{l}+k}$. Finally, since $p_{l_{\alpha_1}} \geq p_l$ for all locations $l$, it follows that $l_{\alpha_2} = l_{\alpha_1}$ and hence that the payoff $P_R(\sigmaa^{\alpha_1'}) = P_R(\sigmaa^{\alpha_2'})$, which establishes our desired inequality on the payoff functions for case (a).

\textbf{Case (b):} Next, suppose case (b) holds for the parameter $\alpha_1$. In this case, the payoff accrued by the strategy $\sigmaa^{\alpha_1'}$ is given by $P_R(\sigmaa^{\alpha_1'}) = p_{l_{\alpha_1}} \frac{d_{l_{\alpha_1}}}{d_{l_{\alpha_1}}+k}$. We now show that the corresponding payoff accrued by $\Tilde{\sigmaa}^{\alpha_2}$ is (weakly) higher. To see this, we consider two cases for $l_{\alpha_2}$: (i) $l_{\alpha_2} = \argmax_{l \in L} \alpha_2 p_l$ and (ii) $l_{\alpha_2} = \argmax_{l \in L} (k \Lambda_l + \alpha_2 p_l) \frac{d_l}{d_l+k}$. 

In the first case, it holds that $\alpha_2 p_{l_{\alpha_2}} \geq \alpha_2 p_{l_{\alpha_1}}$ by the definition of $l_{\alpha_2}$, i.e., it follows that $p_{l_{\alpha_2}} \geq p_{l_{\alpha_1}}$. Thus, the payoff accrued under the parameter $\alpha_2$ is given by 
\begin{align*}
    P_R(\sigmaa^{\alpha_2'}) = p_{l_{\alpha_2}} \geq p_{l_{\alpha_1}} \geq p_{l_{\alpha_1}} \frac{d_{l_{\alpha_1}}}{d_{l_{\alpha_1}}+k} = P_R(\sigmaa^{\alpha_1'}),
\end{align*}
which implies that $P_R(\sigmaa^{\alpha_2'}) \geq P_R(\sigmaa^{\alpha_1'})$ in the first case.

In the second case, it holds that $(k \Lambda_{l_{\alpha_2}} + \alpha_2 p_{l_{\alpha_2}}) \frac{d_{l_{\alpha_2}}}{d_{l_{\alpha_2}}+k} \geq (k \Lambda_{l_{\alpha_1}} + \alpha_2 p_{l_{\alpha_1}}) \frac{d_{l_{\alpha_1}}}{d_{l_{\alpha_1}}+k}$ by the definition of $l_{\alpha_2}$. Moreover, by the definition of $l_{\alpha_1}$ it holds that $(k \Lambda_{l_{\alpha_2}} + \alpha_1 p_{l_{\alpha_2}}) \frac{d_{l_{\alpha_2}}}{d_{l_{\alpha_2}}+k} \leq (k \Lambda_{l_{\alpha_1}} + \alpha_1 p_{l_{\alpha_1}}) \frac{d_{l_{\alpha_1}}}{d_{l_{\alpha_1}}+k}$. Adding these two inequalities, we get:
\begin{align*}
    \alpha_2 p_{l_{\alpha_2}} \frac{d_{l_{\alpha_2}}}{d_{l_{\alpha_2}}+k} + \alpha_1 p_{l_{\alpha_1}} \frac{d_{l_{\alpha_1}}}{d_{l_{\alpha_1}}+k} \geq \alpha_1 p_{l_{\alpha_2}} \frac{d_{l_{\alpha_2}}}{d_{l_{\alpha_2}}+k} + \alpha_2 p_{l_{\alpha_1}} \frac{d_{l_{\alpha_1}}}{d_{l_{\alpha_1}}+k}.
\end{align*}
By rearranging the above inequality, we obtain that:
\begin{align*}
    (\alpha_2 - \alpha_1) p_{l_{\alpha_2}} \frac{d_{l_{\alpha_2}}}{d_{l_{\alpha_2}}+k} \geq (\alpha_2 - \alpha_1) p_{l_{\alpha_1}} \frac{d_{l_{\alpha_1}}}{d_{l_{\alpha_1}}+k}.
\end{align*}
Since, $\alpha_2 > \alpha_1$, we cancel this expression from the above equation and obtain in the second case that the payoff
\begin{align*}
    P_R(\sigmaa^{\alpha_2'}) = p_{l_{\alpha_2}} \frac{d_{l_{\alpha_2}}}{d_{l_{\alpha_2}}+k} \geq p_{l_{\alpha_1}} \frac{d_{l_{\alpha_1}}}{d_{l_{\alpha_1}}+k} = P_R(\sigmaa^{\alpha_1'}),
\end{align*}
which implies that $P_R(\sigmaa^{\alpha_2'}) \geq P_R(\sigmaa^{\alpha_1'})$ in the second case.

\paragraph{Proof of Claim (ii):}

%Our proof of this claim follows from an analogous line of reasoning used to prove claim (i). 

We first introduce some notation. In particular, in step 1 of Algorithm~\ref{alg:GreedyContractDeterministic}, we let $S_1$ be the set of locations to which resources are allocated given a parameter $\alpha_1$ and let $S_2$ be the set of locations to which resources are allocated given a parameter $\alpha_2$. Furthermore, we define $S_1^1$ as the set of locations among $S_1$ such that $\frac{\alpha_1 p_l (d_l+k)}{d_l} \geq k \Lambda_l + \alpha_1 p_l$ and $S_1^2$ as the set of locations such that $\frac{\alpha_1 p_l (d_l+k)}{d_l} < k \Lambda_l + \alpha_1 p_l$. Analogously, we can define the sets $S_2^1$ and $S_2^2$. Then, the total payoff given the parameter $\alpha_1$ corresponding to the allocation computed using step 1 of Algorithm~\ref{alg:GreedyContractDeterministic} is given by $P_R(\Tilde{\sigmaa}^{\alpha_1}) = \sum_{l \in S_1^2 \backslash \{ l'\}} p_l \frac{d_l}{d_l+k} + \sum_{l \in S_1^1 } p_l + \Tilde{\sigma}_{l'}^{\alpha_1} p_{l'}$ for at most one location $l'$, where $\sum_{l \in L} \Tilde{\sigma}_l^{\alpha_1} \leq R$. Here, we assume without loss of generality that the location $l'$ with an allocation $\sigma_{l'}^{\alpha_1} \in (0, \frac{d_{l'}}{d_{l'}+k})$ belongs to the set $S_1^2$, but note that the following arguments can be readily generalized to cover the other case when $l' \in S_1^1$ as well.

%Analogously, we have that $W_R(\Tilde{\sigmaa}^{\alpha_2}) = \sum_{l \in S_2^2} v_l \frac{d_l}{d_l+k} + \sum_{l \in S_2^1 \backslash \{ l''\}} v_l + \Tilde{\sigma}_{l''}^{\alpha_2} v_{l''}$ for at most one location $l''$, where $\sum_{l \in L} \Tilde{\sigma}_l^{\alpha_1} \leq R$ and $\sum_{l \in L} \Tilde{\sigma}_l^{\alpha_2} \leq R$. We note here that without loss of generality, we assume that the locations $l'$ and $l''$ correspond to the sets $S_1^1$ and $S_2^1$, but note that the below arguments can be readily generalized to cover the other cases when $l' \in S_1^2$ or $l'' \in S_2^2$ as well.

To prove our desired result, we consider two cases, $\sum_{l \in L} \Tilde{\sigma}_l^{\alpha_1} < R$ and $\sum_{l \in L} \Tilde{\sigma}_l^{\alpha_1} = R$.

First, consider the case when the total resource spending $\sum_{l \in L} \Tilde{\sigma}_l^{\alpha_1} < R$. Note that in this case it must hold by the nature of step 1 of Algorithm~\ref{alg:GreedyContractDeterministic} that $\sum_{l \in L} \frac{d_l}{d_l+k} < R$. Consequently, $\Tilde{\sigma}_l^{\alpha_1} = \frac{d_l}{d_l+k}$ for all $l$. Analogously, by the nature of step 1 of Algorithm~\ref{alg:GreedyContractDeterministic}, it also holds that $\Tilde{\sigma}_l^{\alpha_2} = \frac{d_l}{d_l+k}$ for all locations $l$, i.e., $\Tilde{\sigmaa}^{\alpha_1} = \Tilde{\sigmaa}^{\alpha_2}$. Then, the total payoff under given the parameter $\alpha_1$ corresponding to the allocation computed using step 1 of Algorithm~\ref{alg:GreedyContractDeterministic} is given by $P_R(\Tilde{\sigmaa}^{\alpha_1}) = \sum_{l \in S_1^2} p_l \frac{d_l}{d_l+k} + \sum_{l \in S_1^1} p_l$. 

We now show that if $l \in S_1^1$, then $l \in S_2^1$. To see this, we first recall by the definition of $S_1^1$ that $\frac{\alpha_1 p_l (d_l+k)}{d_l} \geq k \Lambda_l + \alpha_1 p_l$ for all locations $l \in S_1^1$. Then, we have at $\alpha_2$ that the following relation holds for all $l \in S_1^1$:
\begin{align*}
    \frac{\alpha_2 p_l (d_l+k)}{d_l} &= \frac{\alpha_1 p_l (d_l+k)}{d_l} + \frac{(\alpha_2 - \alpha_1) p_l (d_l+k)}{d_l}, \\
    &\geq  k \Lambda_l + \alpha_1 p_l + (\alpha_2 - \alpha_1) p_l, \\
    &= k \Lambda_l + \alpha_2 p_l,
\end{align*}
where the inequality follows by the fact that $l \in S_1^1$ and thus $\frac{\alpha_1 p_l (d_l+k)}{d_l} \geq k \Lambda_l + \alpha_1 p_l$, and the fact that $\frac{d_l+k}{d_l} \geq 1$. Thus, the above relation implies that if $l \in S_1^1$, it holds that $\max \{ \frac{\alpha_2 p_l (d_l+k)}{d_l}, k \Lambda_l + \alpha_2 p_l \} = \frac{\alpha_2 p_l (d_l+k)}{d_l}$. Consequently, $l \in S_2^1$ as well, i.e., $S_1^1 \subseteq S_2^1$.

Finally, we have the following relation for the payoff given the parameter $\alpha_2$ corresponding to the allocation computed using step 1 of Algorithm~\ref{alg:GreedyContractDeterministic}:
\begin{align*}
    P_R(\Tilde{\sigmaa}^{\alpha_2}) &= \sum_{l \in S_2^2} p_l \frac{d_l}{d_l+k} + \sum_{l \in S_2^1} p_l , \\
    &\stackrel{(a)}{\geq} \sum_{l \in S_2^2} p_l \frac{d_l}{d_l+k} + \sum_{l \in S_1^1} p_l + \sum_{l \in S_2^1 \backslash \{ S_1^1\}} p_l, \\
    &\stackrel{(b)}{\geq} \sum_{l \in S_2^2} p_l \frac{d_l}{d_l+k} + \sum_{l \in S_1^1} p_l + \sum_{l \in S_2^1 \backslash \{ S_1^1\}} p_l \frac{d_l}{d_l+k}, \\
    &= \sum_{l \in S_1^2} p_l \frac{d_l}{d_l+k} + \sum_{l \in S_1^1} p_l, \\
    &= P_R(\Tilde{\sigmaa}^{\alpha_1}),
\end{align*}
where (a) follows from the fact that $S_1^1 \subseteq S_2^1$ and (b) follows as $\frac{d_l}{d_l+k} \leq 1$. Consequently, in the setting when $\sum_{l \in L} \Tilde{\sigma}_l^{\alpha_1} < R$, we have established our desired relation that $P_R(\Tilde{\sigmaa}^{\alpha_2}) \geq P_R(\Tilde{\sigmaa}^{\alpha_1})$.

%In the following we assume, without loss of generality, that both $\sum_{l \in L} \Tilde{\sigma}_l^{\alpha_1} = R$ and $\sum_{l \in L} \Tilde{\sigma}_l^{\alpha_2} = R$. Note that if either $\sum_{l \in L} \Tilde{\sigma}_l^{\alpha_1} < R$ and $\sum_{l \in L} \Tilde{\sigma}_l^{\alpha_2} < R$, then it must be by the nature of step 1 of Algorithm~\ref{alg:GreedyContractDeterministic} that $\sum_{l \in L} \frac{d_l}{d_l+k} < R$, and thus it follows that $\Tilde{\sigma}_l^{\alpha_1} = \frac{d_l}{d_l+k} = \Tilde{\sigma}_l^{\alpha_2}$ for all locations $l$. 

Thus, for the remainder of this proof, we consider the setting when $\sum_{l \in L} \Tilde{\sigma}_l^{\alpha_1} = R$ and $\sum_{l \in L} \Tilde{\sigma}_l^{\alpha_2} = R$. In the following, recall that, given any contract parameter $\alpha$, in step one of Algorithm~\ref{alg:GreedyContractDeterministic} resources are allocated to locations in descending order of their bang-per-buck ratios given by $\max \{ \frac{\alpha p_l (d_l+k)}{d_l}, k \Lambda_l + \alpha p_l \}$. We now analyse two cases separately. First, as in the setting when $\sum_{l \in L} \Tilde{\sigma}_l^{\alpha_1} < R$, we consider the location set $S_1^1$ and show that if $l \in S_1^1$ then $l\in S_2^1$ as well. Then, we consider the location set $S_1^2$. % and show that 

%, i.e., $\sum_{l \in L} \Tilde{\sigma}_l^{\alpha_1} = \sum_{l \in L} \Tilde{\sigma}_l^{\alpha_2}$.

%Next, we note that, given any contract parameter $\alpha$, in step one of Algorithm~\ref{alg:GreedyContractDeterministic} resources are allocated to locations in descending order of their bang-per-buck ratios given by $\max \{ \frac{\alpha v_l (d_l+k)}{d_l}, k \Lambda_l + \alpha v_l \}$. 

\textbf{Locations in set $S_1^1$:} We now show that all locations in the set $S_1^1$ also belong to the set $S_2^1$, where it holds that $\Tilde{\sigma}_l^{\alpha_2} = \Tilde{\sigma}_l^{\alpha_1}$. We show this using two intermediate results. First, we show that if $l \in S_1^1$, then it holds that $\max \{ \frac{\alpha_2 p_l (d_l+k)}{d_l}, k \Lambda_l + \alpha_2 p_l \} = \frac{\alpha_2 p_l (d_l+k)}{d_l}$. Next, we show that if $l \in S_1^1$, and there are locations with a lower bang-per-buck ratio than location $l$ under the parameter $\alpha_1$, then so too will the bang-per-buck ratios of all those locations remain lower than that of $l$ under the parameter $\alpha_2$.

%, where $\Tilde{\sigma}_l^{\alpha_2} = \Tilde{\sigma}_l^{\alpha_1}$ for all $l \in S_1^1 \backslash \{ l' \}$. 

To show the first result, we first recall by the definition of $S_1^1$ that $\frac{\alpha_1 p_l (d_l+k)}{d_l} \geq k \Lambda_l + \alpha_1 p_l$ for all locations $l \in S_1^1$. Then, we have at $\alpha_2$ that the following relation holds for all $l \in S_1^1$:
\begin{align*}
    \frac{\alpha_2 p_l (d_l+k)}{d_l} &= \frac{\alpha_1 p_l (d_l+k)}{d_l} + \frac{(\alpha_2 - \alpha_1) p_l (d_l+k)}{d_l}, \\
    &\geq  k \Lambda_l + \alpha_1 p_l + (\alpha_2 - \alpha_1) p_l, \\
    &= k \Lambda_l + \alpha_2 p_l,
\end{align*}
where the inequality follows by the fact that $l \in S_1^1$ and thus $\frac{\alpha_1 p_l (d_l+k)}{d_l} \geq k \Lambda_l + \alpha_1 p_l$, and the fact that $\frac{d_l+k}{d_l} \geq 1$. Thus, the above relation implies that if $l \in S_1^1$, it holds that $\max \{ \frac{\alpha_2 p_l (d_l+k)}{d_l}, k \Lambda_l + \alpha_2 p_l \} = \frac{\alpha_2 p_l (d_l+k)}{d_l}$, which establishes the first result.

To show the second result, consider a location $l \in S_1^1$ and another location $l_1$ with a lower bang-per-buck ratio under the parameter $\alpha_1$. Then, it holds that $\frac{\alpha_1 p_l (d_l+k)}{d_l} \geq \frac{\alpha_1 p_{l_1} (d_{l_1}+k)}{d_{l_1}}$ for all locations $l \in S_1^1$, which implies that $\frac{p_l (d_l+k)}{d_l} \geq \frac{p_{l_1} (d_{l_1}+k)}{d_{l_1}}$ and consequently that $\frac{\alpha_2 p_l (d_l+k)}{d_l} \geq \frac{\alpha_2 p_{l_1} (d_{l_1}+k)}{d_{l_1}}$. Furthermore, it holds that $\frac{\alpha_1 p_l (d_l+k)}{d_l} \geq k \Lambda_{l_1} + \alpha_1 p_{l_1}$ for all locations $l \in S_1$. Consequently, it holds that:
\begin{align*}
    \frac{\alpha_2 p_l (d_l+k)}{d_l} &= \frac{\alpha_1 p_l (d_l+k)}{d_l} + \frac{(\alpha_2 - \alpha_1) p_l (d_l+k)}{d_l}, \\
    &\stackrel{(a)}{\geq}  k \Lambda_{l_1} + \alpha_1 p_{l_1} + \frac{(\alpha_2 - \alpha_1) p_{l_1} (d_{l_1}+k)}{d_{l_1}}, \\
    &\stackrel{(b)}{=} k \Lambda_{l_1} + \alpha_1 p_{l_1} + (\alpha_2 - \alpha_1) p_{l_1}, \\
    &= k \Lambda_l + \alpha_2 p_{l_1},
\end{align*}
where (a) follows as $\frac{\alpha_1 p_l (d_l+k)}{d_l} \geq k \Lambda_{l_1} + \alpha_1 p_{l_1}$ and that $\frac{p_l (d_l+k)}{d_l} \geq \frac{p_{l_1} (d_{l_1}+k)}{d_{l_1}}$ and (b) follows as $\frac{d_l+k}{d_l} \geq 1$. The above relations imply that if $l \in S_1^1$, and a location $l_1$ has a lower bang-per-buck ratio at $\alpha_1$, then it also has a lower bang-per-buck ratio at $\alpha_2$. 

The above results have two key implications. First, the bang-per-buck ratios of all locations in set $S_1^1$ retain the same ordering after sorting under $\alpha_2$ as they did under $\alpha_1$. Moreover, since all locations $l_1 \notin S_1$ have a lower bang-per-buck ratio compared to any location in the set $S_1^1$ at $\alpha_1$, it also holds that all allocation $l_1 \notin S_1$ have a bang-per-buck ratio compared to any location in the set $S_1^1$ at $\alpha_2$. Consequently, given that step 1 of Algorithm~\ref{alg:GreedyContractDeterministic} will still allocate $R$ resources under $\alpha_2$, it holds that all locations $l \in S_1^1$ also satisfy $l \in S_2^1$, i.e., $S_1^1 \subseteq S_2^1$.

\textbf{Locations in Set $S_1^2$:} We note that unlike for locations in set $S_1^1$, there may be locations in set $S_1^2$ that may not be in the set $S_2^2$. In particular, there are three possibilities for locations in the set $S_1^2$ under $\alpha_1$ when the contract parameter is changed to $\alpha_2$: (i) a location $l \in S_1^2$ also belongs to $S_2^2$, (ii) a location $l \in S_1^2$ belongs to $S_2^1$, (iii) a location $l \in S_1^2$ no longer belongs to $S_2$ and instead is replaced by a location $l_1 \notin S_1$. We denote the location set in case (i) as $L_1$, the location set in case (ii) as $L_2$, and the locations that are no longer allocated in $S_1^2$ as $L_3'$ and the new locations in $L \backslash S_1$ that are allocated under $\alpha_2$ as $L_3''$.

Clearly, the payoff corresponding to the locations in case (i) remain unchanged, including the payoff accrued and the resource spending. Furthermore, the payoff corresponding to locations in case (ii) increases from $p_l \frac{d_l}{d_l+k}$ to $p_l$ for all locations $l \in L_2$ without a change in the resource spending. Thus, we finally consider case (iii). Case (iii) only happens under two possibilities for some location $l\in S_1^2$ and $l_1 \notin S_2$: (a) $\max\{k \Lambda_{l} + \alpha_2 p_{l}, \frac{\alpha_2 p_{l} (d_{l}+k)}{d_{l}} \} \leq k \Lambda_{l_1} + \alpha_2 p_{l_1}$ or (b) $\max\{k \Lambda_{l} + \alpha_2 p_{l}, \frac{\alpha_2 p_{l} (d_{l}+k)}{d_{l}} \} \leq \frac{\alpha_2 p_{l_1} (d_{l_1}+k)}{d_{l_1}}$.

In case (a), we first note that since $l \in S_1^2$ that $k \Lambda_{l} + \alpha_1 p_{l} \geq k \Lambda_{l_1} + \alpha_1 p_{l_1}$. Thus, the inequality in case (a) under $\alpha_2$ can only hold if it holds that $p_{l_1} \geq p_l$.

Next, we note for case (b) that $\frac{p_{l_1} (d_{l_1}+k)}{d_{l_1}} \geq \frac{p_{l} (d_{l}+k)}{d_{l}}$, i.e., the payoff bang-per-buck of location $l_1$ is higher than that of location $l$. By our assumption that higher payoff bang-per-buck ratios correspond to higher payoffs, we have that $p_{l_1} \geq p_l$ in case (b) as well.

\begin{comment}
Next, by our assumption that the valuations are correlated with the benefits, i.e., higher valuations $v_l$ imply higher benefits $d_l$, we obtain that
\begin{align*}
    \frac{v_{l_1} (d_{l_1}+k)}{d_{l_1}} \geq \frac{v_{l} (d_{l}+k)}{d_{l}} = v_l + \frac{v_lk}{d_l} \geq v_l + \frac{v_lk}{d_{l_1}} = \frac{v_l (d_{l_1}+k)}{d_{l_1}},
\end{align*}
where the final inequality follows as 
i.e., $v_{l_1} \geq v_l$.
\end{comment}

Thus, from the fact that $p_{l_1} \geq p_l$ in both cases (a) and (b), it follows that the total payoff corresponding to the locations in case (iii) can only increase, as all the resource spending on locations in set $L_3'$ is allocated to locations with an increased payoff in set $L_3''$.

Finally, using the fact that $S_1^1 \subseteq S_2^1$ and that the payoff is higher under the contract parameter $\alpha_2$ for all three cases corresponding to the $S_1^2$ analyzed above, we have that, it is straightforward to see that $P_R(\Tilde{\sigmaa}^{\alpha_2}) \geq P_R(\Tilde{\sigmaa}^{\alpha_1})$, which establishes our claim.

%as $\sum_{l \in L_3'} \sigma_l v_l \leq \sum_{l \in L_3''} \sigma_l v_l$, as the valuations in the set $L_3''$ are always at least as high compared to the valuations in the set $L_3'$.

%We now show for every location $l \in S_1^1$ that the total amount of resources allocated allocated satisfy $\Tilde{\sigma}_^{\alpha_2} \geq \Tilde{\sigma}_^{\alpha_1}$

%We will now use an analogous line of reasoning used to prove claim (i) to show that for each location to which resources are allocated

\begin{comment}
\begin{align*}
    \alpha_2 v_{l_{\alpha_2}} \geq (k \Lambda_{l_{\alpha_1}} + \alpha_2 v_{l_{\alpha_1}}) \frac{d_{l_{\alpha_1}}}{d_{l_{\alpha_1}}+k} \geq (k \Lambda_{l_{\alpha_1}} + \alpha_1 v_{l_{\alpha_1}}) \frac{d_{l_{\alpha_1}}}{d_{l_{\alpha_1}}+k},
\end{align*}
where the first inequality follows by the definition of $l_{\alpha_2}$ and the second inequality follows as $\alpha_2>\alpha_1$. 
\end{comment}

%Moreover, the administrator selects a location $l_{\alpha}$ for a given contract parameter $\alpha$ such that: $l_{\alpha} = \argmax_{l \in L} \{  \}$

\section{Additional Details on Numerical Experiments: Parking Enforcement}

\subsection{Costs of Different Parking Permit Types} \label{apdx:costs-parking-permit-types}

Table~\ref{tab:permit-cost}\footnote{The costs of the parking permits are determined based on the university's transportation and parking website: \href{https://transportation.stanford.edu/parking}{https://transportation.stanford.edu/parking}} presents the cost of purchasing the different parking permits at the university campus for the parking enforcement case study in Section~\ref{sec:numerical-experiments-parking-enforcement}.

\begin{table}[]
\caption{\small \sf Costs of different parking permit types per day at university campus for parking enforcement case-study.}
\centering
\small
\begin{tabular}[b]{c|c} \toprule
           Permit Type  & Permit Cost (\$ per day) \\ \midrule
A            & 6.65                        \\
C            & 1.23                         \\
Resident     & 1.50                         \\
Resident/C   & 1.38                       \\
Visitor      & 35.68                        \\
Other Permit & 22.40  \\ \bottomrule      
    \end{tabular} \label{tab:permit-cost}
\end{table}

\subsection{Additional Numerical Results: Parking Enforcement} \label{apdx:additional-numerics-parking}

Tables~\ref{tab:counterfactual1} and~\ref{tab:counterfactual2} compare the permit earnings achieved by Algorithm~\ref{alg:GreedyWelMaxProb} to that achieved under the status quo enforcement mechanism under counterfactual one as the proportion of strategic users is varied and under counterfactual two as the citation multipliers used to calibrate the respective exponential distribution parameters are varied, respectively.

\begin{table}[t] 
%\vspace{-15pt}
\centering
\caption{{\small \sf Comparison of the percentage of permit earnings achieved by Algorithm~\ref{alg:GreedyWelMaxProb} to that achieved under the status quo enforcement mechanism under counterfactual one as the proportion of strategic users is varied. }  }
\begin{tabular}[b]{ccc}\toprule
        Proportion of & Status Quo & Algorithm~\ref{alg:GreedyWelMaxProb} \\ Strategic Users &  (Permit Earning \%) &  (Permit Earning \%)  \\ \midrule
        20\% & 85.5 & 90.2 \\
        50\% & 63.8 & 75.6 \\
        100\% & 27.6 & 51.2 \\ \bottomrule
\end{tabular} \label{tab:counterfactual1}
\end{table}

\begin{table}[t] 
%\vspace{-15pt}
\centering
\caption{{\small \sf Comparison of the percentage of permit earning achieved by Algorithm~\ref{alg:GreedyWelMaxProb} to that achieved under the status quo enforcement mechanism under counterfactual two for different citation multipliers. }  }
\begin{tabular}[b]{ccc}
\toprule
        Citation & Status Quo & Algorithm~\ref{alg:GreedyWelMaxProb} \\ Multiplier &  (Permit Earning \%) &  (Permit Earning \%)  \\ \midrule
        1 & 96.9 & 98.7 \\
        2 & 90.3 & 95.6 \\
        5 & 83.6 & 92.3 \\ \bottomrule
\end{tabular} \label{tab:counterfactual2}
\end{table}

\section{Numerical Experiments: Queue Skipping in IPT Services} \label{apdx:additional-numerics}

In this section, we first present additional details of our experimental setup and model calibration procedure based on the application case of queue jumping in intermediate public transport services in Mumbai, India (Appendix~\ref{apdx:experimentalSetup}). Next, in Appendix~\ref{apdx:gap-welfare-revenue-outcomes}, we present results depicting the contrast in the outcomes corresponding to revenue and payoff maximization administrator objectives. Finally, in Appendix~\ref{apdx:additional-numerical-experiments}, we present a further discussion and analysis of the results presented in Figure~\ref{fig:social-welfare-proportion-vary-rk} in Section~\ref{sec:numerical-experiments-optimal-contracts}.

\subsection{Experimental Setup} \label{apdx:experimentalSetup}

We design a numerical experiment based on an application case of queue jumping in the context of intermediate public transport services in Mumbai, India. In particular, we consider a problem instance with $L = 448$ locations, representing the share-taxi and share-auto-rickshaw stands across the greater metropolitan region in Mumbai~\cite{mmrda-website}, where users can potentially engage in fraudulent activities, e.g., queue jumping. We ssume that each location has one type, i.e., $|\I| = 1$, where the total number of (potentially) fraudulent users that arrive at each location are exponentially distributed with rate $80$ for all locations $l$, i.e., $\Lambda_l \sim Exp(80)$ for all $l$, and the benefits $d_l$ at each location from engaging in fraud are exponentially distributed with rate $20$, i.e., $d_l \sim Exp(20)$. 

These numbers are calibrated based on observational data collected on the number of users that arrived at a share-auto-rickshaw location in an hour at the Aakhruli Mhada share auto-rickshaw stand in Kandivali West in Mumbai and their corresponding average waiting times. In particular, the total number of users that arrive in an hour is about 125 at the studied share-auto-rickshaw location, where about 80 were males, which we assume as the group of potentially fraudulent users that may engage in queue jumping (see Section~\ref{sec:examples-pertinent}). Consequently, we assume that the number of fraudulent users that arrive at each location are exponentially distributed with rate $80$ for all locations, which accounts for the variability in the number of fraudulent users arriving across these locations. 

Moreover, users waited between 4-5 minutes on average from their time of arrival to enter a share-auto (though, in general, the wait time for IPT services can often be on the order of hours~\cite{economist-website}); thus, we calibrated the mean of the benefits of users from engaging in queue jumping as the product of the reduction in the wait time (of 5 minutes, i.e., $\frac{1}{12}$ hours) and the average hourly wage of $Rs. 240$ in Mumbai, resulting in an average gain of $Rs. 20$ from engaging in queue jumping. We note that $Rs. 20$ can be quite substantial for daily wage workers in Mumbai. As with the number of arriving users, we assume that the benefits of engaging in fraud are exponentially distributed with rate $20$, which accounts for the variability in the quantity across the different locations. We also note here that despite an average wait time of about five minutes, the wait times faced by some users, particularly women, were often observed to be as high as 15 minutes, which is attributable to the fact that some users engage in queue jumping. More generally, in other share-taxi, share auto-rickshaw, or mini-bus locations across the world, the wait times are often variable and can be quite high for passengers.

We note here that obtaining more granular information for each of the 448 share-taxi and share-auto-rickshaw locations would result in more accurate results, but the above calibration process serves as a natural starting point to derive key insights and sensitivity relations in our studied security game.

%Here, we allow for the values of the number of arriving users and the corresponding gains from engaging in queue jumping to be exponentially distributed with the above computed means to account for the variability in these quantities across locations. We note that obtaining more granular information beyond the single share taxi location would result in more accurate results, but the above calibration process serves as a natural starting point to derive insights.

Moreover, for our experiments we vary the number of resources $R$ to lie in the range from one to thirty, with an increment of one, i.e., $R \in \{ 1, 2, \ldots, 30\}$ and the fine to lie in the range from $Rs. 50$ to $Rs. 500$, with increments of $Rs. 50$, i.e., $k \in \{ 50, 100, \ldots, 500 \}$. We note here that these fines are on the same order of magnitude as that for traditional traffic violations in Mumbai~\cite{traffic-fines-website}. Furthermore, we consider a payoff function given by $p_l = \Lambda_l (d_l)^x$ for $x$ lying in the range $\{ 1, 1.25, 1.5, 1.75, 2\}$ for all locations $l$. Our choice of the payoff function stems from the fact that $p_l = \Lambda_l d_l$ corresponds to the total benefits derived by fraudulent users at location $l$, which, in the context of IPT services, can serve as a proxy to capture the \emph{additional} value of the wait time faced by users that have not engaged in fraudulent behavior, as users engaging in fraud skip the queue and take their place in the vehicle. Moreover, setting $p_l = \Lambda_l (d_l)^x$ with an exponent $x>1$ serves as a proxy to capture the fact that the administrator may place a greater value in reducing the additional wait time of non-defaulting users, which is likely a more realistic scenario in practice. We note that the payoff $p_l$ at each location $l$ can be quite general and context dependent, and we use the above functional form of the payoff for the purposes of the experiments in Section~\ref{sec:numerical-experiments-optimal-contracts}.

%Moreover, we consider a setting where the welfare corresponding to reducing fraud from a given location is given by $v_l = 2 \Lambda_l d_l$, which corresponds to the total additional wait time faced by users that are not fraudulent multiplied by a factor of two, which captures the fact that an administrator may place a greater weight on users that wait an additional period of time in the line. We note that many other functional forms of the valuations $v_l$ can be chosen; however, we select the above defined functional form for the purposes of our experiments. Finally, for our experiments, we vary the number of security resources $R$ from one to thirty, with an increment of one, and the fine from $Rs. 50$ to $Rs. 500$, with increments of $Rs. 50$, where the fines are on the same order of magnitude as that for traditional traffic violations [CITE].

\subsection{Gap Between Revenue and Payoff Maximization Outcomes} \label{apdx:gap-welfare-revenue-outcomes}

Figure~\ref{fig:social-welfare-revenue-plot-vary-resources} depicts the variation in the payoff and revenue of the allocation corresponding to Algorithm~\ref{alg:GreedyFraudminimizationDeterministic} in the payoff maximization setting and that of Algorithm~\ref{alg:GreedyRevenueMaximization} in the revenue maximization setting with the number of resources for a fine of $k = 500$ and a payoff function $p_l = \Lambda_l (d_l)^{1.25}$. We note from Figure~\ref{fig:social-welfare-revenue-plot-vary-resources} that the payoff corresponding to Algorithm~\ref{alg:GreedyFraudminimizationDeterministic} in the payoff maximization setting monotonically increases in the number of resources as does the revenue corresponding to the revenue-maximizing solution computed using Algorithm~\ref{alg:GreedyRevenueMaximization}. However, the results demonstrate that the objectives of maximizing revenue and payoffs are at odds for the above defined problem instance. In particular, the solution computed using Algorithm~\ref{alg:GreedyRevenueMaximization} in the revenue maximization setting achieves only a small fraction of payoff of about $8 \%$ of that of the allocation corresponding to Algorithm~\ref{alg:GreedyFraudminimizationDeterministic} in the payoff maximization setting. Moreover, the outcome computed using Algorithm~\ref{alg:GreedyFraudminimizationDeterministic} results in almost no revenues due to the nature of the best-response of users to a payoff-maximizing administrator (see Equation~\eqref{eq:best-response-users}) and the structure of the optimal solution of the bi-level Program~\eqref{eq:admin-obj-fraud}-\eqref{eq:bi-level-con-fraud} (see Proposition~\ref{prop:opt-mixed-strategy-solution}). Furthermore, the payoff corresponding to the revenue maximizing outcome also monotonically increases with the number of resources as now more resources can be allocated, which not only results in more revenues but also improved payoffs. Finally, we note that the payoff corresponding to both Algorithms~\ref{alg:GreedyRevenueMaximization} and~\ref{alg:GreedyFraudminimizationDeterministic} and the revenue corresponding to the revenue-maximizing outcome appear concave in the total number of resources, indicating a diminishing marginal returns in these quantities as the number of resources is increased.

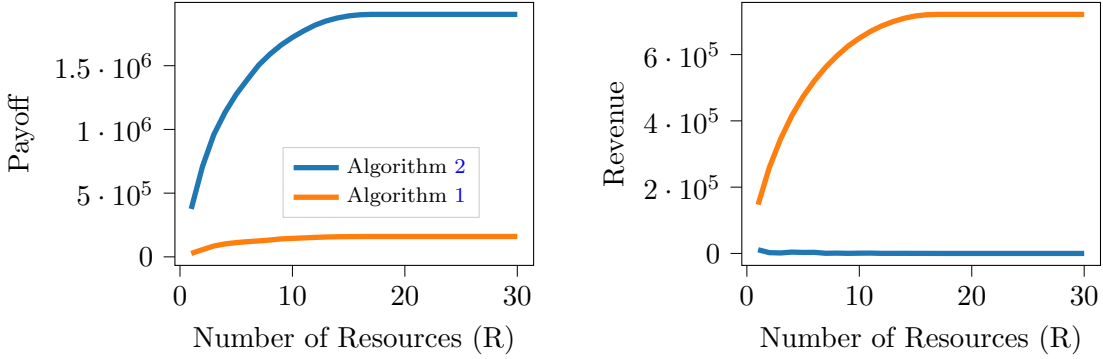
\begin{figure}
    \centering \hspace{-50pt}
\begin{subfigure} {0.4\columnwidth}
    % This file was created with tikzplotlib v0.10.1.
\begin{tikzpicture}

\definecolor{darkgray176}{RGB}{176,176,176}
\definecolor{darkorange25512714}{RGB}{255,127,14}
\definecolor{steelblue31119180}{RGB}{31,119,180}

\begin{axis}[
width=2.5in,
height=2in,
tick align=outside,
tick pos=left,
x grid style={darkgray176},
xlabel={Number of Resources (R)},
xmin=-0.45, xmax=31.45,
xtick style={color=black},
y grid style={darkgray176},
ylabel={Payoff},
ymin=-66934.9174975263, ymax=1996534.92054854,
ytick style={color=black},
ytick={-250000,0,500000,1000000,1500000,2000000},
legend style={
  at={(0.3,0.45)},
  anchor=north west,
  draw=white!80!black,
  font = {\scriptsize\arraycolsep=2pt}
},
]
\addplot [line width = 2pt, steelblue31119180]
table {%
1 375977.939599964
2 708838.63508663
3 957187.243065655
4 1132878.3526539
5 1275394.84769862
6 1392211.07773048
7 1504846.62221499
8 1590529.2821707
9 1662993.91937391
10 1722418.32794874
11 1773056.57520871
12 1817931.73059602
13 1851107.82174731
14 1874598.24880075
15 1890848.39393685
16 1899776.54558678
17 1902740.83700099
18 1902740.83700099
19 1902740.83700099
20 1902740.83700099
21 1902740.83700099
22 1902740.83700099
23 1902740.83700099
24 1902740.83700099
25 1902740.83700099
26 1902740.83700099
27 1902740.83700099
28 1902740.83700099
29 1902740.83700099
30 1902740.83700099
};
\addlegendentry{Algorithm~\ref{alg:GreedyFraudminimizationDeterministic}}
\addplot [line width = 2pt, darkorange25512714]
table {%
1 26859.166050022
2 56601.1540588824
3 84731.4117060983
4 101460.419908018
5 111819.825434853
6 119231.16870258
7 125051.184206762
8 132116.004787505
9 140959.586733422
10 144612.345588115
11 148775.70881221
12 152598.910307024
13 155774.70067006
14 157718.263451923
15 158728.02880528
16 159345.068215892
17 159847.180166746
18 159847.180166746
19 159847.180166746
20 159847.180166746
21 159847.180166746
22 159847.180166746
23 159847.180166746
24 159847.180166746
25 159847.180166746
26 159847.180166746
27 159847.180166746
28 159847.180166746
29 159847.180166746
30 159847.180166746
};
\addlegendentry{Algorithm~\ref{alg:GreedyRevenueMaximization}}
\end{axis}

\end{tikzpicture}
\end{subfigure} \hspace{30pt}
\begin{subfigure} {0.4\columnwidth}
    % This file was created with tikzplotlib v0.10.1.
\begin{tikzpicture}

\definecolor{darkgray176}{RGB}{176,176,176}
\definecolor{darkorange25512714}{RGB}{255,127,14}
\definecolor{steelblue31119180}{RGB}{31,119,180}

\begin{axis}[
width=2.5in,
height=2in,
tick align=outside,
tick pos=left,
x grid style={darkgray176},
xlabel={Number of Resources (R)},
xmin=-0.45, xmax=31.45,
xtick style={color=black},
y grid style={darkgray176},
ylabel={Revenue},
ymin=-36065.0724936645, ymax=757366.522366955,
ytick style={color=black}
]
\addplot [line width = 2pt, steelblue31119180]
table {%
1 11069.0833421055
2 2384.18048222887
3 1444.28387117154
4 4028.78650021813
5 3012.17120500678
6 3265.00969450909
7 472.844618587046
8 954.085598306744
9 251.622776714527
10 768.553885676909
11 1020.69159462428
12 196.438971300817
13 167.008808311993
14 323.211194201196
15 106.008358494903
16 149.484108623312
17 0
18 0
19 0
20 0
21 0
22 0
23 0
24 0
25 0
26 0
27 0
28 0
29 0
30 0
};
%\addlegendentry{Welfare Maximizing Outcome}
\addplot [line width = 2pt, darkorange25512714]
table {%
1 147400.888558835
2 259344.808220674
3 345418.540290178
4 415718.964030563
5 473123.51469561
6 520644.446538526
7 561525.901374627
8 595839.193765741
9 625582.212862355
10 649435.947257157
11 670044.658671391
12 686634.416164049
13 699494.642261788
14 709666.171751709
15 716628.715506126
16 720179.895436049
17 721301.449873291
18 721301.449873291
19 721301.449873291
20 721301.449873291
21 721301.449873291
22 721301.449873291
23 721301.449873291
24 721301.449873291
25 721301.449873291
26 721301.449873291
27 721301.449873291
28 721301.449873291
29 721301.449873291
30 721301.449873291
};
%\addlegendentry{Revenue Maximizing Outcome}
\end{axis}

\end{tikzpicture}
\end{subfigure}
\vspace{-20pt}
    \caption{\small \sf Depiction of the variation in the payoff and revenue corresponding to the solutions computed using Algorithms~\ref{alg:GreedyRevenueMaximization} and~\ref{alg:GreedyFraudminimizationDeterministic} as the number of resources $R$ is varied for a fine $k = Rs. 500$ and a payoff $p_l = \Lambda_l (d_l)^{1.25}$.}
    \label{fig:social-welfare-revenue-plot-vary-resources}
\end{figure}

Figure~\ref{fig:social-welfare-revenue-plot-vary-fine} depicts the variation in the payoff and revenue of the allocation corresponding to Algorithm~\ref{alg:GreedyFraudminimizationDeterministic} in the payoff maximization setting and that of Algorithm~\ref{alg:GreedyRevenueMaximization} in the revenue maximization setting with the fine for $R=15$ resources and a payoff function $p_l = \Lambda_l (d_l)^{1.25}$. Our results demonstrate that the payoff corresponding to Algorithm~\ref{alg:GreedyFraudminimizationDeterministic} and the revenue corresponding to Algorithm~\ref{alg:GreedyRevenueMaximization} monotonically increase in the fine. Such monotonicity relations are expected as a higher fine implies that the threshold fraction of $\frac{d_l}{d_l+k}$ resources to deter users from engaging in fraud at any given location $l$ is reduced; thus, more resources can be spent at other locations that could not have been targeted with a smaller fine. Furthermore, altering the fine has almost no impact on the revenue of the allocation computed using Algorithm~\ref{alg:GreedyFraudminimizationDeterministic} in the payoff maximization setting. Such a result holds as revenues are only accumulated at a single location where $\sigma_l < \frac{d_l}{d_l+k}$ under the payoff maximization objective due to the best-response of users to a payoff-maximizing administrator (see Equation~\eqref{eq:best-response-users}) and the structure of the optimal solution of the bi-level Program~\eqref{eq:admin-obj-fraud}-\eqref{eq:bi-level-con-fraud} (see Proposition~\ref{prop:opt-mixed-strategy-solution}). Thus, the revenues corresponding to the allocation computed using Algorithm~\ref{alg:GreedyFraudminimizationDeterministic} is negligible compared to the revenue obtained by Algorithm~\ref{alg:GreedyRevenueMaximization} in the revenue maximization setting. 

From Figure~\ref{fig:social-welfare-revenue-plot-vary-fine}, we also observe that the revenue maximizing solution computed using Algorithm~\ref{alg:GreedyRevenueMaximization} obtains a payoff that monotonically decreases from about $44 \%$ of the payoff corresponding to Algorithm~\ref{alg:GreedyFraudminimizationDeterministic} in the payoff maximization setting for a fine of $k = 50$ to $8.4 \%$ of the payoff achieved by Algorithm~\ref{alg:GreedyFraudminimizationDeterministic} for $k = 500$. To parse this result, we first note that the payoff of the revenue-maximizing outcome serves as a constant fraction of the optimal payoff (unlike the revenue of the allocation corresponding to Algorithm~\ref{alg:GreedyFraudminimizationDeterministic} that is negligible compared to the optimal revenue) as payoffs are accumulated at all locations to which resources are allocated. In particular, by the best-response Problem~\eqref{eq:best-response-users-rev-max} of users under the revenue-maximization objective, the total payoff corresponding to a revenue-maximizing solution is given by $\sum_l \sigma_l v_l$, where it holds that the resource constraint is satisfied, i.e., $\sum_l \sigma_l \leq R$, and the total resources spent at any location corresponding to a revenue maximizing solution satisfy $\sigma_l \in [0, \frac{d_l}{d_l+k}]$ for all locations $l$. 

Thus, as the fine $k$ increases, the total spending under the revenue-maximizing solution on the locations that have been allocated resources at lower fines decreases while the administrator can now spend the remaining resources on new locations that it did not spend on at lower fines. Given this observation, notice that if the payoff $p_l$ at each location are independent of the revenue, then the total payoff is likely to not change much in response to the fine for all allocation strategies satisfying $\sum_l \sigma_l = R$. However, the payoff function at each location satisfies $p_l = \Lambda_l d_l^{1.25}$, which is positively correlated with the revenue function which has the term $\Lambda_l$ in the objective. 

As a result, we obtain a monotonically decreasing relation between the fine and the payoff of the revenue maximizing solution, as fewer resources will be deployed in locations with higher values of $\Lambda_l$ (which is positively correlated with the payoffs $p_l$) with an increase in the fine while the remaining resources are spent on locations with lower values of $\Lambda_l$, which is correlated with lower payoffs $p_l$. Such a monotonically decreasing relation in the payoff corresponding to the revenue maximizing outcome suggests that simply increasing the fines may not be a solution to reducing fraud, particularly in the presence of a revenue-maximizing administrator, which highlights the importance of setting low to moderate fines, as is often the case in many practical applications, e.g., road traffic fines are typically not arbitrarily large.

\begin{figure}
    \centering \hspace{-50pt}
\begin{subfigure} {0.4\columnwidth}
    % This file was created with tikzplotlib v0.10.1.
\begin{tikzpicture}

\definecolor{darkgray176}{RGB}{176,176,176}
\definecolor{darkorange25512714}{RGB}{255,127,14}
\definecolor{steelblue31119180}{RGB}{31,119,180}

\begin{axis}[
width=2.5in,
height=2in,
tick align=outside,
tick pos=left,
x grid style={darkgray176},
xlabel={Fine (k)},
xmin=27.5, xmax=522.5,
xtick style={color=black},
y grid style={darkgray176},
ylabel={Payoff},
ymin=72702.4973250334, ymax=1989885.5198427,
ytick style={color=black},
ytick={-250000,0,500000,1000000,1500000,2000000},
legend style={
  at={(0.3,0.65)},
  anchor=north west,
  draw=white!80!black,
  font = {\scriptsize\arraycolsep=2pt}
},
]
\addplot [line width = 2pt, steelblue31119180]
table {%
50 905054.188686064
100 1135554.8083664
150 1322462.54395558
200 1477967.36807972
250 1597264.00957281
300 1691788.67660251
350 1766105.02791397
400 1824553.39759873
450 1866129.08989789
500 1890848.39393685
};
\addlegendentry{Algorithm~\ref{alg:GreedyFraudminimizationDeterministic}}
\addplot [line width = 2pt, darkorange25512714]
table {%
50 347524.404374511
100 348220.488561264
150 301688.106209424
200 265711.043435291
250 238960.617459615
300 221354.620022281
350 202683.886981639
400 186124.804077458
450 172288.953202978
500 158728.02880528
};
\addlegendentry{Algorithm~\ref{alg:GreedyRevenueMaximization}}
\end{axis}

\end{tikzpicture}
\end{subfigure} \hspace{30pt}
\begin{subfigure} {0.4\columnwidth}
    % This file was created with tikzplotlib v0.10.1.
\begin{tikzpicture}

\definecolor{darkgray176}{RGB}{176,176,176}
\definecolor{darkorange25512714}{RGB}{255,127,14}
\definecolor{steelblue31119180}{RGB}{31,119,180}

\begin{axis}[
width=2.5in,
height=2in,
tick align=outside,
tick pos=left,
x grid style={darkgray176},
xlabel={Fine (k)},
xmin=27.5, xmax=522.5,
xtick style={color=black},
y grid style={darkgray176},
ylabel={Revenue},
ymin=-36065.0724936645, ymax=757366.522366955,
ytick style={color=black}
]
\addplot [line width = 2pt, steelblue31119180]
table {%
50 148.03474057775
100 2690.62809013956
150 2439.03842481804
200 1223.70308347635
250 785.163127400909
300 284.104313599247
350 764.217224658808
400 601.438526354081
450 436.508415871435
500 106.008358494903
};
\addplot [line width = 2pt, darkorange25512714]
table {%
50 187855.64863638
100 317487.068305786
150 416147.709610845
200 494438.890588689
250 556691.799294328
300 606892.511498005
350 646482.566933554
400 677238.456333835
450 700341.124577786
500 716628.715506126
};
\end{axis}

\end{tikzpicture}
\end{subfigure}
\vspace{-20pt}
    \caption{\small \sf Depiction of the variation in the welfare and revenue corresponding to the solutions computed using Algorithms~\ref{alg:GreedyRevenueMaximization} and~\ref{alg:GreedyFraudminimizationDeterministic} as the fine $k$ is varied for $R = 15$ resources and a payoff $p_l = \Lambda_l (d_l)^{1.25}$.}
    \label{fig:social-welfare-revenue-plot-vary-fine}
\end{figure}
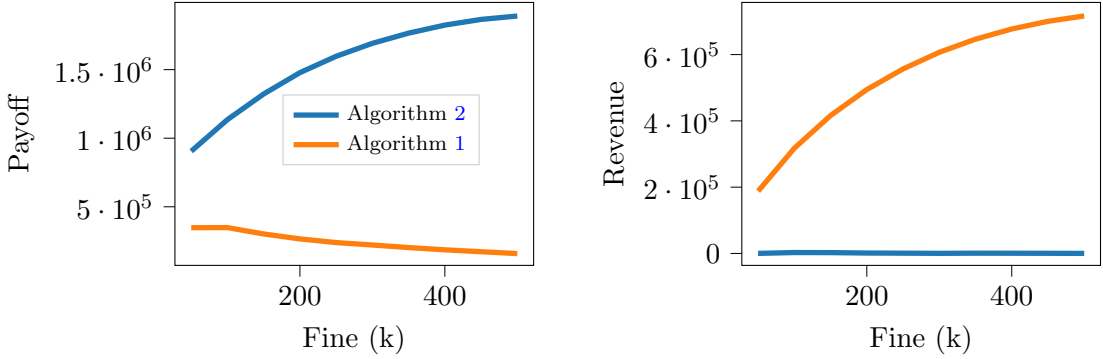

\subsection{Additional Analysis of Numerical Experiments} \label{apdx:additional-numerical-experiments}

In this section, we present some additional analysis of the results presented in the left of Figure~\ref{fig:social-welfare-proportion-vary-rk} in Section~\ref{sec:numerical-experiments-optimal-contracts}.

We note from Figure~\ref{fig:social-welfare-proportion-vary-rk} (left) that as we increase the exponent $x$ in the payoff function $p_l = \Lambda_l (d_l)^x$, the corresponding fraction of the payoff achieved by the administrator strategy computed using Algorithm~\ref{alg:GreedyContractDeterministic} to that achieved using Algorithm~\ref{alg:GreedyFraudminimizationDeterministic} in the payoff maximization setting increases for each contract level $\alpha$. Such a relation naturally follows as the payoff term in the administrator's Objective~\eqref{eq:obj-contract-alpha} in the contract game increasingly dominates the associated revenues from the collected fines at each location with an increase in the exponent of the payoff function. We note that in the setting when the exponent $x=1$, the administrator is only incentivized to change its resource allocation strategy from the one that maximizes revenues from the collected fines when the contract level $\alpha$ is one, as the payoff term in the administrator's objective is not large enough to dominate the revenues the administrator obtains through its collected fines unless the contract level $\alpha = 1$. On the other hand, when the exponent $x>1$, the payoff term in the administrator's objective begins to dominate the revenue obtained from the collected fines for lower levels of $\alpha$, which thus incentivizes the administrator to alter its resource allocation strategy from one that solely maximizes revenues from the collected fines to one that achieves a higher payoff for lower values of $\alpha$. Finally, we note that while $\Lambda_l d_l$ represents the total amount of benefit gained by defaulting users (if they are not found defaulting), the term $\Lambda_l (d_l)^x$ for an exponent $x>1$ serves as a reasonable proxy for the externality imposed on non-defaulting users (see Appendix~\ref{apdx:experimentalSetup}), which is likely to be valued higher, as non-defaulting users' costs may be strictly convex (rather than just linear) in the wait times they incur. Consequently, from the left of Figure~\ref{fig:social-welfare-proportion-vary-rk}, our results, for the studied payoff functions with an exponent $x>1$, demonstrate that using even small values of the contract $\alpha$ can recover most of the payoff achieved by Algorithm~\ref{alg:GreedyFraudminimizationDeterministic}, thereby bridging the gap between the payoff and revenue-maximizing outcomes.

\section{Sub-optimality of Ordering Locations without Affordability Thresholds in Algorithm~\ref{alg:GreedyFraudminimizationDeterministic}} \label{apdx:sub-optimal-less0.5}

In this section, we present an example which demonstrates the necessity of defining affordability thresholds to compute the associated bang-per-buck ratios of the locations in Algorithm~\ref{alg:GreedyFraudminimizationDeterministic} to achieve the desired half-approximation guarantee.

To this end, consider the following problem instance with three locations with a total resource budget of $R = 0.405$ and where all locations have a single type. Furthermore consider the following parameters for the three locations:
\begin{itemize}
    \item Payoffs: $p_1 = 1$, $p_2 = 1$, and $p_3 = 2.2$
    \item Threshold probabilities: $\frac{d_1}{d_1+k} = 0.2$, $\frac{d_2}{d_2+k} = 0.2$, and $\frac{d_3}{d_3+k} = 0.41$
\end{itemize}
In defining the above problem parameters, we drop the super-script $i$ in the notation as we are in the homogeneous user setting with a single type at each location.

We now consider an algorithm analogous to Algorithm~\ref{alg:GreedyFraudminimizationDeterministic}, which does not perform the additional pre-processing step of computing the ``affordable'' bang-per-buck ratios as in Algorithm~\ref{alg:GreedyFraudminimizationDeterministic} and instead just orders locations in the descending order of the ratios $\frac{p_l}{\frac{d_l}{d_l+k}}$. We show that this algorithm will not achieve the desired half approximation that Algorithm~\ref{alg:GreedyFraudminimizationDeterministic} achieves.

To see this, we first note that the optimal allocation corresponds to the enforcement strategy $\sigmaa^* = (0.2, 0.2, 0.05)$, which achieves a total payoff of $2.11$. Next, we note that the location with the highest bang-per-buck ratio $\frac{p_l}{\frac{d_l}{d_l+k}}$ is location three. Consequently, the payoff attained via step one of the algorithm is $0.405 \times 2.2 = 0.891$, as all the resources would be allocated to the first location. On the other hand, the payoff attained in step two of the algorithm would correspond to maximum payoff from spending all resources at a single location, which is one for this instance. Thus, it holds that the total payoff attained by the algorithm is $\max \{ 0.891, 1\} < \frac{1}{2} \times 2.11$, i.e., the approximation ratio of this algorithm is strictly less than $0.5$.

On the other hand, note that the payoff achieved using Algorithm~\ref{alg:GreedyFraudminimizationDeterministic} for this problem instance is $2.11$, i.e.,  Algorithm~\ref{alg:GreedyFraudminimizationDeterministic} is optimal, as it computes the affordability thresholds of the locations, which results in location three having a lower affordable bang-per-buck ratio compared to the first two locations. Thus, the above example demonstrates the importance of defining affordability thresholds to compute the associated bang-per-buck ratios of the locations in Algorithm~\ref{alg:GreedyFraudminimizationDeterministic} to achieve the desired half-approximation guarantee.

\section{Further Directions for Future Research} \label{apdx:model-extensions}

In this section, we present several natural extensions and generalizations of the model studied in this work, which opens directions for future work.

\paragraph{Fines Varying Across Locations:} In this work, we considered a setting where the fine remains fixed across all locations in the system. While a fixed fine across all nodes or locations is natural for many applications, there are often settings when the fines levied on defaulting users vary across locations, e.g., the fine for violating a traffic light is typically higher in a city center compared to a rural area. We note that we can model the variations in the fine across locations through a location-specific fine $k_l$ for each location $l$ and that the techniques and algorithms developed in this work naturally generalize and apply to this setting with slightly more cumbersome notation, as our algorithms and the corresponding proofs do not rely on the fines being fixed across locations.

\paragraph{Incorporating Costs of Deploying Security Resources:} We considered a model where the administrator has a limit on the resources it can allocate to monitor fraudulent or illegal activities. An alternate model that is also of interest is to study equilibrium outcomes in a setting where the administrator additionally incurs a cost for each security personnel it allocates, which influences the administrator's objective function. %In particular, we can model this additional cost to the administrator through a quantity $c_l$ for deploying a security resource to location $l$. 

\paragraph{Strategic Security Personnel:} While we consider a setting when the administrator maximizes its revenues from allocating $R$ security resources (e.g., police officers), another natural setting to consider is one where each of the $R$ security resources are individual decision makers seeking to maximize their own revenues. In this case, we note that rather than one bi-level program to characterize the revenue maximization problem of the administrator, the optimal strategy of the security resources can be characterized through a sequence of $R$ bi-level programs, where each bi-level program corresponds to the optimization problem of a single security resource seeking to maximize its individual revenue, given the actions of the other security resources. In particular, the revenue maximization problem (in the setting with homogeneous user types at each location) for a security resource $j$, given the strategies $(\sigmaa^{\Tilde{j}})_{\Tilde{j} \neq j}$ adopted by a all other security resources, can be formulated as
\begin{maxi!}|s|[2]   
    {\substack{\sigmaa^j \in \Omega_1 \\ y_l(\sigmaa) \in [0, 1], \forall l \in L}}                            
    { Q_R^j(\sigmaa^j, (\sigmaa^{\Tilde{j}})_{\Tilde{j} \neq j}) = \sum_{l \in L} \sigma_l^j y_l (\sigmaa) k \Lambda_l,  \label{eq:admin-obj-revenue222222}}   
    {\label{eq:Eg002222222222}}             
    {}          
    \addConstraint{y_l(\sigmaa)}{\in \argmax_{y \in [0, 1]} U_l(\sigmaa, y), \quad \text{for all } l \in L, \label{eq:bi-level-con-revenue222222222}}
    \addConstraint{\sigmaa}{= \sigmaa^j + \sum_{\Tilde{j} \neq j} \sigmaa^{\Tilde{j}}, \label{eq:bi-level-con-revenue2222233}}
\end{maxi!}
where $\sigmaa^j \in \Omega_1$ for all $j$. We note that in the upper level problem, the administrator deploys a strategy $\sigmaa^j$ to maximize its revenue, given the strategies $(\sigmaa^{\Tilde{j}})_{\Tilde{j} \neq j}$ adopted by all other security resources, to which users best-respond by maximizing their utilities in the lower-level problem. Moreover, note that the aggregate strategy $\sigmaa$, corresponding to the sum of the strategies of all security resources, satisfies Constraint~\eqref{eq:bi-level-con-revenue2222233}.

The above model of the strategic behavior of each security resources to maximize their individual revenues opens up new avenues in terms of studying the resulting equilibrium outcomes that emerge from such selfish behavior of individual security resources. 

\end{document}